\documentclass[reqno,11pt,titlepage]{amsart}

\usepackage[backend=biber,style=alphabetic]{biblatex}
\DeclareFieldFormat{labelalpha}{\thefield{entrykey}}

\DeclareFieldFormat{labelalpha}{%
  \ifnameundef{labelname}
    {\thefield{year}}
    {%
      \printnames[labelalpha][1-1]{labelname}
      \printfield{year}%
    }%
}
\DeclareNameFormat{labelalpha}{%
  \nameparts{#1}%
  \usebibmacro{name:family}{\namepartfamily}{}{}{}%
}

\usepackage{fancyhdr}

\usepackage{comment}
\usepackage{tikz}
\usetikzlibrary{arrows.meta}

\newtheorem{proposition}{Proposition}
\newtheorem{lemma}{Lemma}
\newtheorem{thm}{Theorem}
\newtheorem{prp}[thm]{Proposition}
\newtheorem{cor}[thm]{Corollary}
\newtheorem{hyp}{Hypothesis}
\newtheorem{definition}{Definition}
\newtheorem{example}{Example}

\usepackage{wrapfig} 
\usepackage{graphicx}
\usepackage{amssymb}
\usepackage{amsmath}
\usepackage{amsfonts}
\usepackage{amsthm}
\usepackage{color}
\usepackage[colorlinks=true, allcolors=blue]{hyperref}
\usepackage[mathscr]{eucal}
\usepackage{enumerate}
\usepackage{times}
\usepackage{algorithm}
\usepackage{algorithmic}
\usepackage[skip=10pt]{caption} 

\usepackage{xpatch}

\usepackage{lipsum} 

\makeatletter
\xpatchcmd{\@setaddresses}
  {\addvspace\bigskipamount}
  {\penalty\z@\vskip\bigskipamount}
  {}{}
\makeatother
\usepackage{booktabs}
\usepackage{changepage}
\usepackage{pifont}

\usepackage{longtable} 

\usepackage[a4paper,margin=3cm]{geometry} 

\allowdisplaybreaks[1]

\newtheorem{Def}{Definition}

\newtheorem{Remark}[Def]{Remark}
\newtheorem{Claim}[Def]{Claim}

\newcommand{\beq}{\begin{equation}}
\newcommand{\eeq}{\end{equation}}
\newcommand{\beqa}{\begin{eqnarray}}
\newcommand{\eeqa}{\end{eqnarray}}
\newcommand{\Proof}{\begin{proof}}
\newcommand{\QED}{\end{proof} \noindent}

\newtheorem{Conj}{Conjecture}

\newcommand{\Cay}{\textrm{Cay}}
\newcommand{\Sch}{\textrm{Sch}}
\newcommand{\Inv}{\mathrm{Inv}}
\newcommand{\Area}{\mathrm{Area}}

\usepackage{todonotes}

\title[CayleyPy-4: AI-Holography]{CayleyPy-4: AI-Holography. 
Towards analogs of holographic string dualities for AI tasks.\\
(Preliminary version) }

\author[A.~Chervov]{A. Chervov}
\address{Institut Imagine, Paris, France}
\email[A.~Chervov]{al.chervov@gmail.com}

\author[F.~Levkovich-Maslyuk]{F. Levkovich-Maslyuk}
\address{Centre for Mathematical Science, City St George’s, University of London}
\email[F.~Levkovich-Maslyuk]{fedor.levkovich-maslyuk@citystgeorges.ac.uk}

\author[A.~Smolensky]{A. Smolensky}
\address{Neapolis University Pafos, Cyprus}
\email[A.~Smolensky]{andrei.smolensky@gmail.com}

\author[F.~Khafizov]{F. Khafizov}
\address{University of Texas at Dallas}
\email[Farid~Khafizov]{farid.khafizov@utdallas.edu}

\author[I.~Kiselev]{I. Kiselev}
\address{Accenture}
\email[I.~Kiselev]{igor.kiselev@gmail.com}

\author[D.~Melnikov]{D. Melnikov}
\address{International Institute of Physics}
\email[Dmitry~Melnikov]{dmitry.melnikov@iip.ufrn.br}

\author[I.~Koltsov]{I. Koltsov}
\address{Independent Researcher}
\email[I.~Koltsov]{ivankolt@gmail.com}

\author[S.~Kudashev]{S. Kudashev}
\address{Independent Researcher}
\email[S.~Kudashev]{sergey0474@gmail.com}

\author[D.~Shiltsov]{D. Shiltsov}
\address{Independent Researcher}
\email[D.~Shiltsov]{da.shiltsov@gmail.com}

\author[M.~Obozov]{M. Obozov}
\address{Research Center of the Artificial Intelligence Institute, Innopolis University}
\email[M.~Obozov]{obozovmark9@gmail.com}

\author[S.~Krymskii]{S. Krymskii}
\address{Stanford University}
\email[Stanislav~Krymskii]{skrymskii@stanford.edu}

\author[V.~Kirova]{V. Kirova}
\address{NRNU MEPhI (National Research Nuclear University)}
\email[Valeriia~Kirova]{valeriia.kirova@mephi.ru}

\author[E.~V.~Konstantinova]{E.V. Konstantinova}
\address{Three Gorges Mathematical Research Center, China Three Gorges University, Sobolev Institute of Mathematics, Novosibirsk State University}
\email[E.~V.~Konstantinova]{ e\_konsta@ctgu.edu.cn, e\_konsta@math.nsc.ru}

\author[A.~Soibelman]{A. Soibelman}
\address{IHES}
\email[A.~Soibelman]{asoibelman@gmail.com}

\author[S.~Galkin]{S. Galkin}
\address{PUC-Rio, Departamento de Matem\'atica, Rua Marqu\^es de S\~ao Vicente 225, G\'avea, Rio de Janeiro, Brazil}
\email[S.~Galkin]{sergey@puc-rio.br}

\author[L.~Grunwald]{L. Grunwald}
\address{Sobolev Institute of Mathematics, The Mathematical Center in Akademgorodok}
\email[L.~Grunwald]{mathmanlily@gmail.com}

\author[A.~Kotov]{A. Kotov}
\address{University of Hradec Králové}
\email[Alexei~Kotov]{alexei.kotov@uhk.cz}

\author[A.~Alexandrov]{A. Alexandrov}
\address{IBS Center for Geometry and Physics}
\email[Alexander~Alexandrov]{alexander.alexandrov@ibs.re.kr}

\author[S.~Lytkin]{S. Lytkin}
\address{Kazakh-British Technical University}
\email[S.~Lytkin]{smlytkin@gmail.com}

\author[D.~Fedoriaka]{D. Fedoriaka}
\address{University of Washington}
\email[D.~Fedoriaka]{fedimser@cs.washington.edu}

\author[A.~Chevychelov]{A. Chevychelov}
\address{Independent Researcher}
\email[A.~Chevychelov]{heavy4evy@gmail.com}

\author[Z.~Kogan]{Z. Kogan}
\address{Independent Researcher}
\email[Z.~Kogan]{zahar1991@gmail.com}

\author[A.~Natyrova]{A. Natyrova}
\address{Independent Researcher}
\email[A.~Natyrova]{natyrovaaltana@gmail.com}

\author[L.~Cheldieva]{L. Cheldieva}
\address{Independent Researcher}
\email[L.~Cheldieva]{liuda.tarusina@gmail.com}

\author[O.~Nikitina]{O. Nikitina}
\address{Independent Researcher}
\email[O.~Nikitina]{ol.ya.nik.dev@gmail.com}

\author[S.~Fironov]{S. Fironov}
\address{Independent Researcher}
\email[Sergei~Fironov]{sergei.fironov@iai.spb.ru}

\author[A.~Vakhrushev]{A. Vakhrushev}
\address{Independent Researcher}
\email[Anton~Vakhrushev]{anton.vakhrushev.math@gmail.com}

\author[A.~Lukyanenko]{A. Lukyanenko}
\address{Independent Researcher}
\email[Andrey~Lukyanenko]{andrey.lukyanenko.math@gmail.com}

\author[V.~Ilin]{V. Ilin}
\address{University of Washington}
\email[Vasily~Ilin]{vasilyi@uw.edu}

\author[D.~Gorodkov]{D. Gorodkov}
\address{Independent Researcher}
\email[Denis~Gorodkov]{denis.gorodkov.math@gmail.com}

\author[N.~Bogachev]{N. Bogachev}
\address{University of Toronto}
\email[Nikolay~Bogachev]{n.bogachev@utoronto.ca}

\author[I.~Gaiur]{I. Gaiur}
\address{IHES (L'Institut des Hautes Études Scientifiques)}
\email[Ilia~Gaiur]{ilia.gaiur@ihes.fr}

\author[M.~Zaitsev]{M. Zaitsev}
\address{Higher School of Economics}
\email[Mikhail~Zaitsev]{mrzaytsev@edu.hse.ru}

\author[F.~Petrov]{F. Petrov}
\address{St. Petersburg State University}
\email[Fedor~Petrov]{fedyapetrov@gmail.com}

\author[L.~Petrov]{L. Petrov}
\address{University of Virginia, Charlottesville}
\email[Leonid~Petrov]{lenia.petrov@gmail.com}

\author[T.~Gaintseva]{T. Gaintseva}
\address{Queen Mary University of London}
\email[Tatiana~Gaintseva]{t.gaintseva@qmul.ac.uk}

\author[A.~Gavrilova]{A. Gavrilova}
\address{Independent Researcher}
\email[Alina~Gavrilova]{alinagavrilova2024@gmail.com}

\author[M.~N.~Smirnov]{M. N. Smirnov}
\address{Independent Researcher}
\email[Maxim~N.~Smirnov]{maxim.n.smirnov@gmail.com}

\author[N.~Kalinin]{N. Kalinin}
\address{Guangdong Technion-Israel Institute of Technology}
\email[Nikita~Kalinin]{nikita.kalinin@gtiit.edu.cn}

\author[A.~Khan]{A. Khan}
\address{Independent Researcher} 
\email[Anastasiia~Khan]{boykova.irk@yandex.ru} 

\author[K.~Jung]{K. Jung}
\address{Independent Researcher} 
\email[Kyuseok~Jung]{wjdrbtjr495@gmail.com}

\author[H.~Mousset]{H. Mousset}
\address{Centrale Lyon}
\email[Hugo~Mousset]{hugo.mousset@etu.ec-lyon.fr}

\author[H.~Isambert]{H. Isambert}
\address{Institut Curie, CNRS UMR168, Paris, France}
\email[H.~Isambert]{Herve.Isambert@curie.fr}

\author[O.~Debeaupuis]{O. Debeaupuis}
\address{Institut Curie, CNRS UMR168\\
Imagine Institute, INSERM UMR 1163, Paris, France}
\email[O.~Debeaupuis]{orianne.debeaupuis@curie.fr, orianne.debeaupuis@institutimagine.org, orianne.debeaupuis@gmail.com}

\begin{document} 

\begin{abstract}

This work is the fourth paper in the CayleyPy project, which aims to apply AI-based methods to large graphs. Here, we propose connections to a novel discretized analogue of holographic string dualities originating in theoretical physics. We argue that this perspective can lead to more efficient approaches to a wide range of AI tasks.

Many modern AI problems—such as those addressed by GPT-style language models or reinforcement learning systems—can be viewed as direct analogues of predicting particle trajectories on graphs. We hypothesize that these tasks admit a holographically dual description in terms of discrete strings, and that working in this dual representation can provide a more tractable formulation of the original problems. In particular, strings - holographic images of states are proposed as natural candidates for embeddings, motivated by the “complexity = volume/action” principle in AdS/CFT. 

In a simple illustrative example, the ROC curves serve as holographically dual strings to the nodes of some graphs. 
From a mathematical standpoint, we expect that all properties of graphs can be expressed entirely within this duality framework,  yielding nontrivial  identities. Furthermore, for Cayley graphs of the symmetric group $S_n$, we conjecture that the corresponding dual objects are planar polygons and present various examples. Graph diameters equal  the number of integer lattice points in the $n$-scaled polygon (Ehrhart quasi-polynomials). 
Vertices of graphs can be mapped (``holography") to lattice paths inside the polygon in such a way that word metrics (“gate complexities”) are equal to the areas under the corresponding paths
in accordance to “complexity = volume/action” principle. 
This thus provides an explanation for the quasi-polynomiality conjecture regarding diameters and word metrics from our previous paper.
Stanley-type formulas for counting shortest paths can be reinterpreted as identities relating particle extremals on Cayley graphs to string extremals on the associated polygons. In some cases, the graph Laplacian coincides with integrable spin-chain Hamiltonians and yields conformal field theories in the large-size limit.
We also study the corresponding 
H-polynomials and their properties, including positivity, unimodality, duality, and analogues of the Riemann conjecture.

\bigskip
\begin{center}\large
Project page:
\quad 
\url{https://github.com/CayleyPy/CayleyPy}
\quad
\includegraphics[width=.04\textwidth]{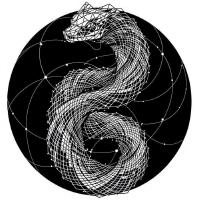}
\end{center}
\end{abstract}

\keywords{Machine learning, reinforcement learning, Cayley graphs}

\maketitle 

\setcounter{tocdepth}{2}
\small
\tableofcontents
\normalsize

\section{Introduction}

The deep learning revolution is one of the most exciting scientific breakthroughs of our time.
The \href{https://en.wikipedia.org/wiki/History_of_string_theory#1994%E2%80%932003:_Second_superstring_revolution}{second superstring revolution} \cite{Witten1995, SchwarzSecondSuperstringRevolution} was  one of these in the past.
The present paper takes a step toward merging these two.
The theme of duality is one of the central ideas in theoretical physics: the belief that for many complex systems there exists a “dual” description in which questions that are difficult in the original (“strong-coupling”) formulation become simpler in the dual (“weak-coupling”) one (e.g. \cite{PolyakovGaugeStrings}).
In particular, the  \href{https://en.wikipedia.org/wiki/AdS/CFT_correspondence}{AdS/CFT} 
holographic duality~\cite{Maldacena1998} marked the culmination of the second superstring revolution with the discovery of the long-sought duality between gauge theory and string theory, and has since become the most cited paper in high-energy theoretical physics.
Here, we present evidence that a similar principle may apply to artificial intelligence tasks, arguably leading to more efficient approaches than those currently available. 

We draw a parallel between dualities in physics and the concept of embedding (“latent representation”) in cognitive theory, both for natural and artificial neural networks. 
The role of embeddings in AI is analogous to the notion of duality in physics: they transform an input representation, in which semantic relations are difficult to analyze, into a representation where these relations become more tractable. 
A well-known example in ML is the relation 
"king - man + woman = queen" 
(\href{https://en.wikipedia.org/wiki/Word2vec}{Word2vec} \cite{MikolovSutskever2013}), where a nontrivial semantic relationship is mapped to a simple vector arithmetic operation. 
Finding appropriate embeddings is  a central problem in modern AI. 
We argue that the perspective of the AdS/CFT correspondence offers a new conceptual framework for approaching this problem.


\textbf{Brief outline of the main ideas.} \textbf{(1)} Setup. The starting point is that AI tasks such as language or RL modeling can be viewed as particle trajectory prediction tasks on edge-labeled graphs. \textbf{(2)} We hypothesize that in many cases such particle systems admit a holographically dual string description (in the spirit of AdS/CFT), converting difficult problems into more tractable ones. We outline how this can be used to build more effective AI systems: embeddings ("latent representations") are strings, training  is based on the "complexity = volume" principle in AdS/CFT. From the mathematical viewpoint, we expect that all properties of graphs can be seen from the dual side, as is typical in string dualities,  leading to new mathematical insights. \textbf{(3)} For the case of $S_n$ Cayley graphs, we present significant evidence that the dual objects are planar rational polygons, and that holography maps nodes of the graphs to lattice paths (discrete strings) inside the polygons. We also present connections to various deep questions and conjectures in mathematics. \textbf{(4)} As simplest examples, ROC curves (Dyck paths) can be used to illustrate the ideas in a clear and straightforward manner.


Surprisingly, the  simplest example -- which is just the Cayley (Schreier) graph generated by neighboring transpositions -- is connected to several advanced areas of research: quantum integrable systems and the Bethe ansatz (since the graph Laplacian coincides with the Hamiltonian of the Heisenberg spin chain), Mahler measures  (related to supersymmetric Landau–Ginzburg models and mirror symmetry for Fano varieties ) for counting spanning trees, 
Burgers and Kardar–Parisi–Zhang (\href{https://en.wikipedia.org/wiki/Kardar%E2%80%93Parisi%E2%80%93Zhang_equation}{KPZ}) equations for description of geodesic flow in a dual formulation for large size, limit shapes of Young diagrams,  
Beilinson's conjectures on special values of $L$-functions, and more.
Moreover, we propose that conformal field theory arising in large size limit for such graphs is dual (in  AdS/CFT manner) to a simple free scalar with radius of compactification related to normalized area of the dual polygon (in this case it plays the role of the dual "string") and propose conjectures
on spectrum of conformal dimensions.

\bigskip

Let us present a more detailed outline.

\begin{enumerate}
    \item \textbf{AI as particle on edge-labeled graphs}.
    AI tasks involving  languages or reinforcement learning can be naturally interpreted 
    as problems of predicting particle dynamics on edge-labeled graphs 
    (for instance, Cayley graphs). 
    The data --- texts or sequences of actions --- correspond to particle trajectories, 
    that is, sequences of traversed edges. Such edge ("token") sequences naturally define 
    the corresponding “texts.”
    Typical tasks then amount to determining the trajectory  given initial (the "prompt" for generative models)
    or boundary conditions (the "masked language modeling objective" for e.g. BERT or "win the game/reach the goal" in RL).

    \item \textbf{Particle-string holographic duality.}
    We hypothesize that many such tasks admit a holographically dual equivalent description 
    in terms of discrete string theory, transforming a difficult problem on the 
    original graph-side ("CFT-side") into a more tractable one on the dual string-side ("AdS-side"), 
    in the spirit of strong--weak coupling duality, thus giving a key to build more efficient AI-systems.
    This perspective is inspired by the landmark AdS/CFT correspondence.
    We further argue that strings (viewed as holographic duals of states) 
    provide natural candidates for embeddings (latent representations), 
    and that unconventional (“tropical”) string actions may be required 
    to capture phenomena intrinsic to discrete settings.
    We remark that “embeddings,” whether in natural or artificial neural systems, 
    play a role analogous to “duality” in physics: 
    they transform input data into representations that are more tractable 
    than the original formulation. The celebrated AdS/CFT principle "complexity = volume/action" can be used as a training objective for string embeddings, and also plays many other key roles.
    From a mathematical point of view, we expect that all properties 
    of the graph admit a complete description in terms of the dual theory, as it is typical in string dualities.
    This dual perspective may provide new results and structural insights, 
    some of which are presented here. We suggest various connections with integrable systems, matrix models, 
    conformal field theory, cluster algebras, the thermodynamic Bethe ansatz, and related structures.

    \item \textbf{$S_n$ Cayley graphs / planar polygon duality.}
    In the case of Cayley graphs of the permutation group $S_n$, 
    we propose that the holographically dual objects are planar rational polygons.
    The holographic correspondence maps graph nodes to lattice paths 
    (discrete strings).
    Diameters and word metrics can then be described in terms of 
    Ehrhart quasi-polynomials of the associated polygons 
    or of their subregions (e.g.\ areas under lattice paths), 
    thus explaining the quasi-polynomiality conjecture from our previous paper.
    This is consistent with the “complexity = volume/action” principle 
    from AdS/CFT: 
    complexities coincide with word metrics in Cayley graphs, 
    while counting lattice points corresponds to discretized volumes.
    Using the CayleyPy AI-based library and methodology, we obtained various results 
    and conjectures concerning the corresponding quasi-polynomials 
    and related $H$-polynomials. 
    That is, we compute “complexities” and demonstrate their agreement with the corresponding “volumes.”

    \item \textbf{Simplest examples. ROC-curves (Dyck like paths) as strings.}
    As a starting point, we present simple and explicit examples 
    in which classical ROC curves can be interpreted as discrete strings 
    holographically dual to nodes of an appropriate graph.
    Several combinatorial facts then admit a natural interpretation 
    in terms of discrete string duality.
     Surprisingly, even this simple example is connected to deep recent 
    mathematical results, and in fact touches upon several open problems.

\end{enumerate}

\subsection{ Main hypothesis: particle-string holographic duality for AI-tasks}
The \href{https://en.wikipedia.org/wiki/AdS/CFT_correspondence}{AdS/CFT} 
holographic duality~\cite{Maldacena1998} 
predicts that difficult questions on the CFT side ("strong coupling") can be computed via
more tractable geometric methods on the AdS side (where string theory at "weak coupling" reduces to gravity). 
We hypothesize the existence,
outline expectations and elaborate several  examples 
of a similar holographic duality in the context of graphs, languages, and reinforcement-learning “environments” — systems of interest in artificial intelligence, mathematics, and physics.
First, we emphasize that many such settings can be viewed as particle trajectories on edge-labeled graphs (e.g. 
\href{https://en.wikipedia.org/wiki/Cayley_graph}
{Cayley graphs}), where the trajectories define a “language,” i.e., sequences of admissible tokens/moves.
Second, we hypothesize that a “particle” moving on such graphs may admit a holographically dual description as a (discrete) “string” living on an appropriate dual object. If such a duality exists, the string description may be more tractable (analogous to a weak-coupling regime) and could provide a key to constructing more powerful AI models.
In particular we expect that strings understood as holographic images of states may provide "good embeddings".
Our examples suggest that one should consider unusual (“tropical”) analogues of discrete string actions that involve functions of the form $\max(\cdot,0)$ (also known as  \href{https://en.wikipedia.org/wiki/Rectified_linear_unit}{ReLU}
). Such terms ensure that, even after fixing initial or boundary conditions, the equations of motion retain substantial local freedom in their solutions — a feature that appears necessary for describing discrete setups.
From a mathematical perspective, we expect that the properties of a graph (or the associated language, etc.) may be computable from the dual object, potentially leading to non-trivial identities and new structural insights. In particular, various combinatorial
results might be reinterpreted  as manifestations of string dualities.


If indeed true as stated, our proposal would imply that string-theoretic holographic dual descriptions, in the sense outlined above, may exist for a wide range of AI systems. Examples of this could potentially include even such settings as English and other natural languages; programming languages; robotic manipulator systems; mathematics viewed as a formal proof system; games such as Go, chess, etc.; the languages of life (admissible protein and DNA sequences); chemical molecules encoded in databases; and so on. One may even argue that one reason a duality or simplification of this kind might exist for e.g. English language is that we indeed can learn it! Both humans and machines master it, which would be impossible without
ability to convert information to some latent representation where complex semantic relations can be seen via simple operations
(which is the basically the duality: complex to simple). Our current proposal is more subtle: we expect that duality may be holographic in the spirit of AdS/CFT and that these efficient latent representations ("embeddings") could be interpreted as strings -- that is, as holographic images of states (sentences in this case). Downstream computations performed on these embeddings may naturally involve Riemannian metrics --reminiscent of the AdS metric -- while more subtle phenomena could correspond to finite-size corrections. More broadly, this perspective may even lead to a deeper geometric understanding of key concepts in AI.


\subsection{ Guiding principle: “complexity = volume/action” }
The celebrated AdS/CFT principle “complexity = volume/action” (L. Susskind et al: \cite{StanfordSusskind2014, BrownSusskind2016}), together with subsequent works  (including \cite{Lin2019} and work by one of the present authors, D. Melnikov,  \cite{CamiloMelnikovNovaesPrudenziati2019}) that bridge this idea with \href{https://en.wikipedia.org/wiki/Cayley_graph}{Cayley
 graphs}, serves as one of the key insights and guiding frameworks for what follows.
 In the graph-theoretic setting, complexity admits a very simple interpretation. The complexity of one node with respect to another is defined as the length of the shortest path between them in the graph.
 Then the "complexity=volume/action" (abbreviated to "C=V/A" below) principle becomes a particular manifestation of particle-string duality. Indeed, the extremal action of a particle -- namely, the length of the shortest path -- can be directly viewed as "complexity". According to the duality principle, it should be equivalent to the extremal action of the corresponding string worldsheet, which is naturally related to a volume-type quantity.

Let us list a few other reasons this principle is important for us.  
One key point is that establishing the duality itself is not expected to be a trivial task, even for relatively simple graphs. In contrast, results supporting the "complexity = volume/action" principle
appear to be more accessible, and in some cases even numerical simulations can provide valuable insight.
In particular, the primary goal of the CayleyPy project is to develop AI-based tools to estimate
``complexities'' (i.e. lengths of shortest paths) using modern machine learning methods.
Previous papers in the project have already achieved state-of-the-art results in this direction.
Moreover, one of the key conjectures formulated in our earlier work is that the ``complexities''
(diameters and word metrics) for $S_n$-Cayley graphs are quasi-polynomials in $n$.
In the present paper, we propose an extension of this conjecture: we observe that these
quasi-polynomials appear to be closely related to
\href{https://en.wikipedia.org/wiki/Ehrhart_polynomial}{Ehrhart quasi-polynomials}
of certain planar polygons.
This conjecture can be viewed as a manifestation and refinement of the ``C=V/A'' principle
for $S_n$-Cayley graphs, since it identifies complexities with the number of lattice points
in corresponding polygons. These lattice point counts may be interpreted as discrete
``volumes,'' or ``volumes''  with  finite-size corrections.

Another reason is that the principle “C = V/A” can itself be viewed as a manifestation of a strong–weak coupling duality. Computing complexity for large systems (such as graphs) is typically extremely difficult—often NP-hard—whereas the computation of geometric quantities like volumes is comparatively more tractable. In this way, a complicated problem is translated into a simpler one through a dual description. 
This perspective also underlies our proposal to interpret strings (i.e., holographic images of states) as good embeddings (latent representations). The primary goal of embeddings—whether in natural neural networks or artificial systems—is to transform an original representation of information into a more tractable format, so that questions which are difficult in the original representation become easier in the embedding space.
From this viewpoint, the principle “C = V/A” can be seen as predicting precisely such a mechanism: by representing states through strings, one effectively moves to a dual geometric description where complexity-like quantities become accessible through simpler geometric computations. We provide a more detailed discussion of this perspective in the main text.

 In the present work, we use the term ``area'' rather than ``volume'', since our examples are planar. 
We also typically use other terms instead of “complexity,” which are more standard in graph theory. 
 Often, a distinguished reference node is chosen—such as the identity element in the case of a Cayley graph—and the complexity of all other nodes is measured relative to this reference. In group-theoretic language, this notion is known as the \href{https://en.wikipedia.org/wiki/Word_metric}{word metric}. In coding theory and bioinformatics, it is related to metric codes and evolution metrics and mutations, respectively (see~\cite{konstantinova2008some} for more details). In computer science, it is closely related to \href{https://en.wikipedia.org/wiki/Circuit_complexity}{circuit  complexity}, while in quantum computing it appears as quantum gate complexity, quantifying the minimal number of elementary gates required to generate a given element. Complexity optimization is one of the central challenges in quantum computing \cite{nam2018automated,ruiz2025quantum}, while the computation of word metrics is a fundamental problem in computational group theory. Previous papers from the CayleyPy project \cite{CayleyPy1Cube,CayleyPy2RL,Cayley3Growth} have provided state-of-the-art solutions to this problem.

\subsection{\texorpdfstring{$S_n$-Cayley graphs to planar polygon duality: lattice paths as discrete strings holographically dual to graph nodes}{Sn-Cayley graphs to planar polygon duality: lattice paths as discrete strings holographically dual to graph nodes}
}
The present paper elaborate the  conjectures 
for $S_n$-Cayley and Schreier graphs. We propose that the corresponding dual objects are rational planar polygons (or line segments in degenerate cases).  
The holographic map associates each node of the graph with a lattice path on the dual polygon.  
Our central conjecture is that the diameter of the graph equals the number of integer lattice points in the polygon, while more generally, the word metric of any node corresponds to the area under its associated lattice path.  
Consequently, both diameters and word metrics can be described as 
\href{https://en.wikipedia.org/wiki/Ehrhart_polynomial}{Ehrhart quasi-polynomials} of the $n$-rescaled polygon, providing a natural explanation for the quasi-polynomiality conjecture in our previous work \cite{Cayley3Growth}.  
These conjectures can be viewed as refinements of the ``complexity = area'' principle, since word metrics—interpreted as computational complexities—are predicted to correspond directly to areas under the dual lattice paths.

The planarity of those polygons is conditioned to the celebrated 50+ years old  open problem \cite{Rubtsov1975}
(reviews \cite{glukhovzubov1999lengths}, \cite{Helfgott2013GrowthIdeas})
that diameters of $S_n$ Cayley graphs are bounded by $n^2$. 
The conjecture can be a seen a special case of the L.Babai conjecture on diameter of any finite
simple group, for which certain progress has been achieved by T.Tao et.al. 
Still both conjectures are widely open,  even establishing a polynomial bound is not achieved.
For further discussion and refinements, we refer to \cite{Cayley3Growth} and discussions below.

Each quasi-polynomial has a naturally associated $H$-polynomial, that is simply the numerator of the generating function, 
which coincides with the Poincar\'e polynomial of the corresponding
toric variety under appropriate smoothness assumptions.
Important questions in combinatorics are whether these $H$-polynomials
satisfy the same properties as those arising from  "good" toric varieties,
namely: positivity, duality (Poincar\'e symmetry), unimodality,
and an analogue of the ``Riemann conjecture'' (i.e., that all roots
lie on the unit circle).
In the present paper, we compute various $H$-polynomials associated with
diameters and word metrics, and investigate their properties, formulating
a number of conjectures.
In some cases, all of the expected properties hold; however, in other
situations we observe that each of them may be violated.

The cases of the Cayley graphs with neighbor transposition (Coxeter) generators of $S_n$ 
are the most classical ones. We present quite an detailed picture of the dualities here.
The holography map can be described 
by e.g. \href{https://en.wikipedia.org/wiki/Lehmer_code}{Lehmer codes}
or by related constructions e.g. ROC-curves. 
"Complexity = Area" reduces to known statements that areas under the ROC-curves correspond to Mann-Whitney
statistics, and similar statements for the Lehmer codes.
Bijective description of  shortest paths on these graphs \cite{Stanley1984,EdelmanGreene1987} can be interpreted as bijections between the extremals of a particle and the corresponding strings, in accordance with the conjectured duality. In this framework, extremal string worldsheets are naturally associated with Young tableaux. 
String action is "tropical" (or "ReLU") analogue of the conventional string action. 
Graph Laplacians can be identified with the Hamiltonians of Heisenberg XXX spin chains. Etc.
The other cases studied in the present paper are generalizations of
neighboring transpositions to $k$-neighbor versions, such as
$k$-consecutive cycles of the form $(i,i+1,i+2,\dots,i+k-1)$ and their variations.
We present evidence that, in this setting, the results can be described
by a $(k-1)$-shrinkage principle: namely, the quantities (in particular dual polyygons) in the standard
case should be rescaled by $(k-1)$ to obtain the corresponding results
in the general case, at least at the level of leading terms.
We  present various conjectures on diameters, word metrics, as well as the
corresponding quasi-polynomials and $H$-polynomials in these generalized
settings, obtained with our CayleyPy library.

In general, determining the dualities and the corresponding dual polygons
is a highly nontrivial task. Already the computation of diameters and
word metrics is well known to be difficult in its own right.
Our approach is based on the CayleyPy library, in particular on its
AI component. The heuristic strategy for identifying dual polygons
is as follows. First, we attempt to determine quasi-polynomials for
the diameters. Second, we try to identify polygons whose Ehrhart
quasi-polynomials match the observed data.
To determine diameters, we proceed in several steps. Using efficient
implementations within CayleyPy, we perform brute-force computations
of diameters for the first several values of $n$, and we  identify
the corresponding longest elements (states). Next, we attempt to
detect patterns in these longest states and formulate conjectural
descriptions valid for all $n$. If this step is successful, we then
use the AI component of CayleyPy to compute word metrics of these
candidate longest states, which allows us to reach significantly
larger values of $n$ than are accessible by brute force alone.
With sufficiently extensive data at hand, we fit quasi-polynomials
for the diameters and compare them with Ehrhart quasi-polynomials
of suitable polygons. We demonstrate that this approach is successful
in a number of cases.


\subsection{ List of  contributions }
Let us outline contributions of the present paper in the itemized form, which mainly follow the order of the exposition:

\begin{itemize}
    \item {\bf Simplest examples. ROC curve (Dyck like paths).}    
    Surprisingly, the basic quality metrics in machine learning such as the \href{https://en.wikipedia.org/wiki/Receiver_operating_characteristic}{ROC curve} and \href{https://scikit-learn.org/stable/modules/generated/sklearn.metrics.roc_auc_score.html}{ROC-AUC score}
    can be used to illustrate our holographic duality for graphs in a very simple and explicit way. In this framework, ROC curves (which are quite similar to \href{https://en.wikipedia.org/wiki/Catalan_number#Applications_in_combinatorics}{Dyck paths}) are strings “holographically dual” to nodes of a certain class of graphs. ($S_n/(S_k \times S_{n-k})$ or "Grassmanians" over "field with one element.")
    Within this setting, the analogue of the AdS/CFT relation “complexity = area” reduces to the familiar equality between the Mann–Whitney statistic and the area under the ROC curve.

    Bijective description of  shortest paths on these graphs \cite{Stanley1984,EdelmanGreene1987} can be interpreted as bijections between the extremals of a particle and the corresponding strings, in accordance with the conjectured duality. In this framework, extremal string worldsheets are naturally associated with Young tableaux. 
    String action is "tropical" (or "ReLU") analogue of the conventional string action. 

    Graph Laplacians can be identified with the Hamiltonians of Heisenberg XXX spin chains. The number of spanning trees is related to the Mahler measure, in accordance with general expectations. We also present an analysis of the limit shapes of ROC curves and outline several open questions.
        
        

    \item {\bf Charting general holographic duality for AI.} 
    We hypothesize that for many graphs and corresponding AI-systems (languages, and RL-environments) there are  dual objects whose features include:
    \begin{itemize}
        \item {\bf \href{https://en.wikipedia.org/wiki/Holographic_principle}{Holography}.} \cite{tHooft1993, Susskind1995Hologram}  Meaning that $d$-dimensional objects on graph side are mapped to $(d+1)$-dimensional ones on the dual side. I.e.  graph nodes to paths ("strings") on the dual side, paths on graph to 2-dimensional objects - like Young tableaux ("string worldsheets"). We expect that the strings which are holographic images of states and "good embeddings".


        \item  {\bf Particle on the graph side (``CFT"-side) = Discrete String on dual side (``AdS"-side).} 
        As is typical in string dualities, we expect that a theory defined on one side is equivalent to a theory defined on the other side, and our proposal consists of duality between particle theory and string theory.



        
        \item  {\bf Corollary: Complexity = Area/Action.} 
        This is a consequence of the previous principle: values of action for particles on extremals are 
        lengths of the shortest path, while for the string these are related to certain areas.

        \item  {\bf Strong coupling to weak coupling.} 
            As in conventional string theory, we expect that difficult problems on the original side are converted into more tractable problems on the dual side.
            In particular we expect that strings dual to nodes of the original graph provide "good" embeddings (from AI point of view).

    \end{itemize}
    
    \item {\bf Case studies and mathematical conjectures.} 
        We outline several concrete examples of duality for 
        Cayley graphs. 
        We explain its relation to quasi-polynomiality hypothesis
        and the $n^2$ conjectural bound for diameters of $S_n$ Cayley graphs,
        which is a celebrated open problem in mathematics for more than 50 years.
        We explain how CayleyPy AI methodology  assists in determining the duality.
        We discuss other examples related to $\mathrm{SL}(2,\mathbb{Z})$ and Farey graphs, etc.

    \begin{itemize}
        \item \textbf{$S_n$-Cayley to planar polygon duality.} 
        For $S_n$-\href{https://en.wikipedia.org/wiki/Cayley_graph}{Cayley} and 
        \href{https://en.wikipedia.org/wiki/Schreier_coset_graph}{Schreier} graphs, we conjecture that the dual objects are rational polygons in the plane (or, in degenerate cases, line segments) such that the graph diameter equals the number of integer lattice points in the $n$-scaled polygon.  
        This provides an illustration of a refined version of the “complexity = area” principle in this context.  
        In other words, diameters are given by the \href{https://en.wikipedia.org/wiki/Ehrhart_polynomial}{Ehrhart quasi-polynomials} of the polygons, explaining the conjecture from our previous work. 
        Such a relation is far from trivial: in particular, it implies the celebrated open problem predicting that the diameters of these graphs are bounded by $n^2$, a conjecture that has resisted the efforts of leading mathematicians for decades \cite{Rubtsov1975} (reviews \cite{glukhovzubov1999lengths}, \cite{Helfgott2013GrowthIdeas}).
        
        \item \textbf{Holography: nodes to lattice paths; word-metrics quasi-polynomiality.}  
        In this framework, "holography" predicts that vertices of the graph can be mapped to lattice paths inside the polygon in such a way that word metrics (or “gate complexities”) coincide with the areas under the corresponding paths.  
        This again mirrors the “complexity = area” principle familiar from AdS/CFT, and implies that word metrics are also described by Ehrhart quasi-polynomials, explaining the conjecture from \cite{Cayley3Growth} that word metrics are quasi-polynomials in $n$ for $S_n$.

        \item \textbf{Holography = Lehmer code (for Coxeter generators graphs).} 
        For Cayley graphs generated by neighboring transpositions, such a mapping is closely related 
        to the \href{https://en.wikipedia.org/wiki/Lehmer_code}{Lehmer code}.     

        \item \textbf{CayleyPy AI-methodology to find diameters and wordmetrics.}  
         We present a methodology how we can find diameters using our AI-based CayleyPy library,
         and present multiple successful examples of computations discussed before are based on it.

        \item \textbf{Behavior under taking $G/H$.} Empirically, we observe the following pattern: the polygon associated with the Schreier coset graph of $G/H$ appears as a subpolygon of the polygon corresponding to Cayley graph of the full group $G$.

        \item \textbf{Consecutive $k$-cycles and variations — $(k-1)$-shrinkage principle.}  
        We present several computations of diameters and word metrics for generators given by consecutive $k$-cycles, explicitly obtaining the corresponding quasi-polynomials.  
        Empirically, we observe that these formulas approximately correspond to shrinking the graph $(k-1)$ times, a behavior that is particularly reflected in the associated dual polygons.

         \item \textbf{ $\mathrm{SL}(2,\mathbb{Z})$ and Farey graphs.}
         We put into the framework of the present paper results from the previous paper by one use:
         (D.Melnikov et.al. \cite{CamiloMelnikovNovaesPrudenziati2019}).

        \item \textbf{Properties of the corresponding $H$-polynomials.}  
        To each quasi-polynomial, one naturally associates an $H$-polynomial, and we study their properties.  
        In many cases, we observe positivity, unimodality, Poincaré duality, and analogues of the Riemann conjecture.  
        While some examples satisfy all of these properties, there are cases in which one or more of them are violated.

        \item \textbf{Limit shapes.}
        We study the limit shape of the corresponding "strings" in large size limit,
        on the example of the Coxeter group, we proposed exact formulas for limit shapes
        with fixed areas under the curves, which are consistent with previous results by A.M.Vershik et.al.

        \item \textbf{Evolution in large size limit, Burgers equation and KPZ.}
        We study numerically the geodesic flow on the graph in the dual picture
        and conjecture its relation with Burgers and KPZ  equations.

        \item \textbf{Spanning trees, Mahler measures, string-like model with polygonal worldsheets, 
        spin chains in thermodynamic limits, spectra of conformal dimensions.}
        We study Laplacian and spectral properties in large size limits.
        We propose various conjectures on numbers of spanning trees.
        We also propose identification of limitings theories 
        with a free bosonic model (playing here the role of a string-like dual) which surprisingly may have polygonal worldsheets.
        The conformal dimensions of primary fields are conjectured to be
        eigenvalues of the Laplacian on the dual planar polygon.

    \end{itemize}

 \end{itemize}

\subsection{ Organization of the paper}
The first five sections present the main examples and core ideas,
while the remaining sections provide further details and technical
computations. Although the overall length of the paper is substantial,
we hope that the first five sections are sufficient for the reader
to grasp the central concepts and motivation. We put reminders on background material to the last section.

\textbf{ AI, mathematics and physics }

The {second superstring revolution} and related developments revealed highly nontrivial
dualities among string-related theories and uncovered far-reaching
consequences \cite{Candelas1991, SeibergWitten1994a, Witten1995,  Polchinski1995,
SchwarzSecondSuperstringRevolution, StromingerVafa1996, vafa1996evidence, StromingerYauZaslow1996, BFSS1997, Connes1998, Maldacena1998, gubser1998gauge, witten1998ads,
kontsevich2003deformation, cattaneo2000path,
GrossNekrasov2000_DynamicsNC, 
de2000holographic,
dijkgraaf2002matrix, 
Nekrasov2003, MinahanZarembo2003,
OkounkovReshetikhinVafa2003_arxiv,
OoguriStromingerVafa2004_OSV,
Gukov2005, dabholkar2005precision, 
RyuTakayanagi2006,  KapustinWitten2007, pestun2007localization,
MironovMorozov2009exactBS,
AldayGaiottoTachikawa2010, verlinde2011origin,
GaiottoMooreNeitzke2013_WallCrossingHitchinWKB,StanfordSusskind2014, 
GrassiHatsudaMarino2014, GaiottoMooreWitten2015_AlgInfrared,  BrownSusskind2016, hijano2016witten}.
These developments have profoundly influenced the evolution of modern
mathematics in a variety of directions.
We hope that the dualities proposed here may extend the scope of
string-theoretic ideas to new domains, including artificial
intelligence, as well as to such areas of mathematics such as
graph theory, group theory and combinatorics.
One of the "pre-AdS/CFT" examples of dualities
was proposed for two-dimensional Yang-Mills
(its string description, "Gross-Taylor formula")
\cite{Witten1992_2DgaugeRevisited, GrossTaylor1993_TwoDimQCD_String, GrossTaylor1993_TwistsWilson, CordesMooreRamgoolam1995_2DYM},
recently it has been connected to 
Cayley graphs via "Yang-Mill/Hurwitz correspondence" \cite{novak20242d},
moreover quasi-polynomial expressions which play key role in the present paper,
also appear in Hurwitz theory \cite{Norbury2010, AndersenChekhovNorburyPenner2018, KramerLewanskiShadrin2018},
it would be tempting to understand relations
to the present paper.

The present time is characterized by growing interest in and number of applications of deep learning methods to mathematics and physics: machine learning has been emerging as ``a tool in theoretical science''~\cite{douglas2022machine}.
In recent years, this has led to several noteworthy applications to mathematical  and physical problems:~\cite{lample2019deep,davies2021advancing, bao2021polytopes, 
romera2024mathematical,
coates2024machine,alfarano2024global, charton2024patternboost,shehper2024makes,swirszcz2025advancing,hashemi2025transformers,he2024ai, lal2024r,lal2025deep,douglas2025diffusion,  
GeorgievGomezSerranoTaoWagner2025, berczi2026flow, ju2026ai, GuevaraLupsascaSkinnerStromingerWeil2026, EllenbergLibedinskyPlazaSimentalWilliamson2026,ChenCumminsOnoEtAl2026, Knuth2026, morozov2026learning}. 
Seewoo Lee created a repository that collects papers in AI for mathematics,  \href{https://seewoo5.github.io/awesome-ai-for-math/}{
Awesome AI for Math}. The present paper can be seen as an attempt at a dual-sided application of AI to mathematics and physics, and vice versa.

\clearpage
\section{Simplest Examples. ROC curve = string dual to $Gr(k,n)$ graph node }

\begin{figure}[H]
	\centering
	\includegraphics[width=0.45\textwidth]{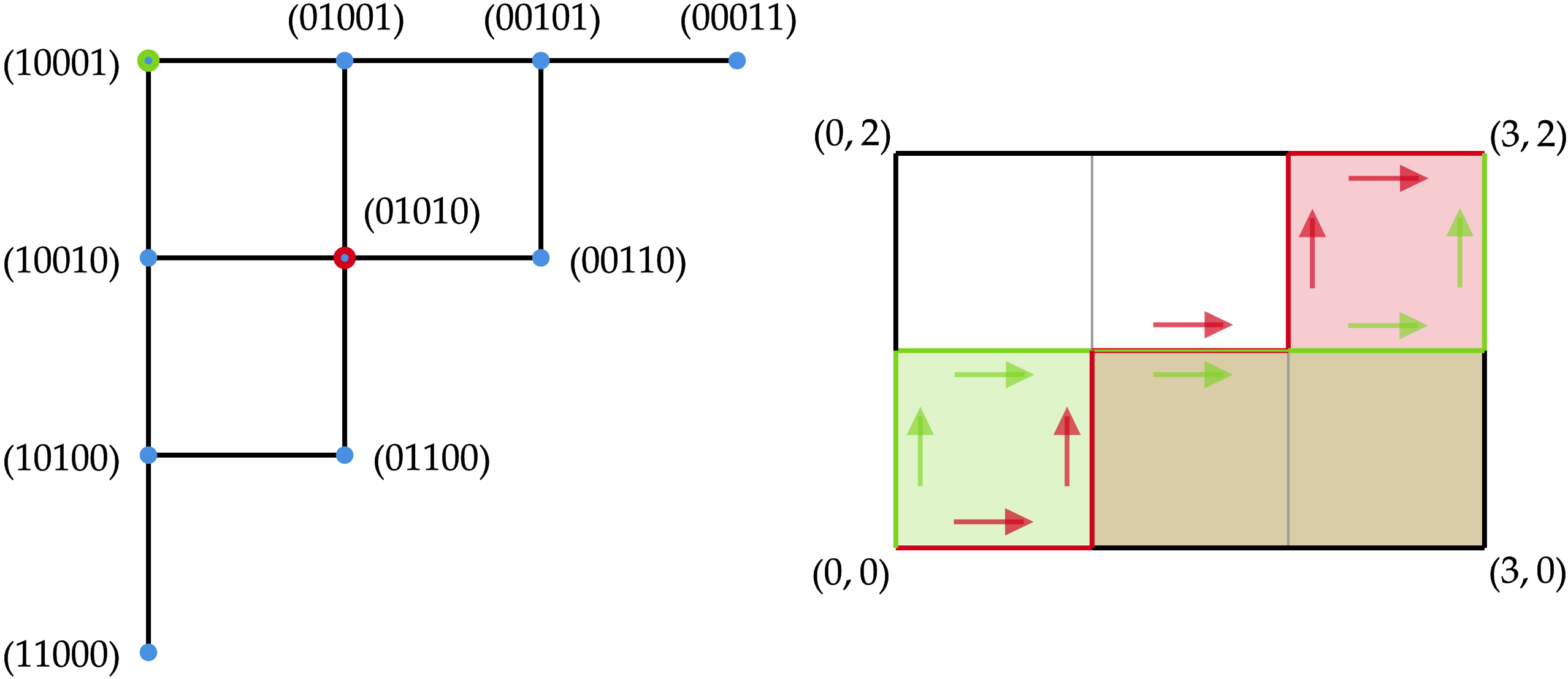}
	\caption{Left to right: the graph $Gr_{2,5}$ with marked vertices (red and green), and a polygon illustrating the corresponding paths (red and green).}
	\label{fig:p_triangles_3120}
\end{figure}

\begin{figure}[H]
	\centering
	\includegraphics[width=0.75\textwidth]{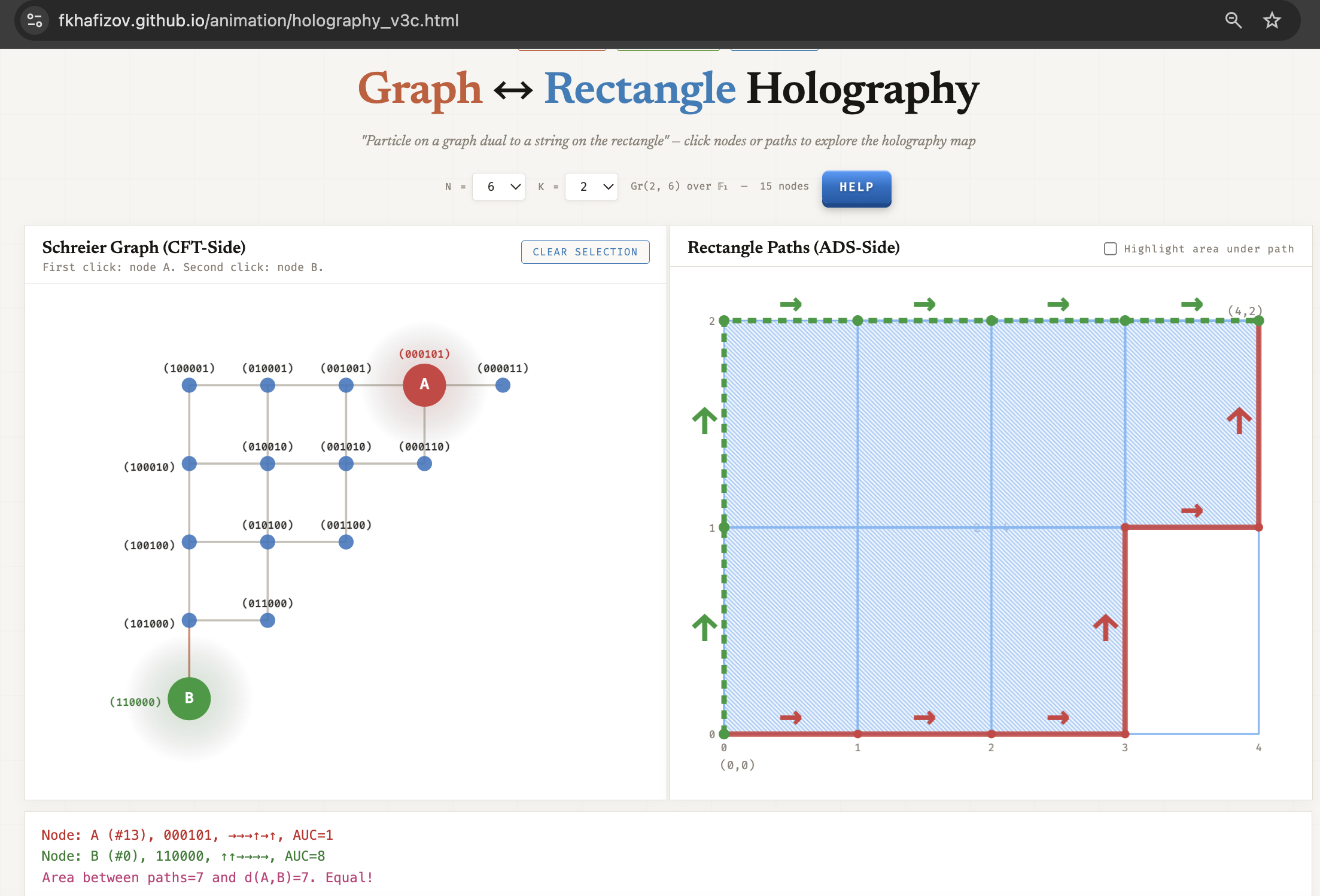}
	\caption{Left to right: the graph $Gr_{2,6}$ with marked vertices (red and green), and a polygon illustrating the corresponding paths (red and green). Interactive widget (F.Khafizov) available at \href{https://fkhafizov.github.io/animation/holography_v3c.html}{link}. The distance between nodes on graph is equal to area between the paths - "complexity = area" principle. Which in that case is equality of the Mann–Whitney statistic and the area under the ROC curve.  }
	\label{fig:graph_rectangle6_2}
\end{figure}

\textbf{Section Outline.} Here we present the basic example for the main ideas of the present paper. We consider a simple Cayley graph and describe the duality for it.

\href{https://en.wikipedia.org/wiki/Receiver_operating_characteristic}{ROC curves} (Receiver Operating Characteristic curves) and 
\href{https://scikit-learn.org/stable/modules/generated/sklearn.metrics.roc_auc_score.html}{area under the ROC curves} are  widely used 
quality metrics in ML. (See e.g. the \href{https://alexanderdyakonov.wordpress.com/2017/07/28/auc-roc-%d0%bf%d0%bb%d0%be%d1%89%d0%b0%d0%b4%d1%8c-%d0%bf%d0%be%d0%b4-%d0%ba%d1%80%d0%b8%d0%b2%d0%be%d0%b9-%d0%be%d1%88%d0%b8%d0%b1%d0%be%d0%ba/ }{exposition} by ~A.~G.~Dyakonov, first Kaggle top-1 Grandmaster.) Surprisingly, these examples can be used to illustrate the idea of holographic string duality, as well as to highlight various non-trivial results and open questions. The “complexity = area” principle reduces here to the familiar \href{https://en.wikipedia.org/wiki/Receiver_operating_characteristic#Area_under_the_curve}{equality}  between the Mann–Whitney statistic and the area under the ROC curve. Figures \ref{fig:p_triangles_3120}, \ref{fig:graph_rectangle6_2} illustrate the discussion below. So we describe below two sides of the correspondence: the graph (``CFT-side") and the rectangle polygon (``AdS-side");
the map from graph nodes to paths on rectangle (``holography map"). We demonstrate its simple, but non-trivial properties
which can be summarized a ``particle on a graph dual to a string on the rectangle". 
The discrete string has an simple action which is a "tropical" ("ReLU") analogue 
of the conventional action. 
The explanations how it is applied in ML tasks as a quality metric are also provided.

{\bf Graph. ``CFT-side".} 
Graph nodes correspond to vectors of length $n$ with entries $k$ zeros and $n-k$ ones.
Two nodes are connected by an edge if there is a transposition of some neighbor elements $(i,i+1)$ which
sends one vector to the other one.  This graph is the \href{https://en.wikipedia.org/wiki/Schreier_coset_graph}{Schreier coset graph} for $S_n$ with neighbor-transposition
generators; it is a quotient of the \href{https://en.wikipedia.org/wiki/Permutohedron}{permutohedron} graph by $S_k \times S_{n-k}$ and
should be thought of as the  \href{https://en.wikipedia.org/wiki/Grassmannian}{Grassmanian} $Gr(k,n) = GL(n)/(GL(k)\times GL(n-k))$ over the \href{https://en.wikipedia.org/wiki/Field_with_one_element}{field with one element}, by the usual analogy $S_n = GL_n(F_1)$.

{\bf Polygon = rectangle. ``AdS-side".} 
Consider the rectangle of the size $k\times (n-k)$ on the plane with integer coordinates. 
Set of our paths (``strings") are paths making steps right and up going from $(0,0)$ to $(k,n-k)$.
These are similar to \href{https://en.wikipedia.org/wiki/Catalan_number}{Dyck paths}, 
but there is no restriction for them to be under the diagonal.

{\bf ``Holography map" from graph nodes to paths (``strings"). }
Take a vector of 0's and 1's and associate to it a path by the rule: each ``1" is a step up, each ``0" is a step right.
Clearly it is a bijection from graph nodes to paths described above. 
As we will discuss below these are precisely the ROC curves for certain machine learning models. 

{\bf  ``Complexity = area". Mann-Whitney = area under the ROC curve. }
Take two nodes on the graph and consider the distance between them, i.e., length of the shortest path, or, equivalently, the minimal number of neighbor transpositions which are needed to transform one vector to the other one (``gate complexity"). 
One can check:

{\bf Proposition 1.} For any two nodes the distance between them on the graph (i.e. the complexity of one with respect to the other) is equal to the area between corresponding paths.

Typically as one of the nodes  we take the sorted vector $0...01...1$, which corresponds to the path going along the bottom-right border. So in that case we get that the area \emph{under} the curve equals the complexity of the other node.

{\bf Use in ML. Notation abuse of ``under'' vs.\ ``above'' the curve.}
Consider a binary classification task with \( n \) observations, of which \( k \) have ground-truth label \(0\) and \( n-k \) have label \(1\).
A machine-learning model assigns a probability score to each observation.
Sorting the observations by these scores produces an ordered binary vector consisting of zeros and ones (ground truth labels).
For a perfect model, this vector is $0^k 1^{n-k}$ (i.e. it is sorted).
An imperfect model produces a non-trivial interleaving of zeros and ones.
Any metric that  quantifies deviation  from the perfectly sorted vector therefore can serve as a natural measure of model quality.
One such metric is obtained by encoding the binary vector as a lattice path, as described above.
The area under this path provides a quantitative measure of deviation from the perfectly sorted vector $0^k 1^{n-k}$.
Indeed, the lattice path corresponding to that vector is  the bottom-right boundary of the corresponding rectangle and hence has zero enclosed area.
Proposition~1 shows that this measure has a natural combinatorial interpretation:
it equals the minimal number of neighboring transpositions required to transform the given vector into the sorted one.
This is precisely the statistic introduced by Mann, Whitney, Wilcoxon, and Kendall in classical rank-based hypothesis testing,
and it admits modern group-theoretic interpretations as developed in foundational work by P.~Diaconis
(see, e.g., \cite{ChatterjeeDiaconis2016} and references therein).
We note a minor abuse of terminology:
in our convention, the area \emph{under} the path measures \textit{deviation} from the perfectly sorted vector,
whereas in the standard ROC-curve exposition, the area under the curve measures \textit{similarity} to the ideal classifier.
Our choice of convention aligns naturally with the guiding principle ``complexity = area''.

\textbf{Stanley--Edelman--Greene correspondence as a bijection of extremals
for a particle on a graph and strings on a polygon.}

\cite{Stanley1984} computed the number of shortest paths between the two most distant vertices of the permutohedron graph; such paths are known as \emph{sorting networks}.
As discussed in the influential paper \emph{Random Sorting Networks} \cite{AngelHolroydRomikVirag2006}:
``another breakthrough was achieved by
Edelman and Greene'' \cite{EdelmanGreene1987},
who constructed a bijection between sorting networks and staircase-shaped standard
\href{https://en.wikipedia.org/wiki/Young_tableau}{Young tableaux} of size \( n \).
We now formulate an analogue of this result for our graph
and provide its string theory  interpretation.

\textbf{Proposition~2.}
Let \( A \) and \( B \) be two vertices of the graph above.
Then the shortest paths between \( A \) and \( B \) are in bijection with Young-type tableaux
associated with the region bounded by the corresponding holographically dual lattice paths.

\textbf{Discrete string action whose extremals are Young tableaux: standard $+$ ReLU.}

The simplest continuum string action is
$\int \int \left( X_a^2 + X_b^2 \right)\, da\, db $.
In the discrete setting, derivatives are replaced by finite differences,
$X_a = X(a,b) - X(a+1,b),  X_b = X(a,b) - X(a,b+1)$.
Consider the  discrete action:
$\sum_{a,b} \bigl( \mathrm{ReLU}(X_a) + \mathrm{ReLU}(X_b) \bigr)$,
where \(\mathrm{ReLU} = \max(\,\cdot\,,0)\) is \href{https://en.wikipedia.org/wiki/Rectified_linear_unit}{ReLU} function.
It follows that the minima of this action  are precisely Young tableaux, i.e. they provide solutions for equations of motion; 
at the same time, the action closely resembles that of a conventional string theory.

\textbf{Corollary.} (Stringy interpretation of analogue of Stanley--Edelman--Greene correspondence).
Extremals of a particle moving on a graph (i.e.\ shortest paths) between vertices \(A\) and \(B\)
are in bijection with extremals of the discrete string action (i.e. Young tableaux) with boundary conditions
defined by the holographic images of \(A\) and \(B\).

\begin{figure}[H]
	\centering
	\includegraphics[width=0.45\textwidth]{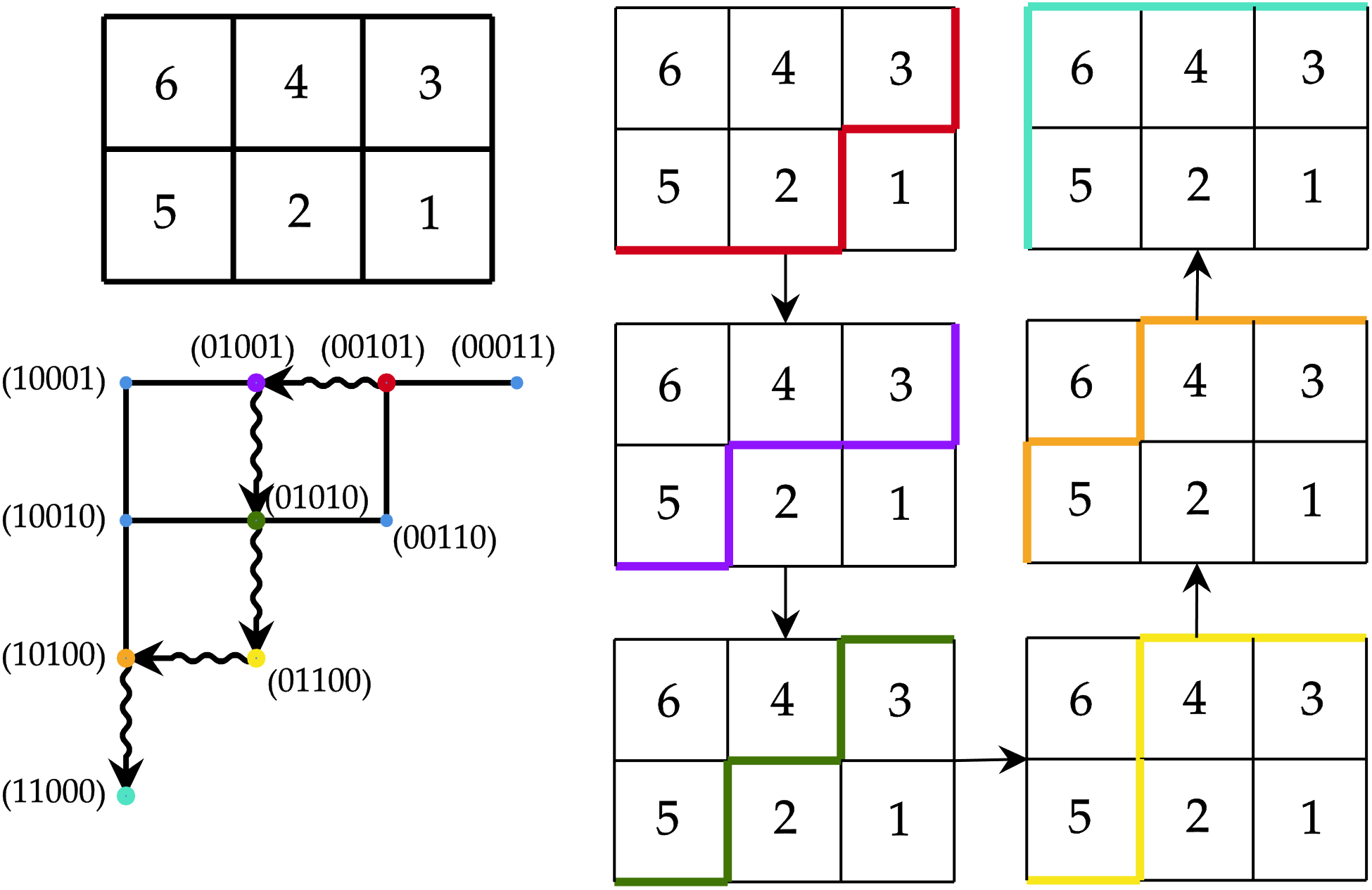}
	\caption{The Young tableau (upper left corner) provides a solution to the string equations of motion.
    It encodes the motion of the string as depicted on the right.
    In the right panel, each position of the string at times $t=1,\dots,6$ is shown in a different color,
    and the number in each box indicates the step at which the string passes through it.
    Each obtained string is a non-decreasing path, thus it belongs to the image of the holography map.
    In the holographically dual picture, the tableau encodes one of the shortest paths on the graph -- indicated by wavy edges.}
	\label{fig:string_dynamics}
\end{figure}

It would be natural to consider other activation functions,
and also other distances on graphs, e.g. diffusion distances (see e.g. the second paper of the project) 
related to supersymmetric LG models.

\paragraph{Further results.}
We discuss further results on this example in the next section "Neighbor transpositions Cayley and Schreier graphs".








\textbf{Open questions.}
We hope that the answers can be naturally formulated using duality.
 
The spectral density of the eigenvalues (as opposed to Bethe roots) is, to the best of our knowledge, unknown.
Similarly, the properties of the resolvent—such as the equations it satisfies—are not known.

The \href{https://en.wikipedia.org/wiki/Abelian_sandpile_model#Sandpile_group}{Jacobian} (also called the sandpile, critical, or Picard) group of these graphs is also unknown.
While its order, given by the number of spanning trees, is known,
the group structure itself has not been determined.

Graph invariants like \href{https://en.wikipedia.org/wiki/Tutte_polynomial}{Tutte polynomial}
are not known. 

In large $n$ limit \href{https://en.wikipedia.org/wiki/Conformal_field_theory}{CFT} 
description is not well-understood, duality may  might be related to a version of  AdS/CFT.
In particular, we may expect the spectrum of conformal dimensions can be related to the eigenvalues 
of the ordinary Laplacian on the plane in the rectangle, by analogy with AdS/CFT principles:   Casimir of conformal algebra maps to Laplacian in the bulk.

\section{AI tasks as predictions of particle trajectories. General AI-holography expectations }~
\textbf{Section Outline.}
Here we first discuss an analogy in which common tasks in AI are viewed as predicting particle trajectories. Then outline the idea of “AI holography”: rather than working directly with a particle  (difficult, "strongly coupled" regime), it is advantageous to seek a dual string description, which provides a more tractable ("weakly coupled") formulation.
We also emphasize that, in a certain sense, modern AI systems—as well as natural neural networks—already operate according to a somewhat similar principle through the use of embeddings. However, we believe that fully exploiting the power of string dualities can lead to a deeper and more systematic understanding of these mechanisms,  enabling the design of more effective AI systems. But to work in discrete
setting it is necessary to consider unusual "tropical" string
actions which involve e.g. "ReLU" ( $\max(\,\cdot\,,0)$) functions,
and lead to some unexpected, but desirable properties of equations of motion.

\subsection{Particle trajectories as texts or action sequences} 
The goal of this subsection is to draw the attention of the physics community to the fact that many core aspects of AI — including input data, prediction objectives, and even methodological approaches — are closely analogous to the study of particle dynamics on graphs with labeled edges. \href{https://en.wikipedia.org/wiki/Cayley_graph}{Cayley graphs} serve as particularly natural and representative examples. Such questions are not uncommon in mathematical physics, and their conceptual and technical tools may prove  beneficial for the development of AI.

\textbf{Texts / robot manipulations / game play / theorem proofs as discrete particle trajectories on edge-labeled graphs.}
Let us bridge the settings and tasks commonly used in AI—such as text corpora or action sequences in reinforcement learning—with more traditional problems in physics, namely particle trajectories on graphs and the study of their dynamics. We also provide a brief review of several AI concepts from this perspective, which may be of interest to physicists.

Modern large language models operate on texts, while other powerful systems—such as those based on reinforcement learning—work with sequences of actions: robot manipulations, moves in games, and so on. In all these cases, the primary data consist of sequences of discrete tokens, for example
$
\texttt{ABBCADB}\ldots
$ 
These tokens may represent letters of the English alphabet, commands for a robotic manipulator, moves in a game such as chess, or the names of theorems and lemmas in a mathematical proof.

To connect this viewpoint with frameworks commonly used in physics, it is natural to interpret such sequences as trajectories of a particle, where each token specifies an elementary increment of the particle’s motion. Pushing the analogy further, one may regard all possible states (for instance, sentences or configurations) as nodes of a graph, with tokens labeling the edges. Appending a new token then corresponds to moving from the current node to a neighboring node along the edge labeled by that token.
Thus toke sequences can be view as trajectories of particle
on graph where edges are marked by the token labels. 

An archetypal {\bf example} of this consideration is provided by \href{https://en.wikipedia.org/wiki/Cayley_graph}{Cayley graphs} in group theory. There, nodes correspond to the elements of a group, and edges correspond to a chosen generating set: two nodes are connected if
$
a = g b
$
for some generator \( g \). Thus, edges labeled by group generators and thus paths (particle trajectories) on the Cayley graph correspond to words formed from these generators. 
One may further restrict attention to those words that correspond to shortest paths in the graph, thereby defining the \emph{geodesic language} of the group. Thus, a quite standard task in physics and mathematics -- to study a free particle on a graph, appears to be  essentially equivalent to studying the graph’s geodesic language.
More generally, one can consider a variety of related languages, such as words whose associated paths deviate from geodesics by at most a prescribed amount, or that satisfy other geometric constraints. These choices lead to different classes of admissible trajectories, interpolating between strictly geodesic motion and more flexible, near-geodesic dynamics.
The considerations above do not require working with groups. Similar arguments apply to arbitrary state-transition graphs whose edges are labeled by tokens, and where edge weights can be naturally incorporated as probabilities of selecting a given transition. In this perspective, recorded sequences of moves in games such as Go or chess play a role analogous to words in group-theoretic settings: each game record corresponds to a path in the underlying state-transition graph, while individual moves act as tokens labeling the edges. Collections of such records therefore form a “language” of admissible trajectories, shaped by the rules of the game.

Let us {\it rephrase the same idea as above} with different emphasizes and provide more examples. 
What is a “language”? One fixes a set of tokens (i.e., an alphabet) and considers a collection of admissible sequences built from these tokens — this collection is, essentially, the language. In other words, a language is a rule (explicit or implicit) that selects, among all possible sequences, those that are allowed.
For example, the set of all texts written in English defines the English language. There are many other such languages: admissable sequences of moves in games such as Go or chess define the corresponding game languages; all possible protein sequences over the 20–amino-acid alphabet define the language of proteins; all possible DNA sequences define the DNA language;  all valid \href{https://en.wikipedia.org/wiki/Simplified_Molecular_Input_Line_Entry_System}{SMILES} encodings define a language for chemical molecules, etc. As discussed above, one can interpret these examples as collections of particle trajectories on a graph whose edges are labeled by elements of the alphabet.
In this sense, the current AI paradigm is closely analogous to experimental physics: one observes a collection of trajectories of a particle — here represented by a language, that is, a set of admissible sequences — and attempts to uncover the hidden laws governing its dynamics. 
It is remarkable that essentially the same techniques originally developed for natural language processing, namely transformer-based models, can be transferred  to other ``languages,'' such as protein language models (e.g. ESM2 \cite{lin2023evolutionary}) or chemical languages based on SMILES representations, (e.g. ChemBERTa \cite{ahmad2022chemberta2}), etc.
This transfer has led to the creation of tools that have become indispensable in modern bioinformatics and cheminformatics, achieving top results on a wide range of benchmarks, for example in the \href{https://en.wikipedia.org/wiki/Critical_Assessment_of_Function_Annotation}{CAFA} protein properties prediction challenge \cite{chervov2024protboost}.

Let us reemphasize that nothing more than a set of sequences—a ``language''—is required to begin AI modeling.
I.e., having a collection of texts—where “texts” may be completely arbitrary sequences over any alphabet—is already sufficient to apply AI techniques.
Modern AI techniques are remarkably successful across a variety of such ``languages,'' many of which are quite different from natural ones. To some extent, this reflects a simple principle: what the human brain can do, artificial neural networks may also learn to do. Since humans can master natural languages, programming languages, and games such as Go or chess, it is perhaps not entirely surprising that similar AI techniques are capable of mastering these domains as well.


Thus, texts in natural languages or sequences of actions in reinforcement learning can be viewed as collections of particle trajectories on a graph whose edges are labeled by tokens, with admissible words corresponding to trajectories that satisfy constraints imposed by both the graph structure and the particle dynamics. Studying dynamics on graphs is a well-established problem in physics and mathematics.


\textbf{AI tasks as prediction of particle trajectories with prescribed initial or boundary conditions.}
Let us discuss the close analogy between typical AI tasks and computing particle trajectories with given initial or boundary conditions. 

The basic training paradigms of modern LLM systems can be broadly divided into two modes. The first is the GPT-style setting of \href{https://en.wikipedia.org/wiki/Generative_model}{generative modeling}, where one is given the beginning of a text and the task is to generate—or predict—its continuation. The second is the \emph{masked} mode (as in classical \href{https://en.wikipedia.org/wiki/Word2vec}{Word2vec} and related models), where both the beginning and the end of a sequence are provided, and the goal is to infer the missing middle portion.

Closely related reinforcement-learning systems—such as those used for robot manipulation, gameplay, or automated theorem proving—are typically generative as well: at each step, the objective is to predict the next action in a sequence.
For example, the task of finding a path on a Cayley graph—which is mathematically equivalent to decomposing a group element into a product of generators—can be viewed as finding a trajectory with prescribed boundary conditions: given element and identity of the group. In the case of the Rubik’s cube group, this corresponds to solving the cube. For games such as Go or chess, the objective is to find a path from an initial position to a position labeled “victory” in an environment where a second player is simultaneously attempting to achieve the same goal. Despite this added complexity, the problem can still be framed as finding a particle trajectory with boundary conditions on a state-transition graph, with the additional complication that the state evolves in response to the opponent’s moves (a kind of randomized environment).

Similarly, proving a mathematical theorem can be seen as finding a particle trajectory on the state-transition graph of all admissible proofs.
Overall, these examples are representative of a broad class of problems in AI: generating a sequence of tokens such that the resulting sequence satisfies specified constraints or desired properties.

From the physical perspective suggested above, these learning tasks admit a natural interpretation in terms of particle dynamics. The generative setting corresponds to predicting a trajectory given fixed initial conditions, whereas the masked or infilling setting corresponds to predicting a trajectory subject to fixed boundary conditions.


\textbf{Counterintuitive - multiple local classical trajectories 
with fixed initial/boundary conditions.}
To what extent are discrete systems—such as particles moving on graphs, or more generally symbolic systems like languages—similar to classical physical systems? In particular, can one meaningfully apply the standard physical language of actions and equations of motion to such settings?

At first sight, there appears to be a serious obstacle to such a description. In classical mechanics, actions are smooth functionals of coordinates and their derivatives, leading to equations of motion in which fixing the initial position and momentum (and possibly higher derivatives) uniquely determines the trajectory of a particle. (In a discrete setup, this is analogous to fixing several positions in the history, which should completely determine the continuation.)
By contrast, in typical discrete systems like on graphs, fixing initial conditions generally allows for a large number of possible continuations. For example, a particle moving on an infinite tree (such as the Cayley graph of a free group) may move in essentially any direction except immediately backtracking (i.e. previous history almost have no effect on next moves), and each such choice produces a geodesic. In this sense, initial conditions impose only weak constraints on the subsequent motion.

A similar phenomenon occurs in languages: the beginning of a sentence often does not uniquely determine its continuation. This apparent non-uniqueness may give the impression that the classical framework of actions and equations of motion is inapplicable to discrete systems.


\textbf{Unusual "tropical/ReLU" discrete string actions - for the rescue.}
Our observation is that modifying classical string actions so as to include
tropically inspired expressions such as $\max(\,\cdot\,,0)$ (also known as the
\href{https://en.wikipedia.org/wiki/Rectified_linear_unit}{ReLU} function) leads to string equations of motion that are well suited to
discrete settings. In particular, the resulting equations admit solutions in
the form of Young tableaux, which are highly non-unique even fixing initial conditions,  yet still subject to
nontrivial global constraints. The example has been decsribed in the previous section. 

In a broader sense, this behavior might be consistent with the perspective of
tropical geometry. Graphs can be viewed as tropical limits of Riemann surfaces,
and it is therefore natural to expect that the corresponding action functionals
should also involve operations characteristic of tropical geometry, such as the $\max$ operation. Although such analogy may not be fully correct.

\textbf{Embeddings: LLMs and natural brains are already performing a form of “duality.” }
In a nutshell, the idea of duality in physics is to replace a given description of a system by a dual one in which previously difficult questions become unexpectedly tractable. This closely parallels a central paradigm of modern AI: before solving a task, one first seeks a more convenient representation of the data called —an "embedding" or "latent space representation"—and then operates primarily within that representation.
Even more striking is the parallel with the AdS/CFT principle that “complexity = area”: in modern AI, a key feature of embeddings is that notions of similarity (difficult to compute in original setting) are reduced to simple geometric quantities, such as dot products or distances. While these ideas are not identical, they share the same essential feature: transforming a hard-to-compute notion of complexity or similarity into an easily computable geometric measure—whether an area, a distance, or a scalar product.

Let us reiterate the points made above, adding some details and providing perspective from both natural brains and artificial neural networks. 
From the perspective of natural brains, all incoming information—whether visual, auditory, or tactile—ultimately affects the neurons, leading to their activation or deactivation. The degree of activation of a neuron can be approximated by a number, so the state of the brain at any moment can be represented as a vector of real numbers indexed by neurons. In other words, the brain essentially converts all incoming information into a vector representation, or an embedding.
The goal of the brain is to operate effectively on these neural states. One can reasonably expect that, through evolution, brains have been optimized to find the most suitable representations of any incoming information. Effective cognitive operations can then be thought of as relatively simple operations on these vectors. In this sense, the brain naturally performs a process analogous to the concept of duality in physics: transforming a complex problem into a representation where simple operations suffice to solve it.

Modern artificial intelligence follows the same principle. Regardless of the task—whether language translation, summarization, sentiment analysis, prediction of protein properties or of chemical molecules, or other applications — the first step is typically not the task itself, but the construction of high-quality vector representations (embeddings).
 Once these embeddings are obtained, all subsequent operations are performed on them, without directly interacting with the original data. Consequently, designing “good” embeddings becomes a central problem in AI.
Analogous to the AdS/CFT principle in physics, where “complexity = area,” embeddings in AI often transform computationally difficult similarity measures of raw information into simple, tractable operations—such as dot products—allowing complex problems to be solved efficiently.
The famous examples on \href{https://en.wikipedia.org/wiki/Word2vec}{Word2vec}  \cite{MikolovSutskever2013} word embeddings which are sometimes  shortened (not quite accurately) to    "King - Man + Woman = Queen" (i.e., the corresponding vector embeddings 
\href{https://en.wikipedia.org/wiki/Word2vec#Preservation_of_semantic_and_syntactic_relationships}{approximately satisfy} this equality) 
- represents the same phenomena: 
semantic and not trivial relations in the language are transformed to simple operation on good embeddings
(just the addition of vectors in that example). 
Which is intriguing, as the constraint was not explicitly imposed during training, yet it emerged naturally.

Despite the success of modern AI, most embedding constructions have been achieved in an ad hoc manner. There is still no clear understanding of how to internally characterize “good” embeddings, nor how to systematically improve them beyond trial-and-error methods. The currently dominant approach in AI is “scaling”: training ever-larger neural networks on ever-larger volumes of data. This strategy indeed works. However, it is somewhat akin to trying to approximate a highly non-trivial function using a trivial one. Increasing the number of parameters and data can improve the fit, but a better approach may be to understand the underlying nature of the function from first principles and to choose an approximation strategy more wisely.
The perspective offered by string dualities may provide valuable insight into these questions, potentially guiding the design of more principled and effective representations.
In our idea, {\bf good embeddings are  holographic dual strings corresponding to the original particle states.} 
To compute downstream quantities from embeddings 
one should take into account possibly not flat Riemannian metric (similar to AdS metric),
which is expected to be a part of the dual description, more subtle effects should take into account
finite-size corrections.

Let us recall the difference between older Word2Vec-style embeddings and newer context-dependent, transformer-based embeddings.
Word2Vec produces embeddings for words themselves, but not for words in context. In direct analogy with our picture, words may be viewed as nodes of a graph (states), and the goal is to construct embeddings for these nodes.
(Moreover, Word2Vec admits a direct generalization to graph embeddings,
as demonstrated in the well-known works DeepWalk and Node2vec
\cite{Perozzi2014DeepWalk,GroverLeskovec2016node2vec}).
In contrast, modern transformer-based architectures generate embeddings sequentially, so that the representation of a word depends on the preceding text—that is, it is context-dependent. In this sense, embeddings are effectively constructed for sequences rather than isolated tokens.
This perspective is fully consistent with our framework: paths on a graph can be identified with their terminal states, and therefore embeddings of these states naturally correspond to context-dependent embeddings of words.
Our setup encompasses both scenarios: one in which multiple paths can lead to the same state, 
and another in which this is not allowed (as in the case of a tree graph).

\subsection{AI-holography} Here we outline our expectations for analogues of holographic string dualities in AI tasks. In brief, we expect that for broad classes of systems there exists a dual formulation in which states of the original system are mapped (‘holographically’) to paths (strings) in a dual space. In this picture, the particle system on one side (the ‘CFT side’) is equivalent to a string theory on the dual side (the ‘AdS side’). This dual description may offer a more tractable framework; in particular, these paths (strings) can serve as ‘good’ embeddings. Thus, providing new approaches for improving AI methods.

As described in the previous subsection, we may view AI systems of interest as defined by a set of admissible token sequences (‘language’), which can be interpreted as paths (‘particle trajectories’) on an edge-labeled graph and basic AI questions
can be view as predicting particle trajectories with initial or boundary conditions. The nodes of this graph represent the states of the original system. (In AdS/CFT terminology that is ‘CFT side’).
We expect the existence of a dual object equipped with a holographic map that sends nodes of the original graph to paths (strings) in the dual space. Such that ‘particle theory’ on the original graph would be equivalent to ‘string theory’ on the dual object: any question about particles could be reformulated and computed via their holographic images, with the expectation that the dual description is more tractable.
Moreover, based on examples and general considerations, we expect that unconventional, ‘tropical’ actions for discrete strings are necessary to describe AI-related systems. In suitable large-size limits, however, these discrete models may converge to more familiar geometric (gravitational) descriptions. Thus, questions that are especially difficult when original system is large
 could translate into geometric (gravitational) computations that are more tractable.
This perspective parallels the original AdS/CFT proposal, where quantum observables such as Wilson loops on CFT-side admit dual gravitational descriptions on AdS-side, as well as subsequent developments, in particular the ‘complexity = volume/action’ principles.


Let us summarize expectations in the itemized form: 
    \begin{itemize}
        \item {\bf \href{https://en.wikipedia.org/wiki/Holographic_principle}{Holography}.} \cite{tHooft1993, Susskind1995Hologram}  Meaning that $d$-dimensional objects of the original system are mapped to $(d+1)$-dimensional, 
        i.e. nodes of graph are mapped to paths (strings) on the dual object. 
        Similarly, paths on graphs 
        (i.e. 1-dimensional objects) 
        are mapped to 2-dimensional surfaces 
        (string \href{https://en.wikipedia.org/wiki/Worldsheet}{worldsheets}). 
        We expect that the strings which are holographic images of states and "good embeddings" for the states in original system.

        \item  {\bf Particle on the graph side (``CFT"-side) = Discrete String on dual side (``AdS"-side).} 
        As is typical in string dualities, we expect that a theory defined on one side is equivalent to a theory defined on the other side. On one side, the dynamics describe a particle moving on a graph, while on the other side they are captured by a discrete string theory defined on the dual object. Holography provides a map between the two descriptions, under which the  quantities computed in one theory are equal to corresponding quantities computed in the other one.
        We illustrate what we mean by discrete string theory through concrete examples, e.g.  worldsheets are naturally associated with Young diagrams, while their images in target space correspond to Young tableaux.   
        Moreover, based on examples and general considerations, we expect that unconventional, ‘tropical’ actions for discrete strings are necessary to describe AI-related systems. E.g. such as 
        $\sum_{a,b} \bigl( \mathrm{ReLU}(X_a) + \mathrm{ReLU}(X_b) \bigr)$,
        where \(\mathrm{ReLU} = \max(\,\cdot\,,0)\) is \href{https://en.wikipedia.org/wiki/Rectified_linear_unit}{ReLU} function
        and $X_a = X(a,b) - X(a+1,b),  X_b = X(a,b) - X(a,b+1)$ - common analogs of discrete derivatives.


        
        \item  {\bf Corollary: Complexity = Area/Action.} (Similarity = Geometric Measure).
        Similar to the principle discovered in AdS/CFT correspondence (L. Susskind et. al.
        \cite{StanfordSusskind2014, BrownSusskind2016}), we expect that lengths of paths on the graph side—serving as measures of complexity—
        are mapped to areas under the corresponding curves on the dual side.
        That is a consequence of the previous principle: values of action for particles on extremals are 
        lengths of the shortest path, while for the string these are related to certain areas.
        However results in that direction are typically far more accessible than establishing the duality.

        \item  {\bf Strong coupling to weak coupling.} 
            As in conventional string theory, we expect that difficult problems on the original side are converted into more tractable problems on the dual side. For example, the computation of complexity is typically NP-hard in general and remains difficult even in specific cases, with the difficulty growing rapidly as the size of the system increases. The key idea of the duality is that it maps this hard computational problem to the evaluation of a geometric quantity—namely, an area—which is often much easier to compute.

    \end{itemize}

\subsection{Further remarks.} 

Let us comment on further analogies with, as well as differences from, the AdS/CFT correspondence. In string theory, graphs (e.g., Feynman diagrams) can be viewed as degenerate string worldsheets; in more mathematical terms, graphs arise as tropical limits of Riemann surfaces. From this perspective, particle theories on graphs may be regarded as degenerations of particle theories on Riemann surfaces. Which are closely related to 
conformal field theories (for metrics of constant curvature on Riemann surfaces). 
One may therefore speculate that the conventional AdS/CFT correspondence for CFT related to Riemann surfaces, in an appropriate tropical limit, could be connected to the considerations proposed in the present paper.

Secondly, in the original AdS/CFT correspondence, the CFT side lives on the boundary of the AdS space, whereas in our proposal, based on the examples considered, there is no direct relation between the graph and its dual object. In particular, there is no bulk/boundary correspondence. To our mind, this is an advantage, indicating that holographically dual theories may arise in more general setups than the conventional bulk/boundary scenario.  
We also expect that holographic duality is not restricted to conformal field theories, which is natural from a general duality perspective.
Moreover, for the $S_n$-Cayley graph, the dual polygon emerges from an abstract mathematical existence conjecture, rather than from a geometric picture. In this sense, the approach is reminiscent of S.~Wolfram's ideas \cite{Wolfram2002}, suggesting that cellular automata might provide insight into physics on the Planck scale, and conventional theories may emerge from such microscopic descriptions.


Modern AI systems are data-driven, as reflected in the well-known phrase ‘fire the linguist — the language system starts working better.’ Meaning that language models can learn to solve tasks on their own, without explicitly encoding structural linguistic knowledge. Indeed, modern LLMs learn entirely  from data and have achieved remarkable success. At the same time, it is widely argued that the human brain learns far more efficiently (even —orders of magnitude more). From our perspective, these may be two sides of the same coin. 
Languages are, of course, constrained by grammatical rules,
ignoring that  may partly explain why current models require more training effort than one might ideally expect.
However, the question of how to incorporate these grammatical rules in the most efficient way and combine with successful AI approaches might not be   trivial.
From our point of view structural linguistic knowledge may provide insight into understanding the holographically dual description of  languages, thereby offering a path toward building more efficient and powerful AI systems than those currently available. In this sense, the goal is not to ‘fire the linguist,’ but to bring linguists together with string theorists in pursuit of further progress.

\section{Case studies}
\subsection{\texorpdfstring{$S_n$ Cayley graphs and polygon duality}{Sn Cayley graphs and polygon duality}}
\subsubsection{General idea: graph-polygon duality and the quasi-polynomiality hypothesis}
Here we argue that the holographic duals of graphs associated with $S_n$
 are planar polygons. This perspective provides a natural—and essentially unique—explanation for the quasi-polynomiality conjectured in our previous paper. The construction of polygons and holography maps from graphs is not expected to be easy in general, 
since it is related to the computation of diameters and word metrics, which are known to be NP-hard.
This task is difficult even in specific cases; for example, it took over 30 years of continuous
effort to determine the diameter of the Rubik's cube \cite{Rokicki2014}.
In subsequent sections we work out some explicit examples.

{\bf Hypothesis: $S_n$-Cayley to polygon duality.} For $S_n$-\href{https://en.wikipedia.org/wiki/Cayley_graph}{Cayley} and
\href{https://en.wikipedia.org/wiki/Schreier_coset_graph}{Schreier} graphs,
we conjecture that the corresponding dual objects are rational polygons in the plane
(or, in degenerate cases, intervals), such that the diameter of the graph equals
the number of integer lattice points in the $n$-scaled polygon.
Moreover, we expect the existence of a \emph{holography map} such that vertices
of the graph are associated with lattice paths inside the polygon in a way that
word metrics (or ``gate complexities'') coincide with the areas under the
corresponding paths. So both diameters and word-metrics are \href{https://en.wikipedia.org/wiki/Ehrhart_polynomial}{Ehrhart quasi-polynomials}
associated with certain rational polygons.
Both expectations can be viewed as refinements and concrete realizations of the
AdS/CFT principle ``complexity = area''.

\textbf{Quasi-polynomiality hypothesis.}
The motivation for these conjectures is as follows.
In earlier work within the CayleyPy project, we computed diameters and word metrics
for a large class of Cayley and Schreier graphs.
Empirically, these quantities were observed to be eventually
\href{https://en.wikipedia.org/wiki/Quasi-polynomial}{quasi-polynomial} functions
of $n$, of degree at most two.
We further hypothesized that this behavior is generic under  conditions,
such as when the generating sets are
\href{https://en.wikipedia.org/wiki/Presburger_arithmetic}{Presburger-definable}.

\textbf{Ehrhart quasi-polynomials.}
Quasi-polynomials arise most naturally as
\href{https://en.wikipedia.org/wiki/Ehrhart_polynomial}{Ehrhart quasi-polynomials}
associated with rational polygons, which count the number of integer lattice points
inside $n$-scaled polygons.
This observation motivates the search for polygons whose Ehrhart quasi-polynomials
coincide with those obtained from the corresponding graphs.
In the present paper, we consider multiple examples and observe that such a
correspondence can indeed be established.
Moreover, this analysis reveals a clear connection with AdS/CFT ideas and
holographic string dualities, providing new examples and insights.

\textbf{Planarity from the $n^2$ conjecture.}
The appearance of plane polygons, rather than higher-dimensional analogues,
is a consequence of the fact that the observed quasi-polynomials are always
of degree at most two.
This behavior is conditional on a celebrated open problem predicting that
the diameters of these graphs are bounded by $n^2$.
This conjecture has resisted the efforts of leading mathematicians for several
decades.
This statement can be viewed as a special case of the Babai conjecture \cite{BabaiSeress1988CayleyDiameter}, for which
partial progress has been achieved (notably by B.~Green, T.~Tao,  and collaborators \cite{Tao11, Tao12}).
Nevertheless, the specific $n^2$ bound for these graph diameters has remained
open for at least fifty years \cite{Rubtsov1975} (reviews \cite{glukhovzubov1999lengths}, \cite{Helfgott2013GrowthIdeas} ).
For refinements and further discussion of this conjecture, we refer to our
previous work \cite{Cayley3Growth}.

\textbf{Behavior under taking $G/H$.}
Empirically, we observe the following pattern: the polygon associated with the
Schreier coset graph of $G/H$ appears as a subpolygon of the polygon corresponding
to Cayley graph of the full group $G$.

\textbf{H-polynomial properties.}
Given any quasi-polynomial, it is natural to associate to it what is known as the
\href{https://en.wikipedia.org/wiki/Ehrhart_polynomial#Ehrhart_series_for_rational_polytopes}{H-polynomial}.
For Ehrhart polynomials of integral polytopes, the corresponding H-polynomial
coincides with the
\href{https://en.wikipedia.org/wiki/Betti_number#Poincar%C3%A9_polynomial}{Poincar\'e polynomial}
of the associated
\href{https://en.wikipedia.org/wiki/Toric_variety}{toric variety} if this toric variety is smooth and projective (equivalently, the polytop admits a unimodular triangulation, which is always the case for polygons).
As a result, it enjoys several strong properties, such as non-negativity of
coefficients, Poincar\'e duality, and unimodality.

A deep question in combinatorics is whether analogous properties continue to hold
for more general quasi-polynomials.
We observe that in many examples the H-polynomials arising from graphs exhibit
these desirable properties.
Nevertheless, there also exist examples in which all of them are violated.

Of particular interest is a combinatorial analog of the Riemann hypothesis,
which predicts that all roots of the H-polynomial have modulus equal to one.
We find that this phenomenon occurs frequently in our examples, although
counterexamples do exist.

\textbf{Bulk-boundary is unnecessary.}
In contrast to conventional AdS/CFT, the graph and its dual object are not related by the requirement that the boundary of the dual object coincide with the graph. Instead, the polygon arises from a purely mathematical existence hypothesis and has no immediate geometric relation to the original graph.

\subsubsection{Permutohedron and Lehmer code}
Here we discuss the well-known example of the 
\href{https://en.wikipedia.org/wiki/Permutohedron}{permutohedron} graph, argue that its holographic dual
is an isosceles right triangle in the plane, identify the holography map with the
\href{https://en.wikipedia.org/wiki/Lehmer_code}{Lehmer code}, and demonstrate that the desired
properties indeed hold. 

\ref{fig:lehmer} serves as an illustration.

\textit{\textbf{Permutohedron graph.}}
Here we discuss another example: the Cayley graph of $S_n$ with respect to the neighbor
transposition generators $(i,i+1)$, which are the Coxeter generators of $S_n$.
The nodes of the graph are all the $n!$ permutation vectors, and an edge exists
between two nodes if they differ by a transposition of neighbors.
This graph is known to be the $1$-skeleton (i.e., the set of vertices and edges) of the
\href{https://en.wikipedia.org/wiki/Permutohedron}{permutohedron} polytope.
It is also called the \href{https://en.wikipedia.org/wiki/Bubble_sort}{bubble-sort} graph,
since the bubble-sort algorithm operates exactly by neighbor transpositions $(i,i+1)$
and, in fact, provides an optimal path-finding method for traversing the graph.

\textit{\textbf{Holography dual polygon is a triangle.}}
We will argue that the dual polygon is just the isosceles right triangle with nodes $(0,0)$, $(n,0)$, $(0,n)$. 
The number of possible lattice paths (without backtracking) that can move right, up, or down starting at  $(0,0)$ with ending in $(n,0)$ and are restricted to the triangle is exactly $n!$ which matches the graph size. 

\textit{\textbf{Holography map via Lehmer code.}}
The \href{https://en.wikipedia.org/wiki/Lehmer_code}{Lehmer code} maps a permutation to a tuple of $n$ integers  such that they satisfy $L(k)<n-k$, for $k=0, \ldots, n-1$.
That is, the image belongs to the triangle above. The map is known to be a bijection. 
One can think of the image as a lattice path tracing the upper boundary of the resulting region.
Such lattice paths lie inside the triangle and consist only of right, up, and down steps,
as illustrated in Figure~\ref{fig:lehmer}.

\begin{figure}[H]
	\centering
	\includegraphics[width=0.45\textwidth]{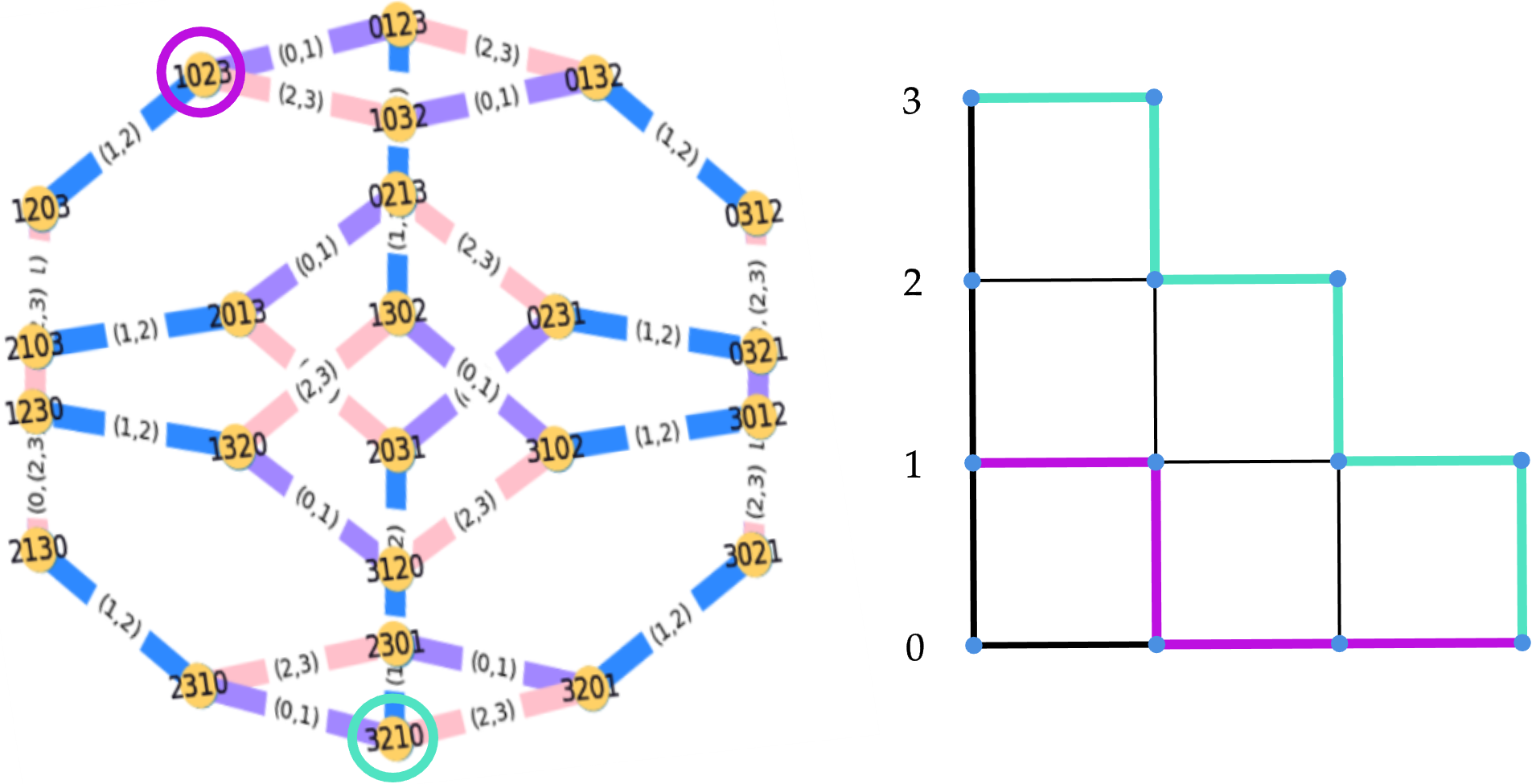}
	\caption{Lehmer-code representation of the permutations $(3,2,1,0)$ (cyan) and $(1,0,2,3)$ (magenta). Their Lehmer codes are $L(3,2,1,0)=(3,2,1,0)$ and $L(1,0,2,3)=(1,0,0,0)$, respectively. On the left, the corresponding vertices of the permutohedron are marked in matching colors. On the right, each permutation is represented by a lattice path given by the upper boundary of its Lehmer diagram inside the staircase Young diagram. The complexities of these permutations are 1 and 6, respectively, which coincide with the areas under the corresponding paths.}
	\label{fig:lehmer}
\end{figure}

\textit{\textbf{``Complexity = area".}}
The analogue of that AdS/CFT principle is  known to experts in this context.
For any permutation, one can compute its complexity, defined as the number of neighbor
transposition generators $(i,i+1)$ in its shortest decomposition
— in other words, the number of steps required by bubble sort to transform it to the sorted form, or equivalently,
the number of inversions in the permutation. (Example: \ref{fig:lehmer}).
It turns out that the number of boxes under the Lehmer code path, or equivalently the sum
$\sum_i L(i)$, exactly equals this quantity.
This means,  that the complexity equals the area under the
holographic image of the node - in accordance with general expectations.

\textit{\textbf{Bijection between particle extremals and string extremals on a polygon.}}
Similar to the previous section on ROC curves, the Stanley--Edelman--Greene
correspondence can be interpreted as a bijection   between shortest paths on a graph
and Young tableaux in the triangle, which represent string extremals.

\textit{\textbf{Open questions.}}
This case appears to be more difficult than the previous ROC curve case, and
even results that are known for ROC curve scenarios remain open here.
In addition to the open questions in the ROC case.
For example, the Laplacian of the graph also serves as the Hamiltonian of the
XXX-Heisenberg spin chain. However, it is not in the spin-$1/2$ representation,
making it less tractable from the Bethe ansatz viewpoint.
Moreover, the number of spanning trees is not known, even asymptotically.

\textit{\textbf{Lehmer code and adjacent transpositions.}}
Fix a permutation $\sigma \in S_{n}$ and its Lehmer code
$$
L(\sigma) = (k_{0}, k_{1}, \dots, k_{i}, k_{i+1}, \dots, k_{n-1}),
$$
where
$$
k_{i} = \#\{\, j>i : \sigma(i) > \sigma(j) \,\}, \qquad 0 \le k_{i} \le n-1-i.
$$

Consider the action of the adjacent transposition $(i,i+1)$ on $\sigma,$
for $i = 0, \dots, n-2.$

Each component $k_{i}$ is the number of inversions of the element $\sigma(i).$ When two elements of the permutation are swapped, we obtain a new permutation $\sigma'.$ The new Lehmer code $L(\sigma')$ can be obtained from the old one $L(\sigma)$ as follows.

The inversion numbers move together with the elements. After the transposition, the inversion numbers are written again according to the positions, as in the definition of the Lehmer code. In particular, the coordinates of the Lehmer vector stay in the same places; only the data attached to the elements move. 

Let
$$
a = \sigma(i), \qquad b = \sigma(i+1).
$$
The next step is to check the following cases.

\begin{itemize}
\item If $\sigma(i) > \sigma(i+1)$, then the pair $(a,b)$ is an inversion. After applying $(i,i+1)$, the element $a$ loses exactly one inversion. Hence the Lehmer code changes as
      $$
      L \longmapsto (k_{0}, k_{1}, \dots, k_{i+1}, k_{i}-1, \dots, k_{n-1}).
      $$

\item If $\sigma(i) < \sigma(i+1)$, then the pair $(a,b)$ is not an inversion. After applying $(i,i+1)$, the element $a$ gains exactly one inversion. Hence the Lehmer code changes as
      $$
      L \longmapsto (k_{0}, k_{1}, \dots, k_{i+1}+1, k_{i}, \dots, k_{n-1}).
      $$
\end{itemize}

In both cases, all other components of the Lehmer code with index $\ne i,i+1$ stay unchanged.

\begin{figure}[H]
\centering
\includegraphics[width=0.45\textwidth]{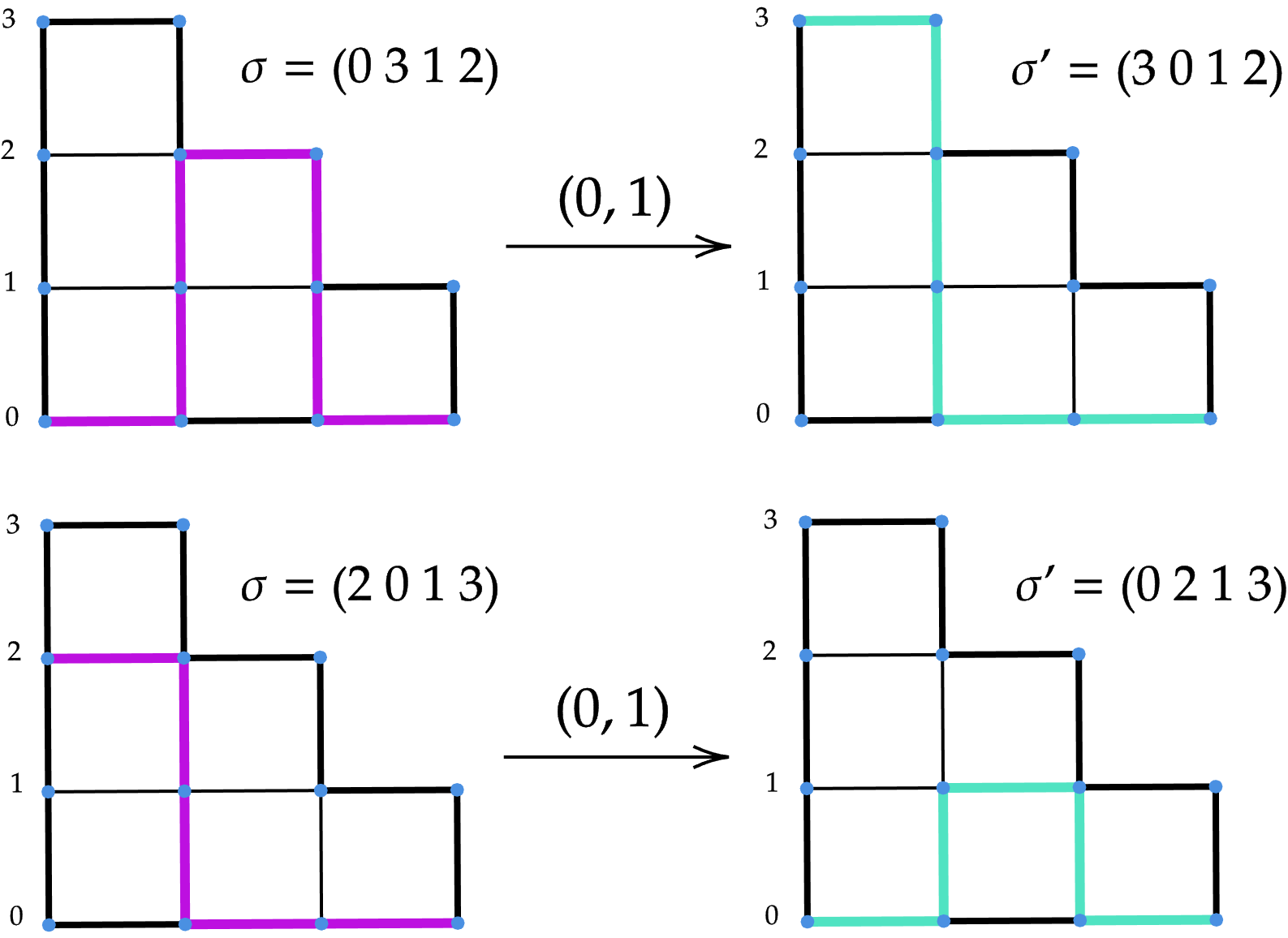}
\caption{Example of the adjacent transposition $(0,1)$. Top: $\sigma=(0,3,1,2)$, $\sigma'=(3,0,1,2)$; one inversion is added. Bottom: $\sigma=(2,0,1,3)$, $\sigma'=(0,2,1,3)$; one inversion is removed. Only the components $k_{0}$ and $k_{1}$ of the Lehmer code change. All other components remain unchanged.}
\label{fig:lehmer}
\end{figure}

\subsection{CayleyPy AI methodology for assisting in determining the duality}

To determine the dual polygon, we currently follow the procedure below.
First, we attempt to determine the diameters, then guess quasi-polynomials for them, and finally
fit polygons whose Ehrhart quasi-polynomials match. Determining the diameters is the most difficult step.
Our currently successful cases use the following workflow, based on the AI-assisted \texttt{CayleyPy} library.

\begin{enumerate}
    \item \textbf{BFS (Brute force).} Compute the entire Cayley graph for small $n$ using brute-force BFS
    (\texttt{CayleyPy} currently allows up to $n \leq 15$).
    
    \item \textbf{Guess the longest elements.} Examine the elements in the last layer
    (i.e., those whose word metric equals the diameter) and identify patterns that may generalize to arbitrary $n$.
    
    \item \textbf{AI pathfinding.} For larger $n$, take these candidate elements and apply AI-based algorithms
    to estimate their word metrics. The current version of \texttt{CayleyPy} allows exact computations (optimal paths)
    up to around $n = 30$ (graph size $\sim 10^{30}$). Extending this to $n = 40$ (graph size $\sim 10^{50}$)
    is a work in progress.
    
    \item \textbf{Fit quasi-polynomials.} Fit quasi-polynomial expressions to the resulting data.
    In particular, one needs to determine both the starting value of $n$ from which  the formula becomes valid
    and the period of the quasi-polynomial.
\end{enumerate}


The most difficult part of the pipeline is Step 2, where one must identify a pattern for the longest elements.
Whether this can be done in general remains unclear, but there are examples where it is possible,
such as for consecutive cycles and related generators -- discussed below.


\subsubsection{\texorpdfstring{$k$-Consecutive cycles $(i,i+1,i+2,\dots,i+k-1)$}{k-Consecutive cycles (i, i+1, i+2, ..., i+k-1)}}

Fix an integer $k$. For $n > k$, we consider elements of $S_n$ given by the cyclic permutations 
$(i, i+1, \dots, i+k-1)$ for $i = 0, \dots, n-k$. For $k = 2$, these are the neighbor transpositions 
(Coxeter or bubble sort generators) considered previously. For odd $k$, they generate $A_n$ 
inside $S_n$, for even $k$ they give $S_n$. There are two natural options: whether or not to include inverses in the generating set.
We consider both cases.
Here we present quasi-polynomial formulas for some of the diameters associated with these generators 
and briefly outline the shape of the dual polygon, leaving the determination of the holography map 
for future investigation. To the best of our knowledge, formulas for the diameters are new for $k > 3$.
From the physical point of view, the Laplacian of such graphs corresponds to the  situation
when $k$ neighboring spins interact, spin chains of that sort appear e.g. in AdS/CFT. 

\textbf{Informally. $k$ leads to $(k-1)$ shrinkage.}
Various effects for general $k$ compared to basic $k=2$ can be informally summarized  
as shrinkage $(k-1)$-times. For example, the diameter of the Cayley graph for $k=2$ is $n(n-1)/2$,
while the leading term for the diameters  for general $k$ is $n(n-1)/(2(k-1))$,
the same concerns various Schreier coset graphs and not only diameters, but also word metrics.

\textbf{Dual polygon as a $(k-1)$-shrinkage.}
For $k = 2$, the Cayley graph has as its dual polygon the triangle with vertices
$(0,0)$, $(n,0)$, and $(0,n)$. For general $k$, we expect the dual polygon to be close
to the triangle with vertices $(0,0)$, $(n/(k-1),0)$, and $(0,n)$, although exact results
have yet to be obtained.
Similarly, for various Schreier coset graphs associated with these generators, we expect
a $(k-1)$-fold shrinkage comparing to $k=2$ case,  along the $x$-axis.

\textbf{Conjecture on diameters of the coset graph.}
Consider the generating set given by consecutive $k$-cycles (without inverses), acting on
$\{0,1\}$-vectors with $\lfloor n/2 \rfloor$ zeros and $n - \lfloor n/2 \rfloor$ ones.
We conjecture that the diameters are as follows for large enough $n$:
\[
\begin{aligned}
n &\equiv 0 \!\pmod{(2k-2)} \;:& \!\!\!
D_k(n) &= \frac{n^2 + 4(k-1)}{4(k-1)}, \\[6pt]
n &\equiv -1 \!\pmod{(2k-2)} \;:& \!\!\!
D_k(n) &=  \frac{n^2 + 4(k-1)-1}{4(k-1)},
\end{aligned}
\]
and else, denoting $p = n \pmod{(2k-2)} $,
\[
\begin{aligned}
 D_k(n)= &  \frac{n(n+2k-p-2)}{4(k-1)}, & p  {\rm \ is \ even \ }
\\[6pt]
 D_k(n)= &   \frac{(n-1)(n+2k-p-2)}{4(k-1)}, & p  {\rm \ is \ odd \ }.
\end{aligned}
\]

\textbf{``Riemann conjecture'' for $H$-polynomials.}
We computed quasi-polynomials and $H$-polynomials for the coset graphs discussed above,
as well as for their modifications with inverse-closed generating sets.
We observe that for small $k$ (up to $k=6$) all zeros of the $H$-polynomials in the inverse-closed case
lie on the unit circle, whereas this is not the case for the non–inverse-closed generators.
In this sense, the ``Riemann conjecture'' \cite{RodriguezVillegas2002},\cite{BumpChoiKurlbergVaaler2000} holds in the inverse-closed setting.
It is unclear whether this pattern persists for larger $k$.
For example, for the full Cayley graph with inverse closed generators the ``Riemann conjecture'' holds up to $k=5$,
but already fails for $k=6$.

\begin{figure}[H]
	\centering
	\includegraphics[width=0.45\textwidth]{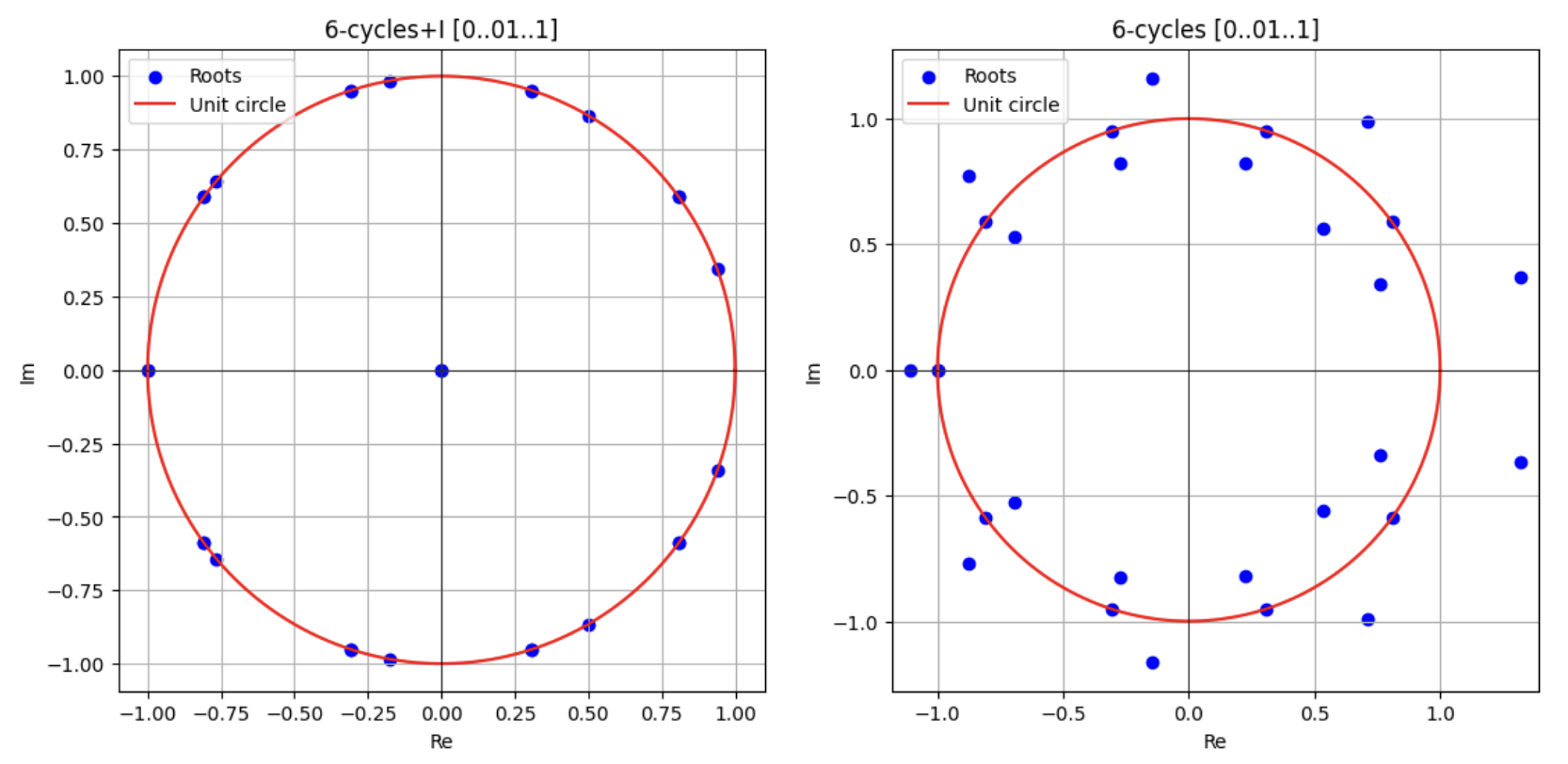}
	\caption{Roots of the $H$-polynomials in the complex plane.
In the left panel, all non-trivial roots have modulus equal to one, so the ``Riemann conjecture'' holds,
whereas in the right panel it is violated.
Both panels correspond to $k$-consecutive cycle generators with $k=6$:
the left uses an inverse-closed generating set, while the right uses generators without inverses.
These correspond to Schreier coset graphs for the subgroup $S_{\lfloor n/2 \rfloor} \times S_{n-\lfloor n/2 \rfloor}$.
Whether the property that all roots have modulus one persists for larger $k$ remains unclear. }
	\label{fig:6cycles_coset0then1}
\end{figure}

\textbf{Conjectures on diameters of the Cayley graph.}
Consider consecutive $k$-cycle generators including inverses and consider the Cayley graph.
(Not the Schreier coset graph as above).
We conjecture: the diameters are given by
$
D_k(n) = \left\lfloor \frac{n(n-1)}{2(k-1)} \right\rfloor + Q_0(n),
$
where $Q_0(n)$ is a periodic function of $n$, that is, a quasi-polynomial of degree zero.
It should be understood as a small correction to the leading term
$\left\lfloor \frac{n(n-1)}{2(k-1)} \right\rfloor$.
The period is $(k-1)$ for even $k$ and $2(k-1)$ for odd $k$.
The quasi-polynomial formulas are valid for $n \ge 2k$.
Explicit quasi-polynomials for $k \le 7$ are presented in the Supplementary Material.

It is worth looking at the corresponding H-polynomial,
e.g. for $k=6$ it is:  $x^{14} + 2x^{13} + 2x^{12} + 3x^{11} + 4x^{10} + 2x^9 + x^8 + 2x^7 + x^6
+ 2x^4 + 2x^3 + x^2 + x + 1$, one can observe it is not unimodal.

\subsection{\texorpdfstring{$SL(2,\mathbb{Z})$ and the Farey graph}{SL(2,Z) and the Farey graph}}

Here we reinterpret the results of by one the present authors (D.Melnikov et.al.~\cite{CamiloMelnikovNovaesPrudenziati2019}) in a way closer to the
present exposition and to the principles of holography and ``complexity $=$ area''.
It was argued there that to each element of $SL(2,\mathbb{Z})$ one may associate a curve in the
hyperbolic plane whose enclosed area approximately corresponds to the complexity of the element
with respect to the standard generators $S$ and $T$.
In this language, the map from a group element (i.e.\ a node of the Cayley graph) to such a curve
may be viewed as a holography map, in agreement with our framework.
The cited results therefore imply that a variant of the ``complexity $=$ area'' principle holds,
at least approximately.
An illustration is shown in Figure~\ref{fig:fareytesselation}.



\begin{figure}[H]
	\centering
	\includegraphics[width=0.45\textwidth]{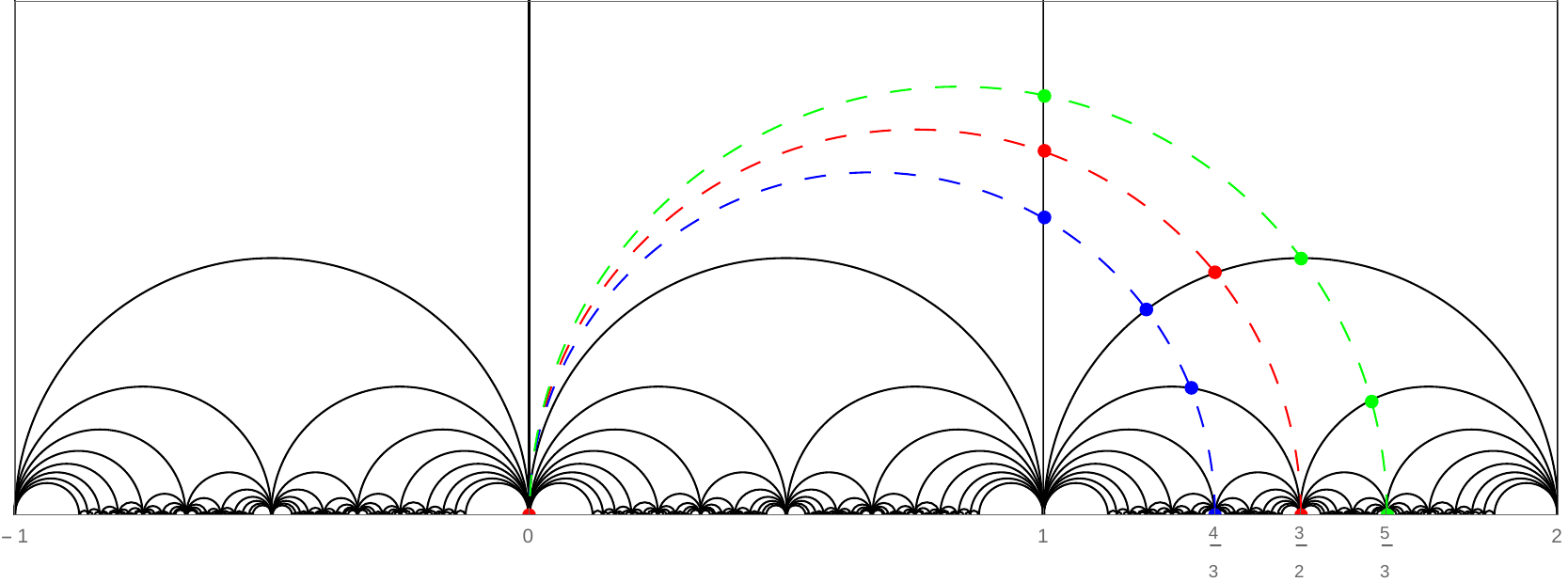}
	\caption{Farey graph on the upper halfplane. From~\cite{CamiloMelnikovNovaesPrudenziati2019}.
    The dashed curves are ``holography images"  of some elements of  $SL(2,\mathbb{Z})$ 
    areas under them correspond to complexities of these elements. ``Complexity = area" holds approximately. }
	\label{fig:fareytesselation}
\end{figure}

\clearpage

\section{Neighbor transpositions Cayley and Schreier graphs}

\subsection{\texorpdfstring{Multiset $\{0^{n_0}\,1^{n_1}\,\dots\,m^{n_m}\}$}{Multiset {0^{n_0} 1^{n_1} ... m^{n_m}}}}

In this section we show how to build polygons for vertices of graphs generated by mulitsets.
\subsubsection*{\texorpdfstring{An example for $\{0^4\,1^3\,2^5\,3^2\}$.}{An example for {0^4 1^3 2^5 3^2}.}}

\begin{figure}[htbp]
	\centering
	\includegraphics[width=0.3\textwidth]{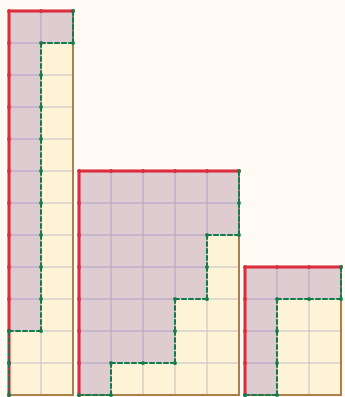}
	\caption{Multiset $\{0^4\,1^3\,2^5\,3^2\}$, i.e.\ $(n_0,n_1,n_2,n_3)=(4,3,5,2)$ ; $N=4+3+5+2=14$.
		The rectangles from left to right are $R_0$ (size $2\times 12$), $R_1$ ($5\times 7$), and $R_2$ ($3\times 4$).
		$E=[0,0,0,0,1,1,1,2,2,2,2,2,3,3]$, $P=[2,1,3,2,2,0,0,2,0,1,2,1,3,0]$,
		and $d(P,E)=A_0+A_1+A_2=11+25+6=42$.}
	\label{fig:three_rectangles_v2b.png}
\end{figure}

We illustrate the construction on a graph $G$ whose vertices are permutations of the multiset
$\{0^4\,1^3\,2^5\,3^2\}$, with edges given by adjacent transpositions $(i,i+1)$.
Let
$E=[0,0,0,0,1,1,1,2,2,2,2,2,3,3]$.
The construction produces a polygon $f(E)$ such that, for any vertex $P$ in $G$, the polygon $f(P)$ is a subset of $f(E)$.
Moreover, the bubble-sort distance from $P$ to $E$ equals the area difference
$\operatorname{area}(f(E))-\operatorname{area}(f(P))$.

Let $(n_0,n_1,n_2,n_3)=(4,3,5,2)$. Consider the rectangles
\begin{itemize}
	\item $R_0$ of size $n_3\times(n_0+n_1+n_2)=2\times(4+3+5)=2\times 12$ (area $24$),
	\item $R_1$ of size $n_2\times(n_0+n_1)=5\times(4+3)=5\times 7$ (area $35$),
	\item $R_2$ of size $n_1\times n_0=3\times 4$ (area $12$).
\end{itemize}
Define $f(E)$ as the union of these rectangles, so that
$\operatorname{area}(f(E))=24+35+12=71$; see Figure~\ref{fig:three_rectangles_v2b.png}.
Given
$P=[2,1,3,2,2,0,0,2,0,1,2,1,3,0]$,
define
\begin{itemize}
	\item $P_0=P=[2,1,3,2,2,0,0,2,0,1,2,1,3,0]$,
	\item $P_1=P_0-\{3\}=[2,1,2,2,0,0,2,0,1,2,1,0]$,
	\item $P_2=P_1-\{2\}=[1,0,0,0,1,1,0]$.
\end{itemize}

Next, in each rectangle $R_i$ we draw a green lattice path from the lower-left corner to the upper-right corner,
moving one step either up or right at each move.
The choice of step is determined by the symbols of $P_i$, read from left to right, as follows:
\begin{itemize}
	\item In $R_0$, read $P_0$: if the symbol is $3$, step right; otherwise step up.
	\item In $R_1$, read $P_1$: if the symbol is $2$, step right; otherwise step up.
	\item In $R_2$, read $P_2$: if the symbol is $1$, step right; otherwise step up.
\end{itemize}

Let $f(P)$ be the union of the regions below these green paths. Then $f(P)\subseteq f(E)$.
Furthermore, if $A_i$ denotes the area in $R_i$ between the corresponding red and green curves, then
the bubble-sort distance satisfies
\[
d(P,E)=\operatorname{area}(f(E))-\operatorname{area}(f(P))=A_0+A_1+A_2=11+25+6=42,
\]
as shown in Figure~\ref{fig:three_rectangles_v2b.png}.

\subsubsection*{Generalization for an arbitrary multiset}
Given a multiset $\{0^{n_0}\,1^{n_1}\,\dots\,m^{n_m}\}$. Let
$N = n_0+n_1+\cdots+n_m$, and let $P=\pi(E)$ for some permutation $\pi\in S_N$, where
$$
E=\big[
\underbrace{0,\dots,0}_{n_0},\ 
\underbrace{1,\dots,1}_{n_1},\ 
\dots,\
\underbrace{m,\dots,m}_{n_m}
\big].
$$

Set $P_0=P=\pi(E)$ and, for $k=1,\dots,m-1$, define $P_k$ to be the sequence obtained from $P_{k-1}$ by deleting all occurrences of the letter $m-k+1$.
Note that $P_k$ contains only letters from $\{0,1,2,\dots,m-k\}$.




\color{black}

For each $k=0,\dots,m-1$, consider  rectangle $R_k$ of size
$n_{m-k}\times\left(\sum_{i=0}^{m-k-1}n_i\right)$.
Inside $R_k$ we draw two lattice paths (red and green); both start at the lower-left corner.
The red path goes straight up to the top edge and then straight right to the upper-right corner.

The green path is constructed by reading $P_k$ from left to right:
if the current letter is $(m-k)$, move one step to the right; otherwise, move one step up.
After $|P_k|$ steps, the path reaches the upper-right corner.
Let $A_k$ denote the area between the red and green paths in $R_k$.
Then   the total bubble-sort swap steps converting $P$ to $E$ is
$$
d(P,E)=\sum_{k=0}^{m-1} A_k.
$$

An example of this construction for the multiset $\{0^4\,1^3\,2^5\,3^2\}$ is shown in Figure~\ref{fig:three_rectangles_v2b.png}.

\subsubsection*{\texorpdfstring{More details for the multi-set $\{0^{n_0}\, 1^{n_1}\, \dots \, m^{n_m}\}$}{More details for the multiset {0,1,...,m} with multiplicities n0,n1,...,nm}}


$$
\begin{array}{c|ccccc}
	\hline\text{count}  & n_0& n_1 & \dots  & n_m \\
	\hline\text{letter} & 0 & 1& \dots & m\\
	\hline
\end{array}
$$

Let's fix notation for vectors
\begin{itemize}
	\item $P_0=P=\pi(E)$ contains letters $\{0,1,2,\dots, m-2,m-1, m\}$
	\item $P_1=P_0 - \{m\}$ contains letters $\{0,1,2,\dots, m-2,m-1\}$
	\item $P_2=P_1 - \{m-1\}$ contains letters $\{0,1,2,\dots, m-2\}$
	\item $\dots$
	\item $P_{m-2}=P_{m-1} - \{3\}$ contains letters $\{0,1,2\}$
	\item $P_{m-1}=P_{m-2} - \{2\}$ contains letters $\{0,1\}$
\end{itemize}

$$
\begin{array}{l|lllll} 
	\hline \text{Rectangle} & R_0 & R_1 & \dots & R_{m-1} \\
	\hline\text{Rectangle Size}  & n_m \times \sum_{i=0}^{m-1} n_i & 
	n_{m-1} \times \sum_{i=0}^{m-2} n_i & \dots
	& n_1 \times n_0 \\
	\hline\text{Vector} & P_0 & P_1 & \dots &P_{m-1} \\	     
	\hline\text{Curve}  & m \,\text{vs}\, \{0, \dots, (m-1)\} & 
	(m-1)  \,\text{vs}\,  \{0,\dots,(m-2)\} & \dots &
	1  \,\text{vs}\,  0 \\	                                 
	\hline
\end{array}
$$

\begin{itemize}
	\item
	For rectangle $R_0$, read the word $P_0=P$ from left to right. If you see $m$, take a step to the right; otherwise move up.
	\item
	For rectangle $R_1$, read the word $P_1 $ from left to right.
	If you see $m-1$, take a step to the right; otherwise move up.
	\item $\dots$
	\item
	For rectangle $R_{m-1}$, read the word $P_{m-1}$ from left to right.
	If you see $1$, take a step to the right; otherwise move up.
\end{itemize}

\subsection{Cosets with 0,1,2 components - alternative representation}
\textbf{Coset space as an orbit of words.}

	Throughout this section, cosets refers to the $S_n$–orbits of words with fixed symbol multiplicities, equivalently to the Schreier coset space $S_n / (S_{\lambda_0}\times S_{\lambda_1}\times S_{\lambda_2})$ under the action by adjacent transpositions.
	Let $E=0^{\lambda_0}1^{\lambda_1}2^{\lambda_2}$.
	The Schreier coset space
	can be identified with the $S_n$--orbit of $E$, i.e.\ the set of all words of length $n$ 
	with exactly $\lambda_0$ $0$'s, $\lambda_1$ $1$'s, and $\lambda_2$ $2$'s:
	\[
	\begin{aligned}
		\mathcal O&=\{\,W\in\{0,1,2\}^n:\#0=\lambda_0,\#1=\lambda_1,\#2=\lambda_2\,\},\\
		|\mathcal O|&=\frac{n!}{\lambda_0!\,\lambda_1!\,\lambda_2!}.
	\end{aligned}
	\]
	%
	Let $s_i=(i\ i{+}1)\in S_n$ be the adjacent transposition. It acts on a word
	$W=w_1\dots w_n$ by swapping adjacent letters:
	$
	s_i\cdot W = s_i\cdot (w_1 \dots w_i w_{i+1}\dots w_n)=(w_1 \dots  w_{i+1} w_i \dots w_n).
	$
	Thus, in the Schreier graph picture, vertices are words in $\mathcal O$ and edges connect
	$W$ to $s_i\cdot W$ for each adjacent transposition $s_i$ (ignoring loops when the swap does not change the word).

	Permuting positions among equal symbols does not change the word, so the stabilizer  in $S_n$ is
	$	H \;=\; S_{\lambda_0}\times S_{\lambda_1}\times S_{\lambda_2}$.
	To a word $W\in 	\mathcal O$ we associate a lattice path $P(W)$ starting at the point $(0,0)$,
	using the following encoding of steps:
	\[
	0:\ (x,y)\mapsto(x+1,y),\qquad
	1:\ (x,y)\mapsto(x,y+1),\qquad
	2:\ (x,y)\mapsto(x-1,y+1).
	\]
	The reference word $E$ is encoded analogously, yielding a path $P(E)$.
	The two paths share the same start and terminal points.

\begin{Claim} 
		Let $S(W,E)$ denote the oriented area of the polygon enclosed by the paths
	$P(W)$ and $P(E)$.
	Then	the area $S(W,E)$ is equal to the number of bubble sort operations
		required to transform the word $W$ into the reference word $E$,
		that is, the minimal number of adjacent transpositions needed to sort
		the symbols of $W$. 
\end{Claim}	

	\begin{proof}[Sketch]
		Let $W$ be a word over the alphabet $\{0,1,2\}$ with fixed multiplicities
		$\lambda_0,\lambda_1,\lambda_2$, and let
		\[
		E = 0^{\lambda_0}1^{\lambda_1}2^{\lambda_2}
		\]
		be the sorted reference word.  
		It is classical that the minimal number of adjacent transpositions required to
		transform $W$ into $E$ equals the total number of inversions with respect to
		the order $0<1<2$,
		\[
		\Inv(W)=\#\{(i,j): i<j,\ W_i>W_j\}.
		\]
		
		The inversions decompose into three disjoint types,
		\[
		\Inv(W)=\Inv_{10}(W)+\Inv_{20}(W)+\Inv_{21}(W),
		\]
		where $\Inv_{ab}(W)$ counts occurrences of the pattern $a$ preceding $b$ with
		$a>b$.
		
		The lattice path encoding
		\[
		0:(1,0),\qquad 1:(0,1),\qquad 2:(-1,1)
		\]
		admits a corresponding decomposition of the oriented area
		$S(W,E)$ into three independent contributions.  
		Indeed, projecting the path $P(W)$ onto the directions associated with the pairs
		$(1,0)$, $(2,0)$, and $(2,1)$ yields three binary lattice paths whose enclosed
		areas coincide with $\Inv_{10}(W)$, $\Inv_{20}(W)$, and $\Inv_{21}(W)$,
		respectively.  
		By the binary case, each such area equals the corresponding inversion count.
		
		Summing over the three pairs, we obtain
		\[
		S(W,E)=\Inv_{10}(W)+\Inv_{20}(W)+\Inv_{21}(W)=\Inv(W).
		\]
		Since bubble sort (or insertion sort) performs exactly one adjacent transposition
		per inversion, the minimal number of adjacent swaps required to transform $W$
		into $E$ equals $S(W,E)$, as claimed.
	\end{proof}

	\color{black}

	Hence, the word metric induced by adjacent swaps admits a geometric
	interpretation as the area between the corresponding lattice paths.

	\begin{example}
		Consider
		\[
		E=[0,1,2],\qquad P=[2,1,0].
		\]
		
		In Fig.~\ref{fig:auc_words_012_210}, the words 
		$012$ and 
		$210$ are encoded by lattice paths, shown in blue and red, respectively. The area enclosed by the polygon formed by these two paths equals $S=3$.
		
	\end{example}
	
	\begin{figure}[ht]
		\centering
		\begin{tikzpicture}[scale=1.2,>=stealth]
			
			
			\coordinate (A0) at (0,0);
			\coordinate (A1) at (1,0);   
			\coordinate (A2) at (1,1);   
			\coordinate (A3) at (0,2);   
			
			\coordinate (B0) at (0,0);
			\coordinate (B1) at (-1,1);  
			\coordinate (B2) at (-1,2);  
			\coordinate (B3) at (0,2);   
			
			\fill[black!15]
			(A0) -- (A1) -- (A2) -- (A3) --
			(B2) -- (B1) -- cycle;
			
			\draw[very thick,blue]
			(A0) -- (A1) -- (A2) -- (A3);
			
			\draw[very thick,red]
			(B0) -- (B1) -- (B2) -- (B3);
			
			\fill (0,0) circle (2pt) node[below right] {$(0,0)$};
			\fill (0,2) circle (2pt) node[above] {$(0,2)$};

			\node at (0.1,1) {\Large $S=3$};
			
		\end{tikzpicture}
		\caption{Encoding the words $012$ and $210$ by the lattice paths.
			The distance between the words coincides with the oriented area $S$ between the associated paths.}
		
		\label{fig:auc_words_012_210}
	\end{figure}
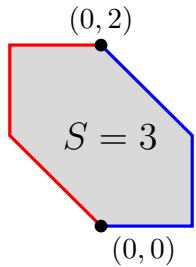
	
	In Fig.~\ref{fig:auc_words_012}, the same encoding is shown for several other pairs of words:
	(a) for the words \(012\) and \(102\), the enclosed area equals \(S=1\);
	(b) for the words \(012\) and \(021\), the enclosed area equals \(S=1\);
	(c) for the words \(012\) and \(120\), the enclosed area equals \(S=2\);
	(d) for the words \(012\) and \(201\), the enclosed area equals \(S=2\).

	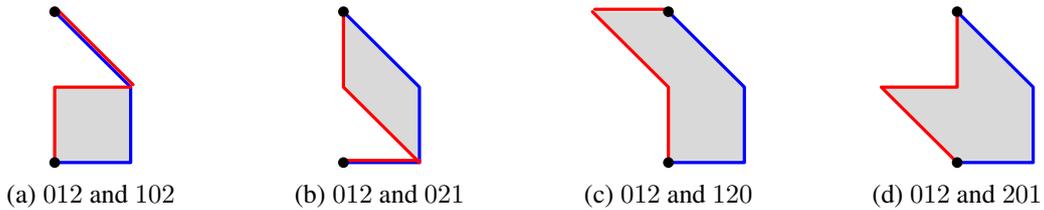
\begin{figure}[ht]
		\centering
		
		\begin{minipage}[t]{0.24\textwidth}
			\centering
			\begin{tikzpicture}[scale=1.0, line cap=round, line join=round]
				
				\tikzset{
					bluepath/.style={very thick, blue},
					redpath/.style={very thick, red},
					area/.style={black!15}
				}
				
				\coordinate (O)  at (0,0);
				\coordinate (A1) at (1,0);
				\coordinate (A2) at (1,1);
				\coordinate (E)  at (0,2);
				
				\coordinate (B1) at (0,1);
				\coordinate (B2) at (1,1);
				
				\fill[area] (O)--(A1)--(A2)--(E)--(B2)--(B1)--cycle;
				
				\draw[bluepath] (O)--(A1)--(A2)--(E);
				\draw[redpath]  (O)--(B1)--(B2);
				
				\begin{scope}[transform canvas={shift={(0.9pt,0.9pt)}}]
					\draw[redpath] (B2)--(E);
				\end{scope}
				
				\fill (0,0) circle (2pt) ;
				\fill (0,2) circle (2pt) ;
				
			\end{tikzpicture}
			
			\centerline{\small (a) $012$ and $102$}
		\end{minipage}
		\hfill
		\begin{minipage}[t]{0.24\textwidth}
			\centering
			\begin{tikzpicture}[scale=1.0, line cap=round, line join=round]
				
				\tikzset{
					bluepath/.style={very thick, blue},
					redpath/.style={very thick, red},
					area/.style={black!15}
				}
				
				\coordinate (O)  at (0,0);
				\coordinate (A1) at (1,0);
				\coordinate (A2) at (1,1);
				\coordinate (E)  at (0,2);
				\coordinate (B2) at (0,1);
				
				\fill[area] (A1)--(A2)--(E)--(B2)--cycle;
				
				\draw[bluepath] (O)--(A1)--(A2)--(E);
				
				\begin{scope}[transform canvas={shift={(0,0.8pt)}}]
					\draw[redpath] (O)--(A1);
				\end{scope}
				\draw[redpath] (A1)--(B2)--(E);
				
				\fill (0,0) circle (2pt) ;
				\fill (0,2) circle (2pt) ;
			\end{tikzpicture}
			
			\centerline{\small (b) $012$ and $021$}
		\end{minipage}
		\hfill
		\begin{minipage}[t]{0.24\textwidth}
			\centering
			\begin{tikzpicture}[scale=1.0, line cap=round, line join=round]
				
				\tikzset{
					bluepath/.style={very thick, blue},
					redpath/.style={very thick, red},
					area/.style={black!15}
				}
				
				\coordinate (O)  at (0,0);
				\coordinate (A1) at (1,0);
				\coordinate (A2) at (1,1);
				\coordinate (E)  at (0,2);
				
				\coordinate (B1) at (0,1);
				\coordinate (B2) at (-1,2);
				
				\fill[area] (O)--(A1)--(A2)--(E)--(B2)--(B1)--cycle;
				
				\draw[bluepath] (O)--(A1)--(A2)--(E);
				\draw[redpath]  (O)--(B1)--(B2);
				
				\begin{scope}[transform canvas={shift={(0.9pt,0.9pt)}}]
					\draw[redpath] (B2)--(E);
				\end{scope}
				
				\fill (0,0) circle (2pt) ;
				\fill (0,2) circle (2pt) ;
				
			\end{tikzpicture}
			
			\centerline{\small (c) $012 $ and $ 120$}
		\end{minipage}
		\hfill
		\begin{minipage}[t]{0.24\textwidth}
			\centering
			\begin{tikzpicture}[scale=1.0, line cap=round, line join=round]
				
				\tikzset{
					bluepath/.style={very thick, blue},
					redpath/.style={very thick, red},
					area/.style={black!15}
				}
				
				\coordinate (O)  at (0,0);
				\coordinate (A1) at (1,0);
				\coordinate (A2) at (1,1);
				\coordinate (E)  at (0,2);
				
				\coordinate (B1) at (-1,1);
				\coordinate (B2) at (0,1);
				
				\fill[area] (O)--(A1)--(A2)--(E)--(B2)--(B1)--cycle;
				
				\draw[bluepath] (O)--(A1)--(A2)--(E);
				\draw[redpath]  (O)--(B1)--(B2)--(E);
				\fill (0,0) circle (2pt) ;
				\fill (0,2) circle (2pt) ;
				
			\end{tikzpicture}
			
			\centerline{\small (d) $012$ and $201$}
		\end{minipage}
		
		\caption{Encoding the words by the lattice paths.}
		
		\label{fig:auc_words_012}
	\end{figure}

	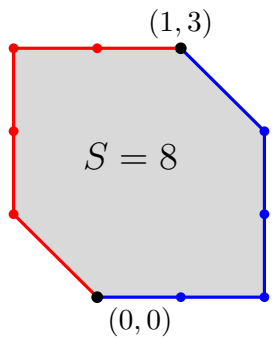
\begin{figure}[ht]
		\centering
		\begin{tikzpicture}[scale=1.1, line cap=round, line join=round]
			
			\tikzset{
				bluepath/.style={very thick, blue},
				redpath/.style={very thick, red},
				area/.style={black!15},
				vblue/.style={circle, fill=blue, inner sep=1.3pt},
				vred/.style={circle, fill=red, inner sep=1.3pt}
			}
			
			\coordinate (O)  at (0,0);
			\coordinate (A1) at (1,0);   
			\coordinate (A2) at (2,0);   
			\coordinate (A3) at (2,1);   
			\coordinate (A4) at (2,2);   
			\coordinate (E)  at (1,3);   
			
			\coordinate (B1) at (-1,1);  
			\coordinate (B2) at (-1,2);  
			\coordinate (B3) at (-1,3);  
			\coordinate (B4) at (0,3);   
			
			\fill[area]
			(O)--(A1)--(A2)--(A3)--(A4)--(E)--
			(B4)--(B3)--(B2)--(B1)--cycle;
			
			\draw[bluepath] (O)--(A1)--(A2)--(A3)--(A4)--(E);
			\draw[redpath]  (O)--(B1)--(B2)--(B3)--(B4)--(E);
			
			\node[vblue] at (O)  {};
			\node[vblue] at (A1) {};
			\node[vblue] at (A2) {};
			\node[vblue] at (A3) {};
			\node[vblue] at (A4) {};
			\node[vblue] at (E)  {};
			
			\node[vred]  at (O)  {};
			\node[vred]  at (B1) {};
			\node[vred]  at (B2) {};
			\node[vred]  at (B3) {};
			\node[vred]  at (B4) {};
			\node[vred]  at (E)  {};
			
			\node at (0.4,1.7) {\Large $S=8$};
			
			\fill (0,0) circle (2pt) node[below right] {$(0,0)$};
			\fill (1,3) circle (2pt) node[above] {$(1,3)$};
			
		\end{tikzpicture}
		\caption{Encoding of the words $00112$ (blue) and $21100$ (red) by lattice paths.}
		\label{fig:auc_words_00112_21100}
	\end{figure}

	\begin{figure}[ht]
		\centering
		\begin{tikzpicture}[scale=1.05, line cap=round, line join=round]
			
			\tikzset{
				bluepath/.style={very thick, blue},
				redpath/.style={very thick, red},
				area/.style={black!15},
				vblue/.style={circle, fill=blue, inner sep=1.3pt},
				vred/.style={circle, fill=red, inner sep=1.3pt}
			}
			
			\coordinate (O)  at (0,0);
			\coordinate (A1) at (1,0);
			\coordinate (A2) at (2,0);
			\coordinate (A3) at (2,1);
			\coordinate (A4) at (2,2);
			\coordinate (A5) at (1,3);
			\coordinate (E)  at (0,4);
			
			\coordinate (B1) at (-1,1);
			\coordinate (B2) at (-2,2);
			\coordinate (B3) at (-2,3);
			\coordinate (B4) at (-2,4);
			\coordinate (B5) at (-1,4);
			
			\fill[area]
			(O)--(A1)--(A2)--(A3)--(A4)--(A5)--(E)--
			(B5)--(B4)--(B3)--(B2)--(B1)--cycle;
			
			\draw[bluepath] (O)--(A1)--(A2)--(A3)--(A4)--(A5)--(E);
			\draw[redpath]  (O)--(B1)--(B2)--(B3)--(B4)--(B5)--(E);
			
			\node[vblue] at (O)  {};
			\node[vblue] at (A1) {};
			\node[vblue] at (A2) {};
			\node[vblue] at (A3) {};
			\node[vblue] at (A4) {};
			\node[vblue] at (A5) {};
			\node[vblue] at (E)  {};
			
			\node[vred]  at (O)  {}; 
			\node[vred]  at (B1) {};
			\node[vred]  at (B2) {};
			\node[vred]  at (B3) {};
			\node[vred]  at (B4) {};
			\node[vred]  at (B5) {};
			\node[vred]  at (E)  {}; 
			
			\node at (0.2,2.2) {\Large $S=12$};
			
			\fill (0,0) circle (2pt) node[below right] {$(0,0)$};
			\fill (0,4) circle (2pt) node[above] {$(0,4)$};
			
		\end{tikzpicture}
		\caption{Encoding of the words $001122$ (blue) and $221100$ (red) by lattice paths.}
		\label{fig:auc_words_001122_221100}
	\end{figure}
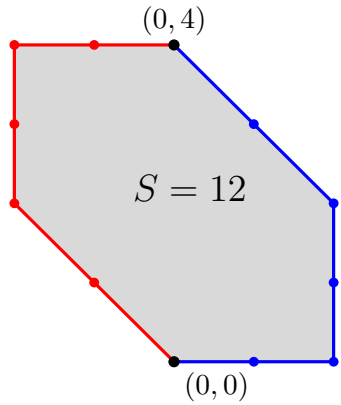

\subsection{Box-Ball System evolution as a deterministic walk on a Schreier/Cayley graph and the AUC--AAC duality}
	
	The \emph{box--ball system} (BBS) is a simple discrete dynamical system on a row of boxes, each either empty ($0$) or containing a ball ($1$). The system evolves over discrete time steps according to a deterministic rule.
	\begin{figure}[htbp]
		\centering
			\includegraphics[width=0.9\textwidth]{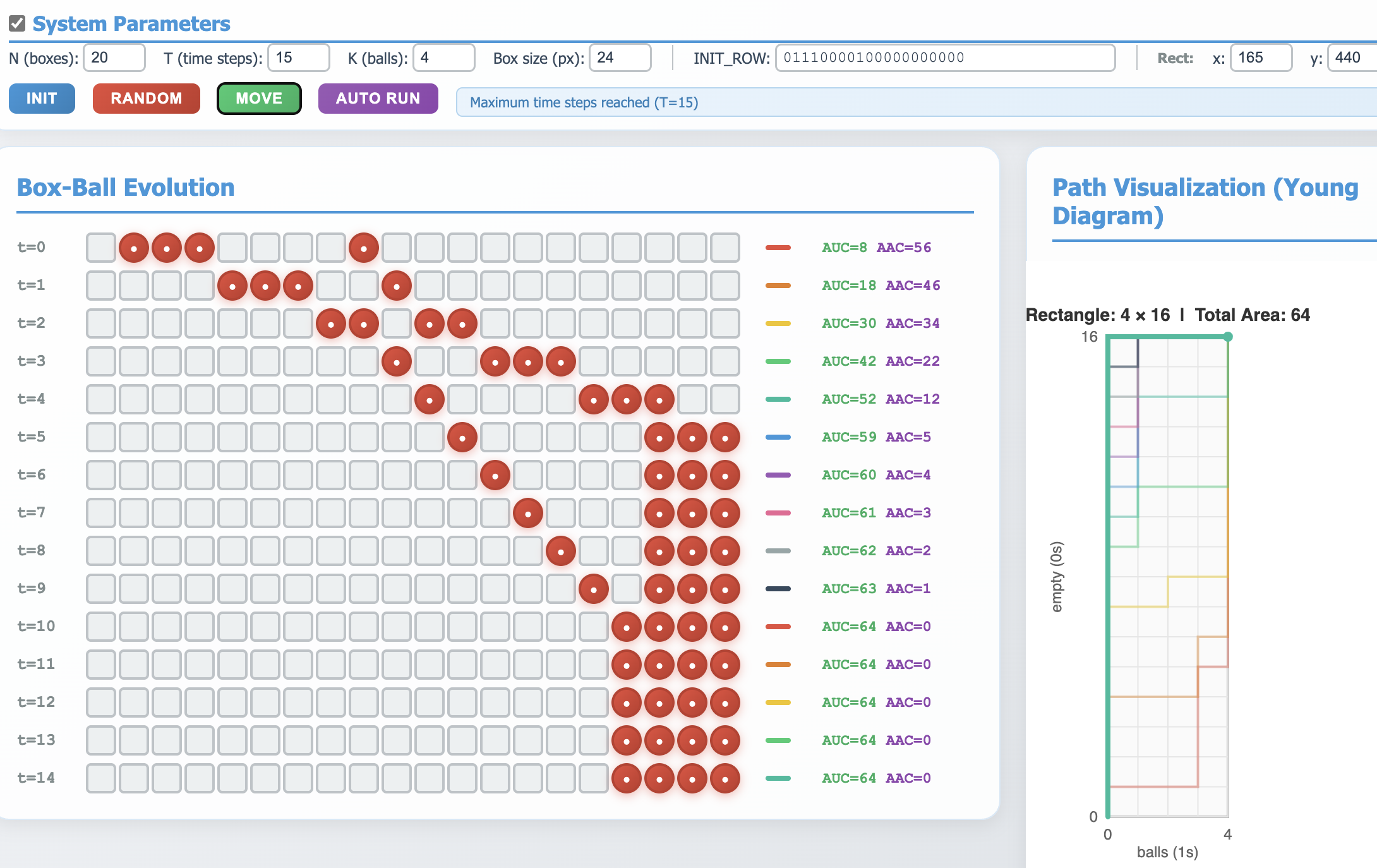}
		\caption{Box-Ball System Illustration. 
			Example for  $K=4$ balls and $N=20$ boxes.
		}
		\label{fig:bbs.png}
	\end{figure}
	At each time step, balls are processed left to right. Each ball moves to the nearest empty box to its right that hasn't already been claimed by another ball in this step. If no such empty box exists, the ball stays put. This rule produces soliton-like behavior -- clusters of balls propagate to the right and interact like solitons (they pass through each other and preserve their sizes).
	
	\href{https://fkhafizov.github.io/animation/bbs_opus45_v3.html}{Box-ball demo application} illustrates this 
	concept. 
	Figure~\ref{fig:bbs.png} shows an orbit of the box--ball  update on the configuration space
	$\Omega_{N,K}=\{x\in\{0,1\}^N:\sum_i x_i=K\}$ (here $N=20$, $K=4$), displayed as rows $t=0,1,\dots,T$. Observe that $N, K, T$ are the system parameters configurable in the application.
	Viewing $\Omega_{N,K}$ as the vertex set of the Schreier graph obtained from the Cayley graph of $S_N$ with generators
	$s_i=(i\ i{+}1)$, each time step is a deterministic walk: the global ``MOVE'' map sends the current vertex $x$ to a new vertex
	$T(x)$, and this map can be realized (conceptually) as a word in adjacent transpositions that performs the needed 
	swaps.
	
	\emph{Young-diagram / lattice-path encoding}. 
	On the right, each row is encoded as a monotone lattice path in the $K\times(N-K)$ rectangle by interpreting entries as follows:	0 (empty box) → move UP; 	1 (ball) → move RIGHT.
	%
%
%
After starting  at $(0,0)$, reading the whole word, the endpoint is $(K,N-K)$, so the path indeed fits the rectangle.
The displayed statistics satisfy the duality
	\[
	\mathrm{AUC}(x)+\mathrm{AAC}(x)=K(N-K)=64,
	\]
	so $\mathrm{AUC}$ and $\mathrm{AAC}$ are complementary potential functions.  In the Cayley/Schreier picture, a single adjacent swap
	$01\leftrightarrow 10$ flips a local corner of the path and changes $\mathrm{AUC}$ by $\pm1$, making area a height function on the graph.
	
	In this particular run, $\mathrm{AUC}$ increases monotonically from $8$ to $64$ while $\mathrm{AAC}$ decreases to $0$, and the final
	state is the extremal sorted word $0^{N-K}1^K$ (all balls packed to the right), whose path maximizes area under the curve.
	Thus the dynamics is visibly flowing toward a distinguished vertex in the Schreier/Cayley geometry, while the AUC/AAC complementarity
	provides the dual (above/below) viewpoint on the same walk.
	
\subsection{Schreier graphs for orbits of tuples}

In this section, we determine the diameters of a certain class of Schreier graphs that generalize the Cayley graph of $S_n$ for Coxeter generators (i.e. the permutohedron) and the $0,1,2$ coset graphs seen earlier. Unlike Cayley graphs, Schreier graphs need no longer be vertex-transitive, so it is not immediately clear the diameter is always the distance to a fixed vertex. For our chosen class of graphs, we explicitly describe the pair of vertices realizing the diameter, as well as compute it based on the graph isomorphism type.

 Let $X_n : = \{0,1, \dots, n-1\}^n$ be the set of $n$-tuples with entries in a finite set of order $n$. The symmetric group $S_n$ acts on $X_n$ by permutation of the components. The stabilizer of each $x \in X_n$ under this action has the form $G_x \cong  G_{\lambda} := S_{\lambda_0} \times \cdots \times S_{\lambda_m}$, where $\lambda = \{\lambda_0, \dots, \lambda_m\}$ and $\lambda_0 + \cdots + \lambda_m = n$ (i.e. $\lambda$ is a partition of $n$). Let $X_{n, \lambda} \subset X_n$ consist of all $n$-tuples with stabilizer isomorphic to $G_{\lambda}$.

Let $S \subset S_n$ be the set of Coxeter generators of $S_n$. Multiplication by elements of $G_x$ (on the right) induces an action on $\Cay(S_n,S)$ by graph automorphisms. Note that the Schreier graph $\Sch(S_n,G_x, S)$ of (left) cosets is obtained as the quotient graph of $\Cay(S_n,S)$ with respect to this action. 

Moreover, the vertices of $\Sch(S_n,G_x, S)$ may be identified with elements of the $S_n$-orbit of $x$, with two vertices connected by an edge if and only if there is a transposition $(i, i+1)$ sending one $n$-tuple to the other. It is not hard to see that all graphs $\Sch(S_n,G_x, S)$ for $x \in X_{n, \lambda}$ are isomorphic. Therefore, from now on, we will assume $x = [0,\dots,0,1, \dots,1,\dots, m, \dots, m]$ is an $n$-tuple consisting of $\lambda_0$ $0$s, followed by $\lambda_1$ $1$s, and so on, such that $\lambda_0 \le \lambda_1 \le \cdots \le \lambda_m$.

\begin{example}
\noindent
\begin{enumerate}[(a)]
\item If $\lambda_i = 1$ for all $0 \le i \le n-1$, then $x = [0,1, \cdots, n-1]$. In this case, the stabilizer of $x$ is trivial, so $\Sch(S_n,G_x, S) \cong \Cay(S_n, S)$.
\item If $\lambda = (n-k, k)$, then $x \in X_{n, \lambda} = X_{n,k}$ is a binary word and $\Sch(S_n,G_x, S) = G_{n,k}$.
\end{enumerate}
\end{example}

 Let $y \in X_n$, and let $L_i(y) := |\{i < j| y_i > y_j \}|$ for each $0 \le i \le n-1$ (we index our $n$-tuple starting from $0$ for the sake of consistency). This is the number of inversions of $y$ for fixed $i$. We call $L(y) := (L_0, \dots, L_{n-1})$  the \textbf{Lehmer code of $y$}. For any $y, y' \in S_n \cdot x$, set 
\[
\Inv(y,y') :=  \sum_{i=0}^{n-1} L_i(y') - L_i(y).
\] 
We have the following:
\begin{prp}
\label{inv_prop}
For any $y = [y_1, \dots, y_n] \in S_n \cdot x$, we have 
\[
\Inv(y,(i, i+1)y)  = \begin{cases}
0 & \text{if } y_i = y_{i+1},\\
-1 & \text{if } y_{i+1} < y_i,\\
1 & \text{if } y_i < y_{i+1}
\end{cases}
\]
\end{prp}
\begin{proof}
 Note the transposition only affects $L_i(y)$ and $L_{i+1}(y)$. The case when $y_i = y_{i+1}$ is clear. Assuming $y_{i+1} < y_i$, transposing these two entries decreases $L_i(y)$ by $1$ and does not change $L_{i+1}(y)$. For any $j > i+1$, we have the following three cases: $y_{i+1} < y_i \le y_j$, $y_j < y_{i+1} < y_i $, and $ y_{i+1} \le y_j < y_i $. In the first two cases, the transposition does not change $ L_i(y), L_{i+1}(y)$. In the third case, it decreases $L_i(y)$ by the number of such $j$ and simultaneously increases $L_{i+1}$ by the same number. This means $\Inv(y, (i, i+1)y) = -1$. A similar argument for $y_{i+1} > y_i$ shows that $\Inv(y,(i, i+1)y) = 1$.
\end{proof}

\begin{Remark}
Actually, the proof of Proposition \ref{inv_prop} implies slightly more, as it completely describes $L((i, i+1)y)$ in terms of $L(y)$ and $y$. Indeed, if $y_i=y_{i+1}$, then $L((i, i+1)y) = L(y)$. If $y_{i+1} < y_i$, then $L((i, i+1)y)$ is obtained from $L(y)$ by subtracting $1$ from $L_{i}(y)$ and exchanging it with $L_{i+1}(y)$. If $y_i < y_{i+1}$, then $L((i, i+1)y)$ is obtained from $L(y)$ by adding $1$ to $L_{i+1}(y)$ and exchanging it with $L_i(y)$.
\end{Remark}

\begin{cor}
\label{ineq_cor}
For any two vertices $y',y \in \Sch(S_n,G_x, S)$, we have $d(y,y') \ge |\Inv(y,y')|$.
\end{cor}
\begin{proof}
By definition, $d(y,y') = l$, where $l$ is the minimal nonnegative integer such that $y'  =  \sigma_l \cdots \sigma_1 y$ and  $\sigma_i \in S$. Therefore, $d(y,y') = \sum_{i=1}^l d(v_{i-1}, v_i)$, where $v_i = \sigma_i v_{i-1}$ and $v_0 = y$. The statement follows by Proposition \ref{inv_prop} and the triangle inequality.
\end{proof}

\begin{cor}
\label{inv_cor}
For any $y \in \Sch(S_n,G_x, S)$, we have $d(x,y) = \Inv(x, y) = \Inv(y)$. The furthest vertex from $x$ in the graph is $x' = [m, \dots, m, \dots, 1, \dots, 1,  0, \dots, 0]$.
\end{cor}
\begin{proof}
By Corollary \ref{ineq_cor} we have $d(x,y) \ge \Inv(x, y) = \Inv(y)$, since the Lehmer code of $x$ is $(0, \dots, 0)$. Starting with the smallest component of $y$, we can apply adjacent transpositions $(i, i+1)$ to move each component of $y$ leftward to its place in $x$. By Proposition \ref{inv_prop} each such action decreases $\Inv(y)$ by exactly $1$. This process inductively constructs a path from $y$ to $x$, so $\Inv(y) \ge d(x,y)$. Thus $d(x,y) = \Inv(y)$. The maximal possible number of inversions is achieved by reversing $x$, which yields $x' = [m, \dots, m, \dots, 1, \dots, 1,  0, \dots, 0]$. 
\end{proof}

\begin{prp}
\label{diam_cor}
The diameter of $\Sch(S_n,G_x, S)$ is equal to $d(x, x') = \sum_{j > i} \lambda_i \lambda_j$.
\end{prp}
\begin{proof}
To prove the diameter is achieved for $x, x' \in \Sch(S_n,G_x, S)$, we show that for all $y, y' \in \Sch(S_n,G_x, S)$ we have $d(y, y') \le d(x,x') = \Inv(x')$ (see Corollary \ref{inv_cor}). Let $\Inv'(y) = |\{i < j| y_i < y_j\}|$. Note that by definition $\Inv(x') = \Inv(y) + \Inv'(y)$.

Let $\sigma \in S_n$ be the longest element. It acts by involution on the graph $\Sch(S_n,G_x, S)$, reversing each $y \in S_n \cdot x$, so that $d(y,y') = d(\sigma y, \sigma y')$. It is also easy to see $\Inv(y) = \Inv'(\sigma y)$. By Corollary \ref{inv_cor}, it follows that $d(y,x') = \Inv'(y)$. By the triangle inequality and Corollary \ref{inv_cor}, we have

\begin{align*}
&2 d(y, y') \le d(x,y) + d(x,y') + d(y,x') +d(y',x')\\ 
&= \Inv(y) + \Inv'(y) + \Inv(y') + \Inv'(y') = 2 \Inv(x').
\end{align*}

Thus, $d(x,x')$ is the diameter, which we can compute as:
\[
 d(x,x') = \Inv(x') = \sum_{i=1}^m \lambda_{i} \sum_{j=0}^{i-1} \lambda_j =  \sum_{j > i} \lambda_i \lambda_j.
\]
\end{proof}

One can interpret the Lehmer code of $y \in S_n \cdot x$ as a path $P_y$ along the lattice $\mathbb{Z}^2 \subset \mathbb{R}^2$, starting at $(0, \lambda_m)$ and ending at $(n,0)$, with only vertical moves and horizontal moves to the right allowed.

 Indeed, for each component of $L_i(y)$ of $L(y)$, define $P_y$ by connecting the points $(i, L_i(y))$ and $(i+1, L_i(y))$ with a horizontal line segment. Then, connect the endpoints of adjacent line segments vertically, as well as the start of the first line segment with the starting point and the end of the last line segment with the ending point. The area bounded by $P_y$ and the two axes is equal to $\Inv(y)$. 
 {TODO: Figure to illustrate?}

\subsection{Large size limits, towards "macroscopic" descriptions}

Let us discuss large-size limits ofr the graphs we considered above.  
The main point is that, in the dual description, the limit shapes 
(i.e., the ``typical'' or ``random'' configurations) of ``strings'' 
can be described rather explicitly by tractable formulas.  
This can be seen as a manifestation of the ``strong-to-weak'' principle of duality,
since the dual description becomes simple when original description is not. 

Moreover, the limiting dynamics can be described by (partial) differential equations, 
which corresponds to the standard microscopic-to-macroscopic change of description.  
In common terms, one may compare this to liquids: on one hand, they consist of atoms, which provide a microscopic description, 
while on the other hand, we typically describe them macroscopically using partial differential equations, 
such as the Navier--Stokes equations.  
A similar picture arises in our setting: in the large-size limit, we can expect the dynamics to be described in terms of partial differential equations.  

In a sense, this picture is analogous to the AdS/CFT correspondence.  
On the CFT side, there is a parameter $N$ — the size of the matrix group — which is exactly analogous to $n$ in our setting, the size of permutation matrices.  
One considers the limit $N \to \infty$, and in this limit the dual theory simplifies: the string-theoretic description reduces to its limiting gravitational description.  
Thus, the observables in the CFT correspond to areas or volumes of certain surfaces in AdS space.  
In our case, instead of AdS and gravity, the $S_n$ duality picture is simpler: planar polygons instead of AdS and
in the large-$n$ limit dynamics is expected to  be described by hydrodynamic equations, such as the inviscid Burgers equation, as discussed below.

We recall classical results and present several conjectures in this direction in the subsequent subsections.  

Modern studies of limit shapes for Young diagrams originate from the seminal works of 
Vershik--Kerov \cite{vershik1977asymptotics} and Logan--Shepp \cite{logan1977variational}, 
and have been greatly extended with numerous applications; 
see, e.g., \cite{borodin2000asymptotics, Olshanski2001CentralLimit, OkounkovReshetikhinVafa2006, 
KenyonOkounkovSheffield2006, AngelHolroydRomikVirag2006, okounkov2006random, 
petrov2013sl2, corwin2012kpz}.

It is honor to mention the landmark paper \cite{tHooft1974LargeN}
which influence on modern mathematical physics
is difficult to overestimate. 
It was proposed that in large $N$ limit gauge theories
admit description similar to string theory. 
It is tempting to think that similar ideas can be applied
in our setting, we hope to elaborate that in future.

\subsection{Vershik's (1996) limit shapes for rectangular Young diagrams, general $q$}

Here we recall the classical analysis of limit shapes for rectangular diagrams, 
which are essentially ROC curves (or Dyck paths). 
This fits into the general line of questions outlined in the previous subsection: 
one expects that, in the large-size limit, dual descriptions of discrete systems 
converge to continuous ones. 
Moreover, these limiting continuous systems often admit explicit descriptions, 
again confirming the ``strong-to-weak'' transformation paradigm, 
since continuous models are typically more tractable.

By a limit shape we mean, roughly speaking, the shape of a ``random'' or 
``typical'' ROC curve. 
This can be made precise by considering an average over all possible ROC curves 
with a given weight; the average is taken pointwise over the family of curves. 
Before proceeding, let us emphasize that the results recalled here are, in a sense, 
$q$-deformations: they depend on a parameter $q$, which has a natural 
``quantum'' interpretation. 
Our primary interest is the case $q=1$, which will be discussed in the next subsection. 
However, this case requires a certain modification of the classical setup, 
so we first review the classical results.

The shape of these limit curves was determined by A.~M.~Vershik~\cite{vershik1996statistical}, 
building on the celebrated earlier joint work with S.~Kerov~\cite{vershik1977asymptotics} 
(see also~\cite{logan1977variational}). 
An interactive simulation widget has been developed by one of the authors (L.~Petrov): 
\href{https://lpetrov.cc/simulations/2025-12-28-q-partition-cftp/}{link}. 
(Press ``About this simulation'' for a detailed description.)

\textbf{Limit Shape:} As $N \to \infty$ with $q = e^{-\gamma/N}$ for fixed $\gamma > 0$, 
the rescaled partition boundary converges to a deterministic curve given by the implicit equation:
\[
A e^{-\gamma y} + B e^{-\gamma x} = 1
\]
where
\[
A = \frac{1 - e^{-\gamma}}{1 - e^{-\gamma(1+a)}}
\qquad \text{and} \qquad
B = \frac{1 - e^{-\gamma a}}{1 - e^{-\gamma(1+a)}}.
\]

\begin{figure}[H]
	\centering
	\includegraphics[width=0.45\textwidth]{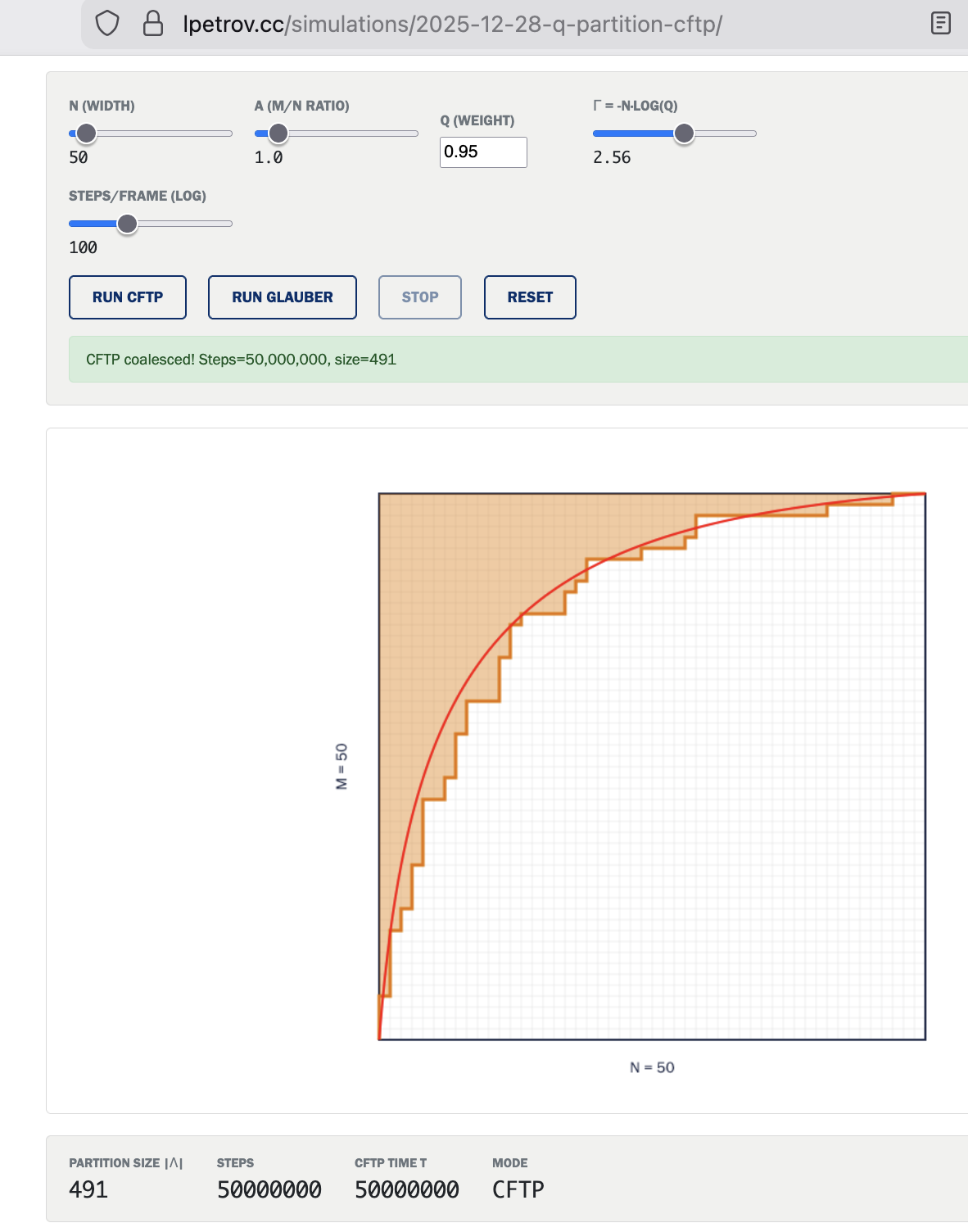}
	\caption{Limit shape for ROC curves (rectangular Young diagrams, or equivalently Dyck paths), 
    defined as the average shape of the curve when all configurations are weighted by $q^{\mathrm{area}}$. 
    (Here $q$ is a parameter; in the plot it is chosen to be $0.95$.) 
    The solid red line represents the theoretical limit curve obtained by A.~M.~Vershik~\cite{vershik1996statistical}. 
    An interactive simulation widget developed by L.~Petrov is available at 
    \href{https://lpetrov.cc/simulations/2025-12-28-q-partition-cftp/}{link}.
    }
	\label{fig:lpetrov_vershik_curve}
\end{figure}

\subsection{Vershik's (1985) limit shapes for rectangular Young diagrams $q=1$}


Here we remind related results on limit shapes for $q=1$.
A Young diagram chosen uniformly at random from the set of all Young diagrams with $n$ boxes exhibits a limit shape phenomenon. After rescaling the diagram by a factor of $1/\sqrt{n}$ (so that the total area is normalized to $1$), the boundary of the diagram concentrates, as $n \to \infty$, around the Vershik curve $\gamma$ given by
\[
e^{-\sqrt{\zeta(2)}\,x} + e^{-\sqrt{\zeta(2)}\,y} = 1,
\]
see~\cite{vershik1985asymptotics} (last formula in the paper, attributed to unpublished work by A.M.Vershik).
Here $\zeta(2) = \pi^2/6$. 

The proofs and actually certain generalizations appeared in works of several mathematicians. 
Here we will rely on ~\cite{petrov2009two},
where the following generalization has been obtain:

\begin{proposition}

    The limit shape for young diagrams in rectangle (ROC-curves) with fixed number of boxes (fixed area under the curve is given by: 
\[
e^{-c(x - x_0)} + e^{-c(y - y_0)} = 1,
\]
for suitable constants $c$, $x_0$, and $y_0$.  Moreover they can be seen as segments of the full Vershik's curve above.  
\end{proposition}

Sketch of proof: consider fixed 
positive $a$ and $b$ with $ab>1$ the Young diagrams with $n$ boxes that fit in the $a\sqrt{n}\times b\sqrt{n}$ rectangle
(in other words, the length is at most $a\sqrt{n}$ and the height
is at most $b\sqrt{n}$). Scaling by $1/\sqrt{n}$ we get a 
random
set $Y_n$ of area 1 inside the rectangle $a\times b$. 

If $ab=2$, the boundary of $Y_n$ approaches (by probability) the diagonal of our $a\times b$. 

If $1<ab<2$, passing to a complement of $Y_n$ to the $a\times b$ rectangle (and making a symmetry with respect to its center),
we get a set of area $(ab-1)<ab/2$, and another scaling reduces 
the question to the $ab>2$ case.

If $ab>2$, one can find the unique point $P=(x_0,y_0)$ below 
the Vershik curve $\gamma$, points $A$ and $B$ on $\gamma$
so that the segment $PB$ is vertical, $PA$ is horizontal, and 
$PB:PA:\sqrt{S}=b:a:1$, where $S$ is the area in the triangle
$BAP$ below $\gamma$. Then this piece of $\gamma$
is (after the scaling by a factor $b:PB=a:PA$) is the limit shape
of $Y_n$. So, the equation of the limit curve is $e^{-c(x-x_0)}+
e^{-c(y-y_0)}=1$ for appropriate constants $c$, $x_0$, $y_0$.

This phenomenon was rediscovered independently by several mathematicians in various forms. 
An analogous result also holds for Young diagrams confined to a strip, that is, under restrictions 
either on the length or on the height.

\subsection{Limit shape for ROC-curves (Dyck paths) }
The first natural question from our perspective regarding the large-$n$ limits of discrete systems is: 
what is the shape of the ``random'' or ``typical'' string with a fixed area under it? 
As discussed above, one can expect tractable formulas for such limit shapes, 
which may be interpreted as a manifestation of a ``strong-to-weak'' holographic duality.
Moreover from Vershik's results above one can expect the answer to be given by the simple equation on the exponential of coordinates.
Based on computational experiments we propose such answer below.

Here, we present the results of simulations and a conjectural answer to this question in the case of ROC curves, 
which, according to our picture, are the holographic dual strings corresponding to the nodes of 
$S_n/(S_k \times S_{n-k})$.  
The same can be described as averaging over Dyck paths  the same result, 
since the dominant contribution comes from the curve's neighborhood (where all monotonic paths are Dyck paths); 
thus, averaging over Dyck paths or over all monotonic paths from $(0,0)$ to $(1,1)$ yields essentially the same outcome.
\begin{Conj}[Limit shape]
Fix $C \in (0,1)$ and let $k \to \infty$ with $L = L(k)$ satisfying $L/k^2 \to C$. 
Then the normalized average path in layer $L$ converges uniformly to the curve $y = y_C(x)$, $x \in [0,1]$, given by
\begin{equation}
y_C(x) = 1 - \frac{1}{\lambda} \ln\bigl(1 + e^{\lambda} - e^{\lambda x}\bigr), 
\quad \lambda \neq 0,
\end{equation}
and 
\begin{equation}
y_{1/2}(x) = x \quad \text{when } C = \frac{1}{2} \; (\text{i.e., } \lambda = 0).
\end{equation}
Here $\lambda = \lambda(C) \in \mathbb{R}$ is the unique solution of
\begin{equation}
C = \frac{\ln^2(1+e^\lambda) + 2\, \mathrm{Li}_2\Bigl(\frac{1}{1+e^\lambda}\Bigr) - \frac{\pi^2}{6}}{\lambda^2}.
\end{equation}

Equivalently, the limit curve satisfies the implicit equation
\begin{equation}
e^{-\lambda(1-x)} + e^{-\lambda y} = 1 + e^{-\lambda}.
\end{equation}

\end{Conj}

Equation (4) is a rescaled and translated segment of the \emph{universal exponential curve}
\[
e^{-\alpha X} + e^{-\alpha Y} = 1,
\]
which arises in the arctic-circle phenomenon for boxed plane partitions \cite{CLP98} and in Okounkov's theory of limit shapes for random surfaces \cite{Ok03}. 
Our constraint to the box $[0,1]^2$ selects the appropriate segment; 
the parameter $\lambda$ is determined by the area (equivalently, the layer number).

\begin{figure}[H]
	\centering
	\includegraphics[width=0.45\textwidth]{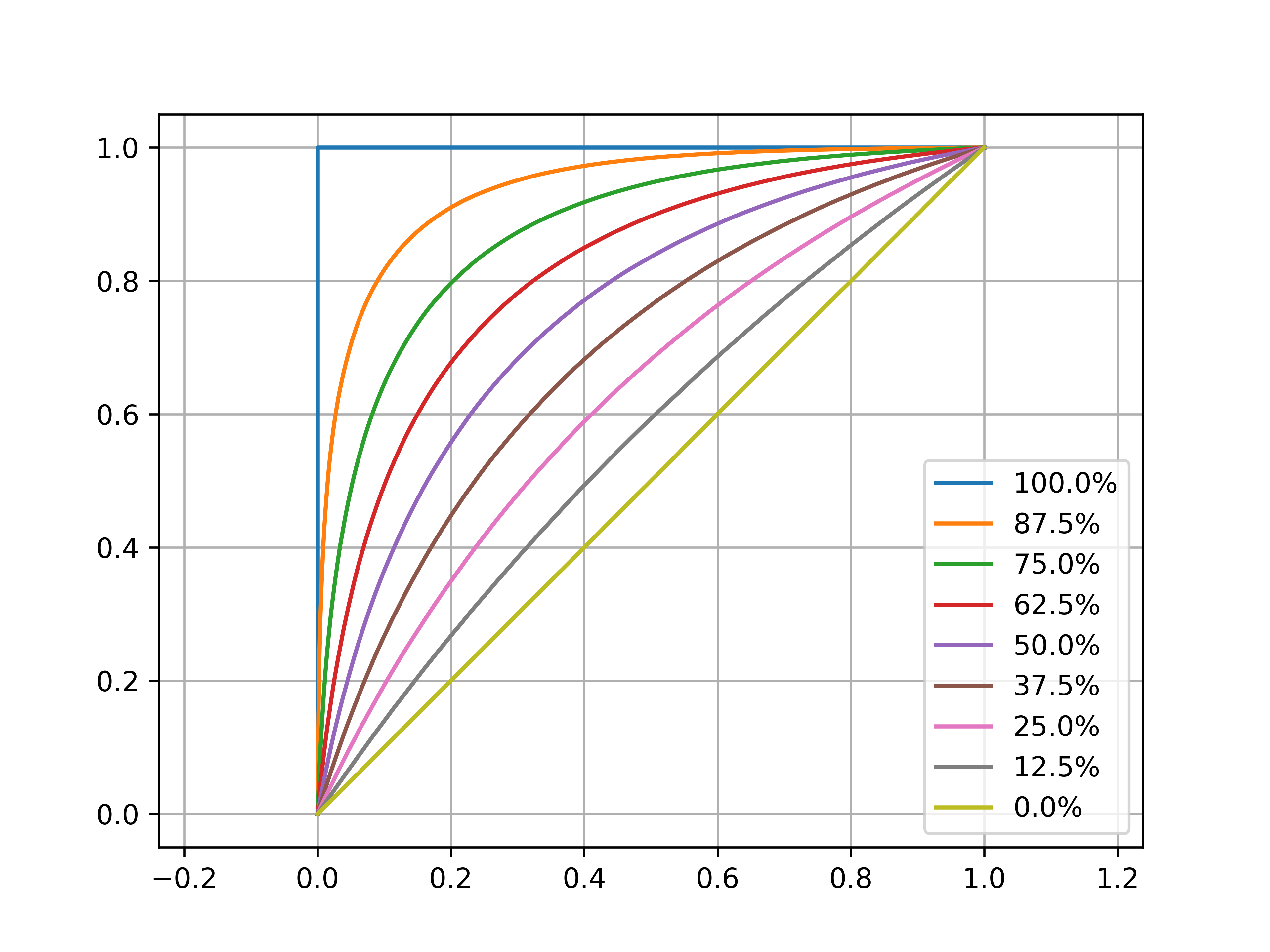}
	\caption{Limit shape for the ROC curves (Dyck paths) with fixed area under curve.
    That is - average of all ROC curves with fixed area under them.
    }
	\label{fig:avg_Dyck_paths_N_256}
\end{figure}

Related results by A.M.Vershik and  \cite{petrov2009two} were discussed in the previous subsection. 

\subsection{Dynamics in large size limits. TASEP-Burgers correspondence. KPZ.   }

As we discussed before Young tableaux can be considered as solutions of discrete string equations of motion
and define the evolution of the paths (discrete strings).
In large size limit it is natural to think that "typical" Young tableaux would produce 
such kind of evolution which would be possible to describe by partial differential equations,
such that their solutions would be monotonic curves going from left-bottom to right-top.
In particular the family of "typical" curves above should be their solution,
but more generally evolution might evolve any such monotonic curve. 

Mathematically speaking we can describe evolution as follows
take a discrete curve take all possible Young bounded by the curve,
make evolution by each of these curves and consider their averages.
That defines an evolution on curves. 

For the case of the ROC-curves, educated guess is that such dynamics can be described by the 
inviscid \href{https://en.wikipedia.org/wiki/Burgers%27_equation}{Burgers equation}
and with fine granularity by \href{https://en.wikipedia.org/wiki/Kardar%E2%80%93Parisi%E2%80%93Zhang_equation}{KPZ-equations} (\cite{KPZ1986}).
Relying on the known results on TASEP-Burgers correspondence \cite{Rost1981,Ferrari2018,AIHPB2009_45_4_1048_0, QuastelTsai2021,QuastelRahman2020, } 
and relation to KPZ \cite{BertiniGiacomin1997,corwin2012kpz}.

Instead of the microscopic dynamics  implemented by random adjacent updates,
i.e.\ TASEP-type evolution on \(\{0,1\}^p\) it is better to consider  macroscopic dynamics.
%
In the large-size (hydrodynamic) regime, with \(p\to\infty\), \(k_p/p\to \alpha\in(0,1)\), and with the empirical initial profile converging to a limiting density \(u_0(x)\), the coarse-grained field is described by a deterministic conservation law.  
Denoting the macroscopic density by \(u(x,s)\), the expected limit is the entropy solution of the inviscid Burgers equation
\[
\partial_s u + \partial_x\!\bigl(u(1-u)\bigr)=0,
\qquad
u(x,0)=u_0(x).
\]


Analogy: think of a one-lane road with cars, where each site can be either occupied (1 = car) or empty (0 = space).
The microscopic motion of the cars — each moving forward if the spot ahead is free — corresponds to a TASEP process or the evolution of nodes on a discrete graph.
Now imagine observing this road from high above in the cosmos: individual cars are invisible, and you only see the average car density along the road.
At this macroscopic level, the density evolves smoothly in time, and its dynamics are governed by the Burgers equation.
If you zoom in and improve the resolution, the random fluctuations of cars around the average density become visible, and their collective dynamics are described by the KPZ equation, capturing the wiggly, stochastic behavior of the system.

Fluctuations around this limit 
under the KPZ scaling $(x \sim \varepsilon^{-1}, t \sim \varepsilon^{-3/2})$, $h^\varepsilon$ converges to the solution of the KPZ equation:
\[
\partial_t h = \nu \partial_{xx} h + \lambda (\partial_x h)^2 + \xi,
\]
where $\xi$ is space-time white noise.  
Thus, TASEP connects to KPZ through the limit of properly rescaled fluctuations around the hydrodynamic density.

Figure~\ref{fig:Burgers} shows the Monte Carlo tab after a completed run (`2000/2000`) for the initial path \(A=(0^k1^{p-k})\) with smoothing \(r=13\): on the left, the heatmap \(p(\text{position},\text{time})\) displays mean occupancy \(u(x,s)\) (blue \(\approx 0\), red \(\approx 1\)) together with an active green `DRAWBOUNDARY` overlay given by the fitted curve \(s(x)\approx 0.259+2.846x-2.574x^2\), which tracks the visible separation between the mixed interior and near-pure corner regions; on the right, the multi-slice profile panel compares exact Burgers curves and MC curves at several times (including highlighted \(s=0.2\), \(s=0.4\), and breaking-time slices), with per-slice \(L^2\) errors indicating stronger agreement at earlier times and larger discrepancies near later, shock-influenced regimes, so the two panels together provide both a spacetime picture of the evolution and a quantitative slice-by-slice validation against the PDE prediction.

\begin{figure}[H]
	\centering
	\includegraphics[width=0.65\textwidth]{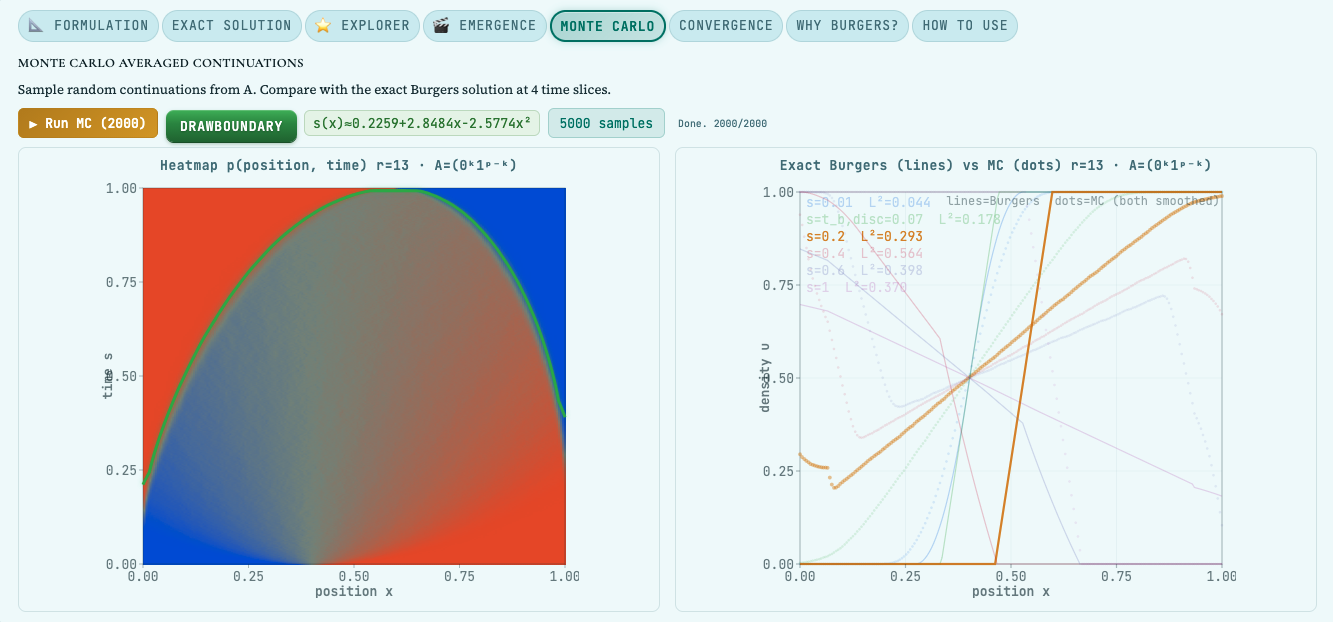} 
	\caption{ Monte Carlo simulation results for starting sequence \(0^k1^{p-k})\): left, the averaged occupancy heatmap \(u(x,s)\) with fitted green boundary \(s(x)\approx0.259+2.846x-2.574x^2\); right, multi-time comparison of exact Burgers profiles and MC data with per-slice \(L^2\) errors.
    An interactive simulation 
    is available at 
    \href{https://fkhafizov.github.io/animation/young_tableaux_claude_v9.html}{link}.
    }
	\label{fig:Burgers}
\end{figure}








\clearpage
\subsection{Spanning trees for finite spin-up sector, 1/n expansion, Benjamini-Schramm limit to lattices,  Mahler measures   }
In this subsection we 
analyze the limit $S_n/ (S_{d_1} \times ... \times S_{d_j} \times S_{n - \sum d_j} )$
for $d_j$ fixed and $n \to infty$.

\textbf{Overview: math part.}
The key message here is that in certain sense
(Benjamini-Schramm) these 
graphs tend just to the standard lattice
graph $Z^{\sum d_j}$. 
Hence all the key quantities for them in the limit
tend to those for the standard lattice.
In particular the number of spanning trees normalized by graph size  converges to logarithmic
\href{https://en.wikipedia.org/wiki/Mahler_measure#Higher-dimensional_Mahler_measure}{multi-variable Mahler measures}. These measures 
count spanning trees for lattices and are interesting objects related to values of $L$-functions, deep Beilinson's conjectures,
supersymmetric Landau-Ginzburg models used in mirror symmetry for Fano varieties, etc. Additionally, we observe that many quantities here admit $1/n$ expansions which are efficient enough
to reproduce numerical data with good accuracy. 

\textbf{Overview: duality and string part. }
In physical language the claim is that the model that describes the graph Laplacian
(the XXX spin chain)
in the large $n$ limit can be described in terms of $d=\sum d_j$ non-interacting particles. 
The appearance of $d$ particles 
has a clear and simple explanation via our duality.
Indeed, let us consider for simplicity the 
case $S_n/(S_d \times S_{n-d})$, with $d$ fixed and 
$n\to \infty$. Our discrete duality describes
that system as a discrete string whose worldsheet is a $d \times (n-d)$ rectangle. In the limit it becomes of size $d \times \infty$, so we can think that one variable on the worldsheet becomes continuous,
while the other one is discrete. The continuous variable leads to particle-like behaviour, and the discrete variable means that we have not just a single particle but $d$ of them. As we described above the action of the discrete system is a discretization 
of the standard string action, so we expect that in the limit it becomes an action for standard particles. In the next subsection we will describe another limit and argue that actual bosonic string-like scalar models appear in the same way.

Our claims are of course closely related to the well-known fact that magnons in XXX models are almost non-interacting at large length.
At the same time we were unable to find in the literature detailed discussions of the specific quantites we consider below, which are rather unusual from the spin chain point of view. 
In particular, surprisingly, it is not clear  how to derive many of the proposals below directly from the Bethe ansatz,
which in principle describes the full spectrum. 

It is interesting to note
 that a  \href{https://en.wikipedia.org/wiki/Van_Hove_singularity}{Van Hove singularity} 
can be clearly seen in the spectrum of $S_n/(S_2\times S_{n-2})$ for large $n$ -- see figure \ref{figs3:VanHove200_2}
(and  \href{https://www.kaggle.com/code/alexandervc/cayleypy-coxeter-eigenvals/}{notebook}).
That confirms that these systems can be described as two non-interacting particles 
since the Van Hove spectrum appears as for the sum of two cosines. (A single cosine is spectrum of line graph
which is $S_n/ S_{n-1}$).

Let us now present our main statements and observations.

\begin{Conj}[Spanning trees, XXX-determinants and Mahler measure]\label{conj:mahler}
Consider Coxeter (neigbour transpositions) generators of $S_n$ and consider  Schreier coset graphs   $S_n/ (S_{d_1} \times ... \times S_{d_j} \times S_{n - \sum d_j} )$ for $d_j$ fixed and $n \to \infty$.
Denote the number of spanning trees by $\tau(\Gamma_n)$ and the graph size 
by $V(\Gamma_n)$, and denote $d=\sum d_j$.
(Remark: $\tau(\Gamma_n)$ is related to the determinant of the Laplace operator
which is the XXX-Heisenberg spin chain Hamiltonian in a specific subsector, up to constants).
With notation as above,
\begin{equation}\label{eq:mahler}
  \frac{\log\tau(\Gamma_n)}{|V(\Gamma_n)|}
  \;\xrightarrow{n\to\infty}\;
  m\!\left(2d - \sum_{i=1}^{d}(z_i+z_i^{-1})\right),
\end{equation}
where the \emph{logarithmic Mahler measure} of a Laurent polynomial
$P\in\mathbb{Z}[z_1^{\pm1},\ldots,z_d^{\pm1}]$ is
\[
  m(P)
  \;=\;
  \int_{\mathbb{T}^d}
  \log\bigl|P(e^{2\pi i\theta_1},\ldots,e^{2\pi i\theta_d})\bigr|\,
  \frac{d\theta_1\cdots d\theta_d}{(2\pi)^d}.
\]
For $d=2$ (i.e.\ $r=2$), the right-hand side equals $4G/\pi$, where
$G$ is Catalan's constant. Numerically this Mahler measure is around  $1.166243$, and for $d=3$ it is around $1.673389$.  It is  related to special values of $L$-functions,
explicitly known for $d=2,3$ and non-explicitly related to Beilinson's conjectures for general $d$ (see discussion below on C.Denininger's results).
For general $d$, the measure equals the same expression for the lattice, i.e. 
logarithm  of number of spanning trees on the  finite lattice again normalized by size,
and it is sometimes called "spanning-tree entropy" of the
$d$-dimensional integer lattice $\mathbb{Z}^d$.

Similarly the traces, characteristic polynomial and entire spectral measure converges to appropriate expressions for the lattice:

\begin{equation}\label{eq:mahler}
  \frac{\log \det(Laplacian_{G_n} - E)}{|V(\Gamma_n)|}
  \;\xrightarrow{n\to\infty}\;
  m\!\left(2d - \sum_{i=1}^{d}(z_i+z_i^{-1}) - E \right),
\end{equation}

\begin{equation}\label{eq:mahler}
  \frac{{\rm Tr} (Laplacian_{G_n} - E)^s}{|V(\Gamma_n)|}
  \;\xrightarrow{n\to\infty}\;
  \int_{\mathbb{T}^d}
  \!\left(2d - \sum_{i=1}^{d}( e^{2\pi  i\theta_d} +e^{-2\pi  i\theta_d} ) - E \right)^s
  \frac{d\theta_1\cdots d\theta_d}{(2\pi)^d}.
\end{equation}

Instead of ${\rm Tr}$ or $\det$ one can put any other function and a similar equality is expected. In other words, moments and entire spectral measure with appropriate normalization
tend to the one of the lattice.
\end{Conj}


Remark. The expressions above are well-known and easy to understand for lattices (grid-graphs), see
e.g. \cite{Lyons2005,SilverWilliams2016}.
The number of spanning trees is of course different for lattices and finite graphs above, 
but with  normalization by graph sizes (which are also different) they coincide in the limit. 

\begin{Conj}
    [Benjamini--Schramm limit]\label{conj:BS}
The sequence $(\Gamma_n)_{n\ge 1}$ converges in the
Benjamini--Schramm (local weak) sense \cite{BenjaminiSchramm2001, Lyons2005}
to the $d$-dimensional integer lattice $\mathbb{Z}^d$,
with its standard nearest-neighbour structure.
Consequently, the spectral measure of $\Gamma_n$ converges
weakly to the spectral measure of the Laplacian on $\mathbb{Z}^d$.
\end{Conj}

\begin{Conj}[$1/n$ expansion and numerical check]
The expressions $ \frac{\log\tau(\Gamma_n)}{|V(\Gamma_n)|} $ and similar expressions for traces
and characteristic polynomial
admit an asymptotic expansion as series in $1/n$.
The constant term is the Mahler measure (given by expressions in  the conjecture above).
Convergence is fast enough such that determining only
few values in $n$ (which is easy numerically) one can 
determine the leading and several subleading coefficients, thus obtaining numerical results for infinite $n$ from several quite small $n$ values.
\end{Conj}


\bigskip 
\textbf{ Informal arguments.} It is clear that the key idea is to understand that the graph is close to the lattice in a certain sense.
The Coxeter generators are very similar to commutative -- all generators
commute except neighbor pairs. 
Moreover, for the case 
$S_n/(S_2 \times S_{n-2})$  is has been already pictured on several figures above
that the graph is quarter of the lattice
("quarter Aztec diamond"). 
The same can be achieved for any $d$ for  $S_n/(S_d \times S_{n-d})$ . Indeed,
nodes of the graph are sequences of 0's and 1's with exactly $d$ zeros.
The embedding of the graph is defined by the following rule: 
associate to a vector $d$ integer numbers which are just the positions
where these zeros appear.
One can see that this gives an embedding of our graph into the
lattice graph $Z^d$. 
Figure \ref{fig:graph3_8like_lattice} illustrates that embedding for $d=3$.
See also figures \ref{fig4:2dlatticeLike},
\ref{figs4:3dlatticeLike1} for $S_{n}/S_{n-2}$
and $S_{n}/S_{n-3}$.

\begin{figure}[H]
    \centering
    \includegraphics[width=0.8\linewidth]{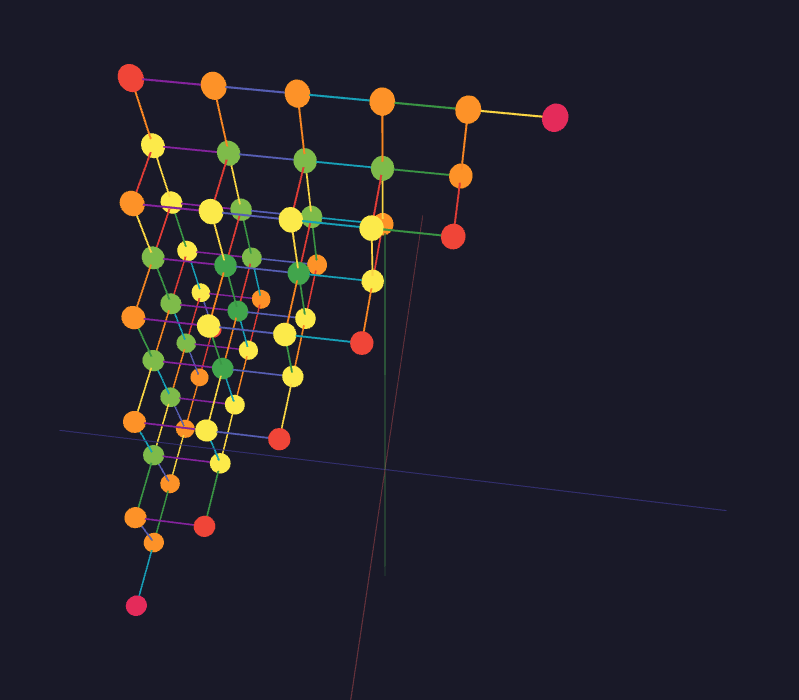}
    \caption{The Schreier graph $S_8/(S_3 \times S_5) $ with Coxeter generators, similarity to lattice can be clearly seen.  Widget is available \href{https://fkhafizov.github.io/animation/schreier_graph_v4_o45.html}{link}. Other graph examples available by the link. }
    \label{fig:graph3_8like_lattice}
\end{figure}

\bigskip 

\textbf{ Numerical checks. } To check the conjecture on Mahler measures numerically 
we rely on the $1/n$ expansion.
We compute the number of spanning trees for several values 
of $n$, and then make a numerical fit for $\log\tau(\Gamma_n)/|V(\Gamma_n)|$ in the form
$c_0 + c_1/n + c_2/n^2 + ...$.
The leading term $c_0$ should match the desired Mahler measure.
We observe excellent correspondence between the two computations:
$c_0$ from finite data and numerical computation of the Mahler
measure. The Mahler measure itself is just an integral, so it can be computed e.g.
in Mathematica with high precision.

Technically, to compute the number of spanning trees by Kirchoff's theorem we take the principal minor
of the Laplace matrix (to avoid problem with zero eigenvalue) and compute the logarithm of the determinant
directly,  e.g. via the numpy function np.linalg.slogdet. In this way we get values up to $n$ around 30.
When we fit the obtained data we exclude the first several values (for example, we start from $n\sim 10$).
We can also fit by polynomials of various degrees, and we observe that degrees from 5 to 10 give nearly 
the same result for the leading term. Figure \ref{fig:fit_demo} provides a screenshot 
from the \href{https://www.kaggle.com/code/alexandervc/cayleypy-fit-spanning-tree-data}{notebook}
with fit showing numerical coincidence up to 5 digits and near-independence on the choice of the degree of the fitting polynomial. Computations of the numerical determinants themselves are available in notebooks:
\href{https://www.kaggle.com/code/lilypilly/determinants-for-large-graph-laplacians}{notebook}, \href{https://www.kaggle.com/code/ivankolt/coxeter-cosets-determinant}{notebook}. 



\begin{figure}[H]
    \centering
    \includegraphics[width=0.8\linewidth]{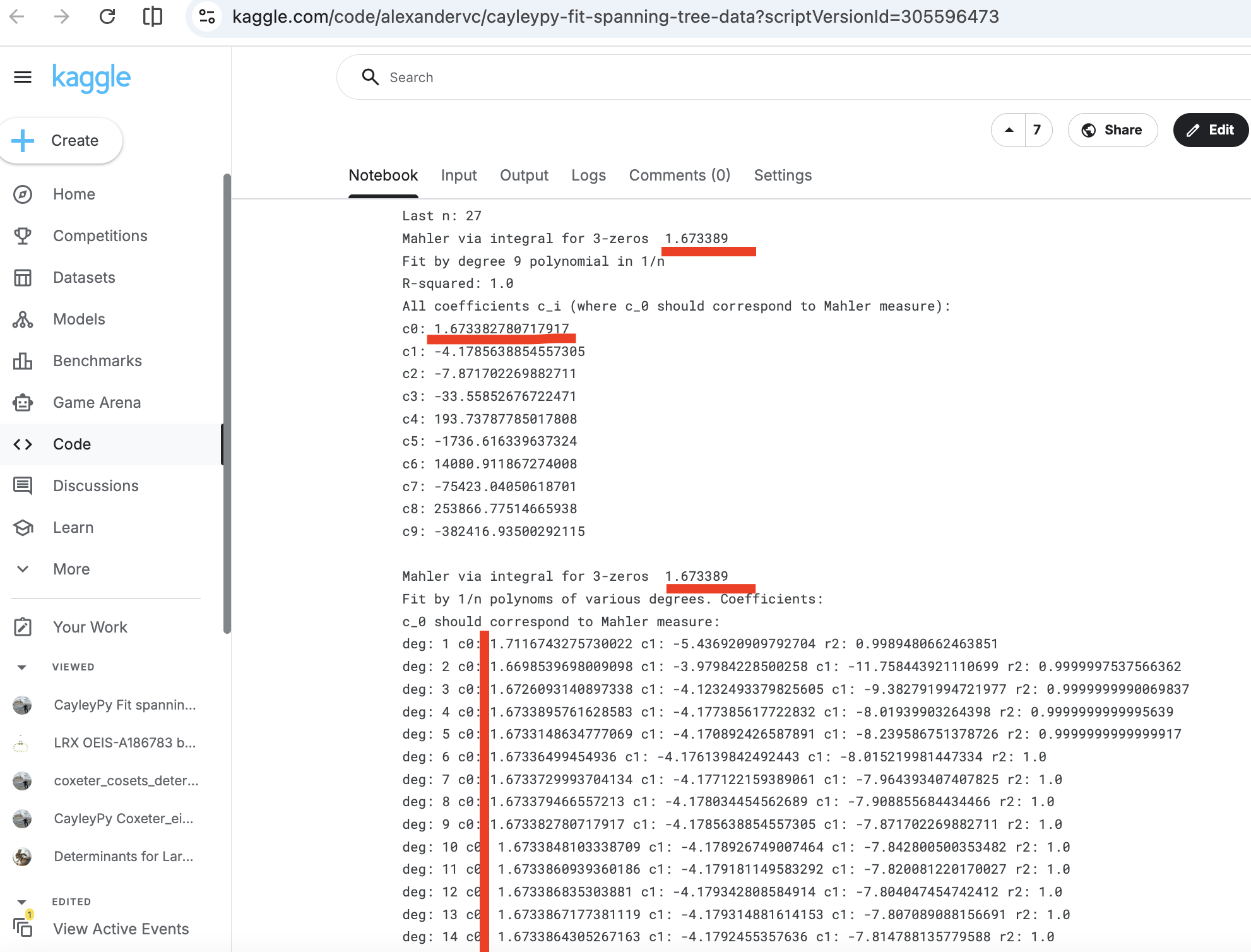 }
    \caption{ Fitted numerical data shows quite good coincidence of $c_0$ with Mahler measure which is around 1.67338.
    Leading coefficient $c_0$ does not depend much on the degree of the polynomial approximation.
    \href{https://www.kaggle.com/code/alexandervc/cayleypy-fit-spanning-tree-data}{Notebook} }
    \label{fig:fit_demo}
\end{figure}

\textbf{Potential route to the  proof.} The two conjectures are related.
The first one follows from the second by a result of \cite{Lyons2005}.
By the visualization and arguments above it is clear that our graphs resemble lattices,
the difficult part is to control the difference and show that is disappears in large size limit -- while clear in principle, it may be nontrivial to do it fully rigorously without also relying e.g. on unproven properties of the Bethe ansatz.

\noindent\textbf{Step 1 (Local structure).}
A vertex of $\Gamma_n$ is an ordered partition
$[n]=B_1\sqcup\cdots\sqcup B_r$.
The Coxeter generator $s_i$ acts non-trivially at $v$ if and only if
$i$ and $i{+}1$ belong to \emph{different} blocks.
For a uniformly random vertex and a fixed radius $R$, the $R$-ball
in $\Gamma_n$ is determined by $O(R)$ consecutive positions of $[n]$,
whose block-membership pattern converges (as $n\to\infty$) to that
of an i.i.d.\ colouring of $\mathbb{Z}$ with $r$ colours and
frequencies $(k_j/n)$.
The resulting limit graph is the Cayley graph of $\mathbb{Z}^{d}$
with the $2d$ standard generators -- exactly $\mathbb{Z}^d$.

\smallskip
\noindent\textbf{Step 2 (Følner / BS convergence).}
The fraction of ``boundary'' vertices -- those whose $R$-neighbourhood
is not isomorphic to a ball in $\mathbb{Z}^d$---is at most
$O(R/n)\to 0$.
This is precisely the Følner condition, which implies
Benjamini--Schramm convergence (Conjecture~\ref{conj:BS}).

\smallskip
\noindent\textbf{Step 3 (Lyons' theorem).}
By \cite[Theorem~1.2]{Lyons2005}, if a sequence of finite connected
graphs with uniformly bounded degree converges in the
Benjamini--Schramm sense to an infinite unimodular random graph $G_\infty$,
then
\[
  \frac{\log\tau(G_n)}{|V(G_n)|}
  \;\longrightarrow\;
  h(G_\infty),
\]
where $h(G_\infty)$ denotes the \emph{tree entropy} of $G_\infty$
and which is logarithmic Mahler measure in case when $G_\infty$ is  $\mathbb{Z}^d$.

\smallskip
\noindent\textbf{Step 4 (Tree entropy $=$ Mahler measure).}
For $G_\infty=\mathbb{Z}^d$, the tree entropy equals
\[
  h(\mathbb{Z}^d)
  \;=\;
  \int_{\mathbb{T}^d}
  \log\!\left(2d-2\sum_{i=1}^d\cos(2\pi\theta_i)\right)
  d\theta_1\cdots d\theta_d
  \;=\;
  m\!\left(2d-\sum_{i=1}^d(z_i+z_i^{-1})\right),
\]
establishing Conjecture~\ref{conj:mahler} conditionally on
Conjecture~\ref{conj:BS}.
The equality of the last two expressions is the tautological unfolding
of the Mahler measure as a torus integral \cite{SilverWilliams2016}.

\textbf{Related works.}
\href{https://en.wikipedia.org/wiki/Mahler_measure#Higher-dimensional_Mahler_measure}{Multi-variable Mahler measures}
are very similar to supersymmetric Landau-Ginzburg models used for mirror symmetry of Fano varieties \cite{Golyshev2007, CoatesCortiGalkinGolyshevKasprzyk2013, GalkinGolyshevIritani2015} and
 earlier works \cite{Witten1993Phases,HoriVafa2000,EguchiKawaiYamada1991, VafaWarner1989,Candelas1991,Kontsevich1995,Givental1998,Batyrev1994}.

From physical point of view counting spanning trees on these graphs is  equivalent to computation
of the determinant of Hamiltonian of XXX spin chain (by Kirchhoff's theorem and identification of the Hamiltonian with graph Laplacian). The conjectures above are closely related to the Bethe ansatz, but establish them fully rigorously and in full generality is still a nontrivial task which seems to not have been addressed yet.


For the most simple case, the  subsector with just two spins-up (i.e. graph $S_n/(S_2\times S_{n-2}$
which is quatter of the Aztec diamond),
there is a non-trivial result by R.Stanley, D.Knuth, T.Chow,
et.al. (\href{https://oeis.org/A007726}{oeis-A007726},
observed and conjectured  in \cite{Stanley1994AztecTrees}, resolved in
\cite{Knuth1997, Chow1997, Ciucu1997}, developed  \cite{KenyonProppWilson2000, Ciucu2008} etc.).
which provides the number of spanning trees exactly for each finite $n$,
(not only asymptotics that we discussed above). Our conjecture can be shown to be true in this case. At the same time, it would be very interesting to derive that exact finite-$n$ result from the Bethe ansatz -- which should be possible to do, but seems rather nontrivial.
In general, the interest in this set of questions during that time seems to have come 
from seminal works on related questions of domino  tilings (dimer model) for Aztec diamond
\cite{elkies1991alternatingsignmatricesdomino}, and more generally remarkable
results on spanning trees \cite{burton1993local}, in particular with relation to conformal field theories
\cite{duplantier1989statistical}.

For Abelian groups,  the relation of spanning trees and Mahler measures
is quite natural and well documented in the literature. Indeed,
by \href{https://en.wikipedia.org/wiki/Kirchhoff%27s_theorem}{Kirchhoff's theorem}, the number of spanning trees of a finite graph equals the determinant of its Laplacian (with one row and column removed). 
In the Abelian case, the group algebra is essentially a quotient of the lattice algebra $\mathbb{Z}^d$. 
Thus group elements can be identified with Laurent polynomials, and convolution (or the action of the Laplacian) corresponds to multiplication by a Laurent polynomial. 
Consequently, the determinant of the Laplacian can be expressed in terms of this Laurent polynomial,
and using the formula $log(\det(M)) = {]rm Tr}(log(M))$ one arrives at the logarithmic Mahler measure. 
On the one hand, it counts spanning trees; on the other hand, in the infinite or periodic limit it naturally leads to the Mahler measure of the corresponding polynomial, modulo some technical details \cite{SilverWilliams2005, grunwald2019number, GrunwaldMednykh2021Cone,}. The non-trivial situation is that
our graphs come from non-abelian groups, but this 
non-abelianity is not that large, and in the appropriate limit
one obtains similar results to the  abelian  case.

Mahler measure and their generalizations play a key role in the number theory:
\cite{Deninger1997,LindWard1990,LindSchmidtWard1990,Deninger1994, arzhakova2021decimation}. Wonderful \href{https://en.wikipedia.org/wiki/Mahler_measure#Some_results_by_Lawton_and_Boyd}{ relations}
were found between Mahler measures and values of L-functions
in special points \cite{Smyth2008,Boyd1981, RodriguezVillegas1999,BoydRodriguezVillegasDunfield2003, } 
  and partly derived by C.Deninger \cite{Deninger1997Deligne} conditioned to Beilinson
conjectures, we refer to \cite{BrunaultZudilin2020, Trieu2023MahlerExact} for further information.
The general expectations according to C.Deninger and previous works by Smith, Boyd et.al. are the following. 
The Mahler measure $m(P_d)$ is commensurable with the first non-vanishing derivative of the $L$-function associated to the  cohomology $H^{d-1}(V_{P_d})$ (smooth compactification of the zero locus of Laurent polynomial) evaluated at $s=0$ (or at the central point of the critical strip, depending on the normalization). 
It might be that $m(P_d) \sim_{\mathbb{Q}^\times} c_d \cdot L'(M_d, 0)$ where $M_d$ is the motive attached to the projective closure of $V_{P_d}$, $c_d$ is a rational normalization factor (often involving powers of $\pi$), and $\sim_{\mathbb{Q}^\times}$ denotes equality up to a rational factor.
Cases $d=2,3$ have well-known explicit description.
In some cases Mahler measure of A-polynomial of the knot gives
hyperbolic volume of the knot complement manifold e.g. \cite{BoydRodriguezVillegasDunfield2003}.
Quantization of A-polynomial is conjecturally related to 
colored Jones polynomials ("AJ-conjecture") \cite{Le2006_ColoredJonesApoly, Gukov2005, FujiGukovSulkowski2012, GrassiHatsudaMarino2014, },
fascinating topic deeply connected to various questions in topological string theory. 

Related important conjecture is \textbf{L\"uck's Determinant Conjecture.} 
Let $G$ be a finitely generated group and $\Delta$ the Laplacian on its Cayley graph. 
Then the Fuglede--Kadison determinant of $\Delta$ satisfies $\det_G(\Delta) \ge 1$. 
More relations to deep conjectures and questions by Connnes, Gromov et.al briefly sketched on page 5/1458 \cite{aldous2007processes}.
For abelian groups this is equivalent to a lower bound on the growth rate of spanning trees in the graph \cite{Luck2002}.
The papers above have more analytic flavor,
related purely algebraic considerations for free group can be found in
\cite{KontsevichNoncommutative, bellissard2007algebraic, Haiman1993}
although terminology is different but constructions are essentially related - the key results show that resolvent of the Laplacian 
is algebraic function, moreover M.Kontsevich extends it to "det(1-tM)" which 
is non-commutative analogue of Mahler measure. Surprisingly, results rely
on \cite{ChomskySchutzenberger1968} theory of context free languages 
by Chomsky and Schutzenberger from 1968, and Kontsevich's generalization
additionally on Grothendieck's conjecture on algebraic solutions of holonomic systems (see also \href{https://www.youtube.com/watch?v=vD6eIQPpNFs}{lecture by M.Kontsevich} and his newer lectures on the subject).

In the case $S_n/(S_d\times S_{n-d}$
the graphs themselves apparently  are special cases of "token" graphs
(applied to just a linear  graph)
\cite{DalfoDuqueFabilaMonroyFiolHuemerTrujilloZaragoza2020}.
Use of Benjamini-Schramm limits
is quite wide-spread technique
nowadays e.g. \cite{bille2023random} section 2.2.

\begin{figure}[H]
    \centering
    \begin{minipage}{0.48\linewidth}
        \centering
        \includegraphics[width=\linewidth]{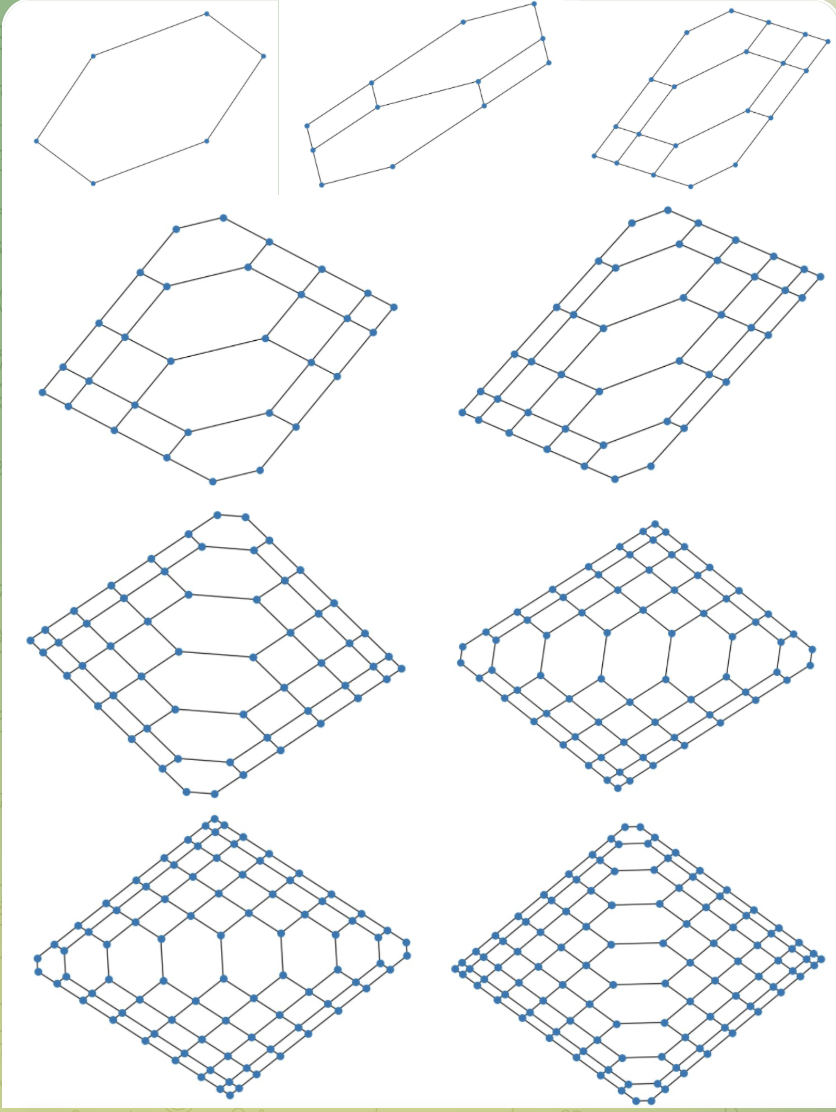}
        \caption{Schreier coset graphs -  Coxeter generators  $S_n/S_{n-2}$. Similarity with lattice is evident.\href{https://www.kaggle.com/code/lilypilly/determinants-for-large-graph-laplacians?scriptVersionId=303973750}{Notebook}. }
        \label{fig4:2dlatticeLike}
    \end{minipage}\hfill
    \begin{minipage}{0.48\linewidth}
        \centering
        \includegraphics[width=\linewidth]{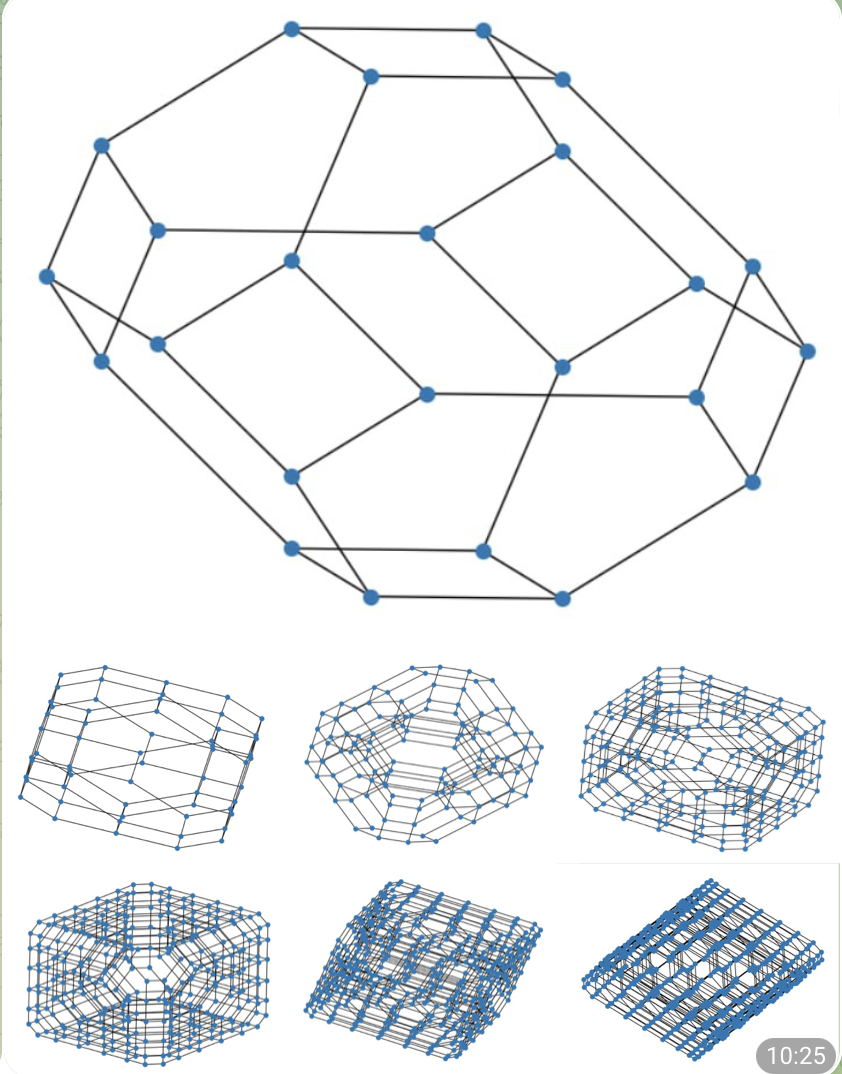}
        \caption{Schreier coset graphs -  Coxeter generators  $S_n/S_{n-3}$.
        Similarity with lattice is evident. \href{https://www.kaggle.com/code/lilypilly/determinants-for-large-graph-laplacians?scriptVersionId=305455260}{Notebook}
         }
        \label{figs4:3dlatticeLike1}
    \end{minipage}
\end{figure}

\begin{figure}[H]
    \centering
    \begin{minipage}{0.48\linewidth}
        \centering
        \includegraphics[width=\linewidth]{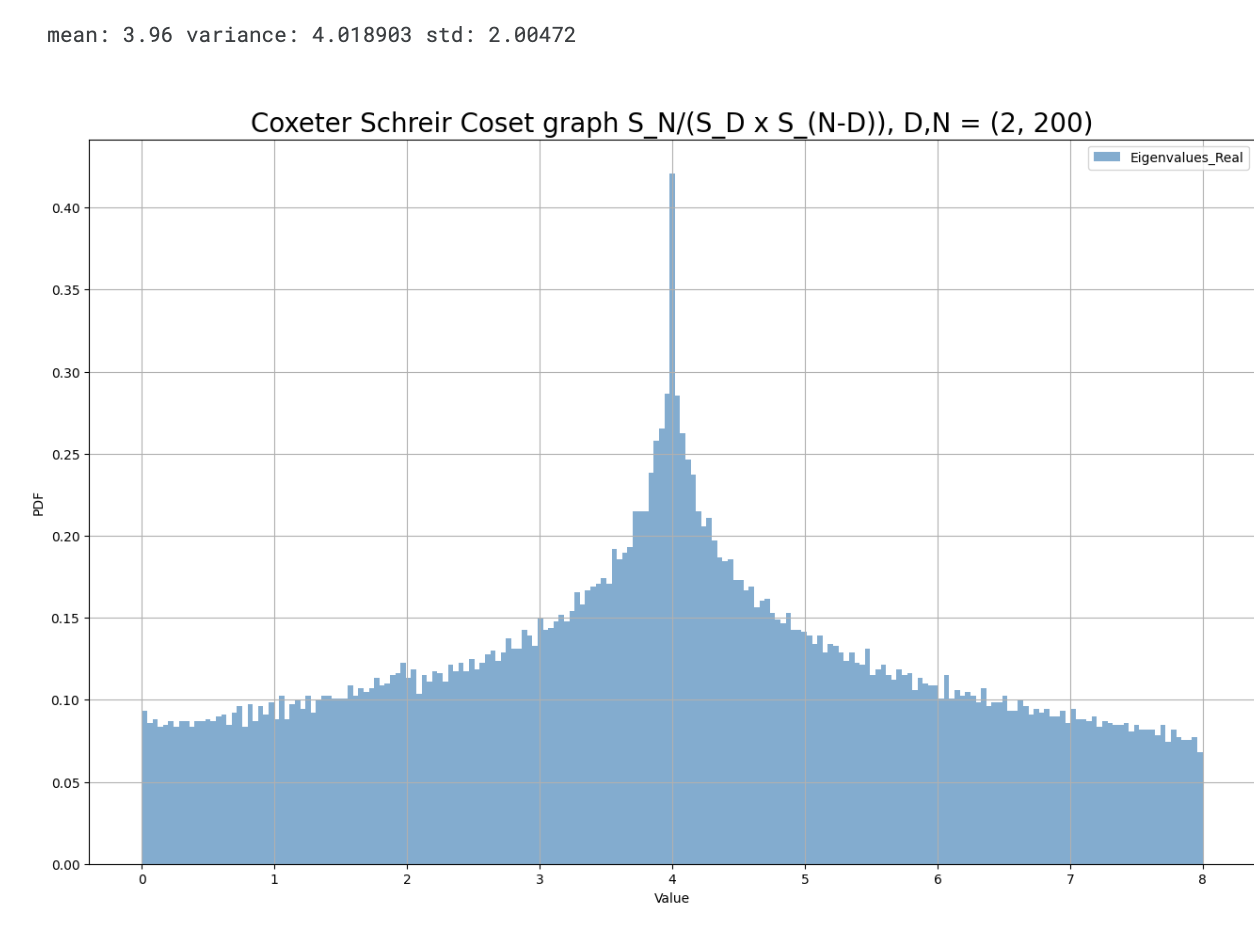}
        \caption{Eigenvalues histogram Coxeter Schreier graph $S_{200}/(S_2\times S_{198})$. 
        Distributions tends to Van Hove distribution (sum of two cosines).
        \href{https://www.kaggle.com/code/alexandervc/cayleypy-coxeter-eigenvals/}{Notebook}.}
        \label{figs3:VanHove200_2}
    \end{minipage}\hfill
    \begin{minipage}{0.48\linewidth}
        \centering
        \includegraphics[width=\linewidth]{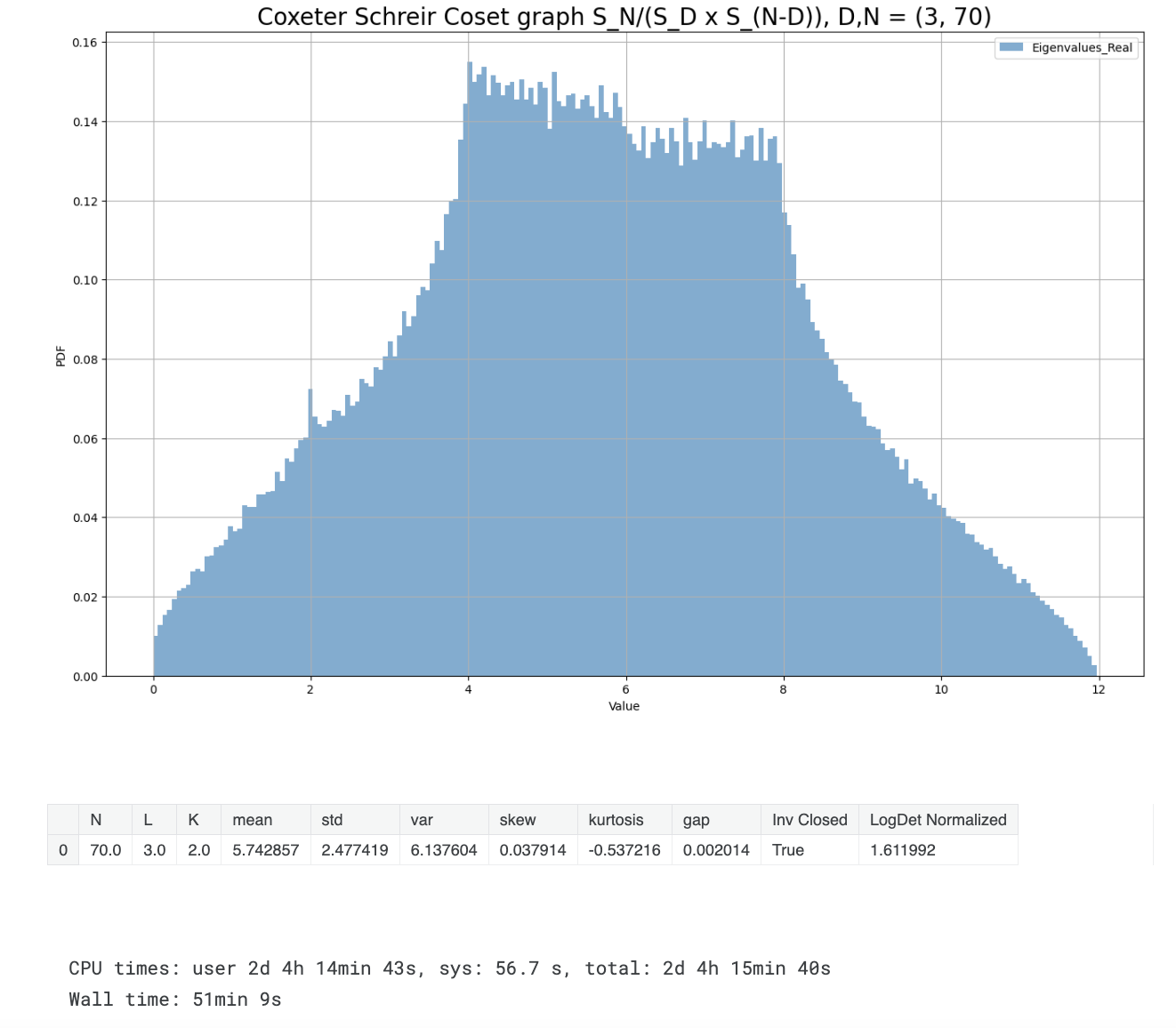}
        \caption{Eigenvalues histogram Coxeter Schreier graph $S_{70}/(S_3\times S_{67})$
        Distributions tends to sum of three cosines.
        \href{https://www.kaggle.com/code/alexandervc/cayleypy-coxeter-eigenvals/}{Notebook}.}
        \label{figs3:Graph3Spectr}
    \end{minipage}
\end{figure}

\clearpage
\subsection{Limit to field theory string-like model. Spanning trees for many spins  up and down sector }
In the present subsection
analyze limit $S_n/ (S_{k_1} \times ... \times S_{k_j} )$
for $k_j/n = p_j$ fixed, and $n\to \infty$.

\textbf{Overview - math part.} We propose that eigenvalue distribution
tends to Gaussian, propose formulas
for mean, variance
and hence conjecture  the number of  spanning trees (Laplacian determinant),
and also spectral gap. Numeric simulations are provided to support.
It is worth mentioning why Benjamini–Schramm technique is not applicable in the setup of the limit considered here.
It is because the degrees of the graph nodes tends to infinity, thus it is obvious that there is no possible limit to usual graph. In some sense what is going on is defining the "renormalization" of the UV-divergent limit. 

\textbf{Overview - string part.}
We also conjecture that limiting theories
are conformal field theories,
which are almost the standard 1-dimensional scalar theories which here play the role of the string-like dual model. What is non-standard
that in general the worldsheet has a non-standard shape in general - the shape of the polygon which appears in our discrete duality.  
For example for $S_n$ itself (without quotients) it is triangle. 

Appearance of the 1d boson is not surprising
-- the limit of XXX is described by such QFT
goes back at least to \cite{luttinger1963exactly} and standard in modern literature.
More surprising is appearance of non-standard worldsheets,  relation to discrete duality  and predictions of spectrum of conformal dimensions.

In that sense we give an answer what CFT may be viewed in some sense as dual to the standard 1-dimensional scalar compactified on a radius $R$ circle -- in our setup it is the CFT 
arising from the graph $S_n/(S_d\times S_{n-d}) $, where the normalized area 
of the dual polygon plays role of $R^2$, i.e. 
$R^2 = d(n-d)/n^2 $, where $n \to \infty$, in the other words this is XXX Heisenberg spin chain on subsector with $k$ spins up, out of $n$  total spins.
The fact again goes back to 
\cite{luttinger1963exactly}, but what
is new is to put into the framework of discrete analogue of AdS/CFT duality.

It it worth to compare the limit in the present section with the previous from string theory perspective -  the picture is surprisingly simple.
Consider $S_n/(S_d \times S_{n-d})$
the dual polygon in our duality 
is rectangle $d \times (n-d)$,
the limit in the previous section
is: $d$ fixed $n\to \infty$,
think of the rectangle as a worldsheet,
imagine that one of the coordinates disappear - so strings reduces to particle, now imagine it disappeared not completely but has $d$ discrete possible values - so we get $d$ independent particles - that is exactly what happened in the previous subsection.
In the present section both $d,n$ goes to infinity and we get usual worldsheet,
so appearance of usual string is natural.

Our duality conjecture implies the answer
on the spectrum of conformal dimensions
for CFT arising as limits of XXX spins 
chains with various spin configurations -  it corresponds  to the spectrum of the Laplace operators
on the planar polygons which
shape is defined by $d_i$ in particular
for the full $S_n$ it is triangle,
and $S_n/(S_d\times S_{n-d}) $
to rectangle, general $d_i$
correspond to cutting down 
small triangles of with edges $d_i$
from the diagonal of the big triangle.

\begin{Conj}
    [Guassian spectrum]\label{conj:Gauss2}
The spectrum of the Laplacian
for 
the Schreier graph
$S_n/ (S_{k_1} \times ... \times S_{k_j} )$
(XXX Hamiltonian)
tends to the Gaussian distribution
in the limit
 $k_i/n = p_i$ fixed, and $n\to \infty$.
\end{Conj}
Informal explanation of the Gaussianity
is quite simple and the following: graph Laplacian coincides with $H_{XXX}$ which up to constants is $  \sum_{i}  \sigma_{i,i+1}$. Where
$\sigma_i$ are commuting $|i-j|>1$,
moreover they acts on independent
sets of spins. Thus eigenvalues 
are sums of individual eigenvalues,
thus leading to Gaussian in large $n$ limit. The effect of the non-commutativity for $j=i+1$ disappears at large $n$. That is closely parallel to classical central limit theorems,
where independence basically appears
from variables acting on different coordinates of the probability space.
And weak dependence of variables do not spoil the Gaussian at the limit.
The question has been studied a lot in related contexts  ...todo-references...

\begin{proposition}
    Mean value equals to $\frac{2}{n}(\sum_{i<j} k_i k_j) $, that means
in grows linearly with $n$. 
In particular for the case of $S_n$
(without quotients) it is $(n-1)$,
for $S_n/(S_d \times S_{n-d})$ it is
$\frac{2}{n} d(n-d) $. 
\end{proposition}
\textbf{Small miracle.}
Factor $(\sum_{i<j} k_i k_j)$
coincides with diameter of the corresponding Schreier graph which is the area of the dual polygon according to our duality. It seems that fact is specific to Coxeter group,
it resembles the classical miracle:
  number of reflections
coincide with length of the longest element for Coxeter groups. 

\begin{proposition}
Variance for the case of $S_n$
(without quotients) is again $(n-1)$ (same as mean),
for $S_n/(S_d \times S_{n-d})$ it is
$\frac{(2d(n-d))^2 + 2d(n-d)(n^2-2n)}{n^2(n-1)} $, which is linear in $n$ in leading order. 
In particular for $n$ even and $d = n/2$,
variance is $3/4n - 1/4 - \frac{1}{4(n-1)}$.
\end{proposition}
In general we expect that variance always grows linearly with respect to $n$ in leading order in $n$. 

Proofs are standard.

Remark. For the case of $S_n$ (without quotients) distribution is symmetric that means skewness and higher analogs exactly vanishes for finite $n$,  not just asymptotically. 

\begin{Conj}
    [Spectral gap]\label{conj:gap}
The spectral gap (i.e. value of the first non-zero eigenvalue of the Laplacian) tends to zero  as $c/n^2$,
for some constant $c$.
\end{Conj}
For example for $S_n$ (without quotients)
the smallest eigenvalue is 
$\lambda_1 = 2 \left( 1 - \cos \left( \frac{\pi}{n} \right) \right)$  
because it comes from the standard representation of $S_n$, it corresponds to line graph,  where spectrum is easy to see: $\lambda_k = 2 \left( 1 - \cos \left( \frac{\pi k}{n} \right) \right), \quad k = 1, \dots, n-1$.

\begin{Conj}
    [Determinant. Spanning trees]\label{conj:det2}
    Logarithm of the number of spanning trees (which is related to determinant of XXX Hamiltonian by Kirchhoff's theorem) normalized by the graph size -- 
    behaves as $log(n)$ at leading order.
    This corresponds to central charge 1 from the CFT perspective, confirming the free boson as the dual model. 
    
    The subleading term (independent of $n$  -- "universal") corresponds to twice the area of the dual polygon $log(2\sum_{i<j}(k_i/n)(k_j/n)) $. Which should be thought as logarithm of "renormalized" number of spanning trees. 
    
    For $S_n/(S_d \times S_{n-d})$, it is expected: 
    
$\frac{1}{\binom{n}{d}-1} \log \tau_n = \log(n) + \log(2a(1-a)) - \frac{1 + 2a(1-a)}{4a(1-a)} \frac{1}{n} + O(n^{-2})$, where $a = d/n$

    For $S_n$ (no quotients), it is expected: 

    $\frac{1}{n!-1} \log \tau_n = \log(n-1) - \frac{1}{2(n-1)}
    - \frac{3}{4(n-1)^2} - \frac{15}{6(n-1)^3} -  ... \frac{(2m-1)!!}{2m (n-1)^{m} } - ...
    +   O(exp{(-n}))$.
    The constant terms disappears
    since $0=log(1) = log(2\times 1/2)$,
    here $1/2$ is triangle area,
    taken twice gives 1, logarithm gives zero.
\end{Conj}
The arguments are quite simple.
Assuming Gaussian distribution 
$log det$ is approximately $log(Mean)$,
the next orders of the approximation
are given by the standard formula
$\frac{1}{N} \log \det' H_n \approx \log \mu - \frac{\sigma^2}{2\mu^2} - \frac{3\sigma^4}{4\mu^4} - ...
 \frac{(2m-1)!! \sigma^{2m}}{2m \mu^{2m} } \dots$,
which simply comes just from the $log(m + delta)=
log(m(1+delta/m))$ = $log(m) + delta/m-delta^2/2m^2...$, then taking mean (over $delta$ which is $N(0,1)$) - odd degrees
disappear, even degrees give $(2m-1)!!$. 

To the best of our knowledge number of spanning trees for such graphs was not known before.
For example for permutohedron graph (i.e. $S_n$
without quotients) it is explicitly mentioned
\cite{ehrenborg2021number}: 
"The associated graph is known as
the 1-skeleton of the permutahedron. An open problem is to determine the number of spanning trees for this graph."
For small values of $n$ 
for permutohedron graph numbers are given at 
\href{https://oeis.org/A337083}{A337083}.

\begin{Conj}
    [Limiting 1d bosonic model ("string side")]\label{conj:boson}
    In the limit above the quantum theory on the graph (XXX chain with corresponding spins set) tends to
    CFT described by 1d bosonic free model
    compactified on $R$ ,
    where 
    $R^2 = \sum_{i<j} k_i k_j $ -- 
    the area of the dual polygon.
    However the worldsheet
    has an unusual shape - exactly the dual polygon.

    For the case $S_{n}/(S_d \times S_{n-d})$ the worldsheet is the standard rectangle, but for the example of 
    $S_n$ itself (no quotients) it is a triangle. 
\end{Conj}

\begin{Conj}
    [Conformal dimensions of primary fields]\label{conj:cftdims}
    In the setup above
    the spectrum of conformal dimensions
    is equal to the spectrum of the Laplacian on the dual polygon.
\end{Conj}
Motivation 1. 
It is exactly the relation in conventional AdS/CFT where 
spectrum of conformal dimensions on CFT side can be read of the Laplacian on AdS side, since SL(2) subgroup of CFT it mapped into isometries of AdS and thus conformal algebra Laplacian is mapped to AdS geometric Laplacian.

Motivation 2. Saleur's formula \cite{itzykson1998conformal}.
The Saleur (or Coulomb Gas) formula expresses the conformal dimensions as $\Delta_{n,m} = \frac{1}{4}(g n^2 + \frac{m^2}{g})$, which represents the energy spectrum of a free boson on a circle of radius $R \propto \sqrt{g}$. 
This formula matches  the formula
for the spectrum of the Laplacian 
on rectangle, indeed
the eigenvalues of the Laplacian on a rectangle with side lengths $L_x$ and $L_y$ are given by $\lambda_{n_x, n_y} = \pi^2 \left( \frac{n_x^2}{L_x^2} + \frac{n_y^2}{L_y^2} \right)$ for $n_x, n_y \in \mathbb{Z}^+$.

\textbf{Numeric simulations.}
Numeric simulations can be found in the \href{https://www.kaggle.com/code/alexandervc/cayleypy-coxeter-eigenvals}
{notebook}. Let us present few figures, more can be found there.

\begin{figure}[H]
    \centering
    \begin{minipage}{0.48\linewidth}
        \centering
        \includegraphics[width=\linewidth]{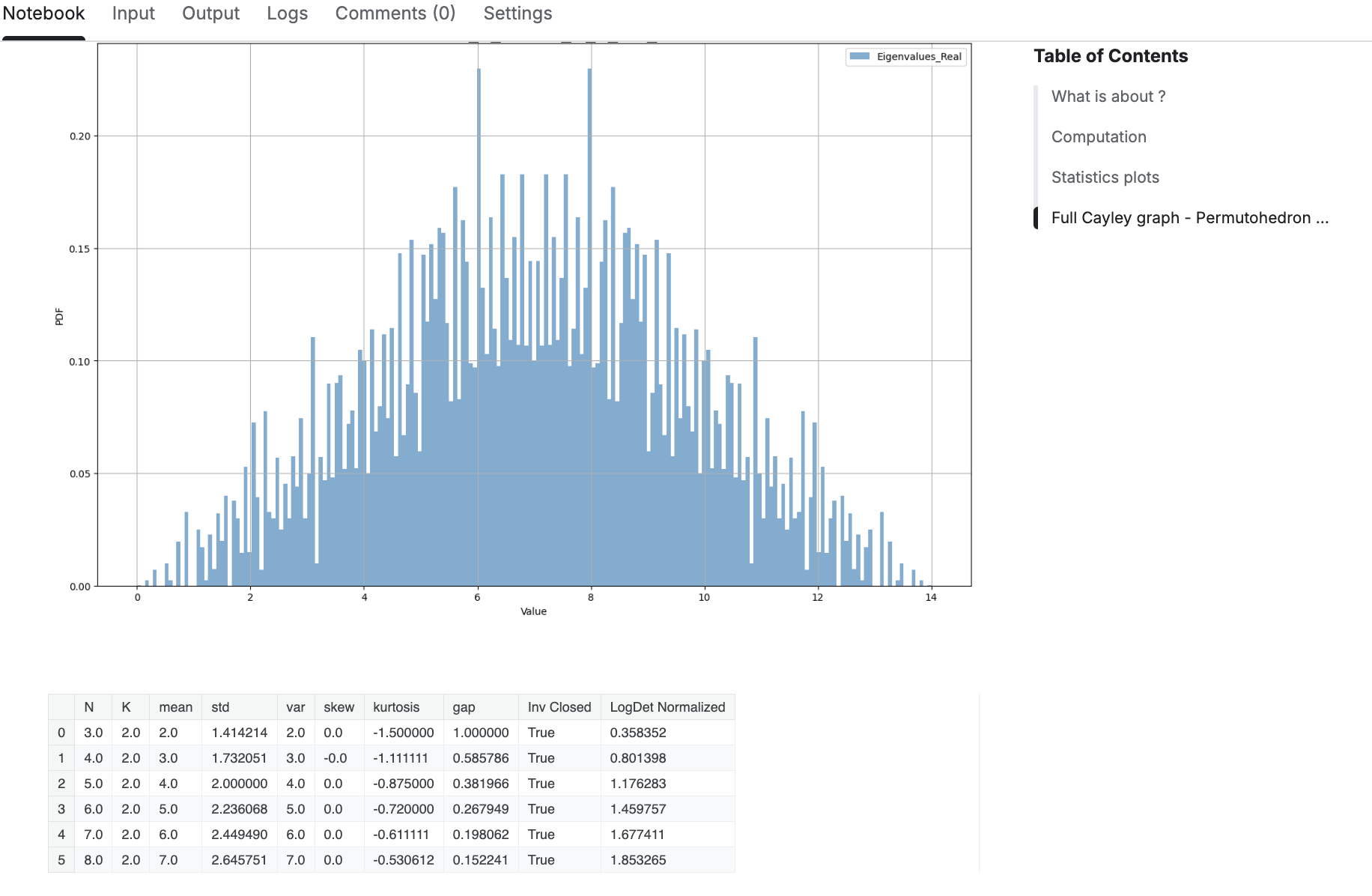}
        \caption{Eigenvalues histogram  Coxeter Cayley graph $S_8$ (Permutohedron graph). 
        Distributions tends to Gaussian.}
        \label{figs3:HistPermutohedron}
    \end{minipage}\hfill
    \begin{minipage}{0.48\linewidth}
        \centering
        \includegraphics[width=\linewidth]{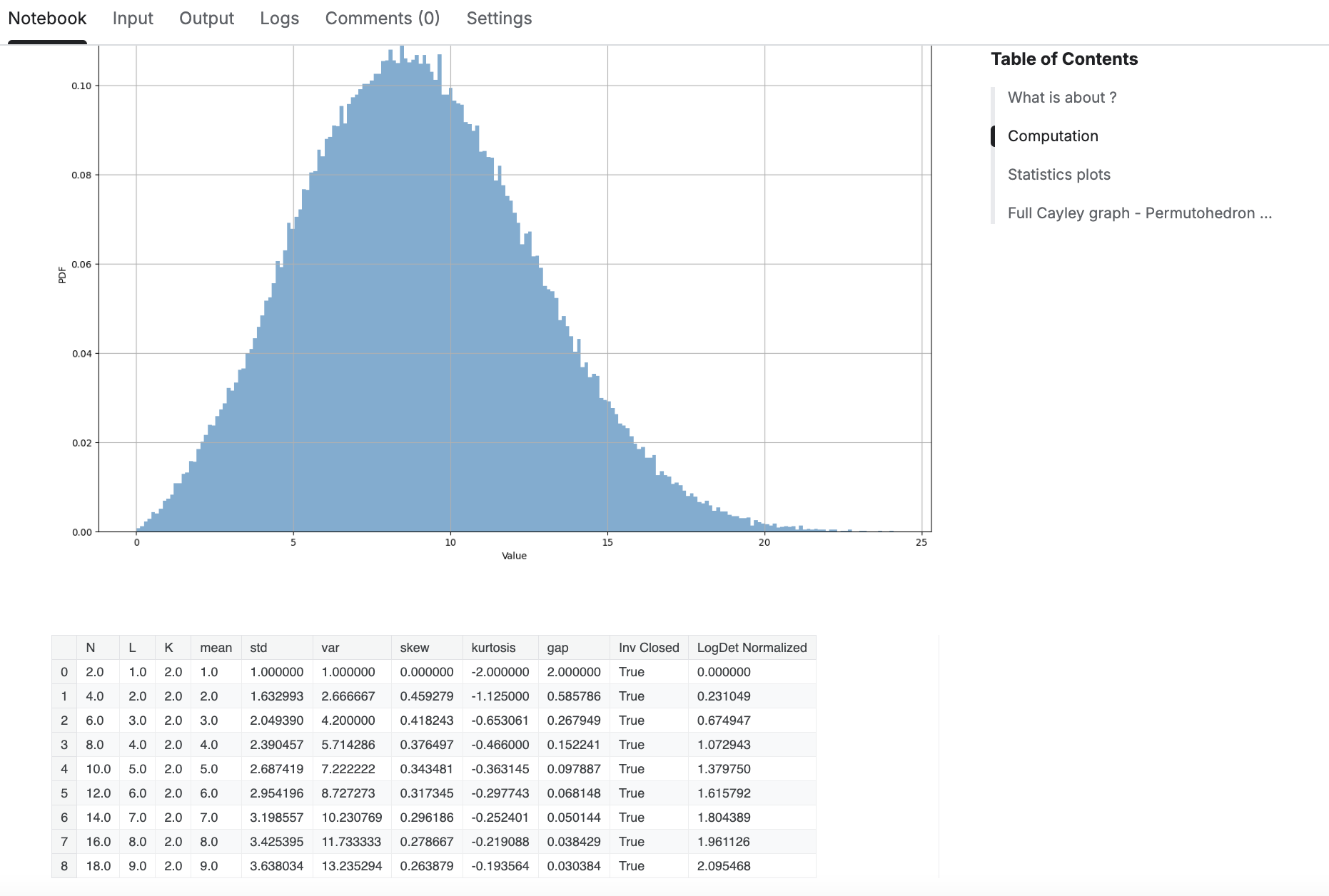}
        \caption{Eigenvalues histogram for Laplacian of Schreier graph $S_{18}/(S_9 \times S_9)$.
        Distributions tends to Gaussian.}
        \label{figs3:histEigsCoset}
    \end{minipage}
\end{figure}

\begin{figure}[H]
    \centering
    \begin{minipage}{0.48\linewidth}
        \centering
        \includegraphics[width=\linewidth]{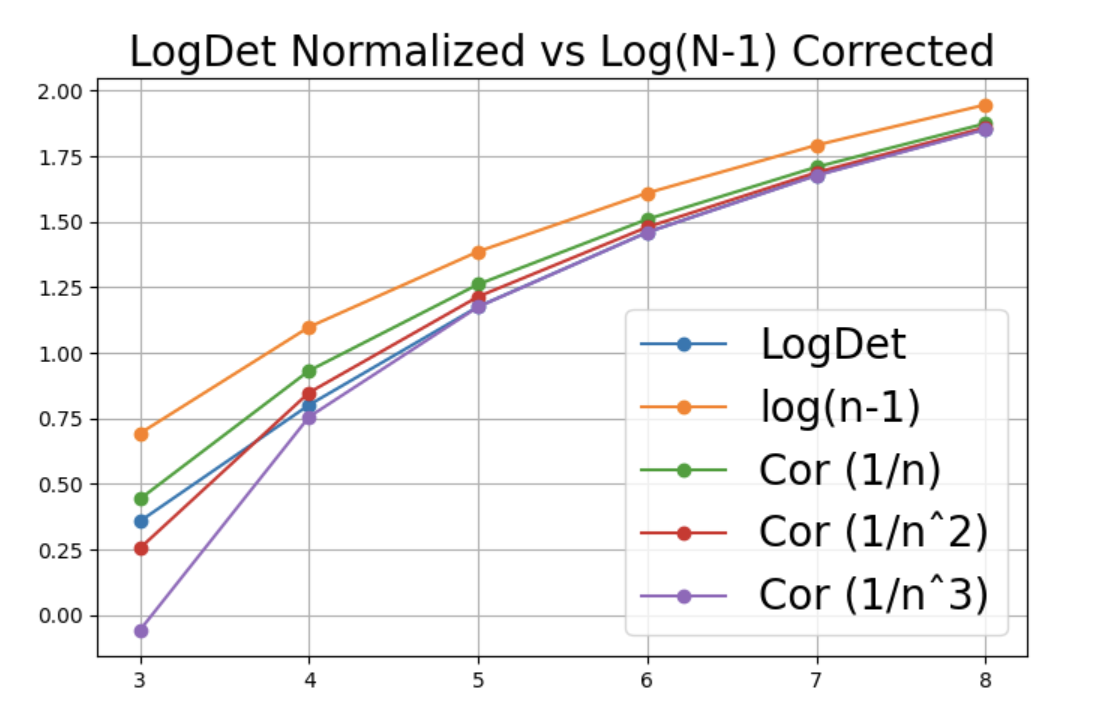}
        \caption{Logarithm of spanning trees number divided by graph size. 
        Coxeter Cayley graphs $S_n$ (Permutohedron graph). X-axis is $n$. }
        \label{fig3s:LogDetPermutohedron}
    \end{minipage}\hfill
    \begin{minipage}{0.48\linewidth}
        \centering
        \includegraphics[width=\linewidth]{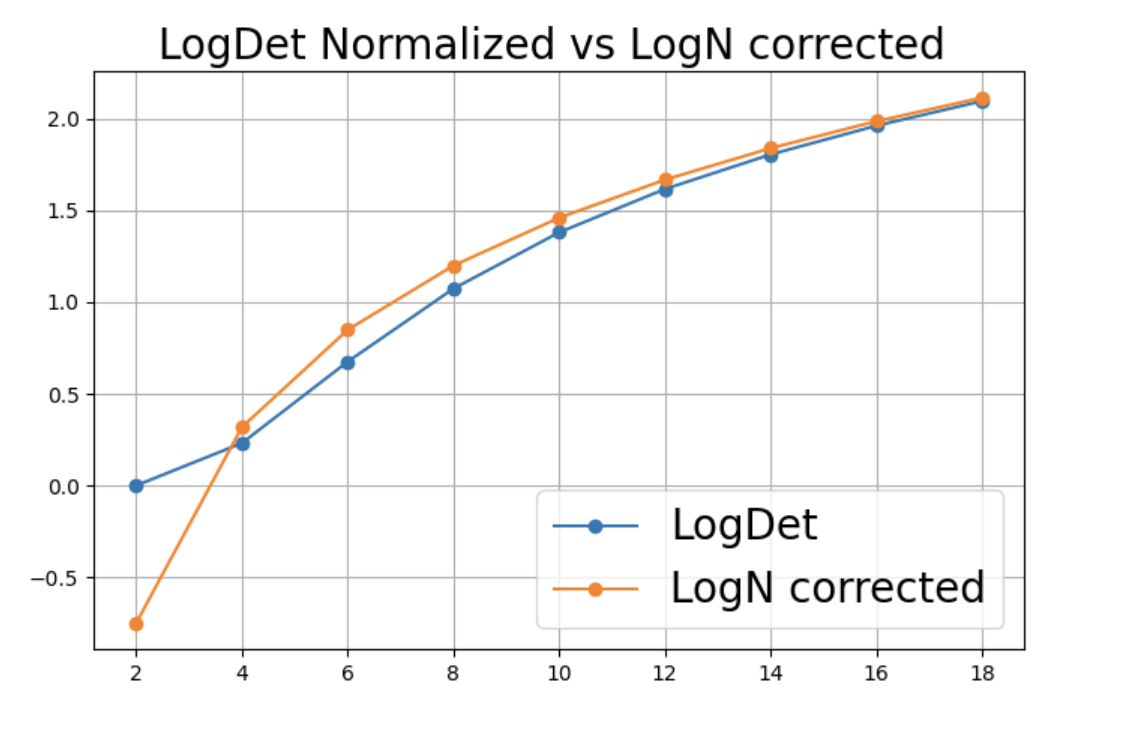}
        \caption{Logarithm of spanning trees number divided by graph size. 
        Schreier graph $S_{n}/(S_{n/2} \times S_{n/2})$.
         X-axis is $n$.}
        \label{figs3:LogDetCoset}
    \end{minipage}
\end{figure}

\begin{figure}[H]
    \centering
    \begin{minipage}{0.48\linewidth}
        \centering
        \includegraphics[width=\linewidth]{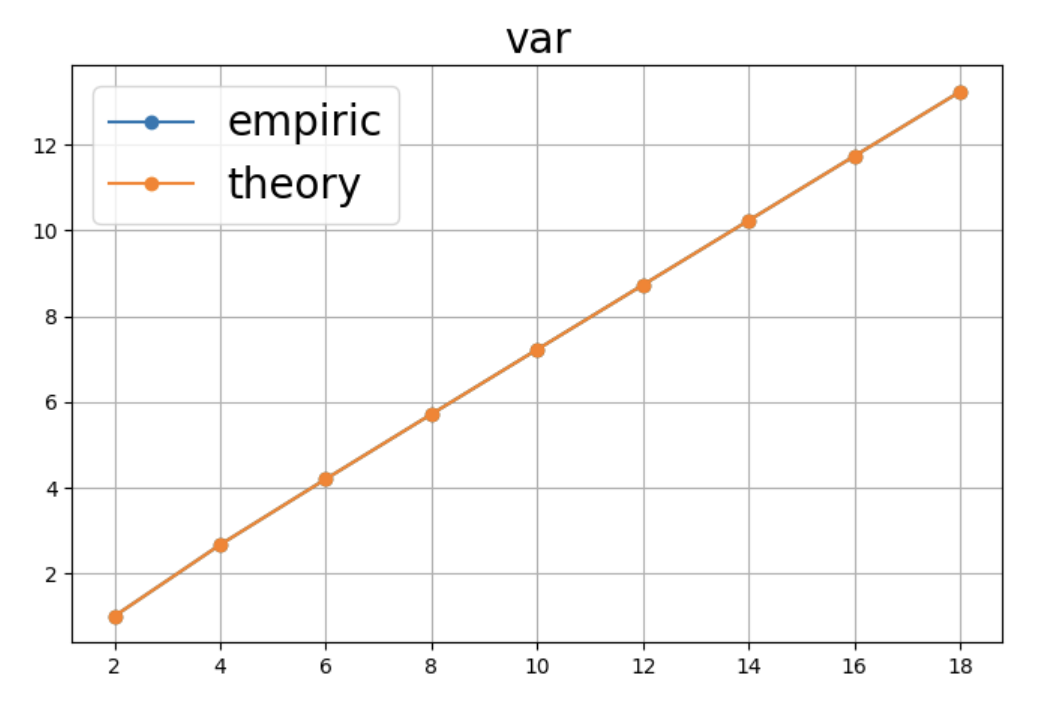}
        \caption{Schreier coset graphs -  Coxeter generators  $S_n/(S_{n/2}\times S_{n/2})$.
        Variance - exact match with proposed theoretical formula.}
        \label{fig3:VarianceCoset}
    \end{minipage}\hfill
    \begin{minipage}{0.48\linewidth}
        \centering
        \includegraphics[width=\linewidth]{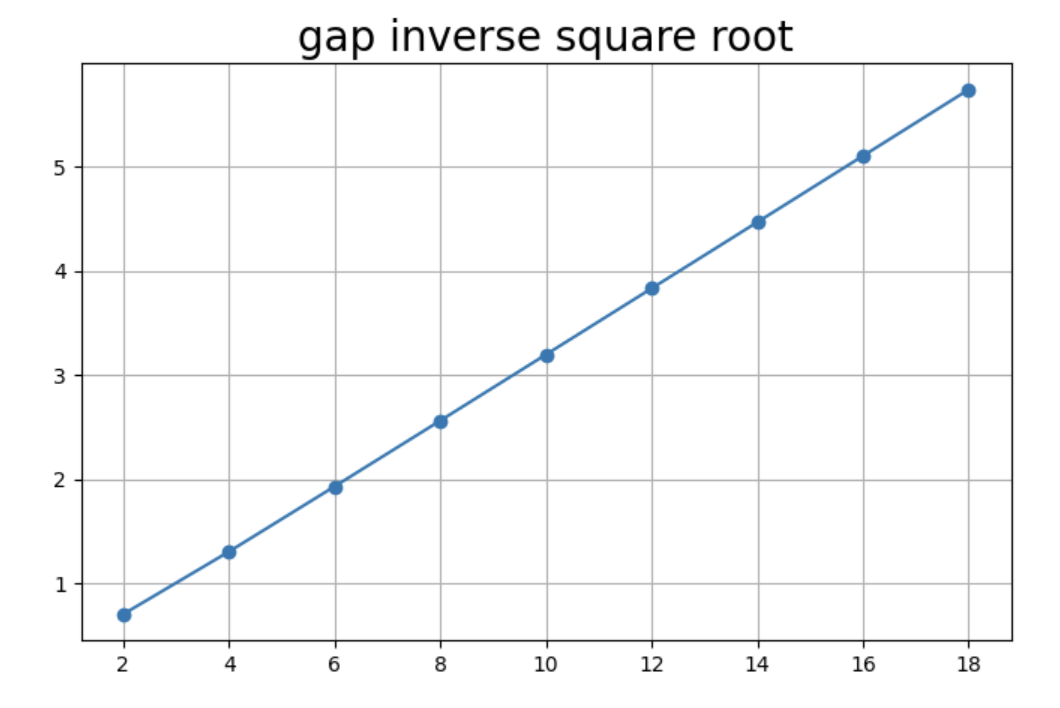}
        \caption{Schreier coset graphs -  Coxeter generators  $S_n/(S_{n/2}\times S_{n/2})$.
        Spectral gap inverse square root - behaves like a line.
        So gap is $c/\sqrt(n)$
         }
        \label{figs3:GapCoset}
    \end{minipage}
\end{figure}


\clearpage
\section{Neighbor transpositions extended by
  \texorpdfstring{$(0,n-1)$}{(0,n-1)}
  (\texorpdfstring{``wrapped''}{wrapped} or \texorpdfstring{``affine''}{affine} case)}

\textbf{Background and related work.}
A small modification of the nearest-neighbor transposition generators—namely, adding the transposition $(0,n-1)$—leads to a more nontrivial Cayley graph, which we refer to as the \emph{cyclic}, \emph{wrapped}, or \emph{affine} case.
In \texttt{CayleyPy}, these generators are denoted
\href{https://cayleypy.github.io/cayleypy-docs/generated/cayleypy.PermutationGroups.html#cayleypy.PermutationGroups.cyclic_coxeter}{\texttt{Cyclic Coxeter}}
and form a special case of
\href{https://cayleypy.github.io/cayleypy-docs/generated/cayleypy.PermutationGroups.html#cayleypy.PermutationGroups.wrapped_k_cycles}{\texttt{wrapped\_k\_cycles}} with $k=2$.
The simplest quotient $S_n/S_{n-1}$ yields a cycle graph, corresponding to the affine Dynkin diagram of type $\widetilde{A}_{n-1}$, and may thus be viewed as an affine analogue of the nearest-neighbor transposition graph.
The Laplacian of this graph corresponds to a \emph{periodic} or \emph{closed} spin chain, which is well known to be related to affine Lie algebras.

The diameter of the graph was finally established to be $\lfloor n^{2}/4 \rfloor$ in \cite{vanZuylen2016cyclic}, building on earlier works cited therein; see also an alternative argument by Ville Salo~\href{https://mathoverflow.net/a/359699/10446}{(MO359699)}. A related recent result~\cite{alon2025circular} establishes the circular version, yielding $\lfloor (n-1)^{2}/4 \rfloor$. The graph was first studied in depth in the foundational work~\cite{jerrum1985complexity}.

Let us note that the corresponding \href{https://en.wikipedia.org/wiki/Ehrhart_polynomial#Ehrhart_series}{$H$-polynomial} is given by
\[
\frac{2x^{4} + 4x^{3} + 2x^{2}}{(1 - x^{2})^{3}},
\]
which is symmetric, has positive coefficients, and whose roots all have modulus~$1$ ("Riemann conjecture holds").

\subsection{Duality for the Wrapped Case}
\label{sec:wrapped_duality} 

In this section we consider a modification of the binary Cayley/Schreier graph
studied above, obtained by adding an extra generator that swaps the first and the
last positions. This produces a natural \emph{wrapped} (or cyclic) version of the
model and leads to a new geometric interpretation of shortest-path distances in
terms of areas between lattice curves minimized over cyclic shifts.

\subsubsection{The wrapped binary graph}

As before, the set of nodes is defined by vectors with $0$ and $1$ components, such that there are $l$-zeros:
\[
Gr^{ \mathrm{aff} }_{l,n}=\{x\in\{0,1\}^n:\ \sum_{i=0}^{n-1} x_i=n-l\},
\]

We will distinguish the sorted vector denoting it by $e$:
\[
e:=0^{\,l}1^{\,n-l}.
\]

Two vertices $x,y\in Gr^{ \mathrm{aff} }_{l,n}$ are connected by an edge if either
\begin{itemize}
\item $y$ is obtained from $x$ by an adjacent swap $01\leftrightarrow 10$, or
\item $y$ is obtained from $x$ by applying the transposition
\[
t=(0,n-1),
\]
which exchanges the first and the last entries:
\[
(t\cdot x)_0=x_{n-1},\quad (t\cdot x)_{n-1}=x_{0} \quad (t\cdot x)_i=x_i
\ \ (1\le i\le n-2).
\]
\end{itemize}

We denote by $d_{\mathrm{wrap}}(x,y)$ the shortest-path distance in
$Gr^{ \mathrm{aff} }_{l,n}$.

The additional generator $(0,n-1)$  
allows symbols to move across the boundary, leading to a metric behavior that is
more complicated than  the non-wrapped case.

\subsubsection{Lattice-path encoding and cyclic shifts}

Each word $x\in Gr^{ \mathrm{aff} }_{l,n}$ canonically determines a monotone lattice path
\[
P(x):(0,0)\longrightarrow (l,n-l),
\]
obtained by reading $x$ from left to right and applying the rule
\[
0\mapsto (1,0),\qquad 1\mapsto (0,1).
\]

Let $\Area(x)$ denote the number of unit lattice squares below $P(x)$ inside the
rectangle $[0,l]\times[0,n-l]$. As shown earlier,
\[
\Area(x)=\Inv(x),
\]
where $\Inv(x)$ is the inversion number of the binary word $x$.

In the wrapped setting, cyclic shifts arise naturally. For
$r\in\{0,1,\dots,n-1\}$ define the cyclic shift
\[
(\mathrm{sh}_r(x))_i := x_{\,i+r \!\!\!\pmod n}.
\]
Each shift determines a path $P(\mathrm{sh}_r(x))$ and an associated area
\[
A_r(x):=\Area(\mathrm{sh}_r(x)).
\]

\subsubsection{Wrapped duality hypothesis}

The main conjectural statement for the wrapped case is the following.

\begin{hyp}[Wrapped duality]
For every $x\in Gr^{ \mathrm{aff} }_{l,n}$, the distance from the sorted vertex $e$ in the wrapped
graph satisfies
\[
d_{\mathrm{wrap}}(e,x)\;=\;\min_{r\in\{0,\dots,n-1\}} A_r(x).
\]
\end{hyp}

Equivalently, if we define the prefix-sum function
\[
f_x(i):=\sum_{j=0}^{i} x_j,\qquad i=0,1,\dots,n-1,
\]
and let
\[
f_{\min}(i):=\min_{r} f_{\mathrm{sh}_r(x)}(i),
\]
then the conjecture asserts that
\[
d_{\mathrm{wrap}}(e,x)
\;=\;
\sum_{i=0}^{n-1}\bigl(f_x(i)-f_{\min}(i)\bigr),
\]
i.e.\ the distance equals the area between the step function associated with $x$
and the lower envelope of all its cyclic shifts.

We illustrate the wrapped construction and highlight a subtlety of the conjecture.

\begin{example}

Let $n=6$, $k=3$, and
\[
e=000111.
\]
Consider the word
\[
x=101010\in X_{6,3}.
\]

\textbf{Non-wrapped area.}
The inversion number of $x$ is
\[
\Inv(x)=6,
\]
hence $\Area(x)=6$.

\textbf{Cyclic shifts.}
The cyclic shift $\mathrm{sh}_1(x)=010101$ has inversion number
\[
\Inv(010101)=3.
\]
Since further shifts repeat these patterns, we obtain
\[
\min_{r} A_r(x)=3.
\]

\textbf{Wrapped distance.}
Using the additional generator $(0,5)$, we have
\[
000111 \xrightarrow{\, (0,5)\,} 100110
\xrightarrow{\, (2,3)\,} 101010=x.
\]
Thus
\[
d_{\mathrm{wrap}}(e,x)\le 2.
\]

So, we have
\[
d_{\mathrm{wrap}}(e,x)=2
\qquad\text{while}\qquad
\min_{r} A_r(x)=3.
\]
\end{example}

This shows that the naive wrapped-duality formula does not hold verbatim with the
linear notion of area. Instead, the wrapped generator introduces a boundary effect
that allows shortcuts not captured by minimizing the usual area over cyclic
shifts.

From the duality perspective, adding the generator $(0,n-1)$ transforms the linear
correspondence
\[
\text{vertex}\;\longleftrightarrow\;\text{lattice path}
\;\longleftrightarrow\;\text{area}
\]
into a wrapped version in which a vertex corresponds to a \emph{family} of cyclic
paths and the graph distance reflects a modified area functional that incorporates
boundary moves. Determining the correct geometric invariant for the wrapped case
remains an interesting open problem.

\subsection{\texorpdfstring{Diameters for the Schreier coset graphs ``few-coincide'': $S_n/S_{d}$}{Diameters for the Schreier coset graphs "few-coincide": S_n/S_d}}

The central state and the initial state are given as \[(0, 1, 2, \dots, n-d-1, n-d, n-d, \dots, n-d)\] and ?? respectively, where $(n-d)$ appears $d$ times.

Here we present some partial results on diameters of the Schreier coset graphs of the form $S_n/S_D$,
which can be equivalently described as graphs with nodes corresponding to vectors where $D$ elements coincide.
In CayleyPy we define them by setting "central\_state" to be 
 $(0,1,2,...n-D-1,n-D, n-D, ... , n-D)$ ($D$ coincide at the end). 
The generators are the  same - cyclic Coxeter (equivalently "wrapped 2 cycles").

\begin{Conj}
For $D \ge 2$, the conjecture was found to be:
\[
D(n,d)=\frac{d-1}{2}\,n-\delta(n,d),
\]
where the offset term $\delta(n,d)$ is given by
\[
\delta(n,d)=
\begin{cases}
\dfrac{d(d-2)}{4}, & \text{if $d$ is even and $n$ is even},\\[8pt]
\dfrac{d(d-2)}{4}+\dfrac12, & \text{if $d$ is even and $n$ is odd},\\[8pt]
\left(\dfrac{d-1}{2}\right)^2, & \text{if $d$ is odd}.
\end{cases}
\]
\end{Conj}

\begin{longtable}{lcccccccc}
\caption{Experimental data ($d$ different, for $2$-cycles)}\\

\toprule
\textbf{n} & \textbf{d=2} & \textbf{d=3} & \textbf{d=4} & \textbf{d=5} & \textbf{d=6} & \textbf{d=7} & \textbf{d=8} & \textbf{d=9}\\
\midrule
\endfirsthead

\toprule
\textbf{n} & \textbf{d=2} & \textbf{d=3} & \textbf{d=4} & \textbf{d=5} & \textbf{d=6} & \textbf{d=7} & \textbf{d=8} & \textbf{d=9}\\
\midrule
\endhead

\midrule
\multicolumn{9}{r}{\emph{Continued on next page}} \\
\endfoot

\bottomrule
\endlastfoot

2  &      &      &      &      &      &      &      &      \\
3  & 1    & 2    &      &      &      &      &      &      \\
4  & 2    & 3    & 4    &      &      &      &      &      \\
5  & 2    & 4    & 5    & 6    &      &      &      &      \\
6  & 3    & 5    & 7    & 8    & 9    &      &      &      \\
7  & 3    & 6    & 8    & 10   & 11   & 12   &      &      \\
8  & 4    & 7    & 10   & 12   & 14   & 15   & 16   &      \\
9  & 4    & 8    & 11   & 14   & 16   & 18   & 19   & 20   \\
10 & 5    & 9    & 13   & 16   & 19   & 21   & 23   & 24   \\
11 & 5    & 10   & 14   & 18   & 21   & 24   & 26   & 28   \\
12 & 6    & 11   & 16   & 20   & 24   & 27   & 30   & 32   \\
13 & 6    & 12   & 17   & 22   & 26   & 30   & 33   & 36   \\
14 & 7    & 13   & 19   & 24   & 29   & 33   & 37   & 40   \\
15 & 7    & 14   & 20   & 26   & 31   & 36   & 40   & 44   \\
16 & 8    & 15   & 22   & 28   & 34   & 39   & 44   & 48   \\
17 & 8    & 16   & 23   & 30   & 36   & 42   & 47   & 52   \\
18 & 9    & 17   & 25   & 32   & 39   & 45   & 51   &      \\
19 & 9    & 18   & 26   & 34   & 41   & 48   & 54   &      \\
20 & 10   & 19   & 28   & 36   & 44   & 51   & 58   &      \\
21 & 10   & 20   & 29   & 38   & 46   & 54   & 61   &      \\
22 & 11   & 21   & 31   & 40   & 49   & 57   & 65   &      \\
23 & 11   & 22   & 32   & 42   & 51   & 60   & 68   &      \\
24 & 12   & 23   & 34   & 44   & 54   & 63   & 72   &      \\
25 & 12   & 24   & 35   & 46   & 56   & 66   &      &      \\
26 & 13   & 25   & 37   & 48   & 59   & 69   &      &      \\
27 & 13   & 26   & 38   & 50   & 61   & 72   &      &      \\
28 & 14   & 27   & 40   & 52   & 64   & 75   &      &      \\
29 & 14   & 28   & 41   & 54   & 66   & 78   &      &      \\
30 & 15   & 29   & 43   & 56   & 69   & 81   &      &      \\
\end{longtable}

\section{Bethe ansatz and graph spectrum}

Consider the Schreier graph for $S_n$ with generators adjacent (wrapped) transpositions, i.e. $(12),(23),\dots,(n-1,n),(n1)$, and central state $1\dots10\dots0$ with $k$ ones and $n-k$ zeros. Its vertices are labelled by all possible permutations of this central state. We will assume its adjacency matrix $A$ is defined such that the sum of all entries in each row is the same\footnote{this means we add an appropriate number of self-loops at each vertex depending on how many generators leave invariant the corresponding permutation} and equal to the number of generators, i.e. $n$. Now let us view each string of 0's and 1's as labelling a basis of a subspace with $k$ flipped spins (corresponding to 1's) in a $2^n$-dimensional space. Then we find that the matrix $A$ exactly matches the XXX spin chain Hamiltonian restricted to this subspace, in the form
\beq
    H=\sum_{i=1}^nP_{i,i+1}
\eeq
where the permutation operators switch the entries of our strings. This means that the spectrum of $A$ can be found via Bethe ansatz.

In general the Bethe equations read
\beq
    \left(\frac{u_a+i/2}{u_a-i/2}\right)^L   =\prod_{b\neq a}^M\frac{u_a-u_b+i}{u_a-u_b-i} \ , \ \ \  \ a=1,\dots,M
\eeq
where we identify the length $L$ with our $n$, while $M$ is the number of magnons. We need eigenvalues corresponding to all states with $k$ flipped spins, so we need to consider $M=k$ but also all lower values $M=0,1,\dots,k-1$ as they correspond to states whose \textit{descendants} will have $k$ flipped spins. Adapting the standard Bethe ansatz formula to the normalisation of our Hamiltonian, the eigenvalues are found from
\beq
    E=n-\sum_{j=1}^M\frac{1}{u_j^2+1/4}
\eeq
For instance, a descendant of the vacuum state will have no Bethe roots and give the eigenvalue equal to $n$ which indeed we always observe for the graph.

As an example, we checked this explicitly for the case $n=5,k=2$, i.e. central state 11000. We find that the Bethe ansatz reproduces all the distinct eigenvalues of $A$ (with more care one should be able to match the eigenvalue multiplicities as well). Concretely, here we need to consider: $M=0$ (one eigenvalue), $M=1$ (four solutions, two different eigenvalues) and $M=2$ (many solutions, three different eigenvalues once we restrict as usual to no repeated roots etc).

\subsection{Spectral gap}

With Bethe ansatz we can compute the gap between the largest (equal to $n$) and next-to-largest eigenvalues of the adjacency matrix. We observe (at least for the central state with two 1's and $n-2$ zeroes) it corresponds to a state with 1 magnon whose Bethe root solves
\beq
    \frac{u+i/2}{u-i/2}=e^{2\pi i/n} \ .
\eeq
Computing the energy for it we find the gap (for large enough $n$, unclear starting from which $n$)
\beq
    2\cos(2\pi/n)-2
\eeq
which at large $n$ gives $\simeq -4\pi^2/n^2$. 

 This gives (for large enough $n$) the 2nd eigenvalue of the Laplacian as $\lambda_2=2-2\cos(2\pi/n)$ which at large $n$ gives an upper bound for the diameter $\sim n^4$ -- a huge overestimate. The bound comes from 
 \beq
     d\leq \frac{N-1}{\lambda_2}
    \eeq
where $N$ is the number of graph's vertices, so for us $N=n(n-1)/2$. 

\subsection{Relation to Laplacian spectrum}

Let us relate the above discussion to the spectrum of the graph Laplacian. The Laplacian is defined as the matrix $L=D-A_0$, where $D$ is the diagonal matrix with degrees of the vertices on the diagonal, and $A_0$ is the adjacency matrix with the only off-diagonal nonzero elements that are equal to $1$ if the corresponding vertices are connected. Note this adjacency matrix is defined differently from $A$ in the above Bethe ansatz discussion, as $A$  there also has diagonal elements equal to the number of self-loops at each vertex. However, nicely, the Laplacian $L$ differs from $A$ only by flipping the overall sign and adding a scalar matrix $nI$, namely $L=nI-A$. This is easy to see since for all $i$ we have $A_{ii}=n-D_{ii}$, because all edges at each vertex are either self-loops or contribute to $D_{ii}$. Here we assume that in $D$ we have the degrees of each vertex not counting any self-loops.  We conclude that the spectrum of the Laplacian is also captured by the Bethe ansatz and are given by 
\beq
    E_{Laplacian}=\sum_{j=1}^M\frac{1}{u_j^2+1/4}
\eeq

\subsection{Schreier graph with non-wrapped 2-cycles}

We can also consider the Schreier graph for $S_n$ with central state of $k$ 1's and $n-k$ zeros but taking as generators only the transpositions $(12),(23),\dots,(n-1,n)$, without the last one $(n1)$. In the spin chain language this corresponds to the open spin chain, with Hamiltonian
\beq
    H=\sum_{i=1}^{n-1}P_{i,i+1}
\eeq
The Bethe equations read
    \beq
    \left(\frac{u_a+i/2}{u_a-i/2}\right)^{2L}   =\prod_{b\neq a}^M\frac{u_a-u_b+i}{u_a-u_b-i}\frac{u_a+u_b+i}{u_a+u_b-i} \ , \ \ \  \ a=1,\dots,M
\eeq
where the length $L=n$. The energy (eigenvalue of the Hamiltonian) is given by 
\beq
    E=n-1-\sum_{j=1}^M\frac{1}{u_j^2+1/4}
\eeq
Experimentally (at least for length $n=4,5$) we find some selection rules for Bethe roots, namely we should not allow roots that are:
\begin{itemize}
    \item zero
    \item equal
    \item equal up to sign
\end{itemize}

For the case of only two 1's, this Schreier graph is in fact the same as the quarter Aztec diamond (i.e. roughly a collection of squares filling the region $1\leq x < y\leq n$) if we add to the latter two edges (and two vertices) at the corners. This is shown in figure \ref{fig:qazd}. 

\begin{figure}[H]
    \centering
    \includegraphics[width=0.6\linewidth]{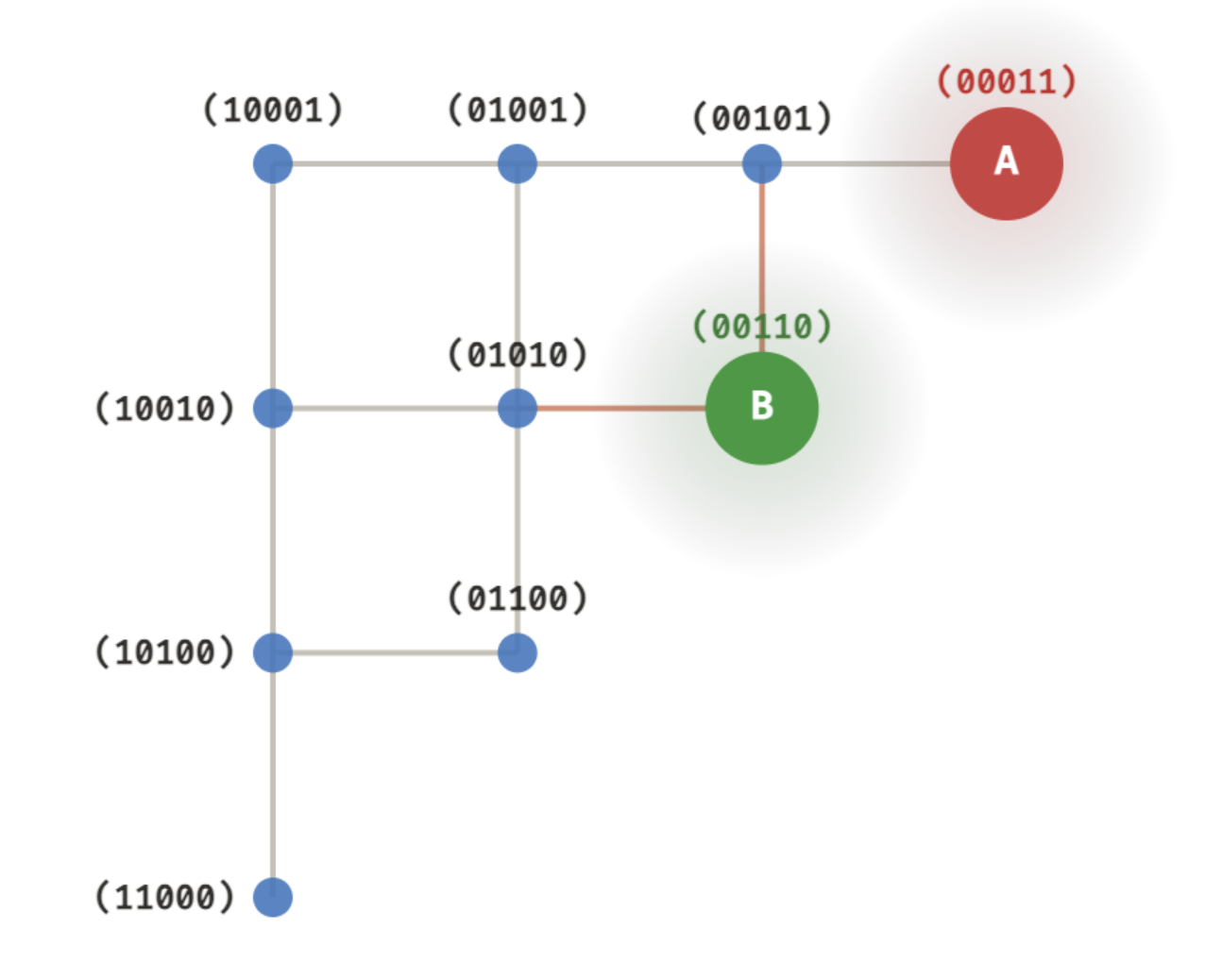}
    \caption{The Schreier graph as a quarter Aztec diamond for $n=5$. Note that $(x,y)$ coordinates are indices of 1s. E.g., for $A=(00011)$  $(x_A,y_A)=(4,5)$ or $B=(00110)$ $(x_B,y_B)=(3,4)$, etc.}
    \label{fig:qazd}
\end{figure}

Notice that adding the extra two edges does not affect the number of spanning trees. We have checked explicitly that taking the product of nonzero eigenvalues for the Schreier graph Laplacian\footnote{defined as for the closed case by appropriately subtracting a diagonal matrix to deal with self-loops at vertices} reproduces the number of spanning trees on the quarter Aztec diamond, sequence \url{https://oeis.org/A007726} in OEIS (what is called the 'order' of the Aztec diamond there is our $n-1$).

\subsection{Relation between Schreier and Johnson graphs}

In the Johnson graph $J(n,k)$ the vertices correspond to subsets of $(1,2,\dots,n)$ consisting of $k$ distinct elements. Two vertices are connected by an edge iff they differ by exactly 1 element.

The Johnson graph $J(n,k)$ is in fact precisely the Schreier graph for $S_n$ with central state being a string of $k$ 1's and $n-k$ 0's, and generators being \textit{all} transpositions. To see this, we can label each string of 0's and 1's by the positions of 1's in it -- this will give precisely a subset of $k$ elements out of $n$ numbers like in $J(n,k)$. It's clear that the rule for placing edges is also the same for both graphs\footnote{since if one string maps to another by a transposition this means we have moved exactly one of the 1's to a new location}.

Let us also note that the Schreier graph for the case of only \textit{adjacent} transpositions will then be a subgraph of $J(n,k)$ (i.e. it will have the same vertices but fewer edges).

\clearpage
\section{Consecutive-(k)-cycles. Quasi-polynomiality, etc.}
\subsection{Section outline}
\textbf{Quasi-polynomiality.}
In this section, we present several quasi-polynomial formulas for the diameters and word metrics associated with the consecutive $k$-cycle generators defined below. We also compute the corresponding $H$-polynomials and study their properties.
Both Cayley graphs and Schreier coset graphs are considered. We also provide some other results like theoretical
lower and upper bounds on diameters which are in agreement with our experimental studies.

As discussed above, the conjectural quasi-polynomiality of diameters and word metrics is one of the main pieces of evidence supporting the \emph{holography} hypothesis. Indeed, if holography holds, then diameters and word metrics arise as
\href{https://en.wikipedia.org/wiki/Ehrhart_polynomial#Ehrhart_quasi-polynomials}{Ehrhart quasi-polynomials},
which naturally explains their quasi-polynomial behavior.
Consequently, working out explicit quasi-polynomial formulas is essential for a deeper understanding of the holography principle.

{\bf Generators definition. Related works.}
Fix an integer $k$. For $n > k$, we consider elements of $S_n$ given by the cyclic permutations 
$(i, i+1, \dots, i+k-1)$ for $i = 0, \dots, n-k$. For $k = 2$, these are the neighbor transpositions 
(Coxeter or bubble sort generators) considered previously. 
In CayleyPy these generators are denoted:
\href{https://cayleypy.github.io/cayleypy-docs/generated/cayleypy.PermutationGroups.html#cayleypy.PermutationGroups.consecutive_k_cycles}{consecutive\_k\_cycles}.
For odd $k$, they generate $A_n$ inside $S_n$, for even $k$ they give $S_n$. There are two natural options: whether or not to include inverses in the generating set. We consider both cases. To the best of our knowledge, these generators were not studied systematically in the literature, e.g. formulas for the diameters are new for $k > 3$ ($k=4$ briefly discussed in our previous work: \cite{Cayley3Growth}). From the physical point of view, the Laplacian of such graphs corresponds to the  situation
when $k$ neighboring spins interact, spin chains of that sort (but not exactly) appear e.g. \cite{KazakovMarshakovMinahanZarembo2004} in AdS/CFT. Let us note that cyclic permutations are used to characterize triples of prefix--reversals generating the whole group $S_n$, see~\cite{blanco2025}, and such triples are also discussed in the previous work. 

\subsection{\texorpdfstring{$(k-1)$-Shrinkage heuristics: Results and Difficulties}{(k-1)-Shrinkage heuristics: Results and Difficulties}}

Before going into details let us first give some informal heuristic principle 
summarizing the results: to obtain results on diameters,
word-metrics and other characteristics one should take results for $k=2$ case 
(which is standard neighbor transposition or Coxeter generators) and just divide them by $k-1$. 
In all considered cases it gives correct leading terms and moreover in some rare cases simple correction
like adding ceil-rounding would suffice to get the correct results. 
However exact results typically are more complicated. 
The only case where we were able to conjecture the result for {\it all } $k$
is the case of the coset Schreier graph with $n//2$ zeros and $n-n//2$ ones
with not inverse closed generators. For the other cases we present conjectures for $k$ up to $5$ to $8$
depending on cases, and only leading terms for all $k$.

Informal motivation of that heuristics is rather simple:  the $k$-cycle 
$(i, i+1, \dots, i+k-1)$ is equal to product of $(k-1)$ neighbor transpositions $(i,i+1)(i+1,i+2)...$,
thus the nodes which are on distance 1 in $k$-cycle case are on distance $k-1$ in standard neighbor
transposition graph - so is some sense $k$-cycle graph shrinks standard graph $k-1$ times. 



\subsection{Theoretical lower and upper bounds on the diameters. Cayley graphs. }
Here we prove lower and upper bounds for the diameters of the Cayley graph (not Schreier coset)
of the consecutive $k$ cycles inverse closed generators. 
They quite correspond to informal $(k-1)$-shrinkage principle above: for the standard neighbor transposition graph
the diameter is equal to $n(n-1)/2$ so we expect that for consecutive $k$-cycles we get 
$n(n-1)/(2(k-1))$. Indeed we can prove bounds with such leading term.

\begin{thm}
    The diameter satisfies the inequality $$D_k(n)\geq \Big\lceil\frac{n(n-1)-2}{2(k-1)}\Big\rceil.$$
\end{thm}
\begin{proof}
    Let $\ell(\cdot)$ denote the Coxeter length on $S_n$ with respect to adjacent transpositions of the form $(i\ i+1)$, i.e.
    $$
    \ell(\sigma) = \#\{(i,j) : 1\le i<j\le n,\ \sigma(i) > \sigma(j)\}.
    $$
    For each consecutive $k$-cycle $\tau \in S_n$, we have $\ell(\tau)=k-1$. On the other hand, $\ell$ is subadditive, i.e. $\ell(\sigma_1 \sigma_2) \le \ell(\sigma_1) + \ell(\sigma_2)$. Hence, if $\sigma$ is the product of $m$ consecutive $k$-cycles, we have $\ell(\sigma) \le m(k-1)$ which implies that $$d(\sigma)=m \ge \frac{\ell(\sigma)}{k-1}.$$ 
    
    We consider the inverse permutation (reversal permutation?) defined by $\sigma_0(i)=n+1-i$, which has $\ell(\sigma_0) = \binom{n}{2} = \frac{n(n-1)}{2}.$
    If $k$ is odd or both $k$ and $\sigma_0$ are even, then $\sigma_0$ lies in the subgroup generated by $S$ and we may take $\sigma=\sigma_0$.
    
    Otherwise, the subgroup generated by $S$ is contained in $A_n$, while $\sigma_0$ is odd (which happens when $k$ is even and $n \equiv 2,3 \pmod 4$). Then we cannot reach $\sigma_0$ itself, so we can consider $\sigma = \sigma_0 s_i$ for some adjacent transposition $s_i=(i \ i+1)$. In that case $\sigma$ is even and hence lies in the subgroup, and
    \[
    \ell(\sigma) = \ell(w_0 s_i) = \ell(w_0)-1 = \frac{n(n-1)}{2}-1 = \frac{n(n-1)-2}{2}.
    \]

    Since the diameter is at least the distance $d(\sigma)$ from the identity $e$ to $\sigma$, we have that $D_k(n) \ge d(\sigma) \ge \frac{\ell(\sigma)}{k-1} \ge \frac{n(n-1)-2}{2(k-1)}$.
\end{proof}

\begin{thm}
    For any $n>3k$, given $k$-cyclic generators, the distance between any permutation $\sigma$ of $n$ elements and the identity is at most $\frac{n(n-1)}{2(k-1)}+O(n)$.
\end{thm}
\begin{proof}
    First of all, we deliver the elements $1,\dots, k-1$ into $[k;2k-2]$ by repeatedly moving them $k-1$ steps to the left or to the right. Then we have done $n+O(k)$ operations or less.

    Next we need to order the elements correctly. For this purpose we rotate $[k;2k-1]$. Whenever an element $i$ is $k-1$ to the right of the correct position, we place the element in the correct position by using the permutation $(i,i+k-1,i+2k-2)^{-1}$. Since $(k,\dots, 2k-1)^{-1}(1,2,\dots, k)^{-1}(k,\dots, 2k-1)(1,2,\dots, k) = (1,k,\dots,2k-1),$ we find that we can place the elements from 1 to $k-1$ onto their rightful places in $O(k)$ moves. Then we can forget about the first $k-1$ elements and reorder the others.

    Repeating these operations with smaller and smaller $n,$ we obtain the required result.
\end{proof}

\subsection{\texorpdfstring{Theoretical diameters estimate. Schreier coset graphs $S_n / \bigl(S_{\lfloor n/2 \rfloor} \times S_{n-\lfloor n/2 \rfloor}\bigr)$}{Theoretical diameters estimate. Schreier coset graphs S_n/(S_{floor(n/2)} x S_{n-floor(n/2)})}}

Here we provide theoretical  diameters estimate for  consecutive $k$ cycles inverse closed generators
for the Schreier coset graph $S_n/(S_{{\lfloor n/2 \rfloor}} \times S_{n-{\lfloor n/2 \rfloor}})$,
which can be alternatively described as graph with nodes corresponding to vectors with components $0$ and $1$
with ${\lfloor n/2 \rfloor} $ zeros, and ${n - \lfloor n/2 \rfloor}$ ones.
Estimate is again consistent with $(k-1)$-shrinkage principle.

\begin{thm}
    The diameter of the coset with $[n/2]$ zeros and $n-[n/2]$ ones is equal to $\frac{n^2}{4(k-1)}+O(n).$
\end{thm}
\begin{proof}
    Consider the sequence of zeros and ones where the zeros are located at places $a_1,a_2,\dots,a_l.$ Create a modified sequence where the zeros are located at $a_1,\dots,a_1+k-2, a_{k},\dots,a_{k}+k-2,\dots, a_{m(k-1)+1},\dots,a_{m(k-1)+k-1},\dots$ This modified sequence has the zeros come in groups of $k-1$. Since a group of size $k-1$ or less can be moved to a distance of 1 by 1 turn of the cycle, the modified position is reachable in $$\sum_{m=0}^{[n/2/(k-1)]}a_{m(k-1)+1} - m(k-1)$$ moves. Which is at most $\lceil n/2\rceil ([n/2/(k-1)]+1)$ moves.
    
    Next we reconstruct the original sequence by guiding every zero to its rightful place in $O(n)$ turns. 
\end{proof}

\subsection{Quasi-polynomials for diameters. Cayley graphs.}
Here we present conjectural quasi-polynomial expressions  for the diameters of the Cayley graph (not Schreier coset)
of the consecutive $k$ cycles inverse closed generators for $k\le 6$, and partial information for larger $k$.
The leading terms of the formulas is $n(n-1)/(2(k-1))$ consistent with the theoretical estimates above and
naive $(k-1)$-shrinkage principle.

For $n\le 14$ the conjectures checked by effective realization of the BFS algorithm provided by CayleyPy.
For larger $n$ the data obtained by AI-component of CayleyPy - according to methodology described above. 

\begin{thm}
\textbf{Case $k=3$ from $n=6$:}
\begin{align*}
    & n \equiv 0,1 \pmod{4}: && D_3(n) = \frac{n(n-1)}{4}, \\
    & n \equiv 2,3 \pmod{4}: && D_3(n) = \frac{n(n-1)}{4}-\frac{1}{2}.
\end{align*}
Or in the other words $D(n) = \lfloor \frac{n(n-1)}{4} \rfloor$
\end{thm}

\begin{Conj}
\textbf{Case $k=4$ from $n=6$:}
\begin{align*}
    & n \equiv 0, 1 \pmod{3}: && D_4(n) = \frac{n(n-1)}{6}-1, \\
    & n \equiv 2 \pmod{3}: && D_4(n) = \frac{n(n-1)}{6}+\frac{2}{3}.
\end{align*}

\begin{equation*}
\mathrm{OGF}
= \frac{-x^{9} + x^{8} + 4x^{6} + x^{5} + 4x^{4} + x^{2} - x}{(1 - x^{3})^{3}}.
\end{equation*}


\begin{figure}[H]
    \centering
    \includegraphics[width=0.7\linewidth]{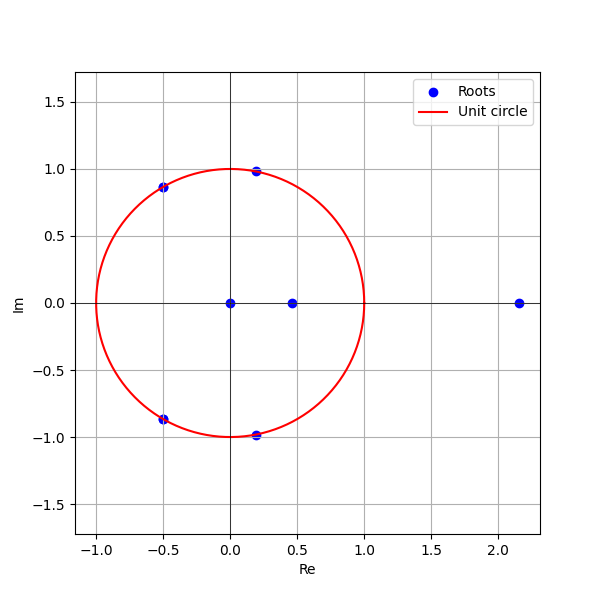}
    \caption{case $k=4$}
\end{figure}

\textbf{Case $k=5$ from $n=8$:}
\begin{align*}
    & n \equiv 0, 1 \pmod{8}: & & D_5(n) = \frac{n(n-1)}{8}, \\
    & n \equiv 2, 7 \pmod{8}: && D_5(n) = \frac{n(n-1)}{8}-\frac{1}{4} = \frac{n(n-1)-2}{8}, \\
    & n \equiv 3, 6 \pmod{8}: && D_5(n) = \frac{n(n-1)}{8}+\frac{1}{4} = \frac{n(n-1)+2}{8}, \\
    & n \equiv 4, 5 \pmod{8}: && D_5(n) = \frac{n(n-1)}{8}+\frac{1}{2} = \frac{n(n-1)+4}{8}.
\end{align*}

\begin{equation*}
\mathrm{OGF}
= \frac{H(x)}{(1 - x^{8})^{3}},
\end{equation*}
where
\begin{multline*}
H(x) = 
x^{22} + 2x^{21} + 3x^{20} + 4x^{19} + 5x^{18} + 7x^{17} + \\ + 9x^{16} + 11x^{15} + 11x^{14} + 11x^{13} + 11x^{12} + 11x^{11} + \\ + 11x^{10} + 9x^{9} + 7x^{8} + 5x^{7} + 4x^{6} + 3x^{5} + 2x^{4} + x^{3}.
\end{multline*}

\begin{table}[h]
\centering
\begin{tabular}{cccc}
\toprule
Nonnegative & Symmetric & Unimodal & Weakly log concave \\
\midrule
True & False & True & False \\
\bottomrule
\end{tabular}
\end{table}

\begin{figure}[H]
    \centering
    \includegraphics[width=0.75\linewidth]{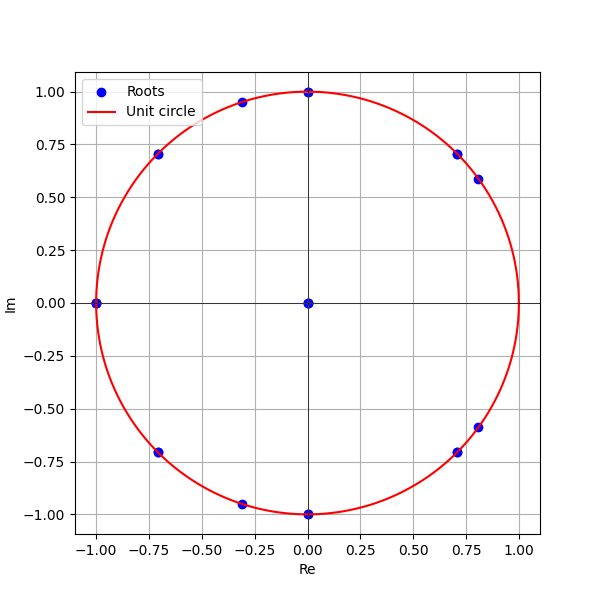}
    \caption{case $k=5$}
\end{figure}

\textbf{Case $k=6$ from $n=10$:}
\begin{align*}
    & n \equiv 0, 1 \pmod{5}: && D_6(n) = \frac{n(n-1)}{10}+1, \\
    & n \equiv 2, 4 \pmod{5}: && D_6(n) = \frac{n(n-1)}{10}+\frac{4}{5} = \frac{n(n-1)+8}{10}, \\
    & n \equiv 3 \pmod{5}: && D_6(n) = \frac{n(n-1)}{10}+\frac{7}{5} = \frac{n(n-1)+14}{10}.
\end{align*}

\begin{equation*}
\mathrm{OGF}
= \frac{H(x)}{(1 - x^{5})^{3}},
\end{equation*}
where
\begin{multline*}
H(x) = 
x^{14} + 2x^{13} + 2x^{12} + 3x^{11} + 4x^{10} + 2x^{9} + \\ + x^{8} + 2x^{7} + x^{6} + 2x^{4} + 2x^{3} + x^{2} + x + 1.
\end{multline*}

\begin{table}[h]
\centering
\begin{tabular}{cccc}
\toprule
Nonnegative & Symmetric & Unimodal & Weakly log concave \\
\midrule
True & False & False & False \\
\bottomrule
\end{tabular}
\end{table}

\begin{figure}[H]
    \centering
    \includegraphics[width=0.75\linewidth]{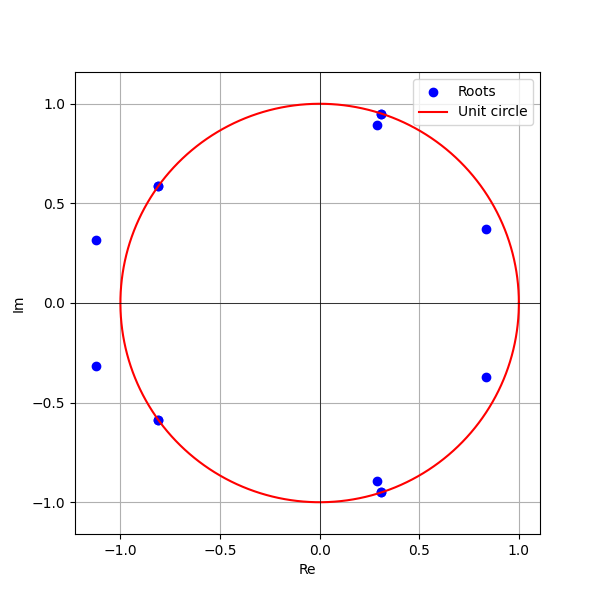}
    \caption{case $k=6$}
\end{figure}

\textbf{Case $k=7$ (period $12$):}
\begin{align*}
    & n \equiv 0,1,4,9 \pmod{12}:
    && D_7(n) = \frac{n(n-1)}{12}, \\
    & n \equiv 2,11 \pmod{12}:
    && D_7(n) = \frac{n(n-1)}{12}-\frac{1}{6}
      = \frac{n(n-1)-2}{12}, \\
    & n \equiv 3,6,7,10 \pmod{12}:
    && D_7(n) = \frac{n(n-1)}{12}+\frac{1}{2}
      = \frac{n(n-1)+6}{12}, \\
    & n \equiv 5,8 \pmod{12}:
    && D_7(n) = \frac{n(n-1)}{12}+\frac{1}{3}
      = \frac{n(n-1)+4}{12}.
\end{align*}

\begin{equation*}
\mathrm{OGF}
= \frac{H(x)}{(1 - x^{12})^3},
\end{equation*}
where
\begin{multline*}
H(x) = 
x^{34} + x^{33} + 2x^{32} + 3x^{31} + 4x^{30} + 5x^{29} + 6x^{28} + 8x^{27} + \\ + 9x^{26} + 11x^{25} + 13x^{24} + 15x^{23} + 15x^{22} + 17x^{21} + 17x^{20} + \\ + 17x^{19} + 17x^{18} + 17x^{17} + 17x^{16} + 15x^{15} + 15x^{14} + 13x^{13} + \\ + 11x^{12} + 9x^{11} + 8x^{10} + 6x^9 + 5x^8 + 4x^7 + 3x^6 + 2x^5 + x^4 + x^3.
\end{multline*}

\begin{table}[h]
\centering
\begin{tabular}{cccc}
\toprule
Nonnegative & Symmetric & Unimodal & Weakly log concave \\
\midrule
True & False & True & False \\
\bottomrule
\end{tabular}
\end{table}

\begin{figure}[H]
    \centering
    \includegraphics[width=0.75\linewidth]{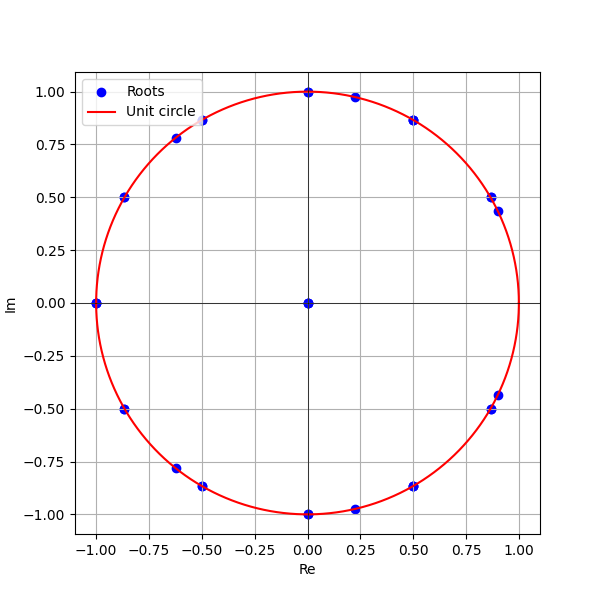}
    \caption{case $k=7$}
\end{figure}

\end{Conj}

More generally: 
\begin{Conj} Let $D_k(n)$ be the diameter of $S_n$  Cayley graphs generated by consecutive cycles
of length $k$ and their inverses. They are given by quadratic quasi-polynomials in $n$ (for $n$ large enough) and the following holds:
\begin{enumerate}

    \item All the statements below  are valid for $ n \ge 2k. $ (The starting point for quasi-polynomials to be valid.)

    \item The periods of  quasi-polynomials  are:
    \[
    \tau =
    \begin{cases}
    k - 1, & \text{if } k \text{ is even}, \\[4pt]
    2(k - 1), & \text{if } k \text{ is odd}.
    \end{cases}
    \]
    
        
    \item The general form of the diameter:
    \[
    D_k(n) =  \frac{n(n-1)}{2(k-1)} + Q_0(n),
    \]    
    where $Q_0(n)$ is quasi-polynomial of degree zero, i.e. just the periodic function of $n$. 
    Moreover in many cases $Q_0(n)$ is just the correction to make the first term expression  to be the  integer 
    (floor/ceil), however it is not always the case.

    \item The last layer always contains either the full reversal
    $
    R = (n{-}1,\, n{-}2,\, \dots,\, 1,\, 0),
    $
    or its modified form with one adjacent inversion:
    $
    R' = R \circ (i,\, i{+}1).
    $
    
    \item The alternation between $R$ and $R'$ in the last layers has a period
    $
    T = k - 1.
    $ I.e.  the sequence $S(n) =$ $R$ or $R'$ whether the last layer contains $R$ or $R'$ will be periodic with the period $k-1$. 


\end{enumerate}

\end{Conj}

For diameter $k$ we have the following formulas for polynomials
\[
H_5(x)=\Phi_2(x)^3\,\Phi_4(x)^2\,\Phi_8(x)^2\,\Phi_{10}(x).
\]

\[
H_7(x)=\Phi_2(x)^3\,\Phi_3(x)^2\,\Phi_4(x)^2\,\Phi_6(x)^3\,\Phi_{12}(x)^2\,\Phi_{14}(x).
\]

\medskip

\noindent\textbf{Conjecture (odd \(k\ge 5\)).}
For odd \(k\),
\[
H_k(x)\;=\;\frac{(x^{2k-2}-1)^2\,(x^k+1)\,(x^{k-4}+1)}{(x-1)^2(x+1)}.
\]

Equivalently, in cyclotomic form,
\[
H_k(x)
=
\Phi_2(x)\!
\prod_{\substack{d\mid(2k-2)\\ d>1}}\Phi_d(x)^2\;
\prod_{\substack{d\mid k\\ d>1}}\Phi_{2d}(x)\;
\prod_{\substack{d\mid (k-4)\\ d>1}}\Phi_{2d}(x).
\]

\clearpage
\subsection{\texorpdfstring{Schreier coset graph: $S_n / \bigl(S_{\lfloor n/2 \rfloor} \times S_{n-\lfloor n/2 \rfloor}\bigr)$ (not inverse-closed)}{Schreier coset graph: S_n/(S_{floor(n/2)} x S_{n-floor(n/2)}) (not inverse-closed)}}


Here we present a conjectural formula for the diameters of Schreier coset graphs of the form
$
S_n \big/ \bigl(S_{\lfloor n/2 \rfloor} \times S_{n-\lfloor n/2 \rfloor}\bigr).$
In this setting, we were able to obtain an explicit formula for all values of $k$.
The graph can be equivalently described as follows: its vertices correspond to binary vectors with components in $\{0,1\}$, containing exactly $\lfloor n/2 \rfloor$ zeros and $n-\lfloor n/2 \rfloor$ ones.
In \texttt{CayleyPy}, we define these graphs by setting the \texttt{central\_state} to be
$
[0]^{n//2} + [1]^{\,n-n//2}.
$
The generating set consists of consecutive $k$-cycles; we consider the case where the generating set is not inverse-closed.

The leading term of the diameter formulas are  $n^2/(4(k-1))$ as it is expected from the $(k-1)$-shrinkage principle,
since for the standard $k=2$ (Coxeter or neighbor transposition) case the diameter is exactly 
${\lfloor n/2 \rfloor} (n-\lfloor n/2 \rfloor)$. 

Informally, one may think of this graph as a ``$k$-shrunken'' version of Grassmannian
$Gr(\lfloor n/2 \rfloor, n)$ over the field with one element. From the point of view of that analogy
diameter corresponds to dimension of the manifold, and Poincare polynomial correspond to growth polynomial of the graph.

\begin{Conj}
Diameters of the Schreier coset graphs $S_n / ( S_{\lfloor n/2 \rfloor} \times S_{n-\lfloor n/2 \rfloor} )$
with consecutive $k$-cycles, not inverse closed case are given by :
\[
\begin{aligned}
n &\equiv 0 \!\pmod{(2k-2)} \;:& \!\!\!
D_k(n) &= \frac{n^2 + 4(k-1)}{4(k-1)}, \\[6pt]
n &\equiv -1 \!\pmod{(2k-2)} \;:& \!\!\!
D_k(n) &=  \frac{n^2 + 4(k-1)-1}{4(k-1)},
\end{aligned}
\]
and else, denoting $p = n \pmod{(2k-2)} $,
\[
\begin{aligned}
 D_k(n)= &  \frac{n(n+2k-p-2)}{4(k-1)}, & p  {\rm \ is \ even \ }
\\[6pt]
 D_k(n)= &   \frac{(n-1)(n+2k-p-2)}{4(k-1)}, & p  {\rm \ is \ odd \ }.
\end{aligned}
\]
\end{Conj}
The formula is expected to be valid for $n$ large enough, exact bound is not clear, 
for $k=2,3,4,5,6$  it starts from $4,4,7,16,31$, so in general it might be for $n>2^k$, which would be
quite surprising since in the other cases it is typical to have linear bound in terms of $k$.

For general $k$ the generating function reads
\beq
    \text{OGF}=\frac{H(x)}{(1-x^{2k-2})^3}
\eeq
where where the $H$-polynomial is
\beqa \nonumber
H(x)&=&\frac{\left(x-x^k\right) \left(x^k+x\right) }{(x-1)^2 x^7 (x+1)} (-x^{4 k}+x^{2 k+2}-2 x^{2 k+3}+k x^{2 k+4}-x^{2 k+4}+x^{2 k+5}
\\ \nonumber &&
-k x^{2 k+6}+x^{4 k+1}+x^{4 k+2}-x^{4
   k+3}+x^8-x^6+x^5
   )
\eeqa
One can show that in fact $H(x)$ is a polynomial.

Below we present picture of roots the $H$-polynomials for small $k$, not all roots have modules equal to 1,
so "Riemann conjecture is not true" for these cases. Nevertheless it is natural to believe that there are some patterns
which remains to be understood.
\begin{figure}[H]
    \centering
    \begin{minipage}{0.48\linewidth}
        \centering
        \includegraphics[width=\linewidth]{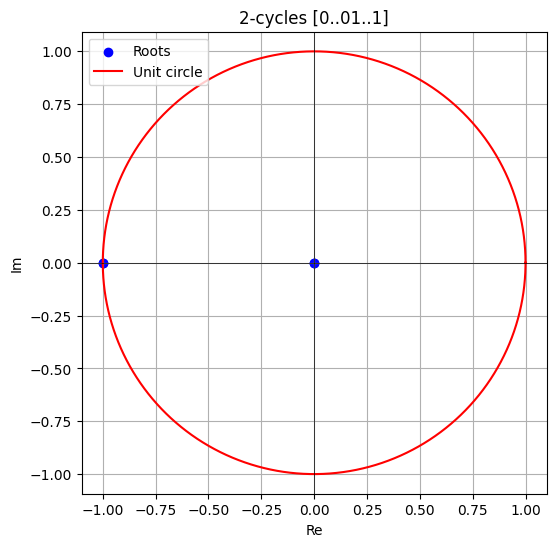}
        \caption{Consecutive 2-cycles [0..01..1]}
        \label{fig:cons2_0011}
    \end{minipage}\hfill
    \begin{minipage}{0.48\linewidth}
        \centering
        \includegraphics[width=\linewidth]{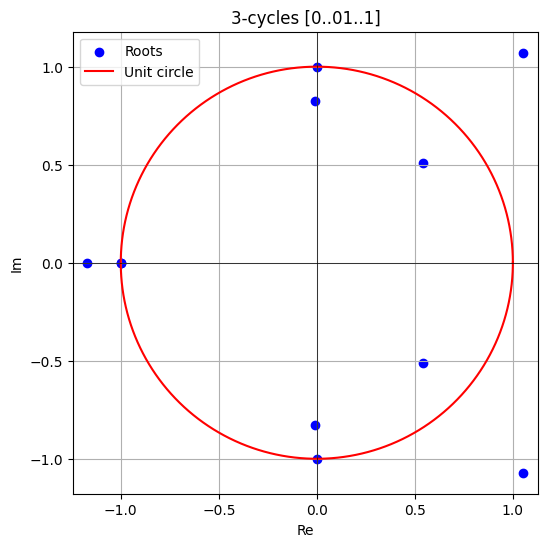}
        \caption{Consecutive 3-cycles [0..01..1]}
        \label{fig:cons3_0011}
    \end{minipage}
\end{figure}

\begin{figure}[H]
    \centering
    \begin{minipage}{0.48\linewidth}
        \centering
        \includegraphics[width=\linewidth]{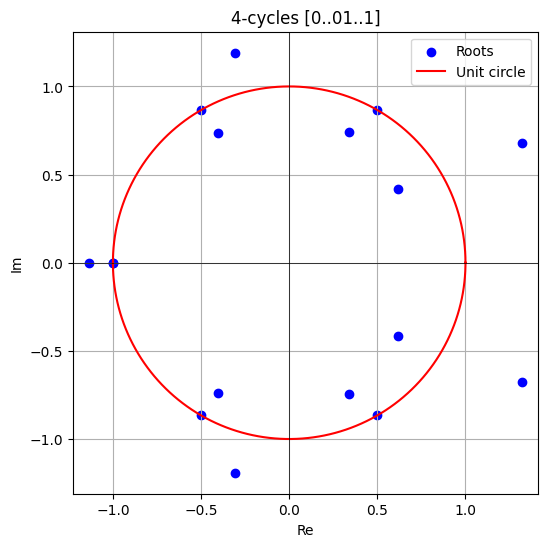}
        \caption{Consecutive 4-cycles [0..01..1]}
        \label{fig:cons4_0011}
    \end{minipage}\hfill
    \begin{minipage}{0.48\linewidth}
        \centering
        \includegraphics[width=\linewidth]{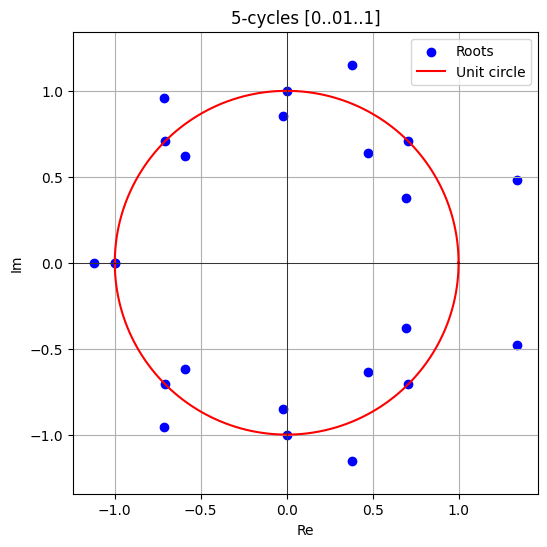}
        \caption{Consecutive 5-cycles [0..01..1]}
        \label{fig:cons5_0011}
    \end{minipage}
\end{figure}





\begin{figure}[H]
    \centering
    \includegraphics[width=0.75\linewidth]{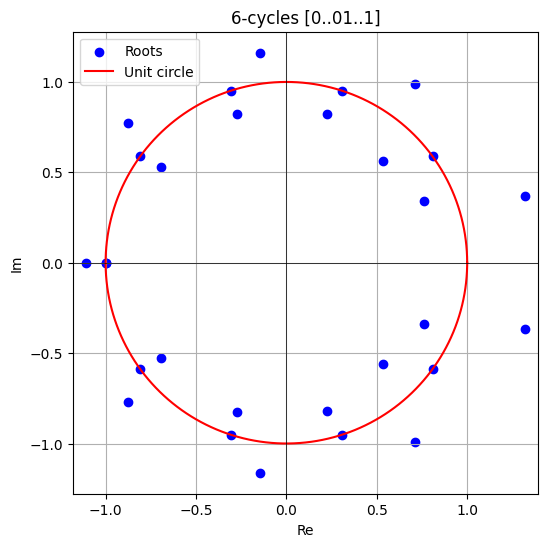}
    \caption{Consecutive 6-cycles [0..01..1]}
    \label{fig:cons6_0011}
\end{figure}

\clearpage
\subsection{\texorpdfstring{Schreier coset graph: $S_n / \bigl(S_{\lfloor n/2 \rfloor} \times S_{n-\lfloor n/2 \rfloor}\bigr)$ (inverse-closed)}{Schreier coset graph: S_n/(S_{floor(n/2)} x S_{n-floor(n/2)}) (inverse-closed)}}

Almost the same setting as in the previous subsection - but now inverse-closed generators.
I.e.  we study the Schreier coset graphs of the form
$
S_n \big/ \bigl(S_{\lfloor n/2 \rfloor} \times S_{n-\lfloor n/2 \rfloor}\bigr).$
The graph can be equivalently described as follows: its vertices correspond to binary vectors with components in $\{0,1\}$, containing exactly $\lfloor n/2 \rfloor$ zeros and $n-\lfloor n/2 \rfloor$ ones.
In \texttt{CayleyPy}, we define these graphs by setting the \texttt{central\_state} to be
$
[0]^{n//2} + [1]^{\,n-n//2}.
$
The generating set consists of consecutive $k$-cycles; we consider the case where the generating set is  inverse-closed.
Informally, one may think of this graph as a ``$k$-shrunken'' version of Grassmanian
$Gr(\lfloor n/2 \rfloor, n)$ over the field with one element. From the point of view of that analogy
diameter corresponds to dimension of the manifold, and Poincare polynomial correspond to growth polynomial of the graph.

We found quasi-polynomials and $H$-polynomials for small $k$. 
Surprisingly  the roots of $H$-polynomials have modules equal one ("Riemann conjecture holds") in the considered cases.
However it is not clear whether it would be true for larger $k$, there are plenty examples when
it holds true for the beginning of the series, but not in general. 

\begin{Conj}
    All the quasi-polynomials below have leading term $n^2/(4(k-1))$ and liner term is absent,
it natural to expect that it holds true for all $k$. 
\end{Conj}

\[
H_k(x)=\frac{\left(x^{2k-2}-1\right)^2\left(x^{2k-3}+1\right)}{(x-1)^2(x+1)}, \quad k\ge 2.
\]

\[
H_k(x)=\Phi_2(x)^2
\prod_{\substack{d\mid(2k-2)\\ d>2}}\Phi_d(x)^2
\prod_{\substack{d\mid(2k-3)\\ d>1}}\Phi_{2d}(x), \quad k\ge 2.
\]

\begin{prp}
Let \(k\ge 2\). Define integers \(a_{k,n}\) for \(0\le n\le 6k-10\) by
\[
a_{k,n}=
\begin{cases}
\left\lfloor \dfrac n2 \right\rfloor+1,
& 0\le n\le 2k-5,\\[1.0ex]
k-1+\bigl((n-(2k-4))\bmod 2\bigr),
& 2k-4\le n\le 4k-6,\\[1.0ex]
k-2-\left\lfloor \dfrac{n-(4k-5)}{2}\right\rfloor,
& 4k-5\le n\le 6k-10.
\end{cases}
\]
Let
\[
H_k(x):=\sum_{n=0}^{6k-10} a_{k,n}\,x^n\in \mathbb Z[x].
\]
Then
\[
H_k(x)=\frac{(x^{2k-2}-1)^2\,(x^{2k-3}+1)}{(x-1)^2(x+1)}\in \mathbb Z[x].
\]
\end{prp}

\begin{proof}
Define
\[
U(x):=\frac{x^{2k-2}-1}{x-1}=\sum_{i=0}^{2k-3}x^i,
\qquad
V(x):=\frac{x^{2k-3}+1}{x+1}=\sum_{r=0}^{2k-4}(-1)^r x^r.
\]
Since \(x^{2k-2}-1\) is divisible by \(x-1\) and \(2k-3\) is odd (hence \(x^{2k-3}+1\) is divisible by \(x+1\)),
the rational expression
\[
H_k^\star(x):=\frac{(x^{2k-2}-1)^2\,(x^{2k-3}+1)}{(x-1)^2(x+1)}
\]
lies in \(\mathbb Z[x]\). Moreover,
\[
H_k^\star(x)=U(x)^2\,V(x).
\]

Set \(W(x):=U(x)V(x)\). We compute \(W(x)\) explicitly:
\[
W(x)=\frac{(x^{2k-2}-1)(x^{2k-3}+1)}{(x-1)(x+1)}
=\frac{(x^{2k-2}-1)(x^{2k-3}+1)}{x^2-1}.
\]
Using
\[
\frac{x^{2k-2}-1}{x^2-1}=1+x^2+x^4+\cdots+x^{2k-4}=\sum_{j=0}^{k-2}x^{2j},
\]
we obtain
\[
W(x)=(1+x^{2k-3})\sum_{j=0}^{k-2}x^{2j}
=\sum_{j=0}^{k-2}x^{2j}+\sum_{j=0}^{k-2}x^{2k-3+2j}.
\]
Therefore
\[
H_k^\star(x)=U(x)W(x)
=\Bigl(\sum_{i=0}^{2k-3}x^i\Bigr)
\Bigl(\sum_{j=0}^{k-2}x^{2j}+\sum_{j=0}^{k-2}x^{2k-3+2j}\Bigr).
\]

Let
\[
S:=\{0,2,4,\dots,2k-4\}\ \cup\ \{2k-3,2k-1,\dots,4k-7\}.
\]
Then the coefficient of \(x^m\) in \(W(x)\) equals \(1\) if \(m\in S\) and \(0\) otherwise.
Since the coefficient of \(x^i\) in \(U(x)\) is \(1\) for \(0\le i\le 2k-3\) (and \(0\) otherwise),
the coefficient of \(x^n\) in \(H_k^\star(x)=U(x)W(x)\) is
\[
c_{k,n}=\#\{m\in S:\ n-(2k-3)\le m\le n\}
=\#\bigl(S\cap [\,n-(2k-3),\,n\,]\bigr).
\]
We show that \(c_{k,n}=a_{k,n}\) for every \(0\le n\le 6k-10\). This proves \(H_k^\star(x)=H_k(x)\).

\medskip
\noindent\textbf{Range I: \(0\le n\le 2k-5\).}
Here \(n<2k-3\), so \(S\cap[\,n-(2k-3),\,n\,]\) contains no element from the odd block
\(\{2k-3,2k-1,\dots\}\), and it contains precisely the even integers in \([0,n]\).
Thus
\[
c_{k,n}=\#\{0,2,4,\dots\le n\}=\left\lfloor \frac n2\right\rfloor+1=a_{k,n}.
\]

\medskip
\noindent\textbf{Range II: \(2k-4\le n\le 4k-6\).}
Put \(L:=n-(2k-3)\) (so \(-1\le L\le 2k-3\) in this range). Split \(S=E\cup O\) with
\[
E:=\{0,2,\dots,2k-4\},\qquad O:=\{2k-3,2k-1,\dots,4k-7\}.
\]
We count \(E\cap[L,n]\) and \(O\cap[L,n]\) separately.

First, since \(n\ge 2k-4\), the interval \([L,n]\) contains all evens up to \(2k-4\) except those \(<L\).
The set \(E\) has exactly \(k-1\) elements. The number of even integers \(<L\) (and \(\ge 0\)) equals
\(\left\lfloor\dfrac{L+1}{2}\right\rfloor\) for every integer \(L\ge -1\).
Hence
\[
\#(E\cap[L,n])=(k-1)-\left\lfloor\frac{L+1}{2}\right\rfloor.
\]
Second, since \(L\le 2k-3\), the interval \([L,n]\) meets the odd block \(O\) starting at \(2k-3\), and it
contains precisely those odds between \(2k-3\) and \(n\). The number of such odds is
\[
\#(O\cap[L,n])=\left\lfloor\frac{n-(2k-3)}{2}\right\rfloor+1=\left\lfloor\frac{L}{2}\right\rfloor+1.
\]
Adding,
\[
c_{k,n}=(k-1)-\left\lfloor\frac{L+1}{2}\right\rfloor+\left\lfloor\frac{L}{2}\right\rfloor+1
= k+\left\lfloor\frac{L}{2}\right\rfloor-\left\lfloor\frac{L+1}{2}\right\rfloor.
\]
Now \(\left\lfloor\frac{L}{2}\right\rfloor-\left\lfloor\frac{L+1}{2}\right\rfloor=0\) if \(L\) is even and equals \(-1\) if \(L\) is odd.
Thus \(c_{k,n}=k\) if \(L\) is even and \(c_{k,n}=k-1\) if \(L\) is odd.
Since \(L=n-(2k-3)\) and \(2k-3\) is odd, the parity condition ``\(L\) even'' is equivalent to ``\(n\) odd''.
Therefore
\[
c_{k,n}=k-1+\bigl(n\bmod 2\bigr)
= k-1+\bigl((n-(2k-4))\bmod 2\bigr)=a_{k,n}.
\]

\medskip
\noindent\textbf{Range III: \(4k-5\le n\le 6k-10\).}
Again put \(L=n-(2k-3)\). Then \(2k-2\le L\le 4k-7\).
The interval \([L,n]\) lies entirely above \(2k-2\), so it contains no element of \(E=\{0,2,\dots,2k-4\}\).
Also, since \(n\ge 4k-5>4k-7\), intersecting with \(O\) truncates at the maximal element \(4k-7\). Hence
\[
c_{k,n}=\#\{t\in \mathbb Z:\ L\le t\le 4k-7,\ t\ \text{odd}\}.
\]
The odd integers in that interval form an arithmetic progression of step \(2\), ending at \(4k-7\), so
\[
c_{k,n}= \left(2k-3\right)-\left\lfloor\frac{L}{2}\right\rfloor
=2k-3-\left\lfloor\frac{n-(2k-3)}{2}\right\rfloor.
\]
Write \(n=(4k-5)+t\) with \(0\le t\le 2k-5\). Then \(n-(2k-3)=2k-2+t\) with \(2k-2\) even, so
\[
\left\lfloor\frac{n-(2k-3)}{2}\right\rfloor
=\left\lfloor\frac{2k-2+t}{2}\right\rfloor
=(k-1)+\left\lfloor\frac{t}{2}\right\rfloor.
\]
Substituting,
\[
c_{k,n}=2k-3-\Bigl((k-1)+\left\lfloor\frac{t}{2}\right\rfloor\Bigr)
= k-2-\left\lfloor\frac{t}{2}\right\rfloor
= k-2-\left\lfloor\frac{n-(4k-5)}{2}\right\rfloor
= a_{k,n}.
\]

\medskip
We have shown \(c_{k,n}=a_{k,n}\) for every \(0\le n\le 6k-10\), hence \(H_k^\star(x)=H_k(x)\), as claimed.
\end{proof}

\textbf{Case $k=3$ from $n=4$}
\[
\begin{aligned}
n &\equiv 0 \pmod{4} :&\quad D_3(n) &= \frac{n^2}{8}, \\[6pt]
n &\equiv 1 \pmod{4} :&\quad D_3(n) &= \frac{n^2 - 1}{8}, \\[6pt]   
n &\equiv 2 \pmod{4} :&\quad D_3(n) &= \frac{n^2 + 4}{8}, \\[6pt]
n &\equiv 3 \pmod{4} :&\quad D_3(n) &= \frac{n^2 - 1}{8}.
\end{aligned}
\]

\[
\mathrm{OGF}
=\frac{x^{10}+x^{9}+2x^{8}+3x^{7}+2x^{6}+3x^{5}+2x^{4}+x^{3}+x^{2}}{(1-x^{4})^{3}}.
\]

\begin{table}[h]
\centering
\begin{tabular}{cccc}
\toprule
Nonnegative & Symmetric & Unimodal & Weakly log concave \\
\midrule
True & False & False & False \\
\bottomrule
\end{tabular}
\end{table}

\begin{figure}[H]
    \centering
    \includegraphics[width=0.75\linewidth]{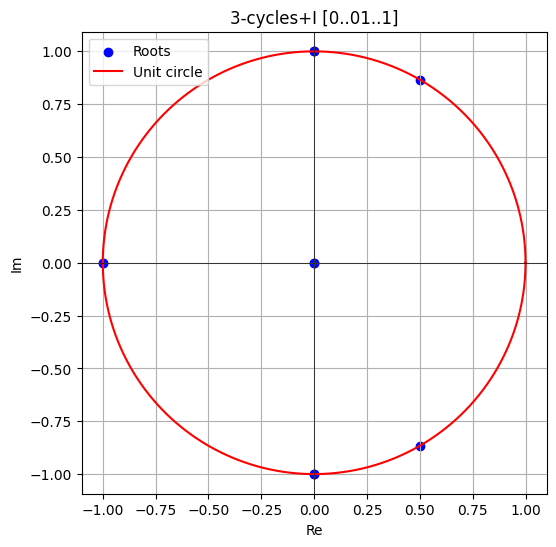}
    \caption{Consecutive 3-cycles +inv [0..01..1]}
    \label{fig:cons3inv_0011}
\end{figure}

\textbf{Case $k=4$ from $n=11$}
\[
\begin{aligned}
n &\equiv 0 \pmod{6}:  &\quad D_4(n) &= \dfrac{n^2}{12},\\[6pt]
n &\equiv 1,5 \pmod{6}: &\quad D_4(n) &= \dfrac{n^2 - 1}{12},\\[6pt]
n &\equiv 2,4 \pmod{6}: &\quad D_4(n) &= \dfrac{n^2 + 8}{12},\\[6pt]
n &\equiv 3 \pmod{6}:   &\quad D_4(n) &= \dfrac{n^2 + 3}{12}.
\end{aligned}
\]

\[
\mathrm{OGF}
=\frac{H(x)}{(1-x^{6})^{3}},
\]
where
\[
\begin{aligned}
H(x)=\;&x^{16}+x^{15}+2x^{14}+2x^{13}+3x^{12}+4x^{11}+3x^{10}+4x^{9} \\
&+3x^{8}+4x^{7}+3x^{6}+2x^{5}+2x^{4}+x^{3}+x^{2}.
\end{aligned}
\]

\begin{table}[h]
\centering
\begin{tabular}{cccc}
\toprule
Nonnegative & Symmetric & Unimodal & Weakly log concave \\
\midrule
True & False & False & False \\
\bottomrule
\end{tabular}
\end{table}

\begin{figure}[H]
    \centering
    \includegraphics[width=0.75\linewidth]{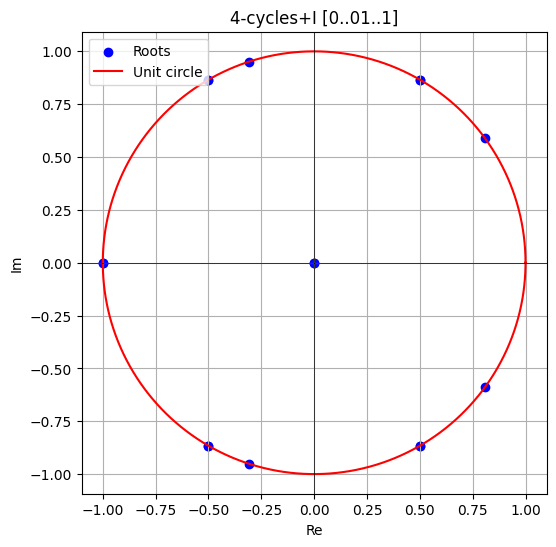}
    \caption{Consecutive 4-cycles +inv [0..01..1]}
    \label{fig:cons4inv_0011}
\end{figure}

\textbf{Case $k=5$ from $n=10$}
\[
\begin{aligned}
n &\equiv 0 \pmod{8}:&\; D_5(n) &= \frac{n^2}{16},\\[4pt]
n &\equiv 1,7 \pmod{8}:&\; D_5(n) &= \frac{n^2-1}{16},\\[4pt]
n &\equiv 2,6 \pmod{8}:&\; D_5(n) &= \frac{n^2+12}{16},\\[4pt]
n &\equiv 3,5 \pmod{8}:&\; D_5(n) &= \frac{n^2+7}{16},\\[4pt]
n &\equiv 4 \pmod{8}:&\; D_5(n) &= \frac{n^2+16}{16}.
\end{aligned}
\]

\[
\mathrm{OGF}
=\frac{H(x)}{(1-x^{8})^{3}},
\]
where
\[
\begin{aligned}
H(x)=\;&x^{22}+x^{21}+2x^{20}+2x^{19}+3x^{18}+3x^{17}+4x^{16}+5x^{15} \\
&+4x^{14}+5x^{13}+4x^{12}+5x^{11}+4x^{10}+5x^{9}+4x^{8}+3x^{7} \\
&+3x^{6}+2x^{5}+2x^{4}+x^{3}+x^{2}.
\end{aligned}
\]

\begin{figure}[H]
    \centering
    \includegraphics[width=0.75\linewidth]{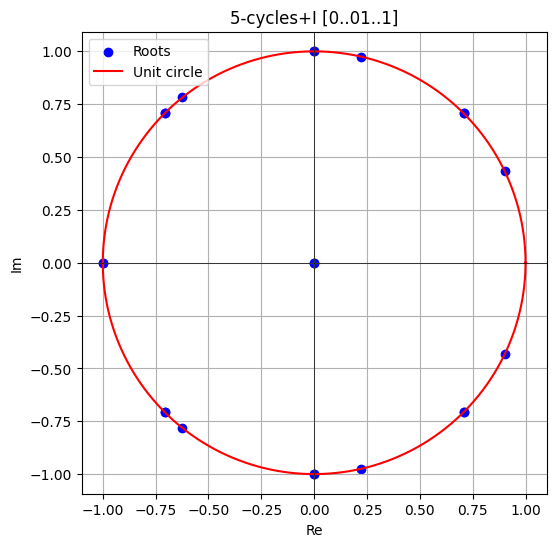}
    \caption{Consecutive 5-cycles +inv [0..01..1]}
    \label{fig:cons5inv_0011}
\end{figure}

\begin{table}[h]
\centering
\begin{tabular}{cccc}
\toprule
Nonnegative & Symmetric & Unimodal & Weakly log concave \\
\midrule
True & False & False & False \\
\bottomrule
\end{tabular}
\end{table}

\textbf{Case $k=6$, inverse-closed (empirical fit from $n\ge 12$)}
\[
\begin{aligned}
n &\equiv 0 \pmod{10} \;:\;& 
D_6(n) &= \frac{n^2}{20}, \\[6pt]
n &\equiv 1 \pmod{10} \;:\;& 
D_6(n) &= \frac{n^2 - 1}{20}, \\[6pt]
n &\equiv 2 \pmod{10} \;:\;& 
D_6(n) &= \frac{n^2 + 16}{20}, \\[6pt]
n &\equiv 3 \pmod{10} \;:\;& 
D_6(n) &= \frac{n^2 + 11}{20}, \\[6pt]
n &\equiv 4 \pmod{10} \;:\;& 
D_6(n) &= \frac{n^2 + 24}{20}, \\[6pt]
n &\equiv 5 \pmod{10} \;:\;& 
D_6(n) &= \frac{n^2 + 15}{20}, \\[6pt]
n &\equiv 6 \pmod{10} \;:\;& 
D_6(n) &= \frac{n^2 + 24}{20}, \\[6pt]
n &\equiv 7 \pmod{10} \;:\;& 
D_6(n) &= \frac{n^2 + 11}{20}, \\[6pt]
n &\equiv 8 \pmod{10} \;:\;& 
D_6(n) &= \frac{n^2 + 16}{20}, \\[6pt]
n &\equiv 9 \pmod{10} \;:\;& 
D_6(n) &= \frac{n^2 - 1}{20}.
\end{aligned}
\]

\[
\mathrm{OGF}
=\frac{H(x)}{(1-x^{10})^{3}},
\]
where
\[
\begin{aligned}
H(x)=\;&x^{28}+x^{27}+2x^{26}+2x^{25}+3x^{24}+3x^{23}+4x^{22}+4x^{21} \\
&+5x^{20}+6x^{19}+5x^{18}+6x^{17}+5x^{16}+6x^{15}+5x^{14}+6x^{13} \\
&+5x^{12}+6x^{11}+5x^{10}+4x^{9}+4x^{8}+3x^{7}+3x^{6}+2x^{5}+2x^{4}+x^{3}+x^{2}.
\end{aligned}
\]

\begin{table}[h]
\centering
\begin{tabular}{cccc}
\toprule
Nonnegative & Symmetric & Unimodal & Weakly log concave \\
\midrule
True & False & False & False \\
\bottomrule
\end{tabular}
\end{table}

\begin{figure}[H]
    \centering
    \includegraphics[width=0.75\linewidth]{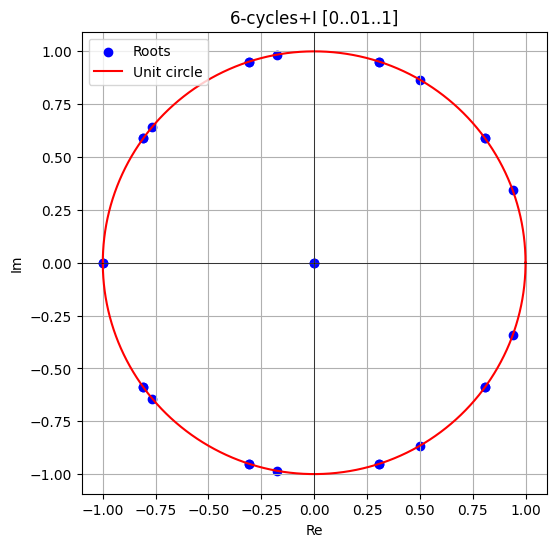}
    \caption{Consecutive 6-cycles +inv [0..01..1]}
    \label{fig:cons6inv_0011}
\end{figure}

\begin{table}[H]
\centering
\small
\setlength{\tabcolsep}{4pt}
\begin{tabular}{rrrrrrrrrrrrr}
\toprule
$n$ & \multicolumn{3}{c}{k=2} & \multicolumn{3}{c}{k=3} & \multicolumn{3}{c}{k=4} & \multicolumn{3}{c}{k=5} \\
\cmidrule(lr){2-4} \cmidrule(lr){5-7} \cmidrule(lr){8-10} \cmidrule(lr){11-13}
 & $d$ & $\Delta d$ & $\Delta^2 d$ & $d$ & $\Delta d$ & $\Delta^2 d$ & $d$ & $\Delta d$ & $\Delta^2 d$ & $d$ & $\Delta d$ & $\Delta^2 d$ \\
\midrule
4 & 4 & 2 & 1 & 2 & 1 & 1 & 2 & 1 & 0 &  &  &  \\
5 & 6 & 3 & 0 & 3 & 2 & -1 & 3 & 1 & -1 & 2 & 1 & 0 \\
6 & 9 & 3 & 1 & 5 & 1 & 1 & 4 & 0 & 2 & 3 & 1 & 0 \\
7 & 12 & 4 & 0 & 6 & 2 & 0 & 4 & 2 & -1 & 4 & 1 & -1 \\
8 & 16 & 4 & 1 & 8 & 2 & 1 & 6 & 1 & 1 & 5 & 0 & 2 \\
9 & 20 & 5 & 0 & 10 & 3 & -1 & 7 & 2 & -1 & 5 & 2 & -1 \\
10 & 25 & 5 & 1 & 13 & 2 & 1 & 9 & 1 & 1 & 7 & 1 & 1 \\
11 & 30 & 6 & 0 & 15 & 3 & 0 & 10 & 2 & 0 & 8 & 2 & -1 \\
12 & 36 & 6 & 1 & 18 & 3 & 1 & 12 & 2 & 1 & 10 & 1 & 1 \\
13 & 42 & 7 & 0 & 21 & 4 & -1 & 14 & 3 & -1 & 11 & 2 & -1 \\
14 & 49 & 7 & 1 & 25 & 3 & 1 & 17 & 2 & 1 & 13 & 1 & 1 \\
15 & 56 & 8 & 0 & 28 & 4 & 0 & 19 & 3 & -1 & 14 & 2 & 0 \\
16 & 64 & 8 & 1 & 32 & 4 & 1 & 22 & 2 & 1 & 16 & 2 & 1 \\
17 & 72 & 9 & 0 & 36 & 5 & -1 & 24 & 3 & 0 & 18 & 3 & -1 \\
18 & 81 & 9 & 1 & 41 & 4 & 1 & 27 & 3 & 1 & 21 & 2 & 1 \\
19 & 90 & 10 & 0 & 45 & 5 & 0 & 30 & 4 & -1 & 23 & 3 & -1 \\
20 & 100 & 10 & 1 & 50 & 5 & 1 & 34 & 3 & 1 & 26 & 2 & 1 \\
21 & 110 & 11 & 0 & 55 & 6 & -1 & 37 & 4 & -1 & 28 & 3 & -1 \\
22 & 121 & 11 & 1 & 61 & 5 & 1 & 41 & 3 & 1 & 31 & 2 & 1 \\
23 & 132 & 12 & 0 & 66 & 6 & 0 & 44 & 4 & 0 & 33 & 3 & 0 \\
24 & 144 & 12 & 1 & 72 & 6 & 1 & 48 & 4 & 1 & 36 & 3 & 1 \\
25 & 156 & 13 & 0 & 78 & 7 & -1 & 52 & 5 & -1 & 39 & 4 & -1 \\
26 & 169 & 13 & 1 & 85 & 6 & 1 & 57 & 4 & 1 & 43 & 3 & 1 \\
27 & 182 & 14 & 0 & 91 & 7 & 0 & 61 & 5 & -1 & 46 & 4 & -1 \\
28 & 196 & 14 & 1 & 98 & 7 & 1 & 66 & 4 & 1 & 50 & 3 & 1 \\
29 & 210 & 15 &  & 105 & 8 &  & 70 & 5 &  & 53 & 4 & -1 \\
30 & 225 &  &  & 113 &  &  & 75 &  &  & 57 & 3 & 1 \\
31 &  &  &  &  &  &  &  &  &  & 60 & 4 & 0 \\
32 &  &  &  &  &  &  &  &  &  & 64 & 4 &  \\
33 &  &  &  &  &  &  &  &  &  & 68 &  &  \\
\bottomrule
\end{tabular}
\caption{Diameters $D_k(n)$ and increments, consecutive $k$-cycle, inverse closed }
\label{tab:consecutive_cycle_growth_k=2..5_inv}
\end{table}

\clearpage
\subsection{\texorpdfstring{Schreier coset graph: $S_n / \bigl(S_{l} \times S_{n-l}\bigr)$ (``$k$-shrunken'' Grassmannian $Gr(l,n,k)$)}{Schreier coset graph: S_n/(S_l x S_{n-l}) ("k-shrunken" Grassmannian Gr(l,n,k))}}


Here we present conjectural formula for the diameters of Schreier coset graphs of the form
$S_n / \bigl(S_{l} \times S_{n-l}\bigr)$. 
The graph can be equivalently described as follows: its vertices correspond to binary vectors with components in $\{0,1\}$, containing exactly $l$ zeros and $n-l$ ones.
In \texttt{CayleyPy}, we define these graphs by setting the \texttt{central\_state} to be
$
[0]^l + [1]^{n-l}.
$
The generating set consists of consecutive $k$-cycles. We discuss both cases
inverse-closed and not in the present section.

The leading term of the formulas are  $l(n-l)/(k-1)$ as it is expected from the $(k-1)$-shrinkage principle,
since for the standard $k=2$ (Coxeter or neighbor transposition) case the diameter is exactly 
$l(n-l)$. 

Informally, one may think of this graph as a ``$k$-shrunken'' version of Grassmannian
$Gr(l, n, k)$ over the field with one element, since for $k=2$ by standard analogies it is indeed
Grassmanian over field with one element. From the point of view of that analogy
diameter corresponds to dimension of the manifold, and Poincare polynomial correspond to growth polynomial of the graph.

{\bf Striking new phenomena - bi-variable quasi-polynomiality} of the diameter formulas,
i.e. formulas behave as quasi-polynomials in both variables $n$ and $n-l$. 
It is expected to be true in larger generality e.g. for more general vector with repeats  ("partial flag manifolds") and
more general families of generators. 

\begin{Conj} (Not inverse closed case).
For consecutive $k$-cycle generators not inverse closed, the diameters of the 
Schreier coset graph: $S_n / \bigl(S_{l} \times S_{n-l}\bigr)$ (``$k$-shrunken''  Grassmannian $Gr(l,n,k)$)
are given by quasi-polynomials in two variables $n$ and $t = n-(l+1)$.

For $k=3$ with  period~$2$:
\[
D_l(n) = l\!\left(\left\lfloor \frac{t}{2} \right\rfloor + 1\right) + (t \bmod 2),
\qquad
t = n - (l+1), \quad n \ge l+1.
\]
Under the condition:  $l \ge 3$. For $l=2$, the formula is: 
$D_2(n) = n - 1, \qquad n \ge 4.$

For $k=4$ with  period~$3$ (again $t = n - (l+1)$):
\[
D_l(n) =
\begin{cases}
l, & t = 0, \\[2pt]
l+1, & t = 1, \\[2pt]
l+2, & t = 2, \\[6pt]
l\!\left( 2 + \left\lfloor \dfrac{t-3}{3} \right\rfloor \right),
& t \ge 3,\; t \equiv 0,1 \pmod{3}, \\[10pt]
l\!\left( 2 + \left\lfloor \dfrac{t-3}{3} \right\rfloor \right) + 1,
& t \ge 3,\; t \equiv 2 \pmod{3}.
\end{cases}
\]

For $k=5$ with  period~$4$ (again $t = n - (l+1)$,
for all $L \ge 4$, starting from
$ t+1 \ge 9 \qquad \text{that is,} \qquad N \ge L + 9$):
\[
D(L,N) = L \left\lfloor \frac{t+1}{4} \right\rfloor + \mathbf{1}_{(t+1) \equiv 0 \; (\mathrm{mod}\ 4)} .
\]

\end{Conj}

\begin{Conj} (Inverse closed case). Same setup as above , but inverse closed case.

For $k=3$ from $L\ge 2$(?): 
\[
D(L, N) = 
\left\lceil \frac{L (N - L)}{2} \right\rceil
\]

For $k=4$ from $L\ge 2$: 
\[
D(L, N) = 
\begin{cases} 
\frac{L}{3}(N - L), & 3 \mid L, \\[10pt]
\left\lfloor \frac{L(N - L) + 2}{3} \right\rfloor, & 3 \nmid L.
\end{cases}
\]

For $k=5$ for $N,L,N-L$ large enough: 
\[
D(L, N) = \left\lfloor \frac{LN + 2}{4} \right\rfloor - 2\left\lfloor \frac{L^2}{8} \right\rfloor - \mathbf{1}_{L \equiv 2 \pmod{4}} \, \mathbf{1}_{N \equiv 2 \pmod{4}}
\]

\end{Conj}

\clearpage
\subsection{Schreier coset graph: "few coincide"}
Here we present some partial results on diameters of the Schreier coset graphs of the form $S_n/S_D$ which can be equivalently described as graphs with nodes corresponding to vectors where $D$ elements coincide.
In CayleyPy we define them by setting "central\_state" to be 
 $(0,1,2,...n-D-1,n-D, n-D, ... , n-D)$ ($D$ coincide at the end). 
The generators are the  same - consecutive $k$ cycles, as in the everywhere in this section.

\subsubsection{Inverse closed case}

\begin{Conj}

For $k=3$ and coincide 4 the next quasi polynomial formulas are correct from $n=5$:

\[
\begin{aligned}
n&\equiv 0,1 \pmod{4}: & D_3(n)&=\frac{n^2-n-12}{4},\\[4pt]
n&\equiv 2,3 \pmod{4}: & D_3(n)&=\frac{n^2-n-10}{4},\\[4pt]
\end{aligned}
\]

For $k=3$ and coincide 3 the next quasi polynomial formulas are correct from $n=2$:

\[
\begin{aligned}
n&\equiv 0,1 \pmod{4}: & D_3(n)&=\frac{n^2-n-4}{4},\\[4pt]
n&\equiv 2,3 \pmod{4}: & D_3(n)&=\frac{n^2-n-6}{4},\\[4pt]
\end{aligned}
\]
\end{Conj}


For $k = 3$ and coincide 3 D: 
\[
\mathrm{OGF}
=\frac{H(x)}{(1-x^{4})^{3}},
\]
where
\[
\begin{aligned}
H(x)=- x^{12} - x^{11} + 2 x^{9} + 7 x^{8} + 9 x^{7} + 9 x^{6} + 7 x^{5} + 2 x^{4} - x^{2} - x.
\end{aligned}
\]

For $k = 3$ and coincide 4 D: 
\[
\mathrm{OGF}
=\frac{H(x)}{(1-x^{4})^{3}},
\]
where
\[
\begin{aligned}
H(x)=- 3 x^{12} - 2 x^{11} - x^{10} + 11 x^{8} + 11 x^{7} + 11 x^{6} + 11 x^{5} - x^{3} - 2 x^{2} - 3 x
\end{aligned}
\]

For the case $k=4$ and the consecutive cycles with inverses coset coincide, God's number $D_{4}(n)$ for different coincidence conditions $D$ is described by the following quasipolynomials:

For coincide 3D (from $n = 7$):

\[
D_{4}(n) = \left\lceil \frac{n(n-1)}{6} \right\rceil - 1
\]

For coincide 4D (from $n = 8$):
\[
D_{4}(n) = \left\lceil \frac{n(n-1)}{6} \right\rceil - 2
\]

For $k = 4$ and coincide 3 D: 
\[
\mathrm{OGF}
=\frac{H(x)}{(1-x^{3})^{3}},
\]
where
\[
\begin{aligned}
H(x)=- x^{9} + 4 x^{6} + 3 x^{5} + 4 x^{4} - x
\end{aligned}
\]

For $k = 4$ and coincide 4 D: 
\[
\mathrm{OGF}
=\frac{H(x)}{(1-x^{3})^{3}},
\]
where
\[
\begin{aligned}
H(x)=- 2 x^{9} - x^{8} - x^{7} + 6 x^{6} + 5 x^{5} + 6 x^{4} - x^{3} - x^{2} - 2 x
\end{aligned}
\]

where $n$ is the length of the graph.

\begin{figure}[H]
    \centering
    \includegraphics[width=0.6\linewidth]{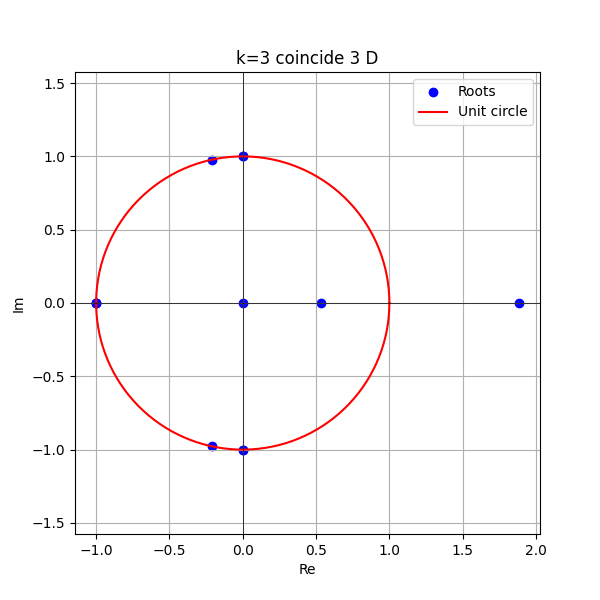}
    \caption{Consecutive k=3 +inv (coincide 3)}
    \label{fig:consk3inv_3D}
\end{figure}

\begin{figure}[H]
    \centering
    \includegraphics[width=0.6\linewidth]{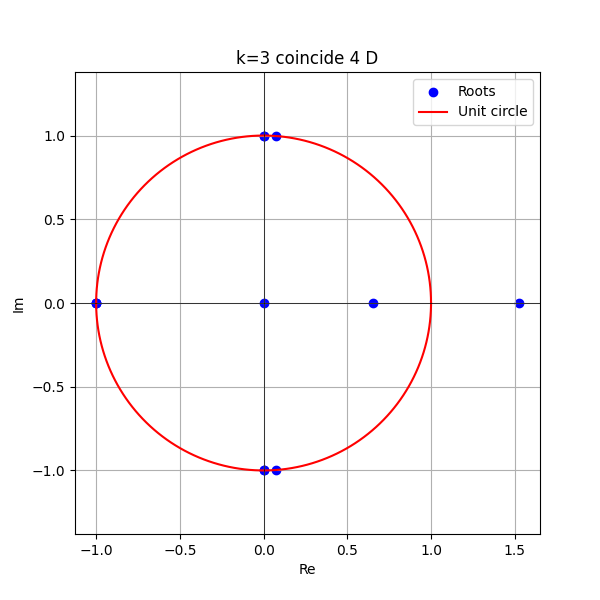}
    \caption{Consecutive k=3 +inv (coincide 4)}
    \label{fig:consk3inv_4D}
\end{figure}

\begin{figure}[H]
    \centering
    \includegraphics[width=0.6\linewidth]{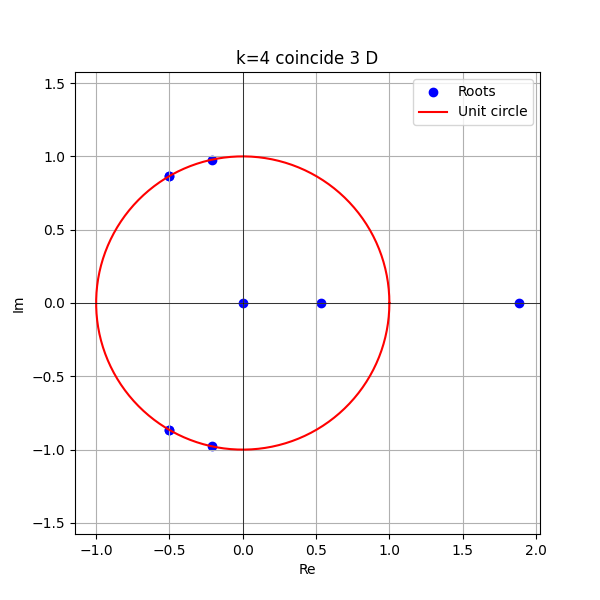}
    \caption{Consecutive k=4 +inv (coincide 3)}
    \label{fig:consk4inv_3D}
\end{figure}

\begin{figure}[H]
    \centering
    \includegraphics[width=0.6\linewidth]{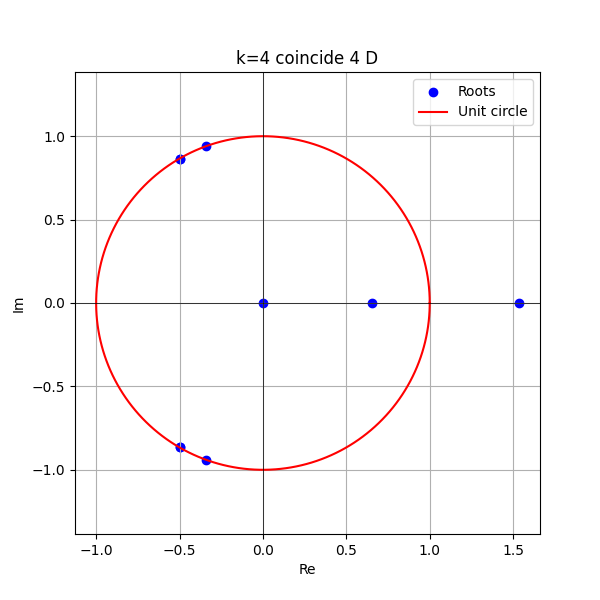}
    \caption{Consecutive k=4 +inv (coincide 4)}
    \label{fig:consk4inv_4D}
\end{figure}

\begin{table}[htbp]
\centering
\caption{God's numbers for consecutive cycles \emph{with inverses} coset: 
central state  $(0,1,2,...n-D-1,n-D, n-D, ... , n-D)$ ($D$ coincide at the end) }
\label{tab:coset-with-inverses}
\begin{tabular}{lccccc}
\toprule
\textbf{k} & \textbf{n} & \textbf{no coincide D} & \textbf{coincide 2 D} & \textbf{coincide 3 D} & \textbf{coincide 4 D}\\
\midrule
3 & 4 & 3 & 3 & 2 & --\\
3 & 5 & 5 & 5 & 4 & 2\\
3 & 6 & 7 & 7 & 6 & 5\\
3 & 7 & 10 & 10 & 9 & 8\\
3 & 8 & 14 & 14 & 13 & 11\\
3 & 9 & 18 & 18 & 17 & 15\\
3 & 10 & 22 & 22 & 21 & 20\\
3 & 11 & 27 & 27 & 26 & 25\\
3 & 12 & 33 & 33 & 32 & 30\\
\midrule[1.2pt]
4 & 5 & 7 & 5 & 3 & 2\\
4 & 6 & 7 & 5 & 5 & 4\\
4 & 7 & 8 & 7 & 6 & 6\\
4 & 8 & 10 & 9 & 9 & 8\\
4 & 9 & 13 & 12 & 11 & 10\\
4 & 10 & 16 & 15 & 14 & 13\\
4 & 11 & 19 & 18 & 18 & 17\\
4 & 12 & 23 & 22 & 21 & 20\\
4 & 13 &  &  & 25 & 24\\
\midrule[1.2pt]
5 & 6 & 9 & 9 & 7 & 4\\
5 & 7 & 8 & 8 & 7 & 5\\
5 & 8 & 8 & 8 & 7 & 7\\
5 & 9 & 9 & 9 & 9 & 8\\
5 & 10 & 11 & 11 & 11 & 10\\
5 & 11 & 14 & 14 & 13 & 13\\
5 & 12 & 17 & 17 & 16 & 15\\
\midrule[1.2pt]
6 & 7 & 13 & 11 & 9 & 8\\
6 & 8 & 12 & 10 & 9 & 7\\
6 & 9 & 12 & 10 & 9 & 8\\
6 & 10 & 12 & 10 & 10 & 9\\
6 & 11 & 13 & 12 & 11 & 10\\
6 & 12 & 14 & 13 & 13 & 12\\
\bottomrule
\end{tabular}
\end{table}

\clearpage
\subsubsection{Not inverse closed case}

\setlength{\LTpre}{0pt}
\setlength{\LTpost}{0pt}

\begin{longtable}{lccccccc}
\caption{God's numbers for consecutive cycles \emph{without inverses} coset: 
central state $(0,1,2,\dots,n-D-1,n-D,n-D,\dots,n-D)$ ($D$ coincide at the end)}
\label{tab:coset-without-inverses}\\

\toprule
\textbf{k} & \textbf{n} & \textbf{no c. D} & \textbf{D = 2} & \textbf{D = 3} & \textbf{D = 4} & \textbf{D = 5} & \textbf{D = 6}\\
\midrule
\endfirsthead

\toprule
\textbf{k} & \textbf{n} & \textbf{no c. D} & \textbf{D = 2} & \textbf{D = 3} & \textbf{D = 4} & \textbf{D = 5} & \textbf{D = 6}\\
\midrule
\endhead

\midrule
\multicolumn{8}{r}{\emph{Continued on next page}}
\\
\endfoot

\bottomrule
\endlastfoot
2 & 3 & 3 & 2 & -- & -- & -- & --\\
2 & 4 & 6 & 5 & 3 & -- & -- & --\\
2 & 5 & 10 & 9 & 7 & 4 & -- & --\\
2 & 6 & 15 & 14 & 12 & 9 & 5 & --\\
2 & 7 & 21 & 20 & 18 & 15 & 11 & 6\\
2 & 8 & 28 & 27 & 25 & 22 & 18 & 13\\
2 & 9 & 36 & 35 & 33 & 30 & 26 & 21\\
2 & 10 & 45 & 44 & 42 & 39 & 35 & 30\\
2 & 11 & 55 & 54 & 52 & 49 & 45 & 40\\
2 & 12 & 66 & 65 & 63 & 60 & 56 & 51\\
2 & 13 & 78 & 77 & 75 & 72 & 68 & 63\\
2 & 14 & ?? & ?? & 88 & 85 & 81 & 76\\
\midrule[1.2pt]

3 & 4 & 4 & 4 & 3 & -- & -- & --\\
3 & 5 & 6 & 6 & 5 & 3 & -- & --\\
3 & 6 & 9 & 9 & 8 & 6 & 4 & --\\
3 & 7 & 12 & 12 & 11 & 9 & 7 & 4\\
3 & 8 & 16 & 16 & 15 & 13 & 11 & 8\\
3 & 9 & 20 & 20 & 19 & 17 & 15 & 12\\
3 & 10 & 25 & 25 & 24 & 22 & 20 & 17\\
3 & 11 & 30 & 30 & 29 & 27 & 25 & 22\\
3 & 12 & 36 & 36 & 35 & 33 & 31 & 28\\
3 & 13 & 42 & 42 & 41 & 39 & 37 & 34\\
3 & 14 & ?? & ?? & 48 & 46 & 44 & 41\\
\midrule[1.2pt]

4 & 5 & 9 & 8 & 5 & 4 & -- & --\\
4 & 6 & 9 & 7 & 6 & 5 & 4 & --\\
4 & 7 & 10 & 9 & 8 & 7 & 6 & 4\\
4 & 8 & 13 & 12 & 11 & 10 & 8 & 7\\
4 & 9 & 16 & 15 & 14 & 13 & 11 & 9\\
4 & 10 & 19 & 18 & 17 & 16 & 14 & 12\\
4 & 11 & 23 & 22 & 21 & 20 & 18 & 16\\
4 & 12 & 27 & 26 & 25 & 24 & 22 & 20\\
4 & 13 & 31 & 30 & 29 & 28 & 26 & 24\\
4 & 14 & ?? & ?? & ?? & 33 & 31 & 29\\
\midrule[1.2pt]

5 & 6 & 11 & 11 & 9 & 6 & 5 & --\\
5 & 7 & 10 & 10 & 9 & 7 & 6 & 5\\
5 & 8 & 11 & 11 & 10 & 8 & 7 & 6\\
5 & 9 & 13 & 13 & 12 & 10 & 9 & 8\\
5 & 10 & 15 & 15 & 14 & 13 & 12 & 10\\
5 & 11 & 18 & 18 & 17 & 16 & 15 & 13\\
5 & 12 & 21 & 21 & 20 & 19 & 18 & 16\\
5 & 13 & 24 & 24 & 23 & 22 & 21 & 19\\
5 & 14 & ?? & ?? & ?? & 26 & 25 & 23\\
\midrule[1.2pt]

6 & 7 & 19 & 16 & 13 & 11 & 7 & 6\\
6 & 8 & 14 & 13 & 11 & 10 & 8 & 7\\
6 & 9 & 14 & 13 & 11 & 10 & 9 & 8\\
6 & 10 & 15 & 13 & 12 & 11 & 10 & 9\\
6 & 11 & 17 & 15 & 14 & 13 & 12 & 11\\
6 & 12 & 19 & 18 & 17 & 16 & 15 & 14\\
6 & 13 & 22 & 21 & 20 & 19 & 18 & 17\\
6 & 14 & ?? & ?? & ?? & 22 & 21 & 20\\

\end{longtable}

\clearpage
\subsection{Schreier coset graph:  "L-Different"  inverse closed}
Here we present some partial results on diameters of the Schreier coset graphs of the form $S_n/S_{n-L}$ which can be equivalently described as graphs with nodes corresponding to vectors where only $L$ elements are different.
In CayleyPy we define them by setting the central state to 
$$
(0,1,2,\ldots,L-2,L-1,\ldots, L-1) = 0^11^12^1\ldots (L-2)^1(L-1)^{n - L - 1}.
$$
The generators are the same --- consecutive $k$ cycles, as everywhere in this section. The next two conjectures for 2- and 3-Different coset graphs are based on the data obtained in this \href{https://www.kaggle.com/code/fedmug/cayleypy-l-different-k-cycles-growth}{Kaggle notebook}. 

\subsubsection{2-Different coset}
\begin{Conj}
The diameter for $k$-cycles 2-Different coset graph with central state of length $n$ is
\begin{equation*}
    d_k(n) = \Big\lfloor \frac{2n + k(k-4)}{2k-2} \Big\rfloor + 1 - \mathbb I[k \text{ is odd}, 2 \leqslant n \mod (k - 1) \leqslant (k-1)/2]
\end{equation*}
\end{Conj}

\subsubsection{3-Different}


\begin{Conj}
We conjecture that the diameter for $k$-cycles 3-Different coset graph with central state of length $n$ equals to
\begin{equation}
    d_k(n) = \left\lfloor \frac{2n + c_k}{k-1} \right\rfloor,\quad
    k \geqslant 2, \quad n \geqslant N_k,
    \label{eq:3diff-formula}
\end{equation}
where $c_2 = -3$, $c_3 = -1$, and for $k\geqslant 4$ the offset $c_k$ is a quadratic function of $k$ whose form depends on the parity of $k$:


\begin{equation}
    c_k = \begin{cases}
        \dfrac{k(k-1)}{2} - 5 & \text{if } k \text{ is even}, \\[6pt]
        \dfrac{(k-1)^2}{2} - 4 & \text{if } k \text{ is odd.}
    \end{cases}
    \label{eq:ck-formula}
\end{equation}
    
\end{Conj}

The indices $N_k$, from which the formula~\eqref{eq:3diff-formula} starts working, and $c_k$ are tabulated in Table~\ref{tab:ldiff-3-2k20}.

\begin{table}[ht]
\centering
\caption{3-Different diameter conjectures for $2 \leqslant k \leqslant 20$}
\label{tab:ldiff-3-2k20}
\begin{tabular}{| l | l | l | l |}
\hline
\textbf{$k$} & \textbf{$N_k$} & \textbf{Offset ($c_k$)} & \textbf{$d_k(n)$} \\
\hline
\textbf{2}  & 3  & \(-3\)  & $2n-3$\\
\textbf{3}  & 4  & \(-1\)  & $n - 1$ \\
\textbf{4}  & 5  & \(1\)   & \(\big\lfloor \frac{2n + 1}{3} \big\rfloor\) \\
\textbf{5}  & 6\rule{0pt}{4mm}  & \(4\)   & \(\big\lfloor \frac{2n + 4}{4} \big\rfloor\) \\
\textbf{6}  & 8\rule{0pt}{4mm} & \(10\)  & \(\left\lfloor \frac{2n + 10}{5} \right\rfloor\) \\
\textbf{7}  & 8\rule{0pt}{4mm}  & \(14\)  & \(\left\lfloor \frac{2n + 14}{6} \right\rfloor\) \\
\textbf{8}  & 10\rule{0pt}{4mm} & \(23\)  & \(\left\lfloor \frac{2n + 23}{7} \right\rfloor\) \\
\textbf{9}  & 18\rule{0pt}{4mm} & \(28\)  & \(\left\lfloor \frac{2n + 28}{8} \right\rfloor\) \\
\textbf{10} & 13\rule{0pt}{4mm} & \(40\)  & \(\left\lfloor \frac{2n + 40}{9} \right\rfloor\) \\
\textbf{11} & 32\rule{0pt}{4mm} & \(46\)  & \(\left\lfloor \frac{2n + 46}{10} \right\rfloor\) \\
\textbf{12} & 25\rule{0pt}{4mm} & \(61\)  & \(\left\lfloor \frac{2n + 61}{11} \right\rfloor\) \\
\textbf{13} & 50\rule{0pt}{4mm} & \(68\)  & \(\left\lfloor \frac{2n + 68}{12} \right\rfloor\) \\
\textbf{14} & 42\rule{0pt}{4mm} & \(86\)  & \(\left\lfloor \frac{2n + 86}{13} \right\rfloor\) \\
\textbf{15} & 72\rule{0pt}{4mm} & \(94\)  & \(\left\lfloor \frac{2n + 94}{14} \right\rfloor\) \\
\textbf{16} & 63\rule{0pt}{4mm} & \(115\) & \(\left\lfloor \frac{2n + 115}{15} \right\rfloor\) \\
\textbf{17} & 98\rule{0pt}{4mm} & \(124\)  & \(\left\lfloor \frac{2n + 124}{16} \right\rfloor\) \\
\textbf{18} & 88\rule{0pt}{4mm} & \(148\) & \(\left\lfloor \frac{2n + 148}{17} \right\rfloor\) \\
\textbf{19} & 128\rule{0pt}{4mm} & \(158\)  & \(\left\lfloor \frac{2n + 158}{18} \right\rfloor\) \\
\textbf{20} & 117\rule{0pt}{4mm} & \(185\) & \(\left\lfloor \frac{2n + 185}{19} \right\rfloor\) \\
\hline
\end{tabular}
\end{table}

The values $N_k$ also eventually seem to obey the quadratic law:
$$
    N_k = \begin{cases}
        \dfrac{k^2 - 9k + 14}2, & k \text{ is even}, \\[6pt]
        \dfrac{(k-3)^2}{2}, & k \text{ is odd.}
    \end{cases}
$$

Our calculations for $L = 4$ show that
$d_k(n)\approx  \left\lfloor \frac{3n}{k-1} + \frac{k-5}2 \right\rfloor$ if $k \geqslant 5$. Combining this observation with two previous conjectures, we propose the following:

\begin{Conj}
The diameter for $k$-cycles L-Different coset graph with central state of length $n$ is
\begin{equation*}
    d_k(n) = \Big\lfloor \frac{Ln}{2(k-1)} + c_k \Big\rfloor +r_k(n),
\end{equation*}
where $c_k = O(k)$, $r_k(n)$ is a "small" remainder (possibly $O(1)$).
\end{Conj}

\subsubsection{Last layer size (3-Different)}

Let $s_k(n)$ denote the size of the last layer of the $k$-cycle 3-Different inverse-closed Schreier coset graph with central state of length $n$. We conjecture that, for fixed $k$, the sequence $s_k(n)$ is eventually periodic in $n$, with behavior depending on the parity of $k$.

\begin{Conj}[Even $k$]
For even $k \ge 6$, we conjecture that, for all sufficiently large $n$, the sequence $s_k(n)$ follows a single periodic pattern of odd values. Let
\begin{equation*}
    C_k = \frac{k(k-1)}{2} - 5,
    \qquad
    m = (2n + C_k) \bmod (k-1).
\end{equation*}
Then
\begin{equation*}
    s_k(n) = m + \big(1 - (m \bmod 2)\big)
    + \left\lfloor \frac{1}{|m - (k-2)| + 1} \right\rfloor
    + 2\left\lfloor \frac{1}{|m - (k-3)| + 1} \right\rfloor.
\end{equation*}
\end{Conj}

\begin{Conj}[Odd $k$]
For odd $k \ge 7$, we conjecture that, for all sufficiently large $n$, the sequence $s_k(n)$ follows a two-half alternating periodic pattern, where each half has length
\begin{equation*}
    H = \frac{k-1}{2}.
\end{equation*}
Let
\begin{equation*}
    C_k = \frac{(k-1)^2}{2} - 4,
    \qquad
    m = \left(n + \frac{C_k}{2} + H\right) \bmod (k-1).
\end{equation*}
Define the two correction locations by
\begin{equation*}
    (m_{\mathrm{origin}}, m_{\mathrm{spike}}) =
    \begin{cases}
        (H,\, 2H-1), & k \equiv 1 \pmod{4},\\
        (0,\, H-1), & k \equiv 3 \pmod{4}.
    \end{cases}
\end{equation*}
Then
\begin{equation*}
    s_k(n) = m + \big(1 - (m \bmod 2)\big)
    + \left\lfloor \frac{1}{|m - m_{\mathrm{origin}}| + 1} \right\rfloor
    + 2\left\lfloor \frac{1}{|m - m_{\mathrm{spike}}| + 1} \right\rfloor.
\end{equation*}
\end{Conj}

\begin{figure}
    \centering
    \includegraphics[width=0.5\linewidth]{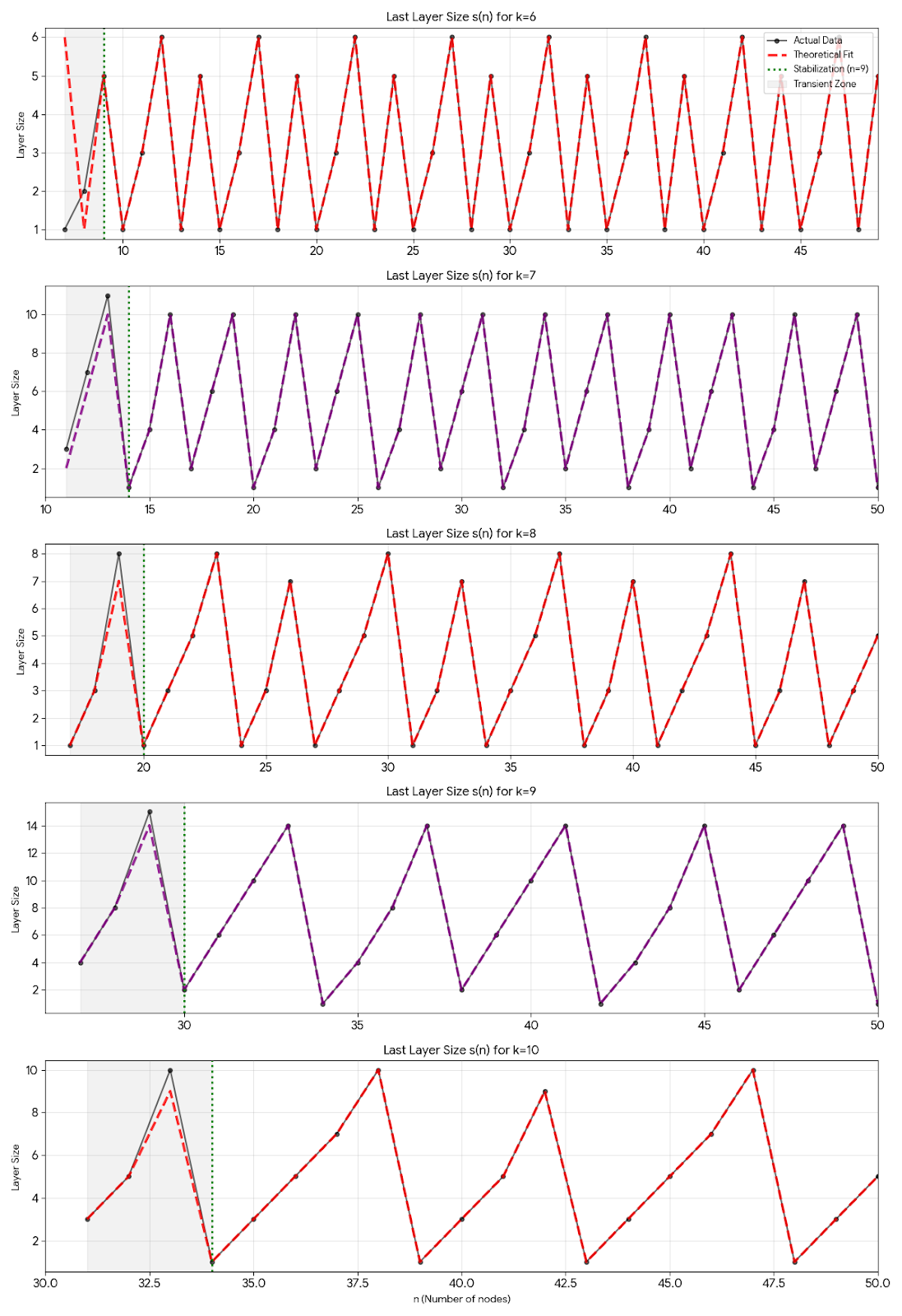}
    \caption{Last layer size for k=6...10. Empirical data coincides with the proposed formula for large enough $n$.}
    \label{fig:ldiff-inv-closed-lastlayer-k-6-10}
\end{figure}

\begin{figure}
    \centering
    \includegraphics[width=0.5\linewidth]{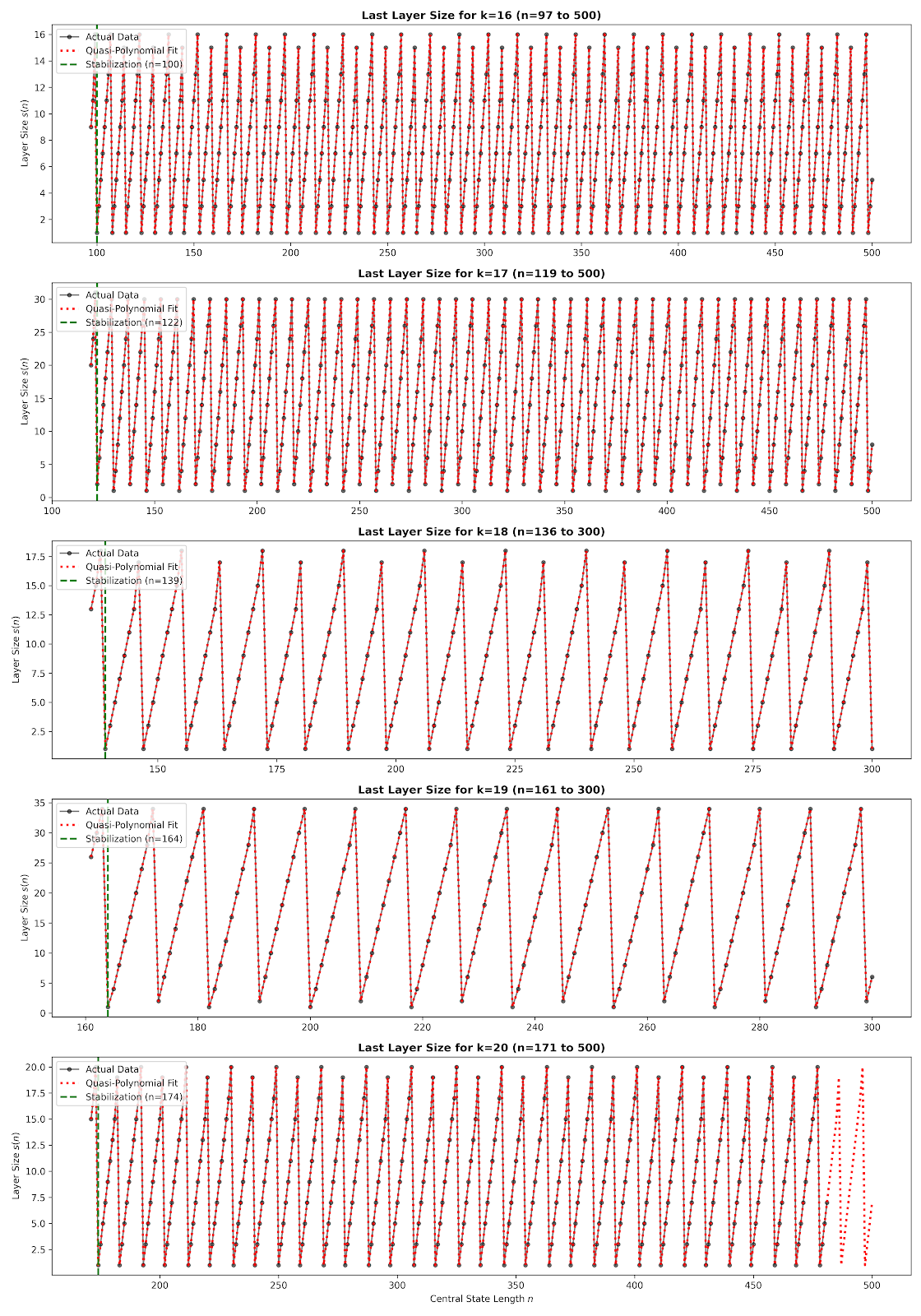}
    \caption{Last layer size for k=16...20. Empirical data coincides with the proposed formula for large enough $n$.}
    \label{fig:ldiff-inv-closed-lastlayer-k-16-20}
\end{figure}

\clearpage
\subsection{Word-metrics for [012]-repeated full flips, 3-cycles}

Here we consider consecutive 3 cycles generators. 
The Schreir coset graph of the form $S_n/(S_{n//3}\times S_{n//3}\times S_{n-2n//3})$, which alternatively can be described as graph with nodes given by taking all vectors with only $0,1,2$ components
where $0$ is repeated $n//3$ times, $1$ is repeated $n//3$ times, and the rest is $2$. 
In CayleyPy such coset Schreier graph is defined by setting "central\_state" appropriately. 

Here we present a conjectural quasi-polynomial expression for the word-metric between two 
families vertices of such graphs.

Let us denote by $r$ residue of $n$ modulo 3 and by $m$ its integer division by 3:
\begin{equation*}
    n = 3m + r \quad \text{with} \quad r \in \{0,1,2\}.
\end{equation*}

Consider the first family of states given by: 
\begin{equation*}
\begin{aligned}
    & [0,1,2]^m \quad \text{if} \quad r = 0 \\
    & [0,1,2]^m + [0] \quad \text{if} \quad r = 1 \\
    & [0,1,2]^m + [0,1] \quad \text{if} \quad r = 2
\end{aligned}    
\end{equation*}
And the other states: 
\begin{equation}
\begin{aligned}
    & [2]^m + [1]^m + [0]^m \quad \text{if} \quad r = 0 \\
    & [2]^m + [1]^m + [0]^{m+1} \quad \text{if} \quad r = 1 \\
    & [2]^m + [1]^{m+1} + [0]^{m+1} \quad \text{if} \quad r = 2
\end{aligned}    
\end{equation}

\begin{Conj}
    \begin{enumerate}
        \item The second states  are  farthest states from the first ones 
        \medskip
        
        \item Starting from $n = 5$ the  distance between them 
        is given by the following degree two quasipolynomial
        \[
        \begin{aligned}
        n &\equiv 0 \pmod{12}:&\; D(n) &= \frac{1}{12}n^2 + \frac{1}{4}n,\\[4pt]
        n &\equiv 1 \pmod{12}:&\; D(n) &= \frac{1}{12}n^2 + \frac{1}{12}n - \frac{1}{6},\\[4pt]
        n &\equiv 2 \pmod{12}:&\; D(n) &= \frac{1}{12}n^2 + \frac{1}{12}n + \frac{1}{2},\\[4pt]
        n &\equiv 3 \pmod{12}:&\; D(n) &= \frac{1}{12}n^2 + \frac{1}{4}n + \frac{1}{2},\\[4pt]
        n &\equiv 4 \pmod{12}:&\; D(n) &= \frac{1}{12}n^2 + \frac{1}{12}n + \frac{1}{3},\\[4pt]
        n &\equiv 5 \pmod{12}:&\; D(n) &= \frac{1}{12}n^2 + \frac{1}{12}n + \frac{1}{2},\\[4pt]
        n &\equiv 6 \pmod{12}:&\; D(n) &= \frac{1}{12}n^2 + \frac{1}{4}n + \frac{1}{2},\\[4pt]
        n &\equiv 7 \pmod{12}:&\; D(n) &= \frac{1}{12}n^2 + \frac{1}{12}n + \frac{1}{3},\\[4pt]
        n &\equiv 8 \pmod{12}:&\; D(n) &= \frac{1}{12}n^2 + \frac{1}{12}n,\\[4pt]
        n &\equiv 9 \pmod{12}:&\; D(n) &= \frac{1}{12}n^2 + \frac{1}{4}n,\\[4pt]
        n &\equiv 10 \pmod{12}:&\; D(n) &= \frac{1}{12}n^2 + \frac{1}{12}n - \frac{1}{6},\\[4pt]
        n &\equiv 11 \pmod{12}:&\; D(n) &= \frac{1}{12}n^2 + \frac{1}{12}n,\\[4pt]
        \end{aligned}
        \]
        of period $12$. Note that the leading term is always $n^2/12$ and the period appears as its denominator. Also note that we have
        \begin{equation*}
            D(0) = D(1) = 0.
        \end{equation*}

        \medskip
        \item The corresponding generating function is given by 
        \begin{equation}\label{eq:generating-function-012}
        \frac{H(x)}{(1 - x^{12})^3},
        \end{equation}
        where
        \begin{multline*}
        H(x) = x^{32} + 2x^{31} + 2x^{30} + 4x^{29} + 5x^{28} + 5x^{27} + 8x^{26} + 9x^{25} + \\9x^{24} + 13x^{23} + 15x^{22} + 15x^{21} + 17x^{20} + 17x^{19} + 17x^{18} + 17x^{17} + \\17x^{16} + 17x^{15} + 15x^{14} + 15x^{13} + 15x^{12} + 11x^{11} + 9x^{10} + 9x^9  \\ + 6x^8 + 5x^7 + 5x^6 + 3x^5 + 2x^4 + 2x^3 + x^2
        \end{multline*}
        The polynomial $H(x)$ has the following properties:
        \begin{itemize}
            \item its coefficients are nonnegative and unimodal;
            \item its coefficients are neither symmetric nor weakly log concave;
            \item its roots are not necessarily on the unit circle (can be both inside and outside).
        \end{itemize}

    \medskip
    \item The $H$-polynomial factors as
    \begin{equation*}
        H(x) = x^2  (x + 1)^3  (x^2 + x + 1)  (x^2 + 1)^2  (x^2 - x + 1)^3  (x^4 - x^2 + 1)^2  (x^7 + x^6 - x^3 + x + 1)
    \end{equation*}
    Hence, the generating function \eqref{eq:generating-function-012} can be rewritten as
    \begin{equation*}
    \frac{x^2 (x^7 + x^6 - x^3 + x + 1)}
    {(1-x)^3  (x^2 + 1) (x^2 + x + 1)^2 (x^4 - x^2 + 1)},
    \end{equation*}
    \end{enumerate}
  
\end{Conj}

The conjecture is checked up to $n=42$. 
The following data was obtained on Kaggle using CayleyPy in this
\href{https://www.kaggle.com/code/msmirnov18/cayleypy-coset012repeats-consecutive-cycles}{notebook}.
It is a dictionary of value pairs $n : v$, where $n$ is as above and $v$ is the (experimentally) shortest path from the initial state to the central state.\footnote{\textbf{There might be some issue with the value for $n = 4$, but we don't need it below anyway.}}
\begin{multline*}
\{4: 3, 5: 3, 6: 5, 7: 5, 8: 6, 9: 9, 10: 9, 11: 11, 12: 15, 13: 15, 
14: 18, 15: 23, 16: 23, \\ 17: 26,  18: 32, 19: 32, 20: 35, 21: 42, 22: 42, 23: 46, 24: 54, 25: 54, 26: 59, 27: 68, 28: 68,\\  29: 73, 30:83, 31:83, 32:88, 33:99, 34:99, 35:105, 36:117, 37:117,  38:124, 39:137, 42:158\}
\end{multline*}
The computations regarding the properties of the $H$-polynomial are done in this 
\href{https://www.kaggle.com/code/msmirnov18/h-polynomial-coset012repeats-consecutive-cycles}{notebook}.

\clearpage
\section{Wrapped ("affine" or "periodic") (k)-consecutive cycles. Quasi-polynomiality, etc.}
\subsection{Section outline}
\textbf{Quasi-polynomiality.}
In this section, we present several quasi-polynomial formulas for the diameters and word metrics associated with the "wrapped" version of the consecutive $k$-cycle generators (also can be called "affine" or "cyclic" or "periodic") defined below. We also compute the corresponding $H$-polynomials and study their properties.
Both Cayley graphs and Schreier coset graphs are considered. We also provide some other results like theoretical
lower and upper bounds on diameters which are in agreement with our experimental studies.

The section is quite parallel to the previous one devoted to consecutive cycles case.


{\bf Generators definition. Related works.}
Fix an integer $k$. For $n > k$, we consider elements of $S_n$ given by the cyclic permutations 
$(i, (i+1) \mod n, \dots, (i+k-1) \mod n )$ for $i = 0, \dots, n-1$. For $k = 2$, these are the neighbor transpositions 
(cyclic Coxeter) considered previously. 
In CayleyPy these generators are denoted:
\href{https://cayleypy.github.io/cayleypy-docs/generated/cayleypy.PermutationGroups.html#cayleypy.PermutationGroups.consecutive_k_cycleshttps://cayleypy.github.io/cayleypy-docs/generated/cayleypy.PermutationGroups.html#cayleypy.PermutationGroups.wrapped_k_cycles}{wrapped\_k\_cycles}.
For odd $k$, they generate $A_n$ 
inside $S_n$, for even $k$ they give $S_n$. There are two natural options: whether or not to include inverses in the generating set.
We consider both cases.
To the best of our knowledge, these generators were not studied systematically in the literature, 
brief discussion is in our previous work: \cite{Cayley3Growth}).
From the physical point of view, the Laplacian of such graphs corresponds to the  situation
when $k$ neighboring spins interact. The simplest case $k=2$ corresponds to "periodic" or "closed" or "affine" spin chain.


\subsection{\texorpdfstring{$2(k-1)$-Shrinkage heuristics: Results and Difficulties}{2(k-1)-Shrinkage heuristics: Results and Difficulties}}

Before going into details let us first give some informal heuristic principle 
summarizing the results: to obtain results on diameters,
word-metrics and other characteristics one should take results for the standard Coxeter (i.e. neighbor
transpositions case)  case 
and just divide them by $2(k-1)$. 
In all considered cases it gives correct leading terms and moreover in some rare cases simple correction
like adding ceil-rounding would suffice to get the correct results. 
However exact results typically are more complicated. 



\subsection{\texorpdfstring{Theoretical diameters estimate. Schreier coset graphs $S_n / \bigl(S_{\lfloor n/2 \rfloor} \times S_{n-\lfloor n/2 \rfloor}\bigr)$}{Theoretical diameters estimate. Schreier coset graphs S_n/(S_{floor(n/2)} x S_{n-floor(n/2)})}}

\begin{thm}
    The diameter of the coset with $\lfloor n/2\rfloor$ zeros and $n-\lfloor n/2\rfloor $ ones is equal to $\frac{n^2}{8(k-1)}+O(n).$
\end{thm}

\begin{proof}
    Since the $k$-cycle is wrapped, one can consider a modified setup where we are to obtain any permutation from the one where the zeros are located at the first $\lfloor n/4\rfloor$ points and the last $\lfloor n/2\rfloor-\lfloor n/4\rfloor$ points.

    Suppose that we want to achieve the sequence of zeros and ones where the zeros are located at places $a_1,a_2,\dots,a_l.$ Create a modified sequence where the zeros are located at $$(k-1)\lfloor a_1/(k-1)\rfloor,\dots,(k-1)\lfloor a_1/(k-1)\rfloor+k-2,$$ $$ (k-1)\lfloor a_k/(k-1)\rfloor,\dots,(k-1)\lfloor a_k/(k-1)\rfloor+k-2,\dots,$$ $$ (k-1)\lfloor a_{m(k-1)+1}/(k-1)\rfloor,\dots,(k-1)\lfloor a_{m(k-1)+1}/(k-1)\rfloor,\dots$$ This modified sequence has the zeros \textit{and} ones come in groups of $k-1$. Since a group of size $k-1$ or less can be moved to a distance of 1 to the right by 1 turn of the cycle, one can exchange two such groups in $k-1$ moves. This time the modified position is reachable in at most $(k-1)\lceil n/4/(k-1)\rceil (\lfloor n/2/(k-1)\rfloor+1)$ moves, because one can choose whether any group will be delivered from the left or from the right.

    Next we reconstruct the original sequence by guiding every zero to its rightful place in $O(n)$ turns. 
\end{proof}

\clearpage
\subsection{\texorpdfstring{Schreier coset graph: $S_n / \bigl(S_{l} \times S_{n-l}\bigr)$ (``$k$-shrunken'' affine Grassmannian $Gr^{aff}(l,n,k)$)}{Schreier coset graph: S_n/(S_l x S_{n-l}) ("k-shrunken" affine Grassmannian Gr^{aff}(l,n,k))}}


Here we present conjectural formula for the diameters of Schreier coset graphs of the form
$S_n / \bigl(S_{l} \times S_{n-l}\bigr)$. 
The graph can be equivalently described as follows: its vertices correspond to binary vectors with components in $\{0,1\}$, containing exactly $l$ zeros and $n-l$ ones.
In \texttt{CayleyPy}, we define these graphs by setting the \texttt{central\_state} to be
$
[0]^l + [1]^{n-l}.
$
The generating set consists of wrapped consecutive $k$-cycles. We discuss both cases
inverse-closed and not in the present section.

The leading term of the formulas are  $l(n-l)/2(k-1)$ as it is expected from the $2(k-1)$-shrinkage principle,
since for the standard $k=2$ (Coxeter or neighbor transposition) case the diameter is exactly 
$l(n-l)$. 

Informally, one may think of this graph as a ``$k$-shrunken'' version of affine Grassmannian
$Gr^{aff}(l, n, k)$ over the field with one element, since for $k=2$ by standard analogies it is indeed
Grassmanian over field with one element. From the point of view of that analogy
diameter corresponds to dimension of the manifold, and Poincare polynomial correspond to growth polynomial of the graph.

{\bf Striking new phenomena - bi-variable quasi-polynomiality} of the diameter formulas,
i.e. formulas behave as quasi-polynomials in both variables $n$ and $n-l$. 
It is expected to be true in larger generality e.g. for more general vector with repeats  ("partial flag manifolds") and
more general families of generators. 

\begin{Conj} For $k=3$, inverse closed, diameters are given by:

\[
d_n = \left\lceil \frac{\ell(n - \ell)}{4} \right\rceil + [n \equiv 0 \pmod 4, \; \ell \equiv 2 \pmod 4]
\]

\end{Conj}

\begin{figure}[H]
    \centering
    \includegraphics[width=0.75\linewidth]{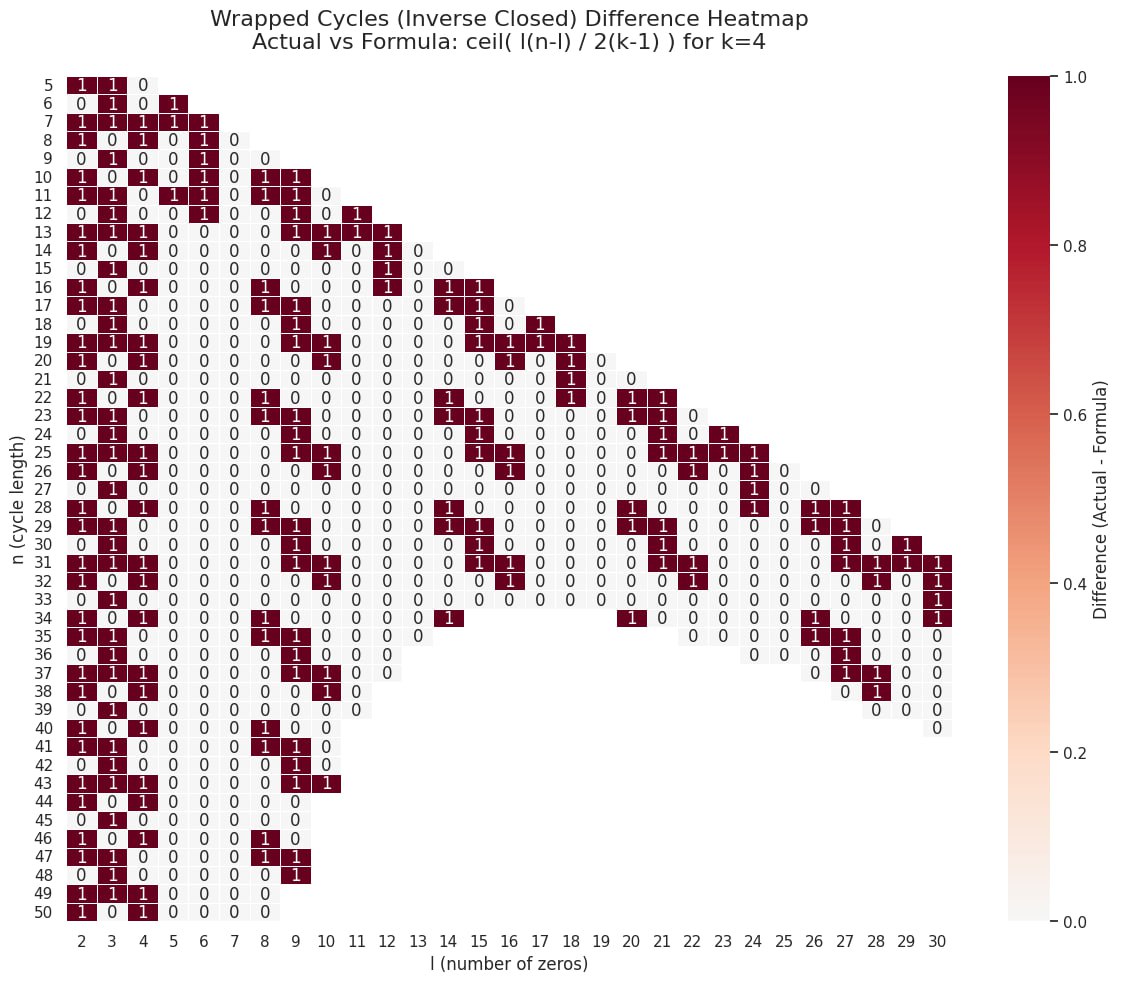}
    \caption{Difference to naive estimation $\lceil l(n-l)/2(k-1) \rceil $ wrapped 4 cycles inverse closed.
    Bi-variable quasi-polynomials pattern.}
    \label{fig:wrapped_4_grassman_inv_closed}
\end{figure}
\begin{figure}[H]
    \centering
    \includegraphics[width=0.75\linewidth]{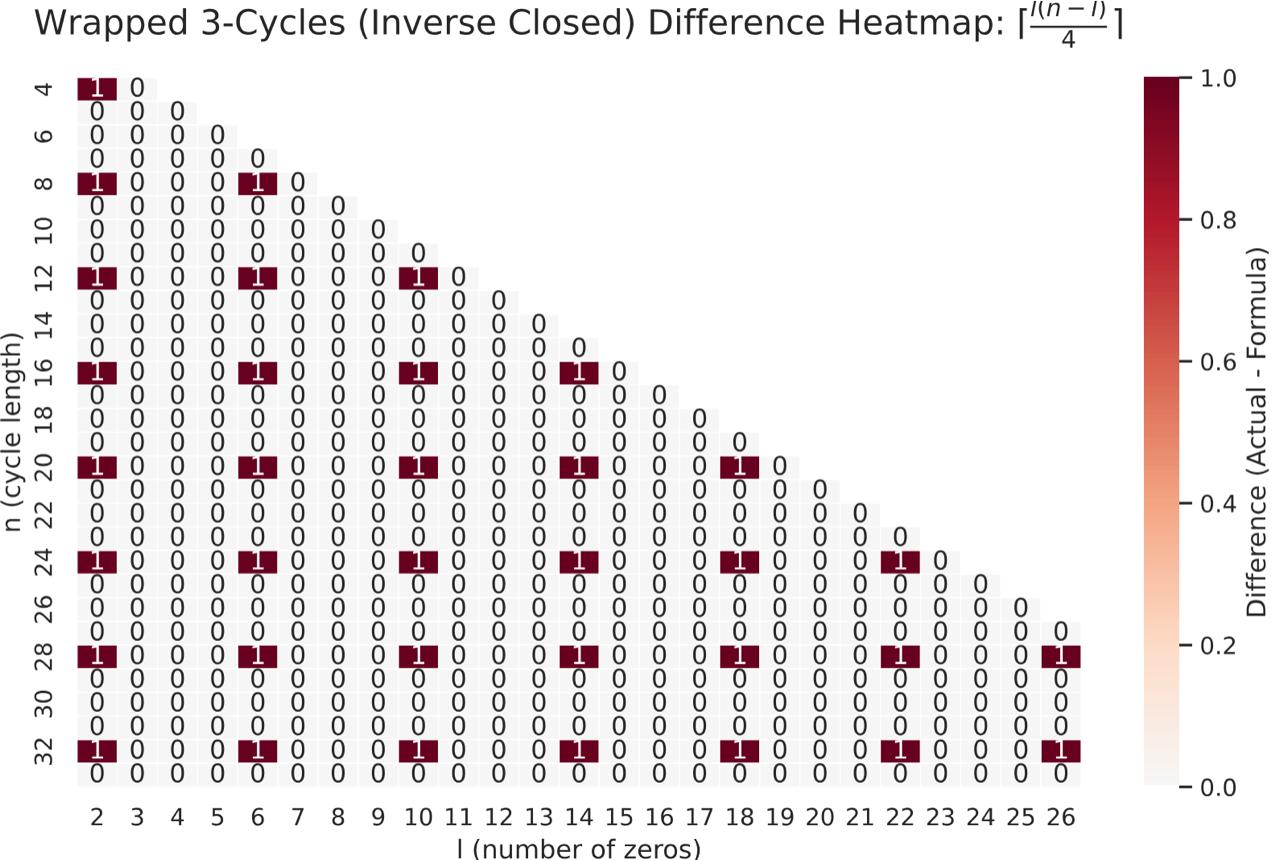}
    \caption{Difference to naive estimation $\lceil l(n-l)/2(k-1) \rceil $ wrapped 3 cycles inverse closed.
    Bi-variable quasi-polynomials pattern.}
    \label{fig:wrapped_3_grassman_inv_closed}
\end{figure}
\begin{figure}[H]
    \centering
    \includegraphics[width=0.75\linewidth]{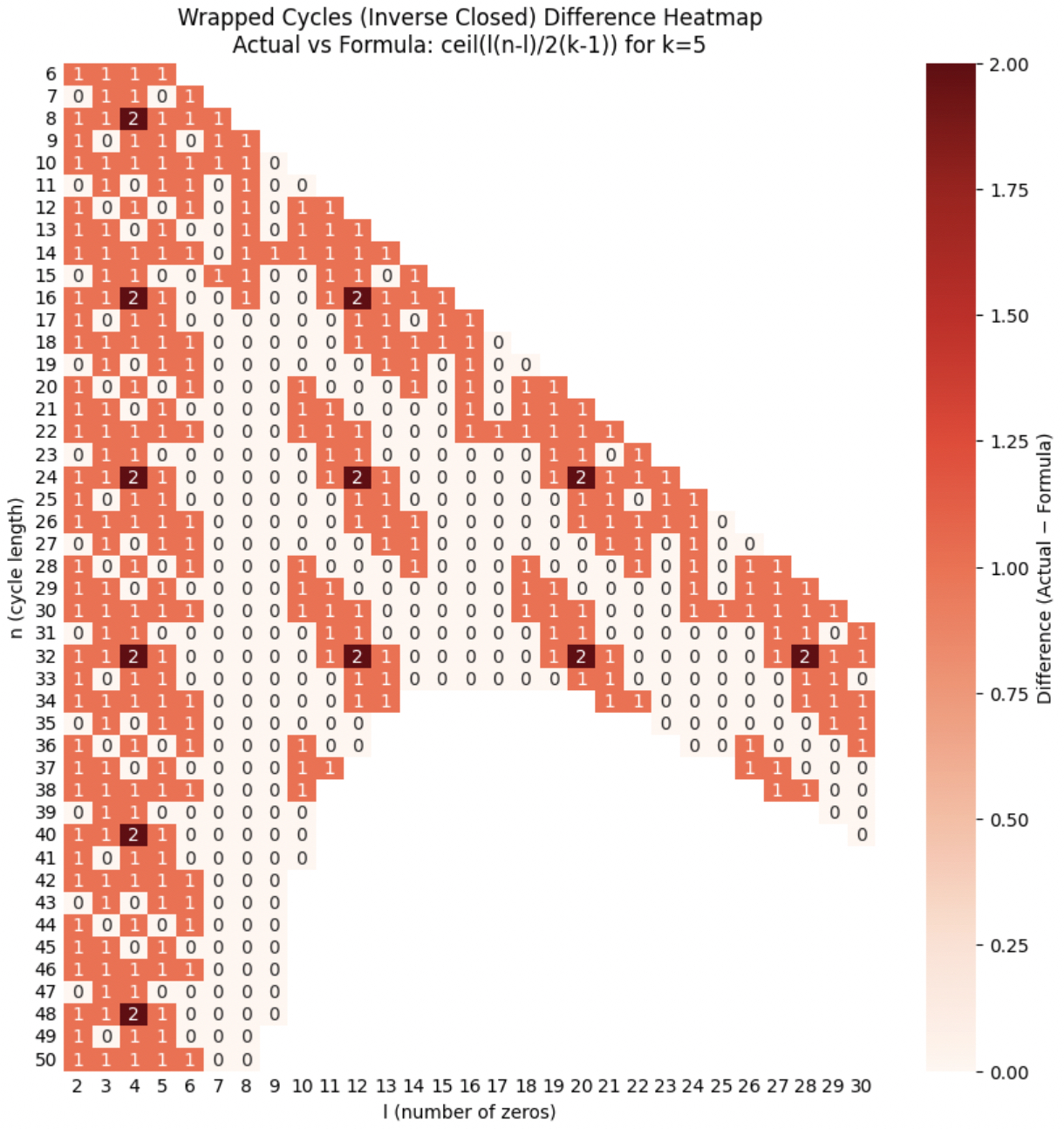}
    \caption{Difference to naive estimation $\lceil l(n-l)/2(k-1) \rceil $ wrapped 5 cycles inverse closed.
    Bi-variable quasi-polynomials pattern.}
    \label{fig:wrapped_5_grassman_inv_closed}
\end{figure}

\clearpage
\subsection{\texorpdfstring{Schreier coset graph: $S_n / \bigl(S_{\lfloor n/2 \rfloor} \times S_{n-\lfloor n/2 \rfloor}\bigr)$ (not inverse-closed)}{Schreier coset graph: S_n/(S_{floor(n/2)} x S_{n-floor(n/2)}) (not inverse-closed)}}


Here we present a conjectural formula for the diameters of Schreier coset graphs of the form
$
S_n \big/ \bigl(S_{\lfloor n/2 \rfloor} \times S_{n-\lfloor n/2 \rfloor}\bigr).$
The graph can be equivalently described as follows: its vertices correspond to binary vectors with components in $\{0,1\}$, containing exactly $\lfloor n/2 \rfloor$ zeros and $n-\lfloor n/2 \rfloor$ ones.
In \texttt{CayleyPy}, we define these graphs by setting the \texttt{central\_state} to be
$
[0]^{n//2} + [1]^{\,n-n//2}.
$
The generating set consists of wrappped consecutive $k$-cycles; we consider the case where the generating set is not inverse-closed.

The leading term of the diameter formulas are  $n^2/(8(k-1))$ as it is expected from the $2(k-1)$-shrinkage principle,
since for the standard $k=2$ (Coxeter or neighbor transposition) case the diameter is exactly 
${\lfloor n/2 \rfloor} (n-\lfloor n/2 \rfloor)$. 

Informally, one may think of this graph as a ``$k$-shrunken'' version of affine Grassmannian
$Gr^{aff}(\lfloor n/2 \rfloor, n)$ over the field with one element. From the point of view of that analogy
diameter corresponds to dimension of the manifold, and Poincare polynomial correspond to growth polynomial of the graph.

We experimentally have found the coefficients and we conjecture that the formula is the following polynomial
\[
H_k(x)=\frac{(x^{4k-4}-1)^2\,(x^{2k-1}+1)}{(x-1)^2(x+1)}
\in \mathbb Z[x].
\]
based on the following proposition.
\begin{prp}[Closed form for \(H_k\)]
\label{prop:Hk-closed}
Let \(k\ge 2\) and \(m:=3(k-1)\). Define integers \(h_i\) for \(0\le i\le 10(k-1)\) by
\[
h_i=
\begin{cases}
1, & i=0,1,\\[0.4ex]
2+\left\lfloor\dfrac{i-2}{2}\right\rfloor, & 2\le i\le 2k-3,\\[1.2ex]
i-(k-2), & 2k-2\le i\le 4k-5,\\[0.8ex]
(m-1)+\bigl((i-(4k-4))\bmod 2\bigr), & 4k-4\le i\le 6k-6,\\[0.8ex]
(m-2)-(i-(6k-5)), & 6k-5\le i\le 8k-10,\\[0.8ex]
(k-1)-\left\lfloor\dfrac{i-(8k-9)}{2}\right\rfloor, & 8k-9\le i\le 10k-14,\\[1.2ex]
1, & i=10k-13,\,10k-12,\\[0.4ex]
0, & i=10k-11,\,10k-10,
\end{cases}
\]
where ``\(\bmod 2\)'' denotes the remainder in \(\{0,1\}\).
Let
\[
H_k(x):=\sum_{i=0}^{10(k-1)} h_i\,x^i\in\mathbb Z[x].
\]
Then
\[
H_k(x)=\frac{(x^{4k-4}-1)^2\,(x^{2k-1}+1)}{(x-1)^2(x+1)}
\in \mathbb Z[x].
\]
\end{prp}

Before a proof we need two lemmas.
\begin{lemma}[Counting even/odd integers in an interval]
\label{lem:count-parity}
Let \(A,B\in\mathbb Z\) with \(A\le B\). Then
\[
\#\{t\in\mathbb Z:\ A\le t\le B,\ t\text{ even}\}
=\left\lfloor\frac{B}{2}\right\rfloor-\left\lfloor\frac{A-1}{2}\right\rfloor,
\]
and
\[
\#\{t\in\mathbb Z:\ A\le t\le B,\ t\text{ odd}\}
=\left\lfloor\frac{B+1}{2}\right\rfloor-\left\lfloor\frac{A}{2}\right\rfloor.
\]
Moreover, for every \(n\in\mathbb Z\),
\[
\left\lfloor\frac{n}{2}\right\rfloor+\left\lfloor\frac{n-1}{2}\right\rfloor=n-1.
\]
\end{lemma}

\begin{proof}
For the even-count, the map \(t\mapsto t/2\) is a bijection between even integers \(t\in[A,B]\) and integers
\(u\in[\lceil A/2\rceil,\ \lfloor B/2\rfloor]\), hence the count equals
\(\lfloor B/2\rfloor-\lceil A/2\rceil+1\).
Using \(\lceil A/2\rceil=\lfloor (A+1)/2\rfloor=\lfloor (A-1)/2\rfloor+1\), we obtain the stated formula.

For the odd-count, odd \(t\in[A,B]\) correspond bijectively to even \(t-1\in[A-1,B-1]\), so the first formula gives
\[
\#\{t\in[A,B]\cap\mathbb Z:\ t\text{ odd}\}
=\left\lfloor\frac{B-1}{2}\right\rfloor-\left\lfloor\frac{A-2}{2}\right\rfloor
=\left\lfloor\frac{B+1}{2}\right\rfloor-\left\lfloor\frac{A}{2}\right\rfloor.
\]

Finally, write \(n=2q\) or \(n=2q+1\). If \(n=2q\), then
\(\lfloor n/2\rfloor+\lfloor(n-1)/2\rfloor=q+(q-1)=2q-1=n-1\).
If \(n=2q+1\), then
\(\lfloor n/2\rfloor+\lfloor(n-1)/2\rfloor=q+q=2q=n-1\).
\end{proof}

\begin{lemma}[A floor identity used in Range II]
\label{lem:range2-floor}
Let \(k\ge 2\) and \(i\in\mathbb Z\). Then
\[
\left\lfloor\frac{i}{2}\right\rfloor+\left\lfloor\frac{i-(2k-1)}{2}\right\rfloor=i-k.
\]
\end{lemma}

\begin{proof}
Write \(i=2q\) or \(i=2q+1\).

If \(i=2q\), then \(i-(2k-1)=2q-2k+1\) is odd, so
\[
\left\lfloor\frac{i}{2}\right\rfloor+\left\lfloor\frac{i-(2k-1)}{2}\right\rfloor
=q+\left\lfloor q-k+\frac12\right\rfloor
=q+(q-k)=2q-k=i-k.
\]

If \(i=2q+1\), then \(i-(2k-1)=2q+2-2k=2(q-k+1)\) is even, so
\[
\left\lfloor\frac{i}{2}\right\rfloor+\left\lfloor\frac{i-(2k-1)}{2}\right\rfloor
=q+(q-k+1)=2q-k+1=i-k.
\]
\end{proof}

\begin{proof}
Define
\[
U(x):=\frac{x^{4k-4}-1}{x-1}=\sum_{a=0}^{4k-5}x^a,
\qquad
V(x):=\frac{x^{2k-1}+1}{x+1}=x^{2k-2}-x^{2k-3}+\cdots-x+1.
\]
Since \(2k-1\) is odd, \(x^{2k-1}+1\) is divisible by \(x+1\), hence \(V(x)\in\mathbb Z[x]\); clearly also \(U(x)\in\mathbb Z[x]\).
Therefore
\[
H_k^\star(x):=\frac{(x^{4k-4}-1)^2\,(x^{2k-1}+1)}{(x-1)^2(x+1)}=U(x)^2V(x)\in\mathbb Z[x].
\]
We prove that \(H_k^\star(x)=H_k(x)\) by comparing coefficients.

Set \(W(x):=U(x)V(x)\). Then
\[
W(x)=\frac{(x^{4k-4}-1)(x^{2k-1}+1)}{(x-1)(x+1)}
=\frac{(x^{4k-4}-1)(x^{2k-1}+1)}{x^2-1}.
\]
Because \(4k-4\) is even,
\[
\frac{x^{4k-4}-1}{x^2-1}=1+x^2+x^4+\cdots+x^{4k-6}=\sum_{j=0}^{2k-3}x^{2j}.
\]
Hence
\[
W(x)=(x^{2k-1}+1)\sum_{j=0}^{2k-3}x^{2j}
=\sum_{j=0}^{2k-3}x^{2j}+\sum_{j=0}^{2k-3}x^{2k-1+2j}.
\]
The first sum has only even exponents, the second only odd exponents; thus the supports are disjoint and every nonzero coefficient of \(W\) equals \(1\).

Let
\[
S:=\{0,2,4,\dots,4k-6\}\ \cup\ \{2k-1,2k+1,\dots,6k-7\}.
\]
Then the coefficient of \(x^t\) in \(W(x)\) is \(1\) if \(t\in S\) and \(0\) otherwise. Since
\[
H_k^\star(x)=U(x)W(x)=\Bigl(\sum_{a=0}^{4k-5}x^a\Bigr)\,W(x),
\]
the coefficient of \(x^i\) in \(H_k^\star(x)\) is
\[
c_i=\#\{t\in S:\ i-(4k-5)\le t\le i\}
=\#\bigl(S\cap[\,i-(4k-5),\,i\,]\bigr).
\]
We show \(c_i=h_i\) for all \(0\le i\le 10(k-1)\), by splitting into the index ranges defining \(h_i\).

\medskip
\noindent\textbf{Range 0: \(i=0,1\).}
Here \(i-(4k-5)\le 0\) and \(i<2\), so \(S\cap[\,i-(4k-5),i\,]=\{0\}\). Thus \(c_i=1=h_i\).

\medskip
\noindent\textbf{Range I: \(2\le i\le 2k-3\).}
Then \(i<2k-1\), so the odd block \(\{2k-1,2k+1,\dots\}\) does not contribute. Also \(i-(4k-5)<0\), hence all even elements of \(S\) up to \(i\) are counted:
\[
c_i=\#\{0,2,4,\dots\le i\}=\left\lfloor\frac{i}{2}\right\rfloor+1
=2+\left\lfloor\frac{i-2}{2}\right\rfloor=h_i.
\]

\medskip
\noindent\textbf{Range II: \(2k-2\le i\le 4k-5\).}
Again \(i-(4k-5)\le 0\), hence \(c_i\) counts all elements of \(S\) in \([0,i]\).
The even contribution equals \(\lfloor i/2\rfloor+1\).
The odd contribution counts the odd numbers \(2k-1,2k+1,\dots\le i\), which is
\(\left\lfloor\frac{i-(2k-1)}{2}\right\rfloor+1\) (this is \(0\) when \(i=2k-2\)).
Thus
\[
c_i=\left\lfloor\frac{i}{2}\right\rfloor+1+\left\lfloor\frac{i-(2k-1)}{2}\right\rfloor+1.
\]
By Lemma~\ref{lem:range2-floor}, \(\lfloor i/2\rfloor+\lfloor(i-(2k-1))/2\rfloor=i-k\), so
\[
c_i=i-k+2=i-(k-2)=h_i.
\]

\medskip
\noindent\textbf{Range III: \(4k-4\le i\le 6k-6\).}
Write \(i=4k-4+s\) with \(0\le s\le 2k-2\). Then the interval becomes
\([\,i-(4k-5),i\,]=[\,1+s,\ 4k-4+s\,]\).
Since \(4k-4+s\le 6k-6<6k-7\), the odd block contributes odds from \(2k-1\) up to \(4k-4+s\), while the even block contributes evens from \(1+s\) up to \(4k-6\):
\[
\#\bigl(\{0,2,\dots,4k-6\}\cap[1+s,\,4k-4+s]\bigr)
=\#\{t\in\mathbb Z:\ 1+s\le t\le 4k-6,\ t\text{ even}\}
=2k-3-\left\lfloor\frac{s}{2}\right\rfloor,
\]
\[
\#\bigl(\{2k-1,2k+1,\dots\}\cap[1+s,\,4k-4+s]\bigr)
=\#\{t\in\mathbb Z:\ 2k-1\le t\le 4k-4+s,\ t\text{ odd}\}
=\left\lfloor\frac{4k-3+s}{2}\right\rfloor-(k-1),
\]
where we used Lemma~\ref{lem:count-parity} in each line.
Adding,
\[
c_i
=\Bigl(2k-3-\left\lfloor\frac{s}{2}\right\rfloor\Bigr)
+\Bigl(\left\lfloor\frac{4k-3+s}{2}\right\rfloor-(k-1)\Bigr)
= k-2+\left\lfloor\frac{4k-3+s}{2}\right\rfloor-\left\lfloor\frac{s}{2}\right\rfloor.
\]
Since \(4k-3\) is odd, writing \(s=2r\) or \(s=2r+1\) gives
\[
\left\lfloor\frac{4k-3+s}{2}\right\rfloor-\left\lfloor\frac{s}{2}\right\rfloor
=
\begin{cases}
2k-2, & s\ \text{even},\\
2k-1, & s\ \text{odd}.
\end{cases}
\]
Hence \(c_i=3k-4\) for even \(s\) and \(c_i=3k-3\) for odd \(s\). Since \(m=3k-3\) and \(s=i-(4k-4)\),
\[
c_i=(m-1)+\bigl(s\bmod 2\bigr)=(m-1)+\bigl((i-(4k-4))\bmod 2\bigr)=h_i.
\]

\medskip
\noindent\textbf{Range IV: \(6k-5\le i\le 8k-10\).}
Write \(i=6k-5+t\) with \(0\le t\le 2k-5\). Then
\([\,i-(4k-5),i\,]=[\,2k+t,\ 6k-5+t\,]\).
Since \(6k-5+t\ge 6k-5>6k-7\), the odd block contributes odd integers from the moving lower bound up to the fixed top \(6k-7\),
and the even block contributes even integers from the same lower bound up to the fixed top \(4k-6\):
\[
\#\{t'\in\mathbb Z:\ 2k+t\le t'\le 6k-7,\ t'\text{ odd}\}
=(3k-3)-\left\lfloor\frac{2k+t}{2}\right\rfloor,
\]
\[
\#\{t'\in\mathbb Z:\ 2k+t\le t'\le 4k-6,\ t'\text{ even}\}
=(2k-3)-\left\lfloor\frac{2k+t-1}{2}\right\rfloor,
\]
again by Lemma~\ref{lem:count-parity}. Adding and using Lemma~\ref{lem:count-parity} with \(n=2k+t\),
\[
c_i=(5k-6)-\left(\left\lfloor\frac{2k+t}{2}\right\rfloor+\left\lfloor\frac{2k+t-1}{2}\right\rfloor\right)
=(5k-6)-(2k+t-1)=3k-5-t.
\]
Since \(m-2=3k-5\) and \(t=i-(6k-5)\), this gives
\(c_i=(m-2)-(i-(6k-5))=h_i\).

\medskip
\noindent\textbf{Range V: \(8k-9\le i\le 10k-14\).}
Write \(i=8k-9+t\) with \(0\le t\le 2k-5\). Then
\([\,i-(4k-5),i\,]=[\,4k-4+t,\ 8k-9+t\,]\).
Because \(4k-4+t>4k-6\), the even block contributes nothing.
The odd block contributes odd integers from the moving lower bound up to the fixed top \(6k-7\), hence
\[
c_i=\#\{t'\in\mathbb Z:\ 4k-4+t\le t'\le 6k-7,\ t'\text{ odd}\}
=(3k-3)-\left\lfloor\frac{4k-4+t}{2}\right\rfloor.
\]
Since \(4k-4\) is even,
\(\left\lfloor\frac{4k-4+t}{2}\right\rfloor=2k-2+\left\lfloor\frac{t}{2}\right\rfloor\),
so
\[
c_i=(3k-3)-\Bigl(2k-2+\left\lfloor\frac{t}{2}\right\rfloor\Bigr)
=(k-1)-\left\lfloor\frac{t}{2}\right\rfloor
=(k-1)-\left\lfloor\frac{i-(8k-9)}{2}\right\rfloor
=h_i.
\]

\medskip
\noindent\textbf{Range VI: \(i=10k-13,\,10k-12\).}
We show \(S\cap[\,i-(4k-5),i\,]=\{6k-7\}\), hence \(c_i=1\).
For \(i=10k-13\), the interval is
\[
[\,i-(4k-5),i\,]=[\,10k-13-(4k-5),\ 10k-13\,]=[\,6k-8,\ 10k-13\,],
\]
which contains \(6k-7\) and is strictly above \(4k-6\), so it meets \(S\) only at the last odd element \(6k-7\).
For \(i=10k-12\), the interval is
\[
[\,i-(4k-5),i\,]=[\,6k-7,\ 10k-12\,],
\]
again containing \(6k-7\) and lying above \(4k-6\). Thus \(c_i=1=h_i\) in both cases.

\medskip
\noindent\textbf{Range VII: \(i=10k-11,\,10k-10\).}
For \(i=10k-11\), the interval is \([\,6k-6,\ 10k-11\,]\); for \(i=10k-10\), it is \([\,6k-5,\ 10k-10\,]\).
In both cases the lower bound is \(>6k-7\), the maximum element of \(S\), so the intersection is empty and \(c_i=0=h_i\).

\medskip
Thus \(c_i=h_i\) for all \(i\), hence \(H_k^\star(x)=H_k(x)\), proving the closed form.
\end{proof}

\begin{cor}[Cyclotomic factorization]
\label{cor:cyclotomic}
For \(k\ge 2\),
\[
H_k(x)=
\prod_{\substack{d\mid(4k-4)\\ d>1}}\Phi_d(x)^2\;
\prod_{\substack{d\mid(2k-1)\\ d>1}}\Phi_{2d}(x),
\]
hence every irreducible factor of \(H_k\) is cyclotomic.
\end{cor}

\begin{proof}
Use \(x^n-1=\prod_{d\mid n}\Phi_d(x)\) and, for odd \(n\), \(x^n+1=\prod_{d\mid n}\Phi_{2d}(x)\).
Also \((x^{4k-4}-1)^2/(x-1)^2\) removes the \(\Phi_1(x)^2\) factor, and \((x^{2k-1}+1)/(x+1)\) removes the \(\Phi_2(x)\) factor.
\end{proof}

\begin{thm}
\textbf{Case $k=2$ from $n=4$}
\[
\begin{aligned}
n &\equiv 0 \pmod{4}: & D_2(n) &= \frac{n^2}{8},\\[4pt]
n &\equiv 1 \pmod{4}: & D_2(n) &= \frac{n^2-1}{8},\\[4pt]
n &\equiv 2 \pmod{4}: & D_2(n) &= \frac{n^2+4}{8},\\[4pt]
n &\equiv 3 \pmod{4}: & D_2(n) &= \frac{n^2-1}{8}.
\end{aligned}
\]
\[
\mathrm{OGF}
=\frac{x^{10}+x^{9}+2x^{8}+3x^{7}+2x^{6}+3x^{5}+2x^{4}+x^{3}+x^{2}}{(1-x^{4})^{3}}.
\]
\end{thm}

\begin{Conj}
\textbf{Case $k=3$ from $n=8$}
\[
\begin{aligned}
n&\equiv 0,6 \pmod{8}: & D_3(n)&=\frac{n^2+2n}{16}=\frac{n(n+2)}{16},\\[4pt]
n&\equiv 1,5 \pmod{8}: & D_3(n)&=\frac{n^2+2n-3}{16},\\[4pt]
n&\equiv 2,4 \pmod{8}: & D_3(n)&=\frac{n^2+2n+8}{16},\\[4pt]
n&\equiv 3,7 \pmod{8}: & D_3(n)&=\frac{n^2+2n+1}{16}=\frac{(n+1)^2}{16}.
\end{aligned}
\]

\[
\mathrm{OGF}
=\frac{H(x)}{(1-x^{8})^{3}},
\]
where
\[
\begin{aligned}
H(x)=\;&x^{20}+x^{19}+2x^{18}+2x^{17}+3x^{16}+4x^{15}+5x^{14}+6x^{13} \\
&+5x^{12}+6x^{11}+5x^{10}+6x^{9}+5x^{8}+4x^{7}+3x^{6}+2x^{5}+2x^{4}+x^{3}+x^{2}.
\end{aligned}
\]

\textbf{Case $k=4$ from $n=12$}
\[
\begin{aligned}
n&\equiv 0,8 \pmod{12}: & D_4(n)&=\frac{n(n+4)}{24},\\[4pt]
n&\equiv 1,7 \pmod{12}: & D_4(n)&=\frac{n^2+4n-5}{24},\\[4pt]
n&\equiv 2,6 \pmod{12}: & D_4(n)&=\frac{n^2+4n+12}{24},\\[4pt]
n&\equiv 3,9,11 \pmod{12}: & D_4(n)&=\frac{n^2+4n+3}{24},\\[4pt]
n&\equiv 4 \pmod{12}: & D_4(n)&=\frac{n^2+4n+16}{24},\\[4pt]
n&\equiv 5 \pmod{12}: & D_4(n)&=\frac{n^2+4n+3}{24},\\[4pt]
n&\equiv 10 \pmod{12}: & D_4(n)&=\frac{n^2+4n+4}{24}.
\end{aligned}
\]

\[
\mathrm{OGF}
=\frac{H(x)}{(1-x^{12})^{3}},
\]
where
\[
\begin{aligned}
H(x)=\;&x^{30}+x^{29}+2x^{28}+2x^{27}+3x^{26}+3x^{25}+4x^{24}+5x^{23} \\
&+6x^{22}+7x^{21}+8x^{20}+9x^{19}+8x^{18}+9x^{17}+8x^{16}+9x^{15} \\
&+8x^{14}+9x^{13}+8x^{12}+7x^{11}+6x^{10}+5x^{9}+4x^{8}+3x^{7} \\
&+3x^{6}+2x^{5}+2x^{4}+x^{3}+x^{2}.
\end{aligned}
\]
\end{Conj}

\begin{figure}[H]
    \centering
    \begin{minipage}{0.32\linewidth}
        \centering
        \includegraphics[width=\linewidth]{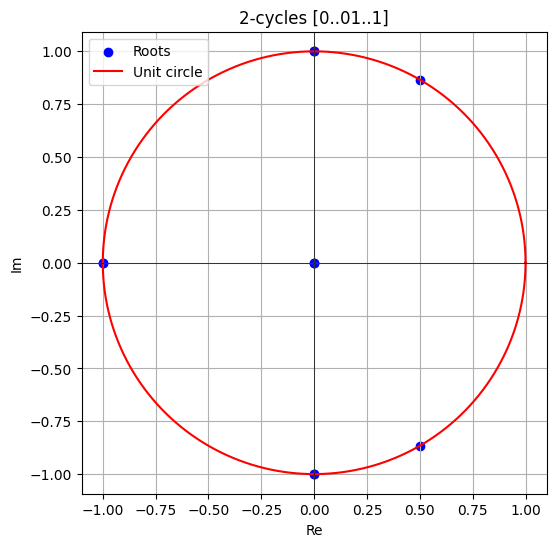}
        wrapped 2-cycles [0..01..1]
    \end{minipage}
    \hfill
    \begin{minipage}{0.32\linewidth}
        \centering
        \includegraphics[width=\linewidth]{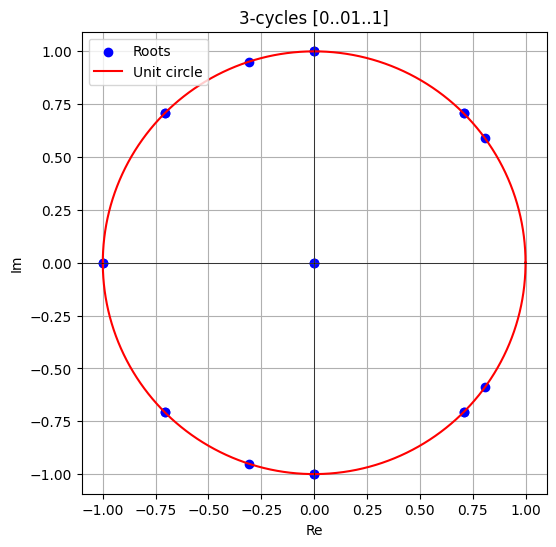}
        wrapped 3-cycles [0..01..1]
    \end{minipage}
    \hfill
    \begin{minipage}{0.32\linewidth}
        \centering
        \includegraphics[width=\linewidth]{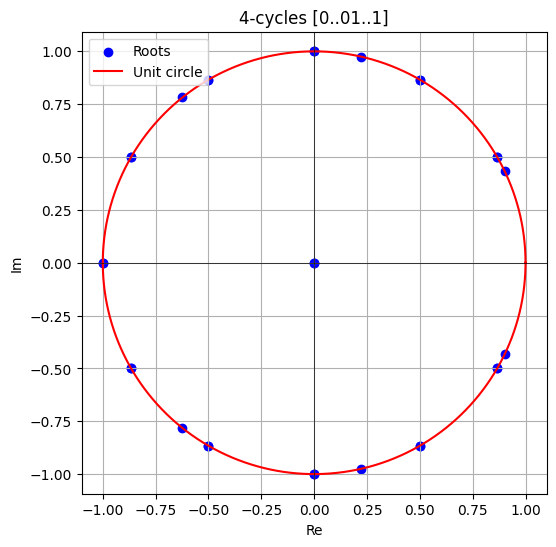}
        wrapped 4-cycles [0..01..1]
    \end{minipage}
    \caption{Zeros of $H$-polynomials for wrapped $k$-cycle Schreier graphs $S_n / \bigl(S_{\lfloor n/2 \rfloor} \times S_{n-\lfloor n/2 \rfloor}\bigr)$ (not inverse-closed)  for $k=2,3,4$.}
\end{figure}




\begin{table}[H]
\centering
\small
\setlength{\tabcolsep}{4pt}
\begin{tabular}{r|rrr|rrr|rrr|rrr|rrr}
\toprule
$n$ &
\multicolumn{3}{c|}{k=2} &
\multicolumn{3}{c|}{k=3} &
\multicolumn{3}{c|}{k=4} &
\multicolumn{3}{c|}{k=5} &
\multicolumn{3}{c}{k=6} \\
\cmidrule(lr){2-4}
\cmidrule(lr){5-7}
\cmidrule(lr){8-10}
\cmidrule(lr){11-13}
\cmidrule(lr){14-16}
 & $d$ & $\Delta d$ & $\Delta^2 d$
 & $d$ & $\Delta d$ & $\Delta^2 d$
 & $d$ & $\Delta d$ & $\Delta^2 d$
 & $d$ & $\Delta d$ & $\Delta^2 d$
 & $d$ & $\Delta d$ & $\Delta^2 d$ \\
\midrule
3 & 1 & 1 & 0 &   &   &   &   &   &   &   &   &   &   &   &   \\
4 & 2 & 1 & 1 & 2 & 0 & 1 &   &   &   &   &   &   &   &   &   \\
5 & 3 & 2 & -1 & 2 & 1 & 0 & 3 & 0 & 0 &   &   &   &   &   &   \\
6 & 5 & 1 & 1 & 3 & 1 & 0 & 3 & 0 & 1 & 4 & 0 & 0 &   &   &   \\
7 & 6 & 2 & 0 & 4 & 1 & 0 & 3 & 1 & 0 & 4 & 0 & 0 & 5 & 0 & 0 \\
8 & 8 & 2 & 1 & 5 & 1 & 1 & 4 & 1 & 0 & 4 & 0 & 1 & 5 & 0 & 0 \\
9 &10 & 3 & -1 & 6 & 2 & -1 & 5 & 1 & 0 & 4 & 1 & 0 & 5 & 0 & 0 \\
10 &13 & 2 & 1 & 8 & 1 & 1 & 6 & 1 & 0 & 5 & 1 & 0 & 5 & 0 & 1 \\
11 &15 & 3 & 0 & 9 & 2 & -1 & 7 & 1 & 0 & 6 & 1 & 0 & 5 & 1 & 0 \\
12 &18 & 3 & 1 &11 & 1 & 1 & 8 & 1 & 1 & 7 & 1 & 0 & 6 & 1 & 0 \\
13 &21 & 4 & -1 &12 & 2 & 0 & 9 & 2 & -1 & 8 & 1 & 0 & 7 & 1 & 0 \\
14 &25 & 3 & 1 &14 & 2 & 0 &11 & 1 & 1 & 9 & 1 & 0 & 8 & 1 & 0 \\
15 &28 & 4 & 0 &16 & 2 & 0 &12 & 2 & -1 &10 & 1 & 0 & 9 & 1 & 0 \\
16 &32 & 4 & 1 &18 & 2 & 1 &14 & 1 & 1 &11 & 1 & 1 &10 & 1 & 0 \\
17 &36 & 5 & -1 &20 & 3 & -1 &15 & 2 & -1 &12 & 2 & -1 &11 & 1 & 0 \\
18 &41 & 4 & 1 &23 & 2 & 1 &17 & 1 & 1 &14 & 1 & 1 &12 & 1 & 0 \\
19 &45 & 5 & 0 &25 & 3 & -1 &18 & 2 & 0 &15 & 2 & -1 &13 & 1 & 0 \\
20 &50 & 5 & 1 &28 & 2 & 1 &20 & 2 & 0 &17 & 1 & 1 &14 & 1 & 1 \\
21 &55 & 6 & -1 &30 & 3 & 0 &22 & 2 & 0 &18 & 2 & -1 &15 & 2 & -1 \\
22 &61 & 5 & 1 &33 & 3 & 0 &24 & 2 & 0 &20 & 1 & 1 &17 & 1 & 1 \\
23 &66 & 6 & 0 &36 & 3 & 0 &26 & 2 & 0 &21 & 2 & -1 &18 & 2 & -1 \\
24 &72 & 6 & 1 &39 & 3 & 1 &28 & 2 & 1 &23 & 1 & 1 &20 & 1 & 1 \\
25 &78 & 7 & -1 &42 & 4 & -1 &30 & 3 & -1 &24 & 2 & 0 &21 & 2 & -1 \\
26 &85 & 6 & 1 &46 & 3 & 1 &33 & 2 & 1 &26 & 2 & 0 &23 & 1 & 1 \\
27 &91 & 7 & 0 &49 & 4 & -1 &35 & 3 & -1 &28 & 2 & 0 &24 & 2 & -1 \\
28 &98 & 7 & 1 &53 & 3 & 1 &38 & 2 & 1 &30 & 2 & 0 &26 & 1 & 1 \\
29 &105 & 8 & -1 &56 & 4 & 0 &40 & 3 & -1 &32 & 2 & 0 &27 & 2 & -1 \\
30 &113 & 7 & 1 &60 & 4 & 0 &43 & 2 & 1 &34 & 2 & 0 &29 & 1 & 1 \\
31 &120 & 8 & 0 &64 & 4 & 0 &45 & 3 & 0 &36 & 2 & 0 &30 & 2 & 0 \\
32 &128 & 8 & 1 &68 & 4 & 1 &48 & 3 & 0 &38 & 2 & 1 &32 & 2 & 0 \\
33 &136 & 9 & -1 &72 & 5 & -1 &51 & 3 & 0 &40 & 3 & -1 &34 & 2 & 0 \\
34 &145 & 8 &   &77 & 4 &   &54 & 3 &   &43 & 2 &   &36 & 2 &   \\
35 &153 &   &   &81 &   &   &57 &   &   &45 &   &   &38 &   &   \\
\bottomrule
\end{tabular}
\caption{Diameters $D_k(n)$ and their increments and second increments for wrapped $k$-cycle coset $S_n / \bigl(S_{\lfloor n/2 \rfloor} \times S_{n-\lfloor n/2 \rfloor}\bigr)$ (not inverse-closed). In that example second increments allow to see periodic structure, i.e. fit quasi-polynomials.}
\label{tab:wrapped_cycle_growth_k=2..6}
\end{table}

\clearpage
\subsection{\texorpdfstring{Schreier coset graph: $S_n / \bigl(S_{\lfloor n/2 \rfloor} \times S_{n-\lfloor n/2 \rfloor}\bigr)$ (inverse-closed)}{Schreier coset graph: S_n/(S_{floor(n/2)} x S_{n-floor(n/2)}) (inverse-closed)}}


The setup is almost identical to the previous subsection, but now consider inverse closed generators.

\begin{Conj}

\textbf{Case $k=3$ from $n=3$}
\[
\begin{aligned}
n&\equiv 0 \pmod{8}: & D_3(n)&=\frac{n^2}{16},\\[4pt]
n&\equiv 1,7 \pmod{8}: & D_3(n)&=\frac{n^2-1}{16},\\[4pt]
n&\equiv 2,6 \pmod{8}: & D_3(n)&=\frac{n^2+12}{16},\\[4pt]
n&\equiv 3,5 \pmod{8}: & D_3(n)&=\frac{n^2+7}{16},\\[4pt]
n&\equiv 4 \pmod{8}: & D_3(n)&=\frac{n^2+16}{16}.
\end{aligned}
\]

\[
\mathrm{OGF}
=\frac{H(x)}{(1-x^{8})^{3}},
\]
where
\[
\begin{aligned}
H(x)=\;&x^{22}+x^{21}+2x^{20}+2x^{19}+3x^{18}+3x^{17}+4x^{16}+5x^{15} \\
&+4x^{14}+5x^{13}+4x^{12}+5x^{11}+4x^{10}+5x^{9}+4x^{8}+3x^{7} \\
&+3x^{6}+2x^{5}+2x^{4}+x^{3}+x^{2}.
\end{aligned}
\]

\textbf{Case $k=4$ from $n=12$}
\[
\begin{aligned}
n&\equiv 0 \pmod{12}: & D_6(n)&=\frac{n^2-2n+48}{24}=\frac{(n-1)^2+47}{24},\\[4pt]
n&\equiv 1,11 \pmod{12}: & D_6(n)&=\frac{n^2-1}{24},\\[4pt]
n&\equiv 2,10 \pmod{12}: & D_6(n)&=\frac{n^2+20}{24},\\[4pt]
n&\equiv 3,9 \pmod{12}: & D_6(n)&=\frac{n^2+15}{24},\\[4pt]
n&\equiv 4,8 \pmod{12}: & D_6(n)&=\frac{n^2+32}{24},\\[4pt]
n&\equiv 5,7 \pmod{12}: & D_6(n)&=\frac{n^2+23}{24},\\[4pt]
n&\equiv 6 \pmod{12}: & D_6(n)&=\frac{n^2+36}{24}.
\end{aligned}
\]

\[
\mathrm{OGF}
=\frac{H(x)}{(1-x^{12})^{3}},
\]
where
\[
\begin{aligned}
H(x)=\;&x^{34}+x^{33}+2x^{32}+2x^{31}+3x^{30}+3x^{29}+4x^{28}+4x^{27} \\
&+5x^{26}+5x^{25}+9x^{24}+7x^{23}+6x^{22}+7x^{21}+6x^{20}+7x^{19} \\
&+6x^{18}+7x^{17}+6x^{16}+7x^{15}+6x^{14}+7x^{13}+x^{12}+5x^{11} \\
&+5x^{10}+4x^{9}+4x^{8}+3x^{7}+3x^{6}+2x^{5}+2x^{4}+x^{3}+x^{2}+2.
\end{aligned}
\]
    
\end{Conj}

\begin{figure}[H]
    \centering
    \begin{minipage}{0.48\linewidth}
        \centering
        \includegraphics[width=\linewidth]{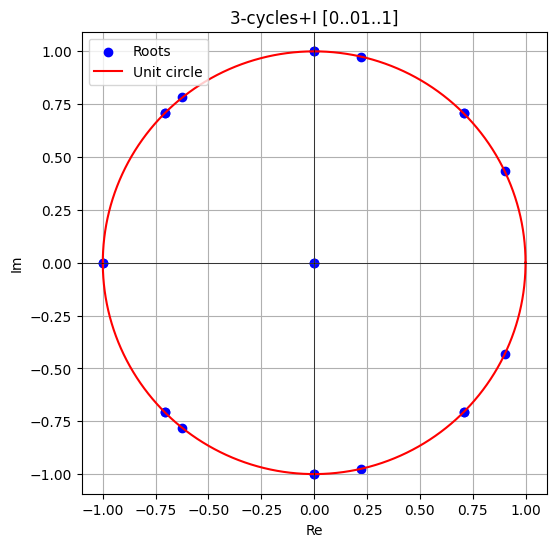}
        \caption{wrapped 3-cycles inverse closed [0..01..1]}
        \label{fig:wrapped3inv}
    \end{minipage}
    \hfill
    \begin{minipage}{0.48\linewidth}
        \centering
        \includegraphics[width=\linewidth]{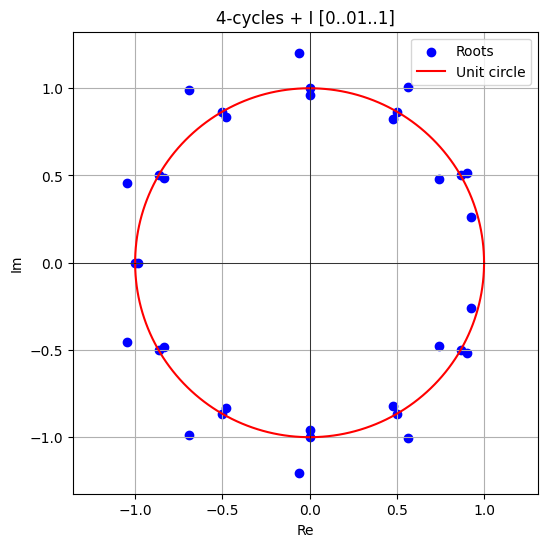}
        \caption{wrapped 4-cycles inverse closed [0..01..1]}
        \label{fig:wrapped4inv}
    \end{minipage}
\end{figure}



\begin{table}[H]
\centering
\small
\setlength{\tabcolsep}{4pt}
\begin{tabular}{r|rrr|rrr|rrr|rrr|rrr}
\toprule
$n$ &
\multicolumn{3}{c|}{k=2 inv} &
\multicolumn{3}{c|}{k=3 inv} &
\multicolumn{3}{c|}{k=4 inv} &
\multicolumn{3}{c|}{k=5 inv} &
\multicolumn{3}{c}{k=6 inv} \\
\cmidrule(lr){2-4}
\cmidrule(lr){5-7}
\cmidrule(lr){8-10}
\cmidrule(lr){11-13}
\cmidrule(lr){14-16}
 & $d$ & $\Delta d$ & $\Delta^2 d$
 & $d$ & $\Delta d$ & $\Delta^2 d$
 & $d$ & $\Delta d$ & $\Delta^2 d$
 & $d$ & $\Delta d$ & $\Delta^2 d$
 & $d$ & $\Delta d$ & $\Delta^2 d$ \\
\midrule
3 & 1 & 1 & 0 &   &   &   &   &   &   &   &   &   &   &   &   \\
4 & 2 & 1 & 1 & 2 & 0 & 1 &   &   &   &   &   &   &   &   &   \\
5 & 3 & 2 & -1 & 2 & 1 & -1 & 2 & 1 & -1 &   &   &   &   &   &   \\
6 & 5 & 1 & 1 & 3 & 0 & 1 & 3 & 0 & 1 & 3 & 0 & 1 &   &   &   \\
7 & 6 & 2 & 0 & 3 & 1 & 0 & 3 & 1 & -1 & 3 & 1 & -1 & 3 & 1 & -1 \\
8 & 8 & 2 & 1 & 4 & 1 & 1 & 4 & 0 & 1 & 4 & 0 & 1 & 4 & 0 & 1 \\
9 &10 & 3 & -1 & 5 & 2 & -1 & 4 & 1 & 0 & 4 & 1 & -1 & 4 & 1 & -1 \\
10 &13 & 2 & 1 & 7 & 1 & 1 & 5 & 1 & 0 & 5 & 0 & 1 & 5 & 0 & 1 \\
11 &15 & 3 & 0 & 8 & 2 & -1 & 6 & 1 & -1 & 5 & 1 & -1 & 5 & 1 & -1 \\
12 &18 & 3 & 1 &10 & 1 & 1 & 7 & 0 & 2 & 6 & 0 & 1 & 6 & 0 & 1 \\
13 &21 & 4 & -1 &11 & 2 & -1 & 7 & 2 & -1 & 6 & 1 & 0 & 6 & 1 & -1 \\
14 &25 & 3 & 1 &13 & 1 & 1 & 9 & 1 & 1 & 7 & 1 & 0 & 7 & 0 & 1 \\
15 &28 & 4 & 0 &14 & 2 & 0 &10 & 2 & -1 & 8 & 1 & -1 & 7 & 1 & -1 \\
16 &32 & 4 & 1 &16 & 2 & 1 &12 & 1 & 1 & 9 & 0 & 2 & 8 & 0 & 1 \\
17 &36 & 5 & -1 &18 & 3 & -1 &13 & 2 & -1 & 9 & 2 & -1 & 8 & 1 & 0 \\
18 &41 & 4 & 1 &21 & 2 & 1 &15 & 1 & 1 &11 & 1 & 1 & 9 & 1 & 0 \\
19 &45 & 5 & 0 &23 & 3 & -1 &16 & 2 & -1 &12 & 2 & -1 &10 & 1 & -1 \\
20 &50 & 5 & 1 &26 & 2 & 1 &18 & 1 & 1 &14 & 1 & 1 &11 & 0 & 2 \\
21 &55 & 6 & -1 &28 & 3 & -1 &19 & 2 & -1 &15 & 2 & -1 &11 & 2 & -1 \\
22 &61 & 5 & 1 &31 & 2 & 1 &21 & 1 & 1 &17 & 1 & 1 &13 & 1 & 1 \\
23 &66 & 6 & 0 &33 & 3 & 0 &22 & 2 & 0 &18 & 2 & -1 &14 & 2 & -1 \\
24 &72 & 6 & 1 &36 & 3 & 1 &24 & 2 & 1 &20 & 1 & 1 &16 & 1 & 1 \\
25 &78 & 7 & -1 &39 & 4 & -1 &26 & 3 & -1 &21 & 2 & -1 &17 & 2 & -1 \\
26 &85 & 6 & 1 &43 & 3 & 1 &29 & 2 & 1 &23 & 1 & 1 &19 & 1 & 1 \\
27 &91 & 7 & 0 &46 & 4 & -1 &31 & 3 & -1 &24 & 2 & -1 &20 & 2 & -1 \\
28 &98 & 7 & 1 &50 & 3 & 1 &34 & 2 & 1 &26 & 1 & 1 &22 & 1 & 1 \\
29 &105 & 8 & -1 &53 & 4 & -1 &36 & 3 & -1 &27 & 2 & -1 &23 & 2 & -1 \\
30 &113 & 7 & 1 &57 & 3 & 1 &39 & 2 & 1 &29 & 1 & 1 &25 & 1 & 1 \\
31 &120 & 8 & 0 &60 & 4 & 0 &41 & 3 & -1 &30 & 2 & 0 &26 & 2 & -1 \\
32 &128 & 8 & 1 &64 & 4 & 1 &44 & 2 & 1 &32 & 2 & 1 &28 & 1 & 1 \\
33 &136 & 9 & -1 &68 & 5 & -1 &46 & 3 & -1 &34 & 3 & -1 &29 & 2 & -1 \\
34 &145 & 8 &   &73 & 4 &   &49 & 2 &   &37 & 2 &   &31 & 1 &   \\
35 &153 &   &   &77 &   &   &51 &   &   &39 &   &   &32 &   &   \\
\bottomrule
\end{tabular}
\caption{Diameters $D_k(n)$ and their increments, wrapped k-cycle, inverse closed, $S_n / \bigl(S_{\lfloor n/2 \rfloor} \times S_{n-\lfloor n/2 \rfloor}\bigr)$ }
\label{tab:wrapped_cycle_growth_k=2..6_inv}
\end{table}

\clearpage
\subsection{\texorpdfstring{Some eccentricities for Schreier coset graph: $S_n / \bigl(S_{\lfloor n/2 \rfloor} \times S_{n-\lfloor n/2 \rfloor}\bigr)$}{Some eccentricities for Schreier coset graph: S_n/(S_{floor(n/2)} x S_{n-floor(n/2)})}}


Here we consider eccentricities for elements  $[0,1]^{\lfloor n/2 \rfloor} + [0]^{\,n - 2\lfloor n/2 \rfloor}$)
for the Schreier coset graph: $S_n / \bigl(S_{\lfloor n/2 \rfloor} \times S_{n-\lfloor n/2 \rfloor}\bigr)$,
wrapped $k$ cycles, both inverse-closed and not cases. And propose quasi-polynomial expressions for them, and study 
corresponding $H$-polynomials.

\begin{Conj}
\textbf{Case $k = 2$ from $n = 5$}
\[
p_r(n) = \frac{n^2}{16} + a_r n + b_r.
\]

\begin{table}[H] 
\centering
\begin{tabular}{r|c|c|l}
\textbf{residue ($r$)} & \textbf{$a_r$} & \textbf{$b_r$} & \textbf{polynomial $p_r(n)$} \\ \hline
$(0)$ & $0$ & $0$ &
$\displaystyle \frac{n^{2}}{16}$ \\[6pt]
$(1)$ & $\tfrac{1}{8}$ & $-\tfrac{3}{16}$ &
$\displaystyle \frac{n^{2}}{16} + \frac{n}{8} - \frac{3}{16}$ \\[6pt]
$(2)$ & $0$ & $\tfrac{12}{16}=\tfrac{3}{4}$ &
$\displaystyle \frac{n^{2}}{16} + \tfrac{3}{4}$ \\[6pt]
$(3)$ & $\tfrac{1}{8}$ & $\tfrac{1}{16}$ &
$\displaystyle \frac{n^{2}}{16} + \frac{n}{8} + \frac{1}{16}$ \\
\end{tabular}
\end{table}
\[
\mathrm{OGF}
=\frac{H(x)}{(1-x^{4})^{3}},
\qquad
H(x)=x^{10}+x^{8}+x^{7}+2x^{5}+x^{4}+x^{3}+x^{2}.
\]

\noindent\textbf{Inverse close or not - coincide, case k = 3 from n = 9}

\[
p_r(n) = \frac{n^2}{32} + a_r n + b_r.
\]

\begin{table}[H] 
\centering
\begin{tabular}{r|c|c|l}
\textbf{residue $r$} & $\mathbf{a_r}$ & $\mathbf{b_r}$ & \textbf{polynomial $p_r(n)$} \\ \hline
0 & $\tfrac{1}{8}$ & $0$ &
$\displaystyle \frac{n^2}{32} + \frac{n}{8}$ \\[6pt]
1 & $\tfrac{3}{16}$ & $-\tfrac{7}{32}$ &
$\displaystyle \frac{n^2}{32} + \frac{3n}{16} - \frac{7}{32}$ \\[6pt]
2 & $\tfrac{1}{8}$ & $\tfrac{5}{8}$ &
$\displaystyle \frac{n^2}{32} + \frac{n}{8} + \frac{5}{8}$ \\[6pt]
3 & $\tfrac{3}{16}$ & $\tfrac{5}{32}$ &
$\displaystyle \frac{n^2}{32} + \frac{3n}{16} + \frac{5}{32}$ \\[6pt]
4 & $\tfrac{1}{8}$ & $0$ &
$\displaystyle \frac{n^2}{32} + \frac{n}{8}$ \\[6pt]
5 & $\tfrac{3}{16}$ & $\tfrac{9}{32}$ &
$\displaystyle \frac{n^2}{32} + \frac{3n}{16} + \frac{9}{32}$ \\[6pt]
6 & $\tfrac{1}{8}$ & $\tfrac{1}{8}$ &
$\displaystyle \frac{n^2}{32} + \frac{n}{8} + \frac{1}{8}$ \\[6pt]
7 & $\tfrac{3}{16}$ & $\tfrac{5}{32}$ &
$\displaystyle \frac{n^2}{32} + \frac{3n}{16} + \frac{5}{32}$ \\
\end{tabular}
\end{table}

\[
\mathrm{OGF}
=\frac{H(x)}{(1-x^{8})^{3}},
\]
where
\[
\begin{aligned}
H(x)=\;&x^{18}+x^{16}+x^{15}+2x^{14}+2x^{13}+3x^{12}+3x^{11}+2x^{10} \\
&+4x^{9}+3x^{8}+3x^{7}+2x^{6}+2x^{5}+x^{4}+x^{3}+x^{2}.
\end{aligned}
\]

\noindent\textbf{Not inverse closed, case $k = 4$ from $n = 7$}
\[
p_r(n) = \frac{n^2}{48} + a_r n + b_r.
\]

\begin{table}[H] 
\centering
\begin{tabular}{r|c|c|l}
\textbf{residue ($r$)} & $\mathbf{a_r}$ & $\mathbf{b_r}$ & \textbf{polynomial $p_r(n)$} \\ \hline
0 & $\tfrac{1}{6}$ & $0$ &
$\displaystyle \frac{n^2}{48} + \frac{n}{6}$ \\[6pt]
1 & $\tfrac{5}{24}$ & $-\tfrac{11}{48}$ &
$\displaystyle \frac{n^2}{48} + \frac{5n}{24} - \frac{11}{48}$ \\[6pt]
2 & $\tfrac{1}{6}$ & $\tfrac{7}{12}$ &
$\displaystyle \frac{n^2}{48} + \frac{n}{6} + \frac{7}{12}$ \\[6pt]
3 & $\tfrac{5}{24}$ & $\tfrac{3}{16}$ &
$\displaystyle \frac{n^2}{48} + \frac{5n}{24} + \frac{3}{16}$ \\[6pt]
4 & $\tfrac{1}{6}$ & $0$ &
$\displaystyle \frac{n^2}{48} + \frac{n}{6}$ \\[6pt]
5 & $\tfrac{5}{24}$ & $\tfrac{7}{16}$ &
$\displaystyle \frac{n^2}{48} + \frac{5n}{24} + \frac{7}{16}$ \\[6pt]
6 & $\tfrac{1}{6}$ & $\tfrac{1}{4}$ &
$\displaystyle \frac{n^2}{48} + \frac{n}{6} + \frac{1}{4}$ \\[6pt]
7 & $\tfrac{5}{24}$ & $\tfrac{25}{48}$ &
$\displaystyle \frac{n^2}{48} + \frac{5n}{24} + \frac{25}{48}$ \\[6pt]
8 & $\tfrac{1}{6}$ & $\tfrac{1}{3}$ &
$\displaystyle \frac{n^2}{48} + \frac{n}{6} + \frac{1}{3}$ \\[6pt]
9 & $\tfrac{5}{24}$ & $\tfrac{7}{16}$ &
$\displaystyle \frac{n^2}{48} + \frac{5n}{24} + \frac{7}{16}$ \\[6pt]
10 & $\tfrac{1}{6}$ & $\tfrac{1}{4}$ &
$\displaystyle \frac{n^2}{48} + \frac{n}{6} + \frac{1}{4}$ \\[6pt]
11 & $\tfrac{5}{24}$ & $\tfrac{3}{16}$ &
$\displaystyle \frac{n^2}{48} + \frac{5n}{24} + \frac{3}{16}$ \\
\end{tabular}
\end{table}

\[
\mathrm{OGF}
=\frac{H(x)}{(1-x^{12})^{3}},
\]

where
\[
\begin{aligned}
H(x)=\;&x^{26}+x^{24}+x^{23}+2x^{22}+2x^{21}+3x^{20}+3x^{19}+4x^{18}+4x^{17} \\
&+5x^{16}+5x^{15}+4x^{14}+6x^{13}+5x^{12}+5x^{11}+4x^{10}+4x^{9} \\
&+3x^{8}+3x^{7}+2x^{6}+2x^{5}+x^{4}+x^{3}+x^{2}.
\end{aligned}
\]

\noindent\textbf{Case $k = 4$ from $n = 7$}

\[
p_r(n) = \frac{n^2}{48} + a_r n + b_r.
\]

\begin{table}[H] 
\centering
\begin{tabular}{r|c|c|l}
\textbf{residue $r$} & $\mathbf{a_r}$ & $\mathbf{b_r}$ & \textbf{polynomial $p_r(n)$} \\ \hline
0 & $0$ & $2$ &
$\displaystyle \frac{n^2}{48} + 2$ \\[6pt]
1 & $\tfrac{1}{24}$ & $\tfrac{31}{16}$ &
$\displaystyle \frac{n^2}{48} + \frac{n}{24} + \frac{31}{16}$ \\[6pt]
2 & $\tfrac{1}{12}$ & $\tfrac{3}{4}$ &
$\displaystyle \frac{n^2}{48} + \frac{n}{12} + \frac{3}{4}$ \\[6pt]
3 & $\tfrac{1}{8}$ & $\tfrac{7}{16}$ &
$\displaystyle \frac{n^2}{48} + \frac{n}{8} + \frac{7}{16}$ \\[6pt]
4 & $\tfrac{1}{12}$ & $\tfrac{1}{3}$ &
$\displaystyle \frac{n^2}{48} + \frac{n}{12} + \frac{1}{3}$ \\[6pt]
5 & $\tfrac{1}{8}$ & $\tfrac{41}{48}$ &
$\displaystyle \frac{n^2}{48} + \frac{n}{8} + \frac{41}{48}$ \\[6pt]
6 & $\tfrac{1}{12}$ & $\tfrac{3}{4}$ &
$\displaystyle \frac{n^2}{48} + \frac{n}{12} + \frac{3}{4}$ \\[6pt]
7 & $\tfrac{1}{8}$ & $\tfrac{5}{48}$ &
$\displaystyle \frac{n^2}{48} + \frac{n}{8} + \frac{5}{48}$ \\[6pt]
8 & $\tfrac{1}{6}$ & $-\tfrac{5}{3}$ &
$\displaystyle \frac{n^2}{48} + \frac{n}{6} - \frac{5}{3}$ \\[6pt]
9 & $\tfrac{5}{24}$ & $-\tfrac{25}{16}$ &
$\displaystyle \frac{n^2}{48} + \frac{5n}{24} - \frac{25}{16}$ \\[6pt]
10 & $\tfrac{1}{6}$ & $-\tfrac{7}{4}$ &
$\displaystyle \frac{n^2}{48} + \frac{n}{6} - \frac{7}{4}$ \\[6pt]
11 & $\tfrac{5}{24}$ & $-\tfrac{29}{16}$ &
$\displaystyle \frac{n^2}{48} + \frac{5n}{24} - \frac{29}{16}$ \\
\end{tabular}
\end{table}

\par 
\[
\mathrm{OGF}
=\frac{H(x)}{(1-x^{12})^{3}},
\]

where
\[
\begin{aligned}
H(x)=\;&-2x^{35}-2x^{34}-2x^{33}-2x^{32}+x^{30}+x^{29}+x^{28}+x^{27}+2x^{26}+4x^{25} \\
&+5x^{24}+5x^{23}+6x^{22}+6x^{21}+7x^{20}+4x^{19}+3x^{18}+3x^{17}+4x^{16}+4x^{15} \\
&+3x^{14}-x^{12}+3x^{11}+2x^{10}+2x^{9}+x^{8}+2x^{7}+2x^{6}+2x^{5}+x^{4}+x^{3}+x^{2}+2x+2.
\end{aligned}
\]

\end{Conj}  

\begin{figure}[H]
    \centering
    \begin{minipage}{0.48\linewidth}
        \centering
        \includegraphics[width=\linewidth]{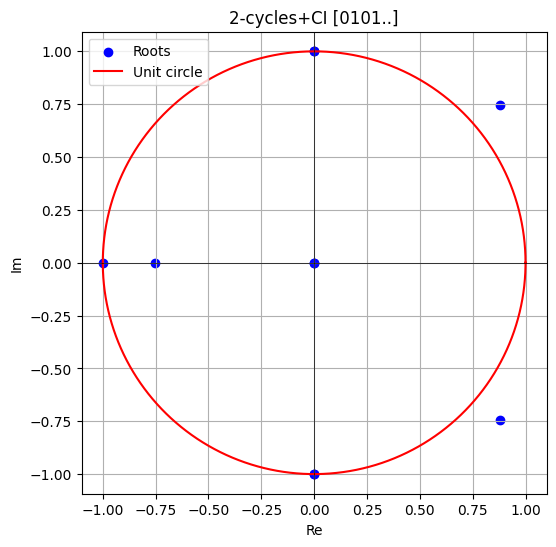}
        \caption{wrapped 2-cycles [0101...]}
        \label{fig:wrapped2inv_0101}
    \end{minipage}
    \hfill
    \begin{minipage}{0.48\linewidth}
        \centering
        \includegraphics[width=\linewidth]{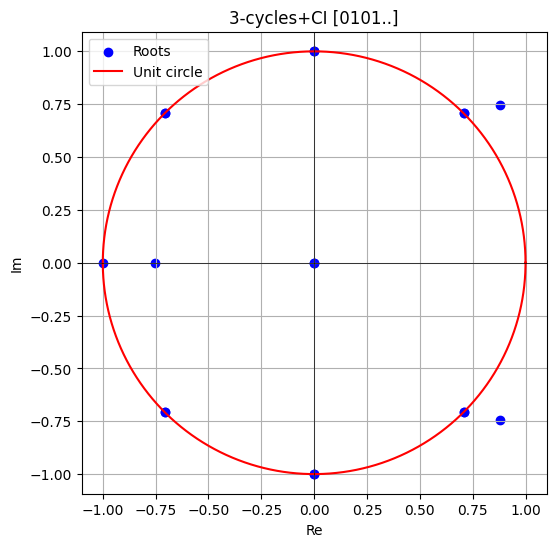}
        \caption{wrapped 3-cycles [0101...] (inverse closed or not coincide)}
        \label{fig:wrapped3inv_0101}
    \end{minipage}
\end{figure}

\begin{figure}[H]
    \centering
    \begin{minipage}{0.48\linewidth}
        \centering
        \includegraphics[width=\linewidth]{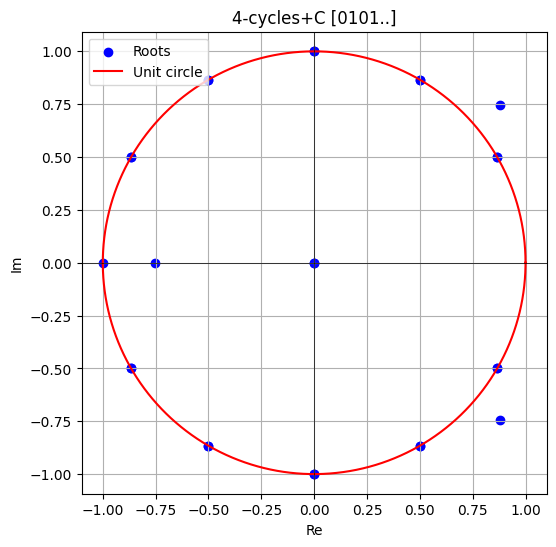}
        \caption{wrapped 4-cycles [0101...]}
        \label{fig:wrapped4_0101}
    \end{minipage}
    \hfill
    \begin{minipage}{0.48\linewidth}
        \centering
        \includegraphics[width=\linewidth]{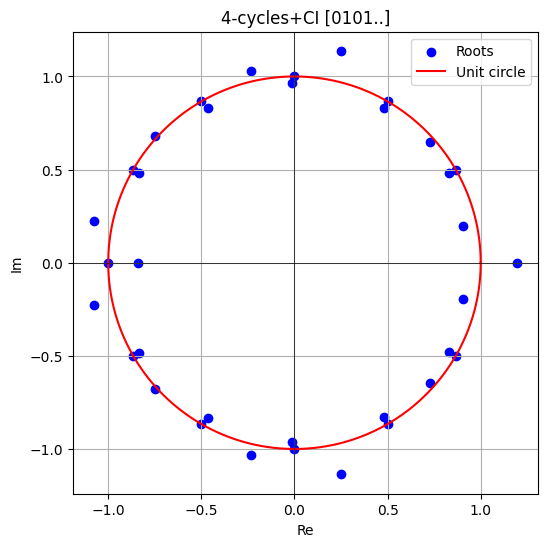}
        \caption{wrapped 4-cycles inverse closed [0101...]}
        \label{fig:wrapped4inv_0101}
    \end{minipage}
\end{figure}





\begin{table}[H]
\centering
\small
\setlength{\tabcolsep}{4pt}
\begin{tabular}{r|rrr|rrr|rrr|rrr}
\toprule
$n$ &
\multicolumn{3}{c|}{k=2} &
\multicolumn{3}{c|}{k=3} &
\multicolumn{3}{c|}{k=4} &
\multicolumn{3}{c}{k=5} \\
\cmidrule(lr){2-4}
\cmidrule(lr){5-7}
\cmidrule(lr){8-10}
\cmidrule(lr){11-13}
 & $d$ & $\Delta d$ & $\Delta^2 d$
 & $d$ & $\Delta d$ & $\Delta^2 d$
 & $d$ & $\Delta d$ & $\Delta^2 d$
 & $d$ & $\Delta d$ & $\Delta^2 d$ \\
\midrule
3  & 1 & 1 & -1 &   &   &   &   &   &   &   &   &   \\
4  & 2 & 0 &  1 & 2 & 0 & 1 &   &   &   &   &   &   \\
5  & 2 & 1 &  0 & 2 & 1 & -1 & 2 & 1 & -1 &   &   &   \\
6  & 3 & 1 & -1 & 3 & 0 &  1 & 3 & 0 &  0 & 3 & 0 & 0 \\
7  & 4 & 0 &  2 & 3 & 1 &  0 & 3 & 0 &  1 & 3 & 0 & 1 \\
8  & 4 & 2 & -1 & 4 & 1 & -1 & 3 & 1 & -1 & 3 & 1 & -1 \\
9  & 6 & 1 &  1 & 5 & 0 &  1 & 4 & 0 &  1 & 4 & 0 & 1 \\
10 & 7 & 2 & -2 & 5 & 1 & -1 & 4 & 1 & -1 & 4 & 1 & -1 \\
11 & 9 & 0 &  3 & 6 & 0 &  2 & 5 & 0 &  1 & 5 & 0 & 1 \\
12 & 9 & 3 & -2 & 6 & 2 & -2 & 5 & 1 &  0 & 5 & 1 & -1 \\
13 &12 & 1 &  2 & 8 & 0 &  2 & 6 & 1 &  0 & 6 & 0 & 1 \\
14 &13 & 3 & -3 & 8 & 2 & -2 & 7 & 1 & -1 & 6 & 1 & -1 \\
15 &16 & 0 &  4 &10 & 0 &  2 & 8 & 0 &  2 & 7 & 0 & 1 \\
16 &16 & 4 & -3 &10 & 2 & -1 & 8 & 2 & -2 & 7 & 1 & 0 \\
17 &20 & 1 &  3 &12 & 1 &  1 &10 & 0 &  2 & 8 & 1 & 0 \\
18 &21 & 4 & -4 &13 & 2 & -2 &10 & 2 & -2 & 9 & 1 & -1 \\
19 &25 & 0 &  5 &15 & 0 &  3 &12 & 0 &  2 &10 & 0 & 2 \\
20 &25 & 5 & -4 &15 & 3 & -3 &12 & 2 & -2 &10 & 2 & -2 \\
21 &30 & 1 &  4 &18 & 0 &  3 &14 & 0 &  2 &12 & 0 & 2 \\
22 &31 & 5 & -5 &18 & 3 & -3 &14 & 2 & -2 &12 & 2 & -2 \\
23 &36 & 0 &  6 &21 & 0 &  3 &16 & 0 &  2 &14 & 0 & 2 \\
24 &36 & 6 & -5 &21 & 3 & -2 &16 & 2 & -1 &14 & 2 & -2 \\
25 &42 & 1 &  5 &24 & 1 &  2 &18 & 1 &  1 &16 & 0 & 2 \\
26 &43 & 6 & -6 &25 & 3 & -3 &19 & 2 & -2 &16 & 2 & -2 \\
27 &49 & 0 &  7 &28 & 0 &  4 &21 & 0 &  3 &18 & 0 & 2 \\
28 &49 & 7 & -6 &28 & 4 & -4 &21 & 3 & -3 &18 & 2 & -2 \\
29 &56 & 1 &  6 &32 & 0 &  4 &24 & 0 &  3 &20 & 0 & 2 \\
30 &57 & 7 & -7 &32 & 4 & -4 &24 & 3 & -3 &20 & 2 & -2 \\
31 &64 & 0 &  8 &36 & 0 &  4 &27 & 0 &  3 &22 & 0 & 2 \\
32 &64 & 8 & -7 &36 & 4 & -3 &27 & 3 & -3 &22 & 2 & -1 \\
33 &72 & 1 &  7 &40 & 1 &  3 &30 & 0 &  3 &24 & 1 & 1 \\
34 &73 & 8 &    &41 & 4 &    &30 & 3 &    &25 & 2 &   \\
35 &81 &   &    &45 &   &    &33 &   &    &27 &   &   \\
\bottomrule
\end{tabular}
\caption{eccentricities  and increments for $[0,1]^{\lfloor n/2 \rfloor} + [0]^{\,n - 2\lfloor n/2 \rfloor}$)
, wrapped k-cycle, not inverse-closed }
\label{tab:wrapped_cycle_growth_k=2..5}
\end{table}

\begin{table}[H]
\centering
\small
\setlength{\tabcolsep}{4pt}
\begin{tabular}{r|rrr|rrr|rrr|rrr}
\toprule
$n$ &
\multicolumn{3}{c|}{k=2 inv} &
\multicolumn{3}{c|}{k=3 inv} &
\multicolumn{3}{c|}{k=4 inv} &
\multicolumn{3}{c}{k=5 inv} \\
\cmidrule(lr){2-4}
\cmidrule(lr){5-7}
\cmidrule(lr){8-10}
\cmidrule(lr){11-13}
 & $d$ & $\Delta d$ & $\Delta^2 d$
 & $d$ & $\Delta d$ & $\Delta^2 d$
 & $d$ & $\Delta d$ & $\Delta^2 d$
 & $d$ & $\Delta d$ & $\Delta^2 d$ \\
\midrule
3  & 1 & 1 & -1 &   &   &   &   &   &   &   &   &   \\
4  & 2 & 0 &  1 & 2 & 0 & 1 &   &   &   &   &   &   \\
5  & 2 & 1 &  0 & 2 & 1 & -1 & 2 & 1 & -2 &   &   &   \\
6  & 3 & 1 & -1 & 3 & 0 &  1 & 3 & -1 &  2 & 3 & -1 & 2 \\
7  & 4 & 0 &  2 & 3 & 1 &  0 & 2 &  1 &  0 & 2 &  1 & 0 \\
8  & 4 & 2 & -1 & 4 & 1 & -1 & 3 &  1 & -1 & 3 &  1 & -1 \\
9  & 6 & 1 &  1 & 5 & 0 &  1 & 4 &  0 &  1 & 4 &  0 & 0 \\
10 & 7 & 2 & -2 & 5 & 1 & -1 & 4 &  1 & -1 & 4 &  0 & 1 \\
11 & 9 & 0 &  3 & 6 & 0 &  2 & 5 &  0 &  1 & 4 &  1 & 0 \\
12 & 9 & 3 & -2 & 6 & 2 & -2 & 5 &  1 & -1 & 5 &  1 & -1 \\
13 &12 & 1 &  2 & 8 & 0 &  2 & 6 &  0 &  1 & 6 &  0 & 1 \\
14 &13 & 3 & -3 & 8 & 2 & -2 & 6 &  1 & -1 & 6 &  1 & -1 \\
15 &16 & 0 &  4 &10 & 0 &  2 & 7 &  0 &  2 & 7 &  0 & 1 \\
16 &16 & 4 & -3 &10 & 2 & -1 & 7 &  2 & -2 & 7 &  1 & -1 \\
17 &20 & 1 &  3 &12 & 1 &  1 & 9 &  0 &  1 & 8 &  0 & 1 \\
18 &21 & 4 & -4 &13 & 2 & -2 & 9 &  1 & -1 & 8 &  1 & -1 \\
19 &25 & 0 &  5 &15 & 0 &  3 &10 &  0 &  2 & 9 &  0 & 1 \\
20 &25 & 5 & -4 &15 & 3 & -3 &10 &  2 & -2 & 9 &  1 & 0 \\
21 &30 & 1 &  4 &18 & 0 &  3 &12 &  0 &  2 &10 &  1 & 0 \\
22 &31 & 5 & -5 &18 & 3 & -3 &12 &  2 & -2 &11 &  1 & -1 \\
23 &36 & 0 &  6 &21 & 0 &  3 &14 &  0 &  2 &12 &  0 & 2 \\
24 &36 & 6 & -5 &21 & 3 & -2 &14 &  2 & -1 &12 &  2 & -2 \\
25 &42 & 1 &  5 &24 & 1 &  2 &16 &  1 &  1 &14 &  0 & 1 \\
26 &43 & 6 & -6 &25 & 3 & -3 &17 &  2 & -2 &14 &  1 & 0 \\
27 &49 & 0 &  7 &28 & 0 &  4 &19 &  0 &  3 &15 &  1 & 1 \\
28 &49 & 7 & -6 &28 & 4 & -4 &19 &  3 & -3 &16 &  2 & -2 \\
29 &56 & 1 &  6 &32 & 0 &  4 &22 &  0 &  2 &18 &  0 & 2 \\
30 &57 & 7 & -7 &32 & 4 & -4 &22 &  2 & -1 &18 &  2 & -2 \\
31 &64 & 0 &  8 &36 & 0 &  4 &24 &  1 &  2 &20 &  0 & 2 \\
32 &64 & 8 & -7 &36 & 4 & -3 &25 &  3 & -3 &20 &  2 & -2 \\
33 &72 & 1 &  7 &40 & 1 &  3 &28 &  0 &  3 &22 &  0 & 2 \\
34 &73 & 8 &    &41 & 4 &    &28 &  3 &    &22 &  2 &   \\
35 &81 &   &    &45 &   &    &31 &   &    &24 &   &   \\
\bottomrule
\end{tabular}
\caption{eccentricities  and increments for $[0,1]^{\lfloor n/2 \rfloor} + [0]^{\,n - 2\lfloor n/2 \rfloor}$)
, wrapped k-cycle, inverse-closed}
\label{tab:wrapped_cycle_growth_k=2..5_inv}
\end{table}

\subsection{\texorpdfstring{Word metrics to full flips for Cayley and Schreier ($S_n/S_d)$ graphs}{Word metrics to full flips for Cayley and Schreier (S_n/S_d) graphs}}


Here we will consider inverse closed wrapped $k$-cycles generators.
And will look on both Cayley and  Schreier ($S_n/S_d)$ graphs. 
We will study the word metric to the "full flipp", i.e. the element
$[(n-d, n-d, \dots, n-d, n-d-1, \dots, 2, 1, 0)$,
of in the other words lengths of the shortest paths between
$(0, 1, 2, \dots, n-d-1, n-d, n-d, \dots, n-d)$  and $(n-d, n-d, \dots, n-d, n-d-1, \dots, 2, 1, 0)$.

Computations are performed with CayleyPy for small $n$ by BFS and for $n>14$ by  AI-compenent.


\begin{longtable}{lccccccc}
\caption{Experimental data ($d$ coincide)}\\

\toprule
\textbf{k} & \textbf{n} & \textbf{no c. d} & \textbf{d = 2} & \textbf{d = 3} & \textbf{d = 4} & \textbf{d = 5} & \textbf{d = 6}\\
\midrule
\endfirsthead

\toprule
\textbf{k} & \textbf{n} & \textbf{no c. d} & \textbf{d = 2} & \textbf{d = 3} & \textbf{d = 4} & \textbf{d = 5} & \textbf{d = 6}\\
\midrule
\endhead

\midrule
\multicolumn{8}{r}{\emph{Continued on next page}}
\\
\endfoot

\bottomrule
\endlastfoot

4 & 5  & 3  & 3  & 2  & 1  &    &    \\
4 & 6  & 4  & 4  & 3  & 2  & 1  &    \\
4 & 7  & 6  & 5  & 5  & 4  & 3  & 1  \\
4 & 8  & 7  & 6  & 6  & 5  & 4  & 3  \\
4 & 9  & 7  & 7  & 6  & 6  & 5  & 4  \\
4 & 10 & 10 & 8  & 8  & 7  & 7  & 6  \\
4 & 11 & 10 & 10 & 9  & 9  & 9  & 8  \\
4 & 12 & 11 & 10 & 10 & 10 & 10 & 10 \\
4 & 13 & 13 & 12 & 12 & 12 & 12 & 11 \\
4 & 14 & 14 & 14 & 14 & 13 & 13 & 13 \\
4 & 15 & 16 & 16 & 16 & 15 & 15 & 15 \\
4 & 16 & 19 & 19 & 19 & 18 & 17 & 17 \\
4 & 17 & 21 & 21 & 21 & 20 & 19 & 19 \\
4 & 18 & 24 & 24 & 24 & 23 & 22 & 22 \\
4 & 19 & 28 & 27 & 26 & 25 & 24 & 24 \\
4 & 20 & 31 & 30 & 29 & 28 & 27 & 27 \\
4 & 21 & 35 & 34 & 33 & 32 & 31 & 29 \\
4 & 22 & 38 & 37 & 36 & 35 & 34 & 32 \\
4 & 23 & 42 & 41 & 40 & 39 & 38 & 36 \\
4 & 24 & 44 & 44 & 43 & 42 & 41 & 39 \\
4 & 25 & 48 & 48 & 47 & 46 & 45 & 43 \\
4 & 26 & 50 & 52 & 52 & 51 & 49 & 48 \\
4 & 27 & 56 & 56 & 56 & 55 & 53 & 52 \\
4 & 28 & 61 & 61 & 61 & 60 & 58 & 57 \\
4 & 29 &    & 65 & 65 & 64 & 62 & 62 \\
4 & 30 &    & 70 & 70 & 69 & 67 & 67 \\
4 & 31 &    & 75 & 74 & 73 & 72 & 70 \\
4 & 32 &    & 80 & 79 & 78 & 77 & 75 \\
4 & 33 &    &    & 85 & 84 & 83 & 83 \\

\end{longtable}

\clearpage

\textbf{Case $k=4$, $d=2$ from $n=12$}
\[
\begin{aligned}
n &\equiv 0 \pmod{12}:&\; D(n) &= \frac{n(n-2)}{12},\\[4pt]
n &\equiv 1 \pmod{12}:&\; D(n) &= \frac{n(n-2)}{12} + \frac{1}{12},\\[4pt]
n &\equiv 2 \pmod{12}:&\; D(n) &= \frac{n(n-2)}{12},\\[4pt]
n &\equiv 3 \pmod{12}:&\; D(n) &= \frac{n(n-2)}{12} - \frac{1}{4},\\[4pt]
n &\equiv 4 \pmod{12}:&\; D(n) &= \frac{n(n-2)}{12} + \frac{1}{3},\\[4pt]
n &\equiv 5 \pmod{12}:&\; D(n) &= \frac{n(n-2)}{12} - \frac{1}{4},\\[4pt]
n &\equiv 6 \pmod{12}:&\; D(n) &= \frac{n(n-2)}{12},\\[4pt]
n &\equiv 7 \pmod{12}:&\; D(n) &= \frac{n(n-2)}{12} + \frac{1}{12},\\[4pt]
n &\equiv 8 \pmod{12}:&\; D(n) &= \frac{n(n-2)}{12},\\[4pt]
n &\equiv 9 \pmod{12}:&\; D(n) &= \frac{n(n-2)}{12} + \frac{3}{4},\\[4pt]
n &\equiv 10 \pmod{12}:&\; D(n) &= \frac{n(n-2)}{12} + \frac{1}{3},\\[4pt]
n &\equiv 11 \pmod{12}:&\; D(n) &= \frac{n(n-2)}{12} + \frac{3}{4}.\\[4pt]
\end{aligned}
\]

The corresponding generating function is given by 
\begin{equation*}
\frac{H(x)}{(1 - x^{12})^3},
\end{equation*}
where
\begin{multline*}
H(x) = x^{35}+x^{34}+2x^{33}+2x^{32}+3x^{31}+4x^{30}+5x^{29}+7x^{28}+\\+8x^{27}+10x^{26}+12x^{25}+14x^{24}+14x^{23}+16x^{22}+16x^{21}+18x^{20}+\\+18x^{19}+18x^{18}+18x^{17}+16x^{16}+16x^{15}+14x^{14}+12x^{13}+10x^{12}+\\+9x^{11}+7x^{10}+6x^{9}+4x^{8}+3x^{7}+2x^{6}+x^{5}+x^{4}
\end{multline*}

\newpage

The polynomial $H(x)$ has the following properties:
    \begin{itemize}
        \item its coefficients are nonnegative and unimodal;
        \item its coefficients are not symmetric;
        \item its roots are not necessarily on the unit circle (can be both inside and outside).
    \end{itemize}

\begin{figure}[H]
    \centering
    \includegraphics[width=0.75\linewidth]{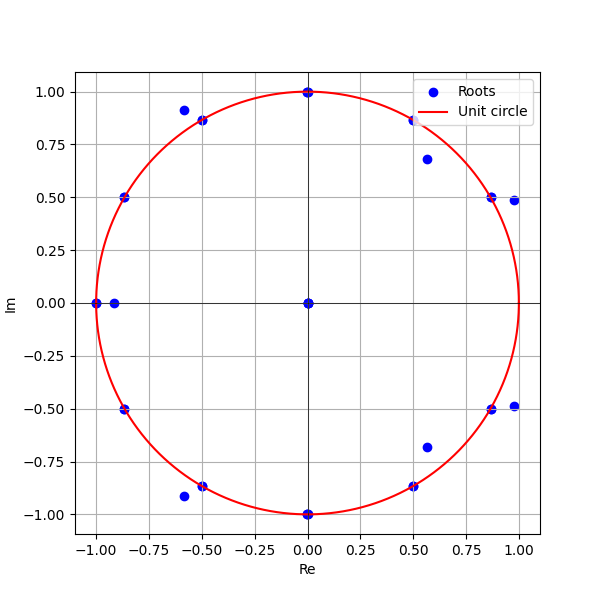}
    \caption{case $k=4$, $d=2$}
    \label{fig:wrapped4_0103}
\end{figure}

The $H$-polynomial factors as
\begin{equation*}
    H(x) = x^{4}(x+1)^{2}(x^{2}+1)^{3}(x^{2}-x+1)^{2}(x^{2}+x+1)^{2}(x^{4}-x^{2}+1)^{2}(x^{7}-x^{6}+x^{3}-x+1)
\end{equation*}
        
Hence, the generating function can be rewritten as
\begin{equation*}
\frac{x^{4}(x^{7}-x^{6}+x^{3}-x+1)}
{(1+x)(1-x)^{3}(x^{2}+x+1)(x^{2}-x+1)(x^{4}-x^{2}+1)}
\end{equation*}

\textbf{Case $k=4$, $d=3$ from $n=14$}
\[
\begin{aligned}
n &\equiv 0 \pmod{12}:&\; D(n) &= \frac{n(n-2)}{12} - 1,\\[4pt]
n &\equiv 1 \pmod{12}:&\; D(n) &= \frac{n(n-2)}{12} - \frac{11}{12},\\[4pt]
n &\equiv 2 \pmod{12}:&\; D(n) &= \frac{n(n-2)}{12},\\[4pt]
n &\equiv 3 \pmod{12}:&\; D(n) &= \frac{n(n-2)}{12} - \frac{1}{4},\\[4pt]
n &\equiv 4 \pmod{12}:&\; D(n) &= \frac{n(n-2)}{12} + \frac{1}{3},\\[4pt]
n &\equiv 5 \pmod{12}:&\; D(n) &= \frac{n(n-2)}{12} - \frac{1}{4},\\[4pt]
n &\equiv 6 \pmod{12}:&\; D(n) &= \frac{n(n-2)}{12},\\[4pt]
n &\equiv 7 \pmod{12}:&\; D(n) &= \frac{n(n-2)}{12} - \frac{11}{12},\\[4pt]
n &\equiv 8 \pmod{12}:&\; D(n) &= \frac{n(n-2)}{12} - 1,\\[4pt]
n &\equiv 9 \pmod{12}:&\; D(n) &= \frac{n(n-2)}{12} - \frac{1}{4},\\[4pt]
n &\equiv 10 \pmod{12}:&\; D(n) &= \frac{n(n-2)}{12} - \frac{2}{3},\\[4pt]
n &\equiv 11 \pmod{12}:&\; D(n) &= \frac{n(n-2)}{12} - \frac{1}{4}.\\[4pt]
\end{aligned}
\]

The corresponding generating function is given by 
\begin{equation*}
\frac{H(x)}{(1 - x^{12})^3},
\end{equation*}
where
\begin{multline*}
H(x) = 
- x^{36} + x^{33} + x^{32} + 2x^{31} + 4x^{30} + 5x^{29} + 7x^{28} + \\ + 8x^{27} + 10x^{26} + 11x^{25} + 16x^{24} + 16x^{23} + 18x^{22} + 18x^{21} + 20x^{20} + \\ + 20x^{19} + 18x^{18} + 18x^{17} + 16x^{16} + 16x^{15} + 14x^{14} + 14x^{13} + 9x^{12} + \\ + 8x^{11} + 6x^{10} + 5x^{9} + 3x^{8} + 2x^{7} + 2x^{6} + x^{5} + x^{4} - x
\end{multline*}

\newpage

The polynomial $H(x)$ has the following properties:
    \begin{itemize}
        \item its coefficients are neither nonnegative nor unimodal;
        \item its coefficients are not symmetric;
        \item its roots are not necessarily on the unit circle (can be both inside and outside).
    \end{itemize}

\begin{figure}[H]
    \centering
    \includegraphics[width=0.75\linewidth]{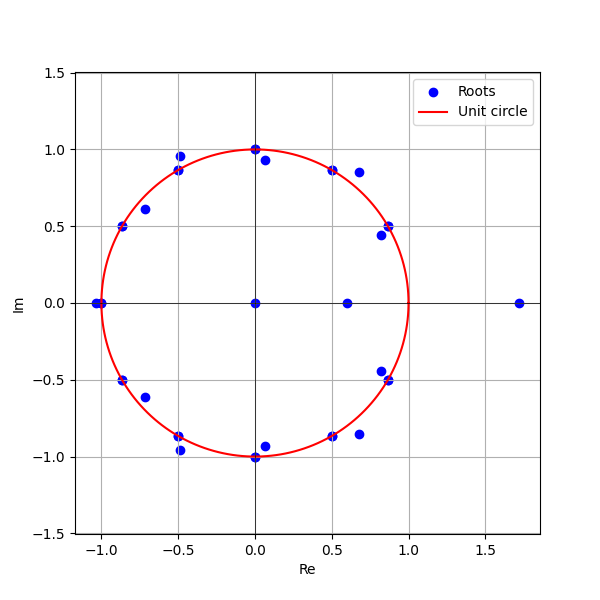}
    \caption{case $k=4$, $d=3$}
    \label{fig:wrapped4_0102}
\end{figure}

The $H$-polynomial factors as
\begin{multline*}
    H(x) = - x (x+1)^2 (x^2+1)^2 (x^2-x+1)^2 (x^2+x+1)^2 (x^4-x^2+1)^2 \\ (x^{13}-2x^{12}+x^{11}-x^{10}+x^9-x^8-x^7+x^6-x^5+x^4-x^3+x^2-2x+1)
\end{multline*}

Hence, the generating function can be rewritten as
\begin{equation*}
\frac{- x (x^{13}-2x^{12}+x^{11}-x^{10}+x^9-x^8-x^7+x^6-x^5+x^4-x^3+x^2-2x+1)}
{(1+x^2)(1+x)(1-x)^3(x^2+x+1)(x^2-x+1)(x^4-x^2+1)}
\end{equation*}

\textbf{Case $k=4$, $d=4$ from $n=14$}
\[
\begin{aligned}
n &\equiv 0 \pmod{12}:&\; D(n) &= \frac{n(n-2)}{12} - 2,\\[4pt]
n &\equiv 1 \pmod{12}:&\; D(n) &= \frac{n(n-2)}{12} - \frac{23}{12},\\[4pt]
n &\equiv 2 \pmod{12}:&\; D(n) &= \frac{n(n-2)}{12} - 1,\\[4pt]
n &\equiv 3 \pmod{12}:&\; D(n) &= \frac{n(n-2)}{12} - \frac{5}{4},\\[4pt]
n &\equiv 4 \pmod{12}:&\; D(n) &= \frac{n(n-2)}{12} - \frac{2}{3},\\[4pt]
n &\equiv 5 \pmod{12}:&\; D(n) &= \frac{n(n-2)}{12} - \frac{5}{4},\\[4pt]
n &\equiv 6 \pmod{12}:&\; D(n) &= \frac{n(n-2)}{12} - 1,\\[4pt]
n &\equiv 7 \pmod{12}:&\; D(n) &= \frac{n(n-2)}{12} - \frac{23}{12},\\[4pt]
n &\equiv 8 \pmod{12}:&\; D(n) &= \frac{n(n-2)}{12} - 2,\\[4pt]
n &\equiv 9 \pmod{12}:&\; D(n) &= \frac{n(n-2)}{12} - \frac{5}{4},\\[4pt]
n &\equiv 10 \pmod{12}:&\; D(n) &= \frac{n(n-2)}{12} - \frac{5}{3},\\[4pt]
n &\equiv 11 \pmod{12}:&\; D(n) &= \frac{n(n-2)}{12} - \frac{5}{4}.\\[4pt]
\end{aligned}
\]

The corresponding generating function is given by 
\begin{equation*}
\frac{H(x)}{(1 - x^{12})^3},
\end{equation*}
where
\begin{multline*}
H(x) = 
-2x^{36} - x^{35} - x^{34} + x^{31} + 3x^{30} + 4x^{29} + 6x^{28} + \\ + 7x^{27} + 9x^{26} + 10x^{25}
+ 18x^{24} + 18x^{23} + 20x^{22} + 20x^{21} + 22x^{20} + \\ + 22x^{19} + 20x^{18} + 20x^{17}
+ 18x^{16} + 18x^{15} + 16x^{14} + 16x^{13} + 8x^{12} + \\ + 7x^{11} + 5x^{10} + 4x^{9}
+ 2x^{8} + x^{7} + x^{6} - x^{3} - x^{2} - 2x
\end{multline*}

\newpage

The polynomial $H(x)$ has the following properties:
    \begin{itemize}
        \item its coefficients are neither nonnegative nor unimodal;
        \item its coefficients are not symmetric;
        \item its roots are not necessarily on the unit circle (can be both inside and outside).
    \end{itemize}

\begin{figure}[H]
    \centering
    \includegraphics[width=0.75\linewidth]{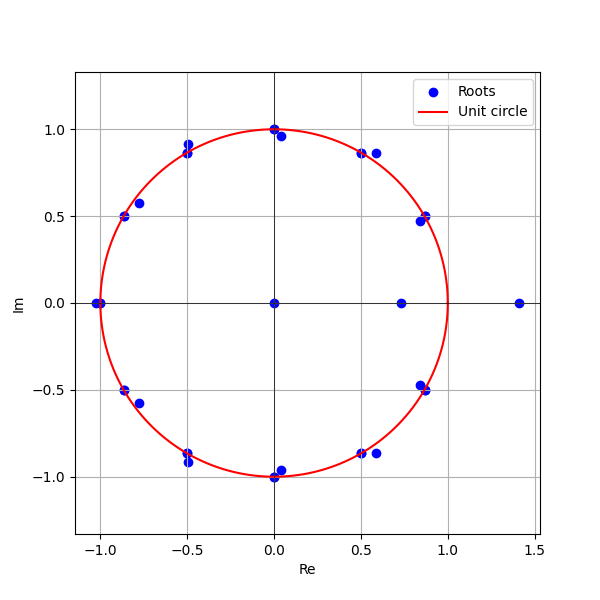}
    \caption{case $k=4$, $d=4$}
    \label{fig:wrapped4_0104}
\end{figure}

The $H$-polynomial factors as
\begin{multline*}
    H(x) = - x (x + 1)^2 (x^2 + 1)^2 (x^2 - x + 1)^2 (x^2 + x + 1)^2 (x^4 - x^2 + 1)^2  \\ (2x^{13} - 3x^{12} + x^{11} - x^{10} + x^9 - x^8 - x^7 + x^6 - x^5 + x^4 - x^3 + x^2 - 3x + 2)
\end{multline*}

Hence, the generating function can be rewritten as
\begin{equation*}
\frac{- x (2x^{13} - 3x^{12} + x^{11} - x^{10} + x^9 - x^8 - x^7 + x^6 - x^5 + x^4 - x^3 + x^2 - 3x + 2)}
{(1+x^2)(1+x)(1-x)^3(x^2+x+1)(x^2-x+1)(x^4-x^2+1)}
\end{equation*}

\textbf{Case $k=4$, $d=5$ from $n=19$}
\[
\begin{aligned}
n &\equiv 0 \pmod{12}:&\; D(n) &= \frac{n(n-2)}{12} - 3,\\[4pt]
n &\equiv 1 \pmod{12}:&\; D(n) &= \frac{n(n-2)}{12} - \frac{35}{12},\\[4pt]
n &\equiv 2 \pmod{12}:&\; D(n) &= \frac{n(n-2)}{12} - 3,\\[4pt]
n &\equiv 3 \pmod{12}:&\; D(n) &= \frac{n(n-2)}{12} - \frac{13}{4},\\[4pt]
n &\equiv 4 \pmod{12}:&\; D(n) &= \frac{n(n-2)}{12} - \frac{8}{3},\\[4pt]
n &\equiv 5 \pmod{12}:&\; D(n) &= \frac{n(n-2)}{12} - \frac{13}{4},\\[4pt]
n &\equiv 6 \pmod{12}:&\; D(n) &= \frac{n(n-2)}{12} - 3,\\[4pt]
n &\equiv 7 \pmod{12}:&\; D(n) &= \frac{n(n-2)}{12} - \frac{35}{12},\\[4pt]
n &\equiv 8 \pmod{12}:&\; D(n) &= \frac{n(n-2)}{12} - 3,\\[4pt]
n &\equiv 9 \pmod{12}:&\; D(n) &= \frac{n(n-2)}{12} - \frac{9}{4},\\[4pt]
n &\equiv 10 \pmod{12}:&\; D(n) &= \frac{n(n-2)}{12} - \frac{8}{3},\\[4pt]
n &\equiv 11 \pmod{12}:&\; D(n) &= \frac{n(n-2)}{12} - \frac{9}{4}.\\[4pt]
\end{aligned}
\]

The corresponding generating function is given by 
\begin{equation*}
\frac{H(x)}{(1 - x^{12})^3},
\end{equation*}
where
\begin{multline*}
H(x) = 
-3x^{36} - 2x^{35} - 2x^{34} - x^{33} - x^{32} + x^{30} + 2x^{29} + 4x^{28} + \\ + 5x^{27} + 7x^{26} + 9x^{25} + 20x^{24} + 20x^{23} + 22x^{22} + 22x^{21} + 24x^{20} + \\ + 24x^{19}
+ 24x^{18} + 24x^{17} + 22x^{16} + 22x^{15} + 20x^{14} + 18x^{13} + 7x^{12} + \\ + 6x^{11}
+ 4x^{10} + 3x^{9} + x^{8} - x^{6} - 2x^{5} - 2x^{4} - 3x^{3} - 3x^{2} - 3x
\end{multline*}

\newpage

The polynomial $H(x)$ has the following properties:
    \begin{itemize}
        \item its coefficients are neither nonnegative nor unimodal;
        \item its coefficients are not symmetric;
        \item its roots are not necessarily on the unit circle (can be both inside and outside).
    \end{itemize}

\begin{figure}[H]
    \centering
    \includegraphics[width=0.75\linewidth]{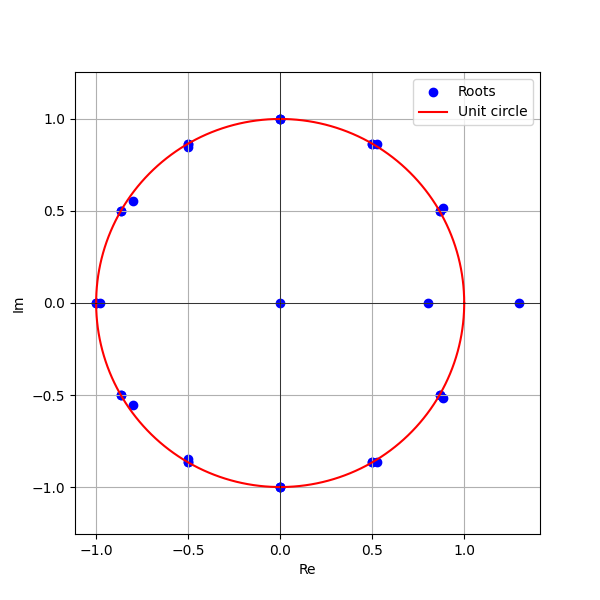}
    \caption{case $k=4$, $d=5$}
    \label{fig:wrapped4_0105}
\end{figure}

The $H$-polynomial factors as
\begin{multline*}
    H(x) = - x (x + 1)^2 (x^2 + 1)^3 (x^2 - x + 1)^2 (x^2 + x + 1)^2 (x^4 - x^2 + 1)^2 \\ (3x^{11} - 4x^{10} - 2x^9 + 3x^8 + 3x^7 - 4x^6 - 3x^5 + 4x^4 + 2x^3 - 3x^2 - 3x + 3)
\end{multline*}

Hence, the generating function can be rewritten as
\begin{equation*}
\frac{- x (3x^{11} - 4x^{10} - 2x^9 + 3x^8 + 3x^7 - 4x^6 - 3x^5 + 4x^4 + 2x^3 - 3x^2 - 3x + 3)}
{(1+x)(1-x)^3(x^2+x+1)(x^2-x+1)(x^4-x^2+1)}
\end{equation*}

\subsection{Coset 2-different} 
Consider inverse-closed coset with central state $011\ldots 11 = 0^1 1^{n-1}$. According to calculations in the \href{https://www.kaggle.com/code/fedmug/cayleypy-consecutive-cycles-growth-diff}{notebook}, one can hypothesize that the diameter of the wrapped $k$-cycles coset is given by

$$
d_k(n) = \Big\lfloor \frac{n+(k-1)(k-2)}{2k - 2} \Big\rfloor + r_k(n),
$$
where $r_k(n)$ is a "small" remainder. This remainder term seems to be $(2k-2)$-periodic function of $n$ for $n \geqslant N_k$. See table~\ref{tab:2diff-r-k-n} for details.

\begin{table}[ht]
\centering
\caption{$r_k(n)$ period}
\label{tab:2diff-r-k-n}
\begin{tabular}{| l | l | l | l |}
\hline
\textbf{$k$} & \textbf{$N_k$} & \textbf{period} & \textbf{$r_k(n)$} \\
\hline
\textbf{3}  & $4$  & $4$  & $(0, 0, 0, 0)$ \\
\textbf{4}  & $5$  & $6$  & $(0, 0, 0, 0, 0, 1)$ \\
\textbf{5}  & $6$  & $8$   & $(0, 0, 0, 0, 0, 0, 0, 0)$ \\
\textbf{6}  & $7$  & $10$   & $(0, 0, 0, 0, 0, 0, -1, 0, 0, 1)$ \\
\textbf{7}  & $9$  & $12$   & $(-1, -1, 0, 0, 0, 0, 0, 0, 0, 0)$ \\
\textbf{8}  & $10$  & $14$   & $(0, -1, 0, 0, 0, 0, 0, -1, -1, -1, 0, 0, 1, 0)$ \\
\textbf{9}  & $14$  & $16$   & $(-1, 0, 0, 0, 0, -1, -1, -1, 0, 0, 0, 0, 0, -1, -1, -1)$ \\
\textbf{10}  & $13$  & $18$   & $(-1, -1, -1, 0, 0, 0, 0, 0, -1, -1, -2, -1, -1, 0, 0, 1, 0, 0)$ \\
\textbf{11}  & $18$  & $20$   & $(-1, 0, 0, 0, 0, -1, -1, -1, -1, -1, 0, 0, 0, 0, 0, -1, -1, -2, -2, -1)$ \\
\textbf{12}  & $20$  & $22$   & $(0, 0, 0, 0, 0,  -2, -2, -2, -2, -2, -1, -1, 0, 0, 1, 0, 0, -1, -1, -2, -1, -1)$ \\
\hline
\end{tabular}
\end{table}

\subsection{Coset 3-different}
The following data correspond to the graph name part4  3Different  $0122\dots 2222$ line in the code.
\[ [t]:=\lfloor t\rfloor,\qquad p:=2k-2,\qquad r:=n\bmod p. \]
For each $k$, the formula below is exact for all observed data points with $n\ge N_k$ in the current spreadsheet,
\[ f_k(n)=2\left[\frac{n+x_k}{2k-2}\right]+e_k(r),\qquad \min_r e_k(r)=0. \]
\begin{itemize}
\item $k=3$ ($p=4$, $x=0$, $N_k=4$, observed to $n=200$; tail covers $\approx 49.25$ periods, well verified):
\[ f_{3}(n)=2\left[\frac{n+0}{4}\right]+e_{3}(r),\qquad r=n\bmod 4. \]
\[ e_{3}=(0,0,1,1). \]

\item $k=4$ ($p=6$, $x=2$, $N_k=7$, observed to $n=200$; tail covers $\approx 32.33$ periods, well verified):
\[ f_{4}(n)=2\left[\frac{n+2}{6}\right]+e_{4}(r),\qquad r=n\bmod 6. \]
\[ e_{4}=(0,1,1,1,0,0). \]

\item $k=5$ ($p=8$, $x=7$, $N_k=6$, observed to $n=200$; tail covers $\approx 24.38$ periods, well verified):
\[ f_{5}(n)=2\left[\frac{n+7}{8}\right]+e_{5}(r),\qquad r=n\bmod 8. \]
\[ e_{5}=(1,0,0,0,0,1,1,1). \]

\item $k=6$ ($p=10$, $x=6$, $N_k=8$, observed to $n=200$; tail covers $\approx 19.30$ periods, well verified):
\[ f_{6}(n)=2\left[\frac{n+6}{10}\right]+e_{6}(r),\qquad r=n\bmod 10. \]
\[ e_{6}=(2,2,2,2,0,1,1,2,1,2). \]

\item $k=7$ ($p=12$, $x=17$, $N_k=17$, observed to $n=200$; tail covers $\approx 15.33$ periods, well verified):
\[ f_{7}(n)=2\left[\frac{n+17}{12}\right]+e_{7}(r),\qquad r=n\bmod 12. \]
\[ e_{7}=(0,1,1,0,0,1,1,0,0,0,0,0). \]

\item $k=8$ ($p=14$, $x=19$, $N_k=35$, observed to $n=200$; tail covers $\approx 11.86$ periods, well verified):
\[ f_{8}(n)=2\left[\frac{n+19}{14}\right]+e_{8}(r),\qquad r=n\bmod 14. \]
\[ e_{8}=(1,2,1,1,1,1,1,2,2,1,0,0,0,0). \]

\item $k=9$ ($p=16$, $x=23$, $N_k=55$, observed to $n=200$; tail covers $\approx 9.12$ periods, well verified):
\[ f_{9}(n)=2\left[\frac{n+23}{16}\right]+e_{9}(r),\qquad r=n\bmod 16. \]
\[ e_{9}=(2,2,2,2,1,2,1,2,2,1,1,1,0,0,1,1). \]

\item $k=10$ ($p=18$, $x=35$, $N_k=69$, observed to $n=200$; tail covers $\approx 7.33$ periods, well verified):
\[ f_{10}(n)=2\left[\frac{n+35}{18}\right]+e_{10}(r),\qquad r=n\bmod 18. \]
\[ e_{10}=(2,1,0,0,0,0,0,0,0,1,1,2,1,1,0,1,1,2). \]

\item $k=11$ ($p=20$, $x=40$, $N_k=88$, observed to $n=200$; tail covers $\approx 5.65$ periods, well verified):
\[ f_{11}(n)=2\left[\frac{n+40}{20}\right]+e_{11}(r),\qquad r=n\bmod 20. \]
\[ e_{11}=(1,1,1,1,0,0,0,0,0,1,2,2,2,2,1,1,1,1,2,2). \]

\item $k=12$ ($p=22$, $x=43$, $N_k=129$, observed to $n=200$; tail covers $\approx 3.27$ periods, well verified):
\[ f_{12}(n)=2\left[\frac{n+43}{22}\right]+e_{12}(r),\qquad r=n\bmod 22. \]
\[ e_{12}=(3,2,1,2,1,0,0,1,1,1,2,2,2,3,2,2,1,2,1,2,2,3). \]

\item $k=13$ ($p=24$, $x=60$, $N_k=129$, observed to $n=200$; tail covers $\approx 3.00$ periods, well verified):
\[ f_{13}(n)=2\left[\frac{n+60}{24}\right]+e_{13}(r),\qquad r=n\bmod 24. \]
\[ e_{13}=(2,2,2,2,2,1,1,1,1,1,2,2,1,1,1,1,0,0,0,0,0,0,1,1). \]

\item $k=14$ ($p=26$, $x=65$, $N_k=165$, observed to $n=200$; tail covers $\approx 1.38$ periods, moderate evidence):
\[ f_{14}(n)=2\left[\frac{n+65}{26}\right]+e_{14}(r),\qquad r=n\bmod 26. \]
\[ e_{14}=(2,3,2,3,2,2,1,1,1,2,2,2,3,2,1,2,1,1,0,0,0,1,0,1,1,2). \]

\item $k=15$ ($p=28$, $x=84$, $N_k=151$, observed to $n=200$; tail covers $\approx 1.79$ periods, moderate evidence):
\[ f_{15}(n)=2\left[\frac{n+84}{28}\right]+e_{15}(r),\qquad r=n\bmod 28. \]
\[ e_{15}=(1,1,1,1,1,0,0,0,0,0,0,0,1,2,2,2,2,2,2,1,0,0,1,1,1,1,2,2). \]

\item $k=16$ ($p=30$, $x=89$, $N_k=149$, observed to $n=200$; tail covers $\approx 1.73$ periods, moderate evidence):
\[ f_{16}(n)=2\left[\frac{n+89}{30}\right]+e_{16}(r),\qquad r=n\bmod 30. \]
\[ e_{16}=(3,2,1,2,1,1,0,0,0,0,1,1,1,2,2,3,2,3,2,2,2,1,1,1,2,2,2,2,2,3). \]

\item $k=17$ ($p=32$, $x=97$, $N_k=127$, observed to $n=200$; tail covers $\approx 2.31$ periods, well verified):
\[ f_{17}(n)=2\left[\frac{n+97}{32}\right]+e_{17}(r),\qquad r=n\bmod 32. \]
\[ e_{17}=(2,2,2,2,2,1,1,0,0,1,1,2,2,2,2,3,3,3,3,3,3,2,2,1,1,2,2,2,3,3,3,2). \]

\item $k=18$ ($p=34$, $x=119$, $N_k=167$, observed to $n=200$; tail covers $\approx 1.00$ periods, moderate evidence):
\[ f_{18}(n)=2\left[\frac{n+119}{34}\right]+e_{18}(r),\qquad r=n\bmod 34. \]
\[ e_{18}=(2,3,3,3,2,2,1,1,0,1,1,2,2,2,2,3,3,2,1,2,1,2,1,0,0,0,0,0,1,1,1,2,2,2). \]

\item $k=19$ ($p=36$, $x=127$, $N_k=230$, observed to $n=300$; tail covers $\approx 1.97$ periods, moderate evidence):
\[ f_{19}(n)=2\left[\frac{n+127}{36}\right]+e_{19}(r),\qquad r=n\bmod 36. \]
\[ e_{19}=(3,3,3,3,3,2,2,1,1,1,2,2,2,2,2,3,3,2,2,2,2,2,2,2,1,0,0,0,0,1,1,1,1,2,2,3). \]

\item $k=20$ ($p=38$, $x=133$, $N_k=262$, observed to $n=300$; tail covers $\approx 1.03$ periods, moderate evidence):
\[ f_{20}(n)=2\left[\frac{n+133}{38}\right]+e_{20}(r),\qquad r=n\bmod 38. \]
\[ e_{20}=(3,4,3,4,3,3,2,2,1,1,2,2,2,3,3,3,3,4,4,3,2,3,2,3,2,2,1,0,0,1,1,2,2,1,2,2,3,4). \]

\item $k=21$ ($p=40$, $x=162$, $N_k=259$, observed to $n=300$; tail covers $\approx 1.05$ periods, moderate evidence):
\[ f_{21}(n)=2\left[\frac{n+162}{40}\right]+e_{21}(r),\qquad r=n\bmod 40. \]
\[ e_{21}=(2,2,2,2,2,2,1,0,0,0,0,0,1,1,2,1,2,2,2,3,3,3,3,3,3,3,2,2,1,1,1,1,2,2,2,3,3,3,2,2). \]

\item $k=22$ ($p=42$, $x=169$, $N_k=253$, observed to $n=300$; tail covers $\approx 1.14$ periods, moderate evidence):
\[ f_{22}(n)=2\left[\frac{n+169}{42}\right]+e_{22}(r),\qquad r=n\bmod 42. \]
\[ e_{22}=(2,3,3,3,2,2,2,1,0,0,0,0,1,1,2,2,2,2,3,3,3,4,4,4,3,4,3,3,2,2,1,1,2,2,2,3,3,3,3,4,4,3). \]
\vspace{2em}
\textbf{Hypothesis (Last-Layer Periodicity, 3-Different Coset).}
Retain the notation of \S11.11: let $p := 2k - 2$ and $r := n \bmod p$. 
The diameter of the 3-different wrapped coset graph is given by
\begin{equation*}
    f_k(n) = 2\left\lceil \frac{n + x_k}{2k-2} \right\rceil + e_k(r), 
    \qquad \min_r\, e_k(r) = 0,
\end{equation*}
where $e_k$ is $p$-periodic for $n \geq N_k$. 
Let $\ell_k(n)$ denote the size of the last BFS layer at depth $n$.

We conjecture that $\ell_k(n)$ is also eventually $p$-periodic 
(i.e., with the same period $p = 2k - 2$) for $n \geq N_k$. 
Equivalently, the quasi-polynomial period of $f_k$ and the period 
of $\ell_k$ coincide.

This is supported empirically for $k = 11, 12, 15, 17$ 
(see Figures~\ref{fig:winvk11}--\ref{fig:winvk17}); in each case both 
$e_k(r)$ and $\ell_k(n)$ exhibit synchronized periodicity 
with period $p = 2k - 2$.

\begin{figure}
    \centering
    \includegraphics[width=0.75\linewidth]{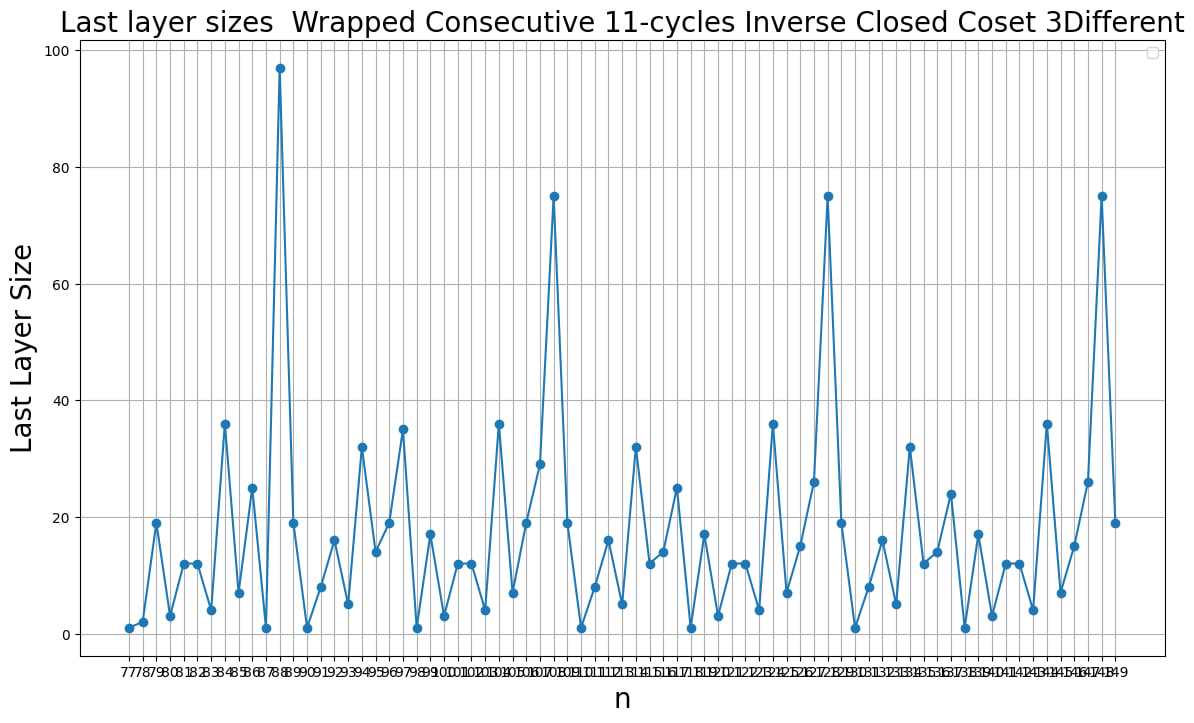}
    \caption{$k = 11$: period $p = 2k-2 = 20$}
    \label{fig:winvk11}
\end{figure}
\begin{figure}
    \centering
    \includegraphics[width=0.75\linewidth]{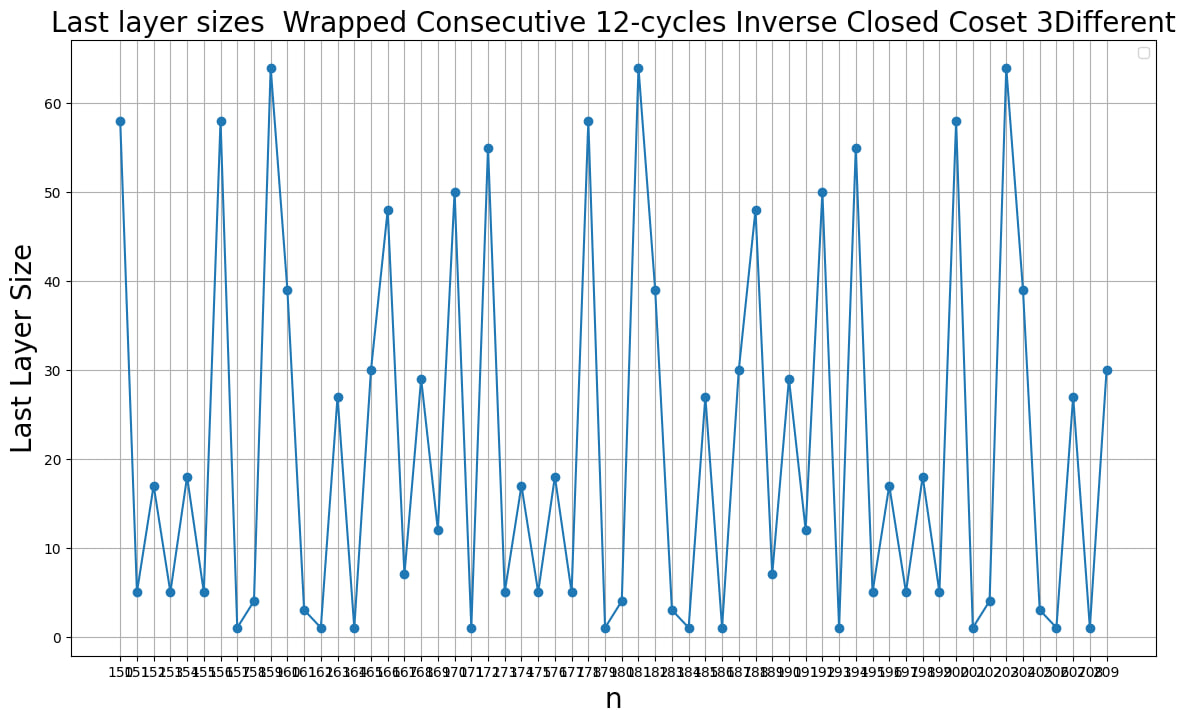}
    \caption{$k = 12$: period $p = 2k-2 = 22$}
    \label{fig:winvk12}
\end{figure}
\begin{figure}
    \centering
    \includegraphics[width=0.75\linewidth]{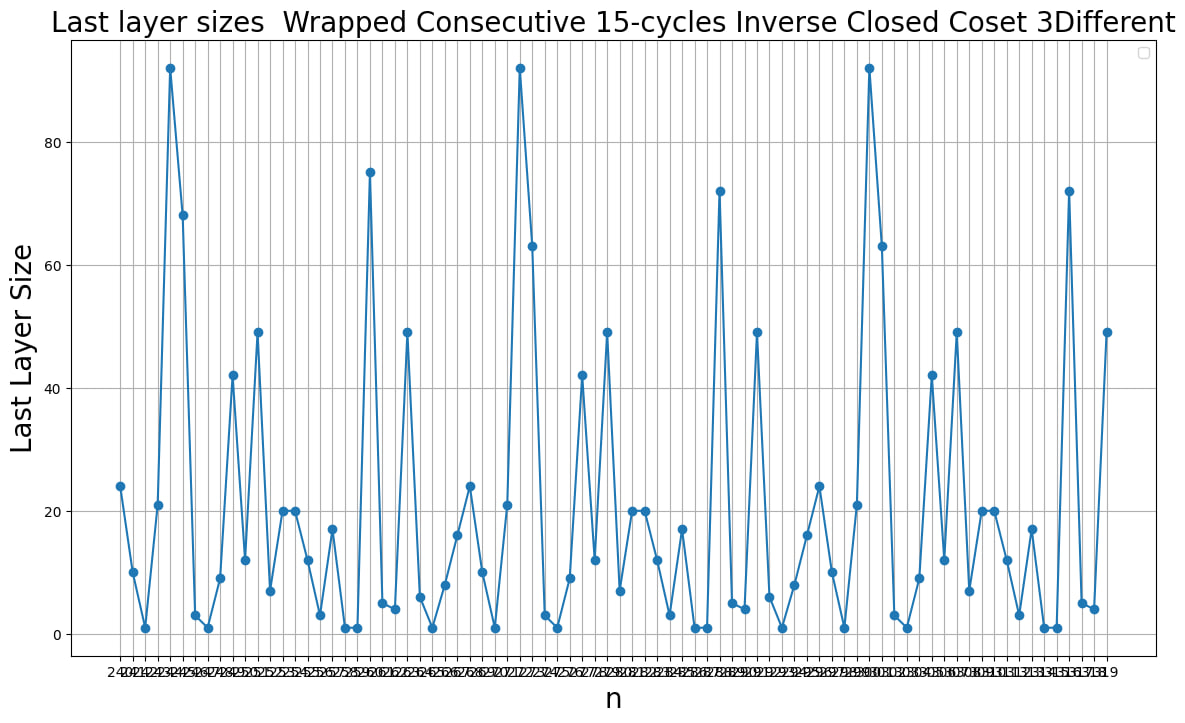}
    \caption{$k = 15$: period $p = 2k-2 = 28$}
    \label{fig:winvk15}
\end{figure}
\begin{figure}
    \centering
    \includegraphics[width=0.75\linewidth]{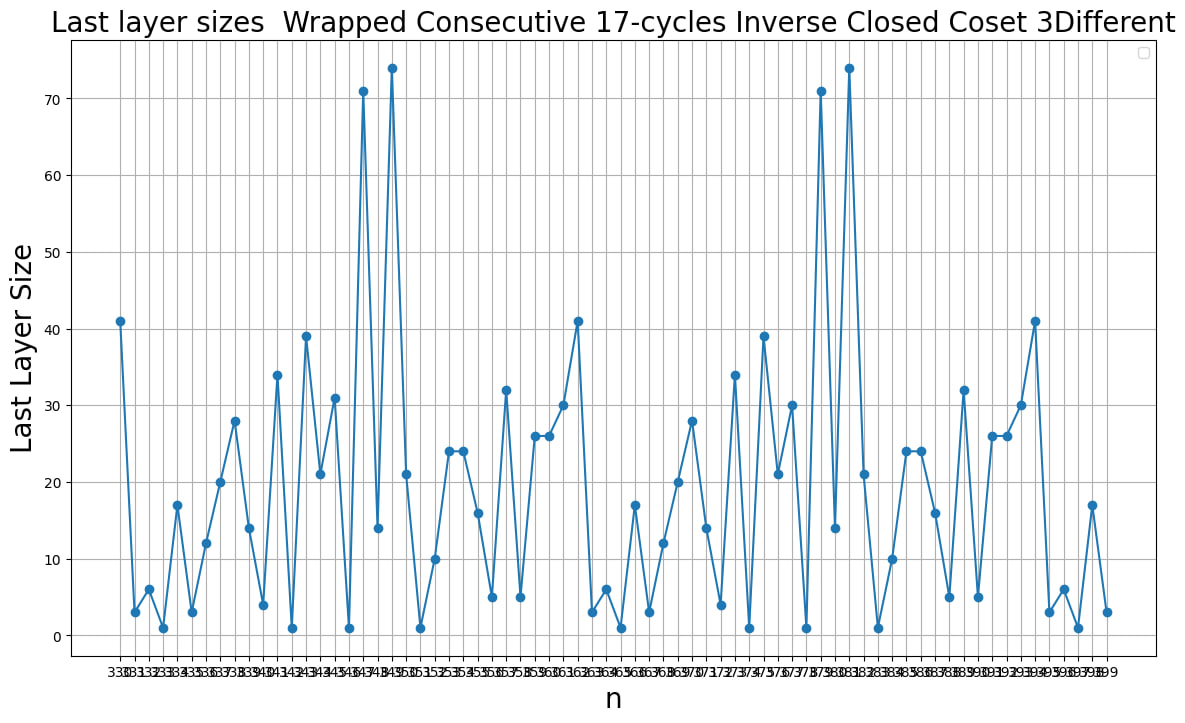}
    \caption{$k = 17$: period $p = 2k-2 = 32$}
    \label{fig:winvk17}
\end{figure}
\end{itemize}
\clearpage
\subsection{Coset 4-different}
\href{https://www.kaggle.com/code/alexandervc/cayleypy-consecutive-cycles-growth}{Notebook link}

\begin{table}[ht]
\centering
\caption{Last layer size period for wrapped consecutive coset 4-different.}
\label{tab:2diff-r-k-n}
\begin{tabular}{| l | l | l | l |}
\hline
\textbf{$k$} & \textbf{$N_k$} & \textbf{period} & \textbf{$l_k(n)$} \\
\hline
\textbf{3} & $8$ & $4$ & $(16, 4, 2, 30)$ \\
\textbf{4} & $11$ & $6$ & $(2, 29, 2, 34, 114, 34)$ \\
\textbf{5} & $22$ & $8$ & $(62, 124, 10, 84, 2, 6, 87, 6)$ \\
\textbf{6} & $37$ & $10$ & $(258, 5, 26, 168, 10, 28, 48, 203, 18, 117)$ \\
\textbf{7} & $44$ & $12$ & $(37, 14, 98, 278, 34, 96, 225, 218, 1, 34, 235, 2)$ \\
\textbf{8} & $76$ & $14$ & $(1, 14, 61, 118, 137, 190, 3, 20, 225, 2, 4, 2, 9, 82)$ \\
\textbf{9} & $92$ & $16$ & $(658, 2, 84, 276, 11, 20, 40, 34, 54, 170, 6, 22, 154, 252, 451, 392)$ \\
\textbf{10}  & $-$  & $-$   & $-$ \\
\textbf{11}  & $-$  & $-$   & $-$ \\
\textbf{12}  & $-$  & $-$   & $-$ \\
\hline
\end{tabular}
\end{table}


\clearpage
\section{Reminder. Background and related works}

\subsection{Cayley and Schreier graphs, diameters, growth}

Here we recollect basic definitions: Cayley and Schreier graphs, diameters, God's numbers, growth, quasi-polynomials, discuss previous and related works. There are not many works on Cayley graphs of consecutive cycles per se, and we focus on more general perspectives.




 Cayley graphs give a way of treating groups as geometric objects. In what follows, $G$ denotes a (finite) group with a set of {\em generators} $S$. 
 
\begin{definition}
The Cayley graph of $G$ with respect to $S$ is a directed graph $\mathrm{Cay}(G,S)$ (or $\Gamma_{G,S}$) such that its set of vertices is precisely $G$, and its oriented edges are all pairs $(g,gs)$ with $g \in G$ and $s \in S$. 
\end{definition}

It is known that Cayley graphs are connected. 

The word metric on $G$ with respect to $S$ can be defined in two natural ways,
depending on whether one considers the directed or the undirected Cayley graph.

The directed word metric $\mathrm{dist}_S^{\rightarrow}(g,h)$ is defined by as the length of the shortest word in letters from $S$ representing the element $g^{-1}h\in G$, if such a word exists. This metric corresponds to the path metric on the directed Cayley graph
$\mathrm{Cay}(G,S)$.

In contrast, the \emph{symmetric word metric} $\mathrm{dist}_S(g,h)$ is defined as the length of the shortest word in letters from $S\cup S^{-1}$ representing $g^{-1}h\in G$. Geometrically, this is the standard path metric on the underlying undirected graph obtained from $\mathrm{Cay}(G,S)$ by forgetting edge orientations.


Unless explicitly stated otherwise, all distances and diameters in this paper are understood with respect to the symmetric word metric induced by $S\cup S^{-1}$.



\begin{figure}
    \centering
    \begin{tikzpicture}[scale=2, >=Stealth]
\foreach \i/\name in {
  0/{$e$},
  1/{(12)},
  2/{(123)},
  3/{(13)},
  4/{(132)},
  5/{(23)}
}{
  \node[fill=black, circle, inner sep=1.5pt] (v\i) at (60*\i:1) {};

  \node[font=\small] at (60*\i:1.25) {\name};
}
\foreach \i/\j in {0/1,1/2,2/3,3/4,4/5,5/0}{
  \draw[->, thick] (v\i) -- (v\j);
  \draw[->, thick] (v\j) -- (v\i);
}
\end{tikzpicture}
    \caption{The Cayley graph $\Gamma_{S_3, S}$ where $S = \{(12), (23)\}$}
    \label{fig:s3-cayley}
\end{figure}
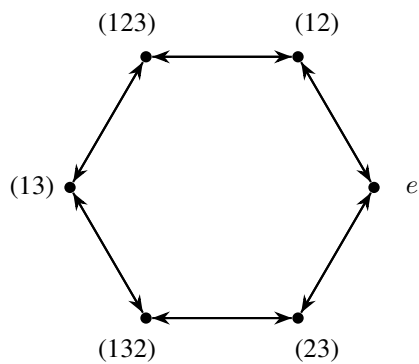

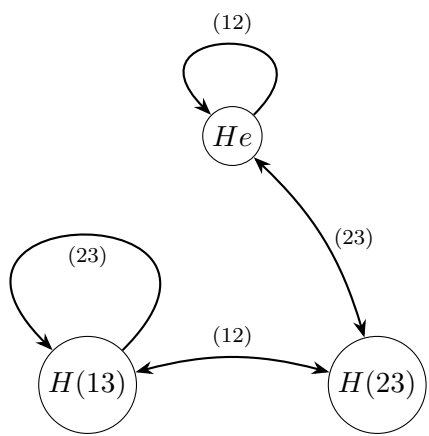
\begin{figure}[t]
\centering
\begin{tikzpicture}[scale=2.2, >=Stealth]

\node[draw, circle, inner sep=2pt] (A) at (90:1.0) {$He$};
\node[draw, circle, inner sep=2pt] (B) at (210:1.0) {$H(13)$};
\node[draw, circle, inner sep=2pt] (C) at (330:1.0) {$H(23)$};

\draw[->, thick, looseness=8] (A) to node[above, font=\scriptsize] {$(12)$} (A);

\draw[<->, thick, bend left=15]
  (B) to node[midway, above, font=\scriptsize] {$(12)$} (C);

\draw[<->, thick, bend left=15]
  (A) to node[midway, right, font=\scriptsize] {$(23)$} (C);
  
\draw[->, thick, looseness=8] (B) to node[below, font=\scriptsize] {$(23)$} (B);
\end{tikzpicture}

\caption{Schreier graph $\mathrm{Sch}(S_3,H,S)$ for $H=\langle(12)\rangle$ and
$S=\{(12),(23)\}$. Vertices are right cosets $Hg$ (not group elements).}
\label{fig:schreier_s3}
\end{figure}

\begin{definition}
Let $\Gamma=(V,E)$ be a (directed or undirected) graph equipped with the path
metric $d(v,w)$, defined as the length of a shortest path from $v$ to $w$
(if such a path exists).
The \emph{diameter} of $\Gamma$ is defined as
\[
\mathrm{diam}(\Gamma):=\sup_{v,w\in V} d(v,w).
\]
If $\Gamma$ is finite and strongly connected (or connected in the undirected
case), the supremum is a maximum.
\end{definition}

\begin{definition}
Let $H \leq G$ be a subgroup. We define the  Schreier coset graph (or simply the Schreier graph) $\mathrm{Sch}(G,H,S)$ as follows:
\begin{itemize}
  \item the vertices are the right cosets $Hg$ of $H$ in $G$, where $g \in G$;
  \item for each vertex $Hg$ and each generator $s \in S$, there is a directed edge
        \[
           Hg \xrightarrow{\;s\;} Hg s.
        \]
\end{itemize}
\end{definition}
Notice that if $H=\{e\}$, the Schreier graph $\mathrm{Sch}(G,H,S)$ coincides with the Cayley graph $\mathrm{Cay}(G,S)$.

In contrast to Cayley graphs, Schreier graphs are generally not vertex-transitive.
As a consequence, metric properties such as growth and distances may depend on
the choice of a basepoint.


\begin{definition}
Let $\Gamma=(V,E)$ be a finite graph endowed with a path metric $d$, and let
$v_0\in V$ be a fixed basepoint.
The \emph{God's number} of $\Gamma$ relative to $v_0$ is defined as
\[
\mathrm{God}(\Gamma,v_0):=\max_{v\in V} d(v_0,v).
\]
\end{definition}
In graph theory, the quantity $\mathrm{God}(\Gamma,v_0)$ defined above is
classically known as the \emph{eccentricity} of the vertex $v_0$, that is,
\[
\mathrm{ecc}(v_0):=\max_{v\in V} d(v_0,v).
\]
We adopt the term \emph{God's number} to emphasize its interpretation as the
maximal number of moves required to reach any state from a fixed initial state,
following the terminology commonly used in the theory of combinatorial puzzles.

\begin{example}
For the Rubik's Cube group $G$ with the standard generating set $S$ of face turns, the diameter of the Cayley graph $\mathrm{Cay}(G,S)$ is $20$. Since Cayley graphs are vertex-transitive, the God's number (equivalently, the eccentricity of any vertex) coincides with the diameter.
\end{example}

If $\Gamma$ is vertex-transitive (in particular, if $\Gamma$ is a Cayley graph),
then $\mathrm{God}(\Gamma,v_0)$ does not depend on the choice of $v_0$ and
coincides with the diameter $\mathrm{diam}(\Gamma)$.
For Schreier graphs this need not be the case.

For a Schreier graph $\mathrm{Sch}(G,H,S)$, the basepoint is typically chosen to
be the trivial coset $H$. In this case, the God's number measures the maximal
distance from $H$ to any coset $Hg$.

\begin{definition}
Let $\Gamma=(V,E)$ be a graph endowed with the path metric $d$, and let $v_0\in V$
be a fixed basepoint.
The growth function of $\Gamma$ relative to $v_0$ is defined as
\[
\gamma_{\Gamma,v_0}(n):=\bigl|\{\, v\in V \mid d(v_0,v)\le n \,\}\bigr|.
\]
\end{definition}

For Cayley graphs we write $\gamma_{G,S}$, omitting the basepoint $v_0$, since by vertex-transitivity the growth function is independent of the choice of basepoint.

We say that the group $G$ (or the Cayley graph $\mathrm{Cay}(G,S)$) has
\begin{itemize}
\item \emph{polynomial growth} if there exist constants $C,d>0$ such that
\[
\gamma_{G,S}(n)\le C n^d \quad \text{for all } n,
\]
\item \emph{exponential growth} if there exists $\lambda>1$ such that
\[
\gamma_{G,S}(n)\ge \lambda^n \quad \text{for all sufficiently large } n.
\]
\end{itemize}

From this point on, when working with Cayley graphs, we denote the diameter by $\mathrm{diam}(G,S)$.
For a finite group $G$, the diameter $\mathrm{diam}(G,S)$ is the minimal radius
$n$ such that the ball of radius $n$ in the Cayley graph exhausts the group, i.e.
\[
\gamma_{G,S}(n)=|G|.
\]
Thus, the diameter measures the extremal behavior of the growth function.

\subsection{Quasi-polynomial functions}

\begin{definition}
A function $f:\mathbb{N}\to\mathbb{Q}$ is called a \emph{quasi-polynomial} if there exists a positive integer $m$ and polynomials
\[
P_0, P_1, \dots, P_{m-1} \in \mathbb{Q}[n]
\]
such that
\[
f(n) = P_r(n) \quad \text{for all } n \in \mathbb{N} \text{ with } n \equiv r \pmod m.
\]
The minimal such $m$ is called the \emph{period} of $f$.
\end{definition}

\begin{example}
Let $S=\{(1,2),(2,3),\dots,(n-1,n)\}$ be the set of Coxeter generators of $S_n$.
Then the diameter of the Cayley graph $\mathrm{Cay}(S_n,S)$ is given by
\[
\mathrm{diam}(S_n,S)=\frac{n(n-1)}{2},
\]
which is a quasi-polynomial of period $1$ (i.e.\ an ordinary polynomial).
\end{example}

\begin{example}
Let $S = \{(1, 2, \ldots, n-1, n),(1, 2)\}$ be the set of LX generators of $S_n$.
In our previous paper~\cite{Cayley3Growth}, we conjectured that the diameter of the LX Cayley graph $\mathrm{Cay}(S_n, S)$ is given by
\[
\mathrm{diam}(S_n,S)=
\begin{cases}
\frac{3n^2}4 - 2n + 3, & n \equiv 0 \pmod 2,\\[4pt]
\frac{3n^2}4 - 2n + \frac 94, & n \equiv 1 \pmod 2,
\end{cases}
\]
which is a quasi-polynomial of period $2$.
\end{example}


\subsection{Ehrhart polynomials }

\begin{definition}
Let $\mathcal{L}\subset \mathbb{R}^d$ be a lattice, and let
$P\subset \mathbb{R}^d$ be a $d$-dimensional convex polytope such that all
vertices of $P$ lie in $\mathcal{L}$.
For a positive integer $t$, let $tP$ denote the $t$-fold dilation of $P$, that is,
the polytope obtained by multiplying the coordinates of each vertex of $P$,
with respect to a fixed basis of $\mathcal{L}$, by the factor $t$.
Define
\[
L(P,t):=\#\bigl(tP\cap \mathcal{L}\bigr),
\]
the number of lattice points contained in the polytope $tP$.
Then $L(P,t)$ is a polynomial in $t$ of degree $d$ with rational coefficients,
called the \emph{Ehrhart polynomial} of $P$.
\end{definition}

\begin{example}
Let
\[
P=\mathrm{conv}\{(0,0),(0,1),(1,0)\}\subset\mathbb{R}^2.
\]
Then
\[
tP=\{(x,y)\in\mathbb{R}^2:\ x\ge 0,\ y\ge 0,\ x+y\le t\}.
\]
Hence the lattice points in $tP$ are exactly the integer pairs $(i,j)$ with
$i,j\ge 0$ and $i+j\le t$, so
\[
L_P(t)=\#(tP\cap\mathbb{Z}^2)=\sum_{i=0}^{t}(t-i+1)=\frac{(t+1)(t+2)}{2}.
\]
Thus the Ehrhart polynomial of $P$ equals
\[
L_P(t)=\frac{t^2+3t+2}{2}.
\]
\end{example}

\begin{definition}[Ehrhart counting function]
Let $\mathcal{L}\subset \mathbb{R}^d$ be a full-rank lattice, and let
$P\subset \mathbb{R}^d$ be a $d$-dimensional convex polytope.
For $t\in\mathbb{Z}_{>0}$ define the lattice-point counting function
\[
L_{\mathcal{L}}(P,t):=\#\bigl(tP\cap \mathcal{L}\bigr),
\]
where $tP=\{tx\mid x\in P\}$ is the $t$-fold dilation of $P$.
\end{definition}

We say that $P$ is \emph{$\mathcal{L}$-integral} if all vertices of $P$ lie in $\mathcal{L}$, and
\emph{$\mathcal{L}$-rational} if all vertices of $P$ lie in $\mathcal{L}\otimes_{\mathbb{Z}}\mathbb{Q}$.
Then:
\begin{itemize}
\item If $P$ is $\mathcal{L}$-integral, then $L_{\mathcal{L}}(P,t)$ is a polynomial in $t$ of degree $d$
(with rational coefficients), called the \emph{Ehrhart polynomial} of $P$ (with respect to $\mathcal{L}$).
\item If $P$ is $\mathcal{L}$-rational, then $L_{\mathcal{L}}(P,t)$ is a quasi-polynomial in $t$ of degree $d$,
called the \emph{Ehrhart quasi-polynomial} of $P$ (with respect to $\mathcal{L}$).
\end{itemize}

In the standard lattice $\mathcal{L}=\mathbb{Z}^d$, if
\[
P=\{x\in\mathbb{R}^d \mid Ax\le b\}, \qquad A\in\mathbb{Q}^{k\times d},\; b\in\mathbb{Q}^k,
\]
then
\[
L_{\mathbb{Z}^d}(P,t)=\#\bigl(tP\cap\mathbb{Z}^d\bigr)
=\#\{x\in\mathbb{Z}^d \mid Ax\le tb\}.
\]
In particular, when $A,b$ are integral (equivalently, $P$ is $\mathbb{Z}^d$-integral), this becomes an Ehrhart
polynomial; otherwise it is an Ehrhart quasi-polynomial.

When $\mathcal L=\mathbb Z^d$,  we write
$L(P,t):=L_{\mathbb Z^d}(P,t).$


\begin{example}
Let
\[
P=\mathrm{conv}\{(0,0),(1,0),(0,\tfrac12)\}\subset\mathbb{R}^2.
\]
For a positive integer $t$, the dilation $tP$ is given by
\[
tP=\{(x,y)\in\mathbb{R}^2 \mid x\ge 0,\ y\ge 0,\ x+2y\le t\}.
\]
Hence the lattice points in $tP$ are exactly the integer pairs
$(x,y)\in\mathbb{Z}_{\ge 0}^2$ satisfying $x+2y\le t$, and therefore
\[
L_P(t)=\#(tP\cap\mathbb{Z}^2)=\sum_{y=0}^{\lfloor t/2\rfloor}(t-2y+1).
\]
A direct computation shows that $L_P(t)$ is a quasi-polynomial of period $2$,
given explicitly by
\[
L_P(t)=
\begin{cases}
\dfrac{t^2}{4}+t+1, & t\equiv 0 \pmod 2,\\[6pt]
\dfrac{t^2+4t+3}{4}, & t\equiv 1 \pmod 2.
\end{cases}
\]
\end{example}

We refer to \cite{stanley1997enumerative1,stanley2001enumerative2} for more information. 

\subsection{ROC curves and AUC }
	Given a ranked list of examples with binary ground-truth labels, the induced label sequence $x\in\{0,1\}^n$ defines the same monotone path $P(x)$. In this encoding, the ROC curve is a scaled version of $P(x)$, while the AUC is the corresponding normalized area. Theorem~ TODO 
    shows that, in the binary orbit model, this area also measures the Cayley-graph distance to the sorted baseline, i.e., the number of adjacent misorderings (inversions) in the ranked sequence.

Examples below  illustrate that for $x$ with $n=10$ and $k=5$, 
$d(e,x)=\text{AAC}$, area above ROC curve measured in unit squares. 
In the context of binary classification we can fix outputs $R$ of a hypothetical binary classificaion model listed in ascending order, 
theshold values $T$, and predictions of classes $\text{PRED}_j$ corresponding to each selected threshold $T_j\in T$.
In this context $x$ should be interpreted as "ground truth". True Positive Rate (TPR) and False Positive Rate (FPR) are used to build corresponding ROC curves.
	$$
	\begin{array}{rcl}
		R&=&[0.05, 0.15, 0.25, 0.35, 0.45, 0.55, 0.65, 0.75, 0.85, 0.95],\\
		T&=&P+0.01,\\
		\text{PRED}_j &=& \text{int}(P>T_j), \quad \forall j=0,\dots,n-1.
	\end{array}
	$$
	
\subsubsection{Example 1}
	
	$$
	\begin{array}{rcl}
		x&=&[0 , 0 , 0 , 0 , 0 , 1 , 1 , 1 , 1 , 1 ],\\
		\mathrm{FPR}&=&[ 1.0, 0.8, 0.6, 0.4, 0.2, 0.0, 0.0, 0.0, 0.0, 0.0, 0],\\
		\mathrm{TPR} &=& [1.0 , 1.0, 1.0,  1.0, 1.0, 1.0, 0.8, 0.6, 0.4, 0.2, 0].
	\end{array}
	$$
	
	In this example the area above the ROC curve is $\text{AAC}=0$ (see  Figure~\ref{fig:roc_curve_perfect}).
	Since $e=x$, we have $d(x,e)=0=\text{AAC}$.

	Now suppose 	$x=[0,0,0,0,1,0,1,1,1,1]$, then $\text{Inv}(x)=1$ and $d(x,e)=1$.
	Observe that exactly one unit square is above ROC curve 
    (see  Figure~\ref{fig:roc_curve_one_swap}). So $\text{AAC}=1=d(x,e)$.

\subsubsection{Example 2}
	
	$$
	\begin{array}{rcl}
		x&=&[0, 0, 0, 1, 0, 1, 1, 0, 1, 1],\\
		\mathrm{FPR}&=&[1,0.8, 0.6, 0.4, 0.4, 0.2, 0.2, 0.2, 0.0, 0.0, 0.0],\\
		\mathrm{TPR} &=& [1,1.0, 1.0, 1.0, 0.8, 0.8, 0.6, 0.4, 0.4, 0.2, 0.0].
	\end{array}
	$$
	
	\begin{figure}[!htbp]
		\centering
		\includegraphics[width=0.7\textwidth]{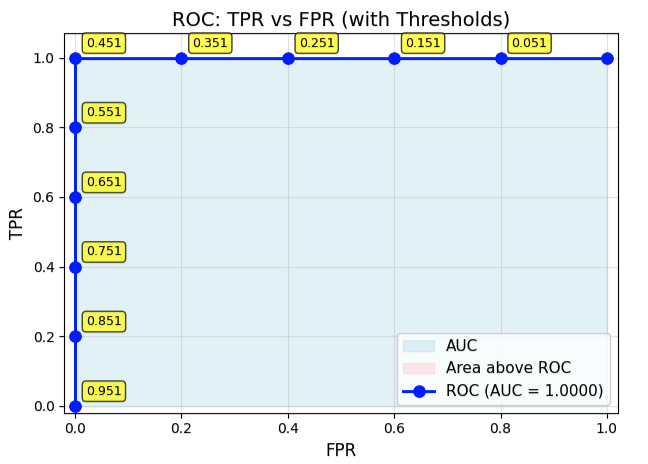}
		\caption{ROC: TPR vs.\ FPR --- perfect predictions.}
		\label{fig:roc_curve_perfect}
	\end{figure}

	Here $d(x,e)=4$.
	On the other hand, the ROC AUC is $0.84$ (see  Figure~\ref{fig:roc_curve_four_swaps}).
	The maximum possible area is $(n-k)k=25$ (a $5\times 5$ grid), so the area above the curve is
	$(1-0.84)\cdot 25=4$.
	Thus $\text{AAC}=4=d(x,e)$.

	\begin{figure}[!htbp]
		\centering
		\includegraphics[width=0.7\textwidth]{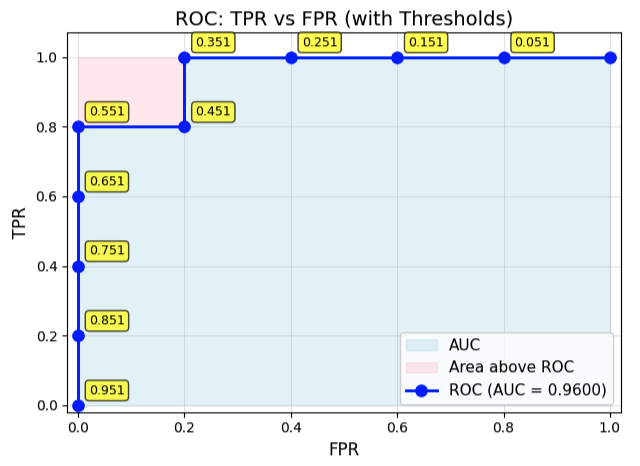}
		\caption{ROC: TPR vs.\ FPR --- one adjacent swap.}
		\label{fig:roc_curve_one_swap}
	\end{figure}

	\begin{figure}[htbp]
		\centering
		\includegraphics[width=0.7\textwidth]{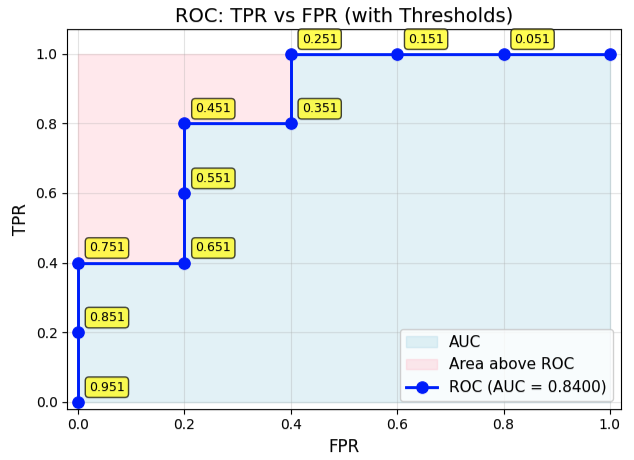}
		\caption{ROC curve showing TPR vs.\ FPR --- four adjacent swaps.}
		\label{fig:roc_curve_four_swaps}
	\end{figure}

\clearpage

\section*{Acknowledgments}
A.C.  is deeply grateful to M. Douglas,  A. Hayt, C .Simpson,  P.A. Melies,  F. Charton, Y.  Fregier, S. Nechaev, V. Rubtsov, G. Williamson,
J. Ellenberg, for stimulating discussions, interest and encouragement,   without whom the project may not exist in the present form. 



A.C. is grateful to J. Mitchel for involving  into the Kaggle Santa 2023 challenge, from which this project originated,
to M.Kontsevich, N.Nekrasov, T.Rokicki,  M.Shapiro, A.Mironov, A.Gorski, V. Fock, V. Gorbunov, V.Golyshev,  Y.Soibelman,  S.Gukov, 
T. Smirnova-Nagnibeda,  D.Osin, V. Kleptsyn, G.Olshanskii, 
 A. Sutherland,   I. Vlassopoulos,  A.Zinovyev, M. Alekseyev,  S.Klevtsov, A. Mellit, V.Dotsenko, L.Rybnikov, D.Grinberg 
for the discussions, interest and comments,
to his wife A.Chervova and daugther K.Chervova for support, understanding and help with computational experiments.

We are deeply grateful to  many colleagues who have contributed to the CayleyPy project at various stages of its development, including: D. Kamenetsky, S.Shakirov,  N.Bukhal, J.Naghiev,  K.Khoruzhii, A.Romanov, A. Naumov, A.Sychev,  A.Lenin, E.Uryvanov,  A. Abramov, M.Urakov, A.Kuchin,  B.Bulatov,  F.Faizullin,  U.Kniaziuk, D.Naumov,  S.Botman, A.Kostin,
R.Vinogradov,   N.Narynbaev, A.Korolkova, N. Rokotyan, S.Kovalev, A.Eliseev, 
A.Ogurtsov, G.Antiufeev, G.Verbii,   A.Rozanov, V.Nelin, S.Ermilov,
 A. Trepetsky, A. Dolgorukova,  N. Narynbaev, S. Nikolenko,  R. Turtayev,
K.Yakovlev, V.Shitov, E.Durymanov,  R.Magdiev, M.Krinitskiy, P.Snopov, M. Evseev ,  A.Aparnev, A.Titarenko,
 M. Litvinov, N. Vilkin-Krom, A. Bidzhiev, A. Krasnyi,  E. Geraseva, E. Koldunov, S. Diner,  E. Kudasheva,
 A. Kravchenko,  V. Zamkovoy, D. Kovalenko, O. Papulov,  D. Mamayeva, M.Kazemina, et. al.
 


The work of F. Levkovich-Maslyuk was supported by the STFC grant APP69281.

\onecolumn
\printbibliography

\end{document}